\documentclass{amsart}
\usepackage[latin1]{inputenc}
\usepackage{latexsym}
\usepackage{amsmath}
\usepackage{amssymb}
\usepackage{amsfonts}
\usepackage{mathrsfs}               
\usepackage{bbm}                      
\usepackage{bbold}                    
\usepackage{textcomp}               
\usepackage{amstext}
\usepackage{enumerate}
\usepackage{enumitem}
\usepackage{units}
\usepackage{stmaryrd}
\usepackage{multicol}
 \newtheorem{thm}{Theorem}[section]
 \newtheorem{cor}[thm]{Corollary}
 \newtheorem{lem}[thm]{Lemma}
 \newtheorem{prop}[thm]{Proposition}
 \theoremstyle{definition}
 \newtheorem{defn}[thm]{Definition}
 \theoremstyle{remark}
 \newtheorem{rem}[thm]{Remark}
 \newtheorem{ex}[thm]{Example}
 \newtheorem{hyp}[thm]{Hypothesis}
 \newtheorem{notation}[thm]{Notation}
 \numberwithin{equation}{section}
\newcommand{\CC}{\mathbb{C}}

\newcommand{\EE}{\mathbb{E}}

\newcommand{\NN}{\mathbb{N}}
\newcommand{\PP}{\mathbb{P}}
\newcommand{\QQ}{\mathbb{Q}}
\newcommand{\RR}{\mathbb{R}}

\newcommand{\ZZ}{\mathbb{Z}}
\newcommand{\BB}{\mathbb{B}}
\newcommand{\supp}{\mathrm{supp}}

\newcommand{\pr}{\mathrm{pr}}

\newcommand{\loc}{\mathrm{loc}}

\newcommand{\ren}{\mathrm{ren}}
\newcommand{\Id}{\mathrm{d}}
\newcommand{\SPn}[2]{\langle #1|#2\rangle} 
\newcommand{\SPb}[2]{\big\langle #1\big|#2\big\rangle} 
\newcommand{\SPB}[2]{\Big\langle \,#1\,\Big|\,#2\, \Big\rangle}

\newcommand{\ol}[1]{\overline{#1}} 
\newcommand{\ul}[1]{\underline{#1}} 
\newcommand{\mr}[1]{\mathring{#1}}
\newcommand{\wh}[1]{\widehat{#1}}
\newcommand{\wt}[1]{\widetilde{#1}}
\newcommand{\nf}[2]{\nicefrac{#1}{#2}}
\newcommand{\eh}{{\nf{1}{2}}}
\newcommand{\mh}{{-\nf{1}{2}}}
\newcommand{\cC}{\mathcal{C}}
\newcommand{\cD}{\mathcal{D}}\newcommand{\cP}{\mathcal{P}} 
\newcommand{\cE}{\mathcal{E}}\newcommand{\cQ}{\mathcal{Q}}
\newcommand{\cF}{\mathcal{F}}

\newcommand{\cU}{\mathcal{U}}

\newcommand{\cL}{\mathcal{L}}\newcommand{\cX}{\mathcal{X}} 
\newcommand{\cM}{\mathcal{M}}       
 
\newcommand{\sN}{\mathscr{N}}

\newcommand{\sD}{\mathscr{D}}\newcommand{\sP}{\mathscr{P}} 
\newcommand{\sE}{\mathscr{E}}
\newcommand{\sF}{\mathscr{F}}
\newcommand{\sS}{\mathscr{S}}
\newcommand{\sT}{\mathscr{T}}
\newcommand{\sU}{\mathscr{U}}
\newcommand{\sV}{\mathscr{V}}
\newcommand{\sK}{\mathscr{K}}\newcommand{\sW}{\mathscr{W}}
\newcommand{\sX}{\mathscr{X}}         
\newcommand{\sM}{\mathscr{M}}       
 
\newcommand{\fA}{\mathfrak{A}}
\newcommand{\fB}{\mathfrak{B}}

\newcommand{\fQ}{\mathfrak{Q}}
\newcommand{\fF}{\mathfrak{F}}

\newcommand{\fH}{\mathfrak{H}}

\newcommand{\V}[1]{\boldsymbol{#1}}
\newcommand{\ve}{\varepsilon}
\newcommand{\vp}{\varphi}
\newcommand{\vo}{\varpi}

\newcommand{\vr}{\varrho}

\newcommand{\vs}{\varsigma}
\newcommand{\id}{\mathbbm{1}}                  
\newcommand{\dom}{\cD}                             
\newcommand{\fdom}{\cQ}                            
\newcommand{\HP}{\mathfrak{h}}                  
\newcommand{\LO}{\mathcal{B}}                   
\newcommand{\ad}{a^\dagger}                      
\newcommand{\ee}{\mathfrak{g}}                   
\renewcommand{\Re}{\mathrm{Re}}
\renewcommand{\Im}{\mathrm{Im}}
\renewcommand{\le}{\leqslant}        
\renewcommand{\ge}{\geqslant}  
\begin{document}

\title[Ultra-violet renormalized Nelson model]{Feynman-Kac formulas for the ultra-violet 
renormalized Nelson model}

\author{Oliver Matte \and Jacob Schach M{\o}ller}

\address{Oliver Matte $\cdot$ Jacob Schach M{\o}ller. Institut for Matematik, Aarhus Universitet, Ny Munkegade 118, DK-8000 Aarhus C, Denmark.}

\email{matte@math.au.dk $\cdot$ jacob@math.au.dk}

\begin{abstract}
We derive Feynman-Kac formulas for the ultra-violet renormalized Nelson Hamiltonian with a
Kato decomposable external potential and for corresponding fiber Hamiltonians in the translation
invariant case. We simultaneously treat massive and massless bosons.
Furthermore, we present a non-perturbative construction of a renormalized Nelson 
Hamiltonian in the non-Fock representation defined as the generator of a corresponding 
Feynman-Kac semi-group. Our novel analysis of the vacuum expectation of the 
Feynman-Kac integrands shows that, if the external potential and the Pauli-principle are dropped, 
then the spectrum of the $N$-particle renormalized Nelson Hamiltonian is bounded from below by 
some negative universal constant times $\ee^4N^3$, for all values of the coupling constant $\ee$. 
A variational argument also yields an upper bound of the same form for large $\ee^2N$. 
We further verify that the semi-groups generated by the ultra-violet renormalized Nelson Hamiltonian 
and its non-Fock version are positivity improving with respect to a natural self-dual cone, if the Pauli 
principle is ignored. In another application we discuss continuity properties of elements in the range 
of the semi-group of the renormalized Nelson Hamiltonian.
\end{abstract}



\maketitle
%
%

\section{Introduction}\label{sec-intro}

\noindent
More than half a century ago, Edward Nelson studied the renormalization theory of a model
for a conserved number of non-relativistic scalar matter particles interacting with a quantized radiation 
field comprised of relativistic scalar bosons. This model is a priori given by a heuristic Hamiltonian 
equal to the sum of the Schr\"{o}dinger operator for the matter particles, the radiation field energy 
operator, and a field operator describing the interaction between the matter particles and the radiation 
field. This heuristic expression is, however, mathematically ill-defined because the physically relevant 
choice of the interaction kernel determining the field operator is not a square-integrable function of the 
boson modes. Hence, one starts out by introducing an artificial ultra-violet cut-off rendering this kernel 
function square-integrable and the Hamiltonian well-defined. The question, then, is whether the 
so-obtained ultra-violet regularized Nelson Hamiltonians converge in a suitable sense as the cut-off is 
removed, possibly after adding cut-off dependent energy shifts that would not harm physical 
interpretations. Nelson approached this mathematical problem by probabilistic methods in 
\cite{Nelson1964proc} and by operator theoretic arguments in~\cite{Nelson1964}.

In his earlier probabilistic investigation Nelson eventually considered the matrix elements of the
unitary groups generated by the regularized Hamiltonians with an explicitly given energy 
renormalization added
and proved that, for strictly positive boson masses and fixed time parameters, these matrix elements 
are convergent as the cut-off is removed, in the weak-$*$ sense as bounded measurable functions of 
the particle mass. While he knew that these limits are non-trivial, he could not yet decide whether they 
define a new unitary group or not. This was clarified in his second work cited where, after adding 
the energy renormalization and applying unitary Gross transformations, he obtained a 
sequence of Hamiltonians converging in the norm resolvent sense. As the Gross transformations 
converge strongly, this implied strong resolvent convergence of the original regularized 
operators plus the energy shifts towards a renormalized Nelson Hamiltonian. 


\subsection*{Ealier results on Nelson's renormalized operator}\label{ssechistory}

Given that the two afore-mentioned papers of Nelson date back to 
1964 the number of mathematical articles explicitly addressing properties of his renormalized model is 
not very large, whence we give a brief, essentially chronological survey in what follows.

The first results following \cite{Nelson1964proc,Nelson1964}
are the construction of asymptotic fields for massive bosons 
by H{\o}egh-Krohn \cite{HoeghKrohn1970} and of renormalized fiber Hamiltonians in the 
translation-invariant case by Cannon \cite{Cannon1971}, who adapted the procedure in 
\cite{Nelson1964}. Cannon also proved the existence of non-relativistic Wightman distributions and, 
for a sufficiently weak matter-radiation coupling, the existence of 
dressed one-particle states as well as analyticity of the corresponding energies and eigenvectors. 
Cannon's smallness assumptions on the coupling depend on the strictly positive boson mass he was
cosidering. His results were pushed forward by Fr\"{o}hlich in \cite{Froehlich1974} who gave 
non-perturbative proofs of Cannon's results, that hold for any strictly positive boson mass or,
alternatively, infra-red cut-off, irrespective of the value of the coupling constant.
In this article Fr\"{o}hlich also found a rich class of Hilbert spaces,
including examples of von Neumann's incomplete direct product spaces, on which the renormalization
procedure of \cite{Nelson1964} can be implemented. Fr\"{o}hlich employed his results in 
\cite{Froehlich1973} to discuss the infra-red problem and aspects of scattering theory for a class of 
models containing the Nelson model. In particular, for vanishing boson mass and without any cut-offs, 
he constructed coherent infra-red representation spaces which are attached to total momenta of the 
matter-radiation system and contain dressed one-particle states that are ground states of, roughly 
speaking, certain non-Fock versions of the renormalized fiber Hamiltonians. He also proved the 
absence of dressed one-particle states in the original Fock space for vanishing boson mass, a 
phenomenon known as infra-particle situation. 

After a gap of more than twenty-five years in the mathematical literature on the renormalized 
Nelson Hamiltonian, its spectral and scattering theory in a confining potential and for massive bosons 
has been worked out by Ammari \cite{Ammari2000}, who proved a HVZ theorem, positive commutator 
estimates, the existence of asymptotic fields, propagation estimates, and asymptotic completeness.
Later on, Hirokawa et al. \cite{HHS2005} considered a system of two particles with charges of equal
sign, one of them static, interacting via a linearly coupled massless boson field. After applying a Gross
transformation to the corresponding ultra-violet regularized Hamiltonian they found a Hamiltonian
for one particle coupled to the radiation field with an additional attractive potential playing the
role of an effective interaction between the two particles. The Gross transformation is actually
infra-red singular for zero boson mass. Thus, an artificial infra-red cut-off is included in its definition. 
By improving some of Nelson's \cite{Nelson1964} relative bounds so as to cover massless bosons, 
Hirokawa et al. removed both the ultra-violet and infra-red cut-offs in their Gross
transformed Hamiltonian and obtained, for sufficiently small coupling constants, a 
self-adjoint operator called the renormalized Nelson Hamiltonian in a non-Fock representation.  
The latter turned out to have a ground state eigenvector. 
Employing this result, Hainzl et al. \cite{HainzlHirokawaSpohn2005} established a formula for the first 
radiative correction to the binding energy of this interacting two-particle system. About ten years ago,
a paper of Ginibre et al. \cite{GinibreNironiVelo2006} appeared concerning a certain partially 
classical limit of the Nelson model for any non-negative boson mass and without cut-offs.


The investigation of the infra-particle nature of the massless Nelson model without cut-offs has been 
revisited more recently by Bachmann et al. \cite{BDP2011} employing Pizzo's iterative analytic 
perturbation theory. While their results hold for sufficiently small coupling only, they provide more 
detailed control on the mass-shell and the dressed one-particle states than earlier results. 

Nelson's operator theoretic renormalization procedure \cite{Nelson1964} has also been implemented
on static Lorentzian manifolds and for position-dependent boson masses by G\'{e}rard et al. 
\cite{GHPS2012}. Quite recently, Nelson's earlier approach of \cite{Nelson1964proc} has been revived 
as well by Gubinelli et al. \cite{GHL2014}, who succeeded by {\em probabilistic} arguments to verify 
strong convergence of the semi-group as an ultra-violet
regularization is removed in a Nelson Hamiltonian for massive bosons. In the same paper, Gubinelli et 
al. also computed effective potentials in the weak coupling limit of the renormalized theory. 
Hiroshima treated infra-red cut-off fiber Hamiltonians along the same lines as well 
\cite{Hiroshima2015a}.

In a recent preprint \cite{AmmariFalconi2016}, Ammari and Falconi proved a Bohr correspondence 
principle showing that, in a classical limit, the time evolution of quantum states generated by a 
renormalized Nelson Hamiltonian for massive bosons converges to the push-forward of a Wigner 
measure under the dynamics of a nonlinear Schr\"{o}dinger-Klein-Gordon system. They also explored 
the idea to carry through a renormalization procedure on the classical level and to Wick quantize the 
result afterwards, which leads to the same renormalized operator in the Nelson model.

Finally, Bley and Thomas \cite{BleyThesis,Bley2016,BleyThomas2015} developed
a general new method to bound a class of exponential moments that often arise when functional
integration techniques are applied in non-relativistic quantum (field) theory. 
Applied to the renormalized Nelson Hamiltonian, with non-negative boson mass and
vanishing exterior potential, this method yields a lower bound on its spectrum of the form 
$-c\ee^{4}N^{3}(1\vee\ln^2([1\vee\ee^2]N))$, \cite{Bley2016}, where $N$ is the number of 
matter particles and the modulus of the coupling constant $\ee$ is either assumed to be sufficiently 
large or sufficiently small. Here we should add that, as we shall do in the present work,
Bley fixes the explicit energy counter terms in the renormalization procedure, which are proportional
to $\ee^2$, in such a way that no contribution of order $\ee^2$ shows up in his lower bound for
the renormalized operator. This differs from the convention in \cite{Nelson1964}.
Using his bound, Bley also provided a non-binding condition in the 
massless Nelson model for two matter particles, whose effective attraction mediated by the radiation 
field is compensated for by a repulsive Coulomb interaction \cite{Bley2016b}.

We restricted the above summary to articles explicitly containing theorems on the 
renormalized Nelson model, as an account on the numerous mathematical papers 
devoted to ultra-violet regularized Nelson Hamiltonians would be far too 
space-consuming. For a general introduction to the model and more references the reader can 
consult, e.g., the textbook \cite{LHB2011}. A renormalization of a translation-invariant Nelson 
type model for a {\em relativistic} scalar matter particle interacting with a massive boson field 
\cite{Gross1973} actually leads to a theory with a flat mass shell \cite{DeckertPizzo2014}.


\subsection*{Description of results}

\noindent
The first main result of the present article is a novel Feynman-Kac formula for the renormalized Nelson
Hamiltonian for $N$ matter particles in a Kato decomposable external potential $V$ and for 
non-negative boson masses. Denoting the latter operator by $H_{N,\infty}^V$ it reads
\begin{equation}
\begin{split}\label{FKintro}
&(e^{-tH_{N,\infty}^V}\Psi)(\ul{\V{x}})=\EE\big[W_{\infty,t}^V(\ul{\V{x}})^*
\Psi(\ul{\V{x}}+\ul{\V{b}}_t)\big],\quad\text{a.e.,}
\\
&W_{\infty,t}^V(\ul{\V{x}})^*=e^{u_{\infty,t}^N(\ul{\V{x}})-\int_0^tV(\ul{\V{x}}+\ul{\V{b}}_s)\Id s}
F_{0,\nf{t}{2}}(-U_{\infty,t}^{N,-}(\ul{\V{x}}))F_{0,\nf{t}{2}}(-U_{\infty,t}^{N,+}(\ul{\V{x}}))^*,
\end{split}
\end{equation}
for every $t>0$, where, in standard notation recalled later on,
\begin{align*}
F_{0,s}(f)&:=\sum_{n=0}^\infty\frac{\ad(f)^n}{n!}e^{-s\Id\Gamma(\omega)},\quad s>0.
\end{align*}
In \eqref{FKintro}, $\Psi$ is a Fock space-valued square-integrable function of $\ul{\V{x}}\in\RR^{3N}$
and $\ul{\V{b}}$ is a $3N$-dimensional Brownian motion. 
The real-valued stochastic process $u_{\infty,t}^N(\ul{\V{x}})$ is called the complex action
following Feynman \cite{Feynman1949} and the $U_{\infty,t}^{N,\pm}(\ul{\V{x}})$ are
continuous adapted stochastic processes with values in the one-boson Hilbert space 
$\HP:= L^2(\RR^3)$. The series defining $F_{0,s}(f)$ converges in the Fock space operator norm 
and defines an analytic function of $f\in\HP$. 

For ultra-violet regularized Nelson Hamiltonians, the special form
\eqref{FKintro} of the Feynman-Kac formula appeared in \cite{GMM2016}. We shall re-prove it to
make this article essentially self-contained and to demonstrate that the Nelson model admits a 
simpler proof than the models in \cite{GMM2016} which in general involve minimally coupled fields as 
well. In fact, our derivation
of \eqref{FKintro} consists in implementing a new renormalization procedure on the level of 
semi-groups in the spirit of \cite{GHL2014,Nelson1964proc} and re-defining $H_{N,\infty}^V$ as the 
generator of the semi-group given by the right hand sides in \eqref{FKintro}. 
We shall actually observe {\em norm} convergence of semi-groups with hardly any technical restriction 
on the details of the ultra-violet regularization; see also \cite{Ammari2000}
as well as Thm.~\ref{thmNelsonHam} and the remarks following it. 
Our definition of $H_{N,\infty}^V$ is manifestly independent of the choice of any cut-off function,
purely and simply as this is the case for the right hand sides in \eqref{FKintro}. 
With only little extra work
we shall also derive new Feynman-Kac formulas for the renormalized Nelson Hamiltonian in the
non-Fock representation and for fiber Hamiltonians in the translation-invariant
renormalized Nelson model. In particular, we shall provide the first non-perturbative construction
of the renormalized Nelson Hamiltonian in the non-Fock representation.

The crucial point about the Feynman-Kac representation \eqref{FKintro} is that it provides a 
fairly simple and tractable formula for a well-defined {\em Fock space operator-valued process}
$W_{\infty}^V(\ul{\V{x}})$ in the Feynman-Kac
integrand and can be applied to {\em every} element $\Psi$ of the Hilbert space for the whole system.
While Nelson and Gubinelli et al. have Feynman-Kac type representations of expectation values
with respect to vectors in certain total subsets of the Hilbert space (involving suitable finite particle 
states \cite{Nelson1964proc} or coherent states \cite{GHL2014} in Fock space), the 
merit of writing the Feynman-Kac formula in the form \eqref{FKintro} is that it allows to first find explicit 
expressions for $U_{\infty,t}^{N,\pm}(\ul{\V{x}})$ containing well-defined $\HP$-valued stochastic
integrals and then to derive {\em operator-norm bounds} on $W_{\infty}^V(\ul{\V{x}})$ with finite 
moments of any order. Furthermore, our formulas permit to verify a Markov property of the 
Feynman-Kac integrand. In particular, we can work out basic features of a semi-group theory in
Fock space-valued $L^p$-spaces in the spirit of \cite{Carmona1979,Simon1982}. Along the way we 
further present a new method to bound the exponential moments
of the complex action $u_{\infty,t}^N(\ul{\V{x}})$ which eventually leads to the improved lower bound
\begin{align}\label{GSEHintro}
\inf\sigma(H_{N,\infty}^0)\ge-c\ee^4N^3,
\end{align}
valid for {\em all} $\ee$ and $N$ with a universal constant $c>0$;
see the introduction to Subsect.~\ref{ssecuge} for more remarks on this new method and a 
discussion of earlier results \cite{BleyThesis,Bley2016,BleyThomas2015,GHL2014,Nelson1964proc}. 
(Here we ignore that the matter particles are supposed to be fermions, i.e., the Pauli principle
is neither taken into account here nor in 
\cite{BleyThesis,Bley2016,BleyThomas2015,GHL2014,Nelson1964proc}.) 
We shall employ a novel bound on an ultra-violet part of $u_{\infty,t}^N(\ul{\V{x}})$ together with
a more standard trial function argument to derive the upper bound in
\begin{align}\nonumber
(&16\pi)^2\ee^4N^3-c'\ee^2N^2\le\inf\sigma(H_{N,\infty}^0)
\\\label{asympintro}
&\le8\pi^4E_{\mathrm{P}}\ee^4N^3+c''(1+\mu+\ln(\ee^2N))\ee^2N^2,
\quad\text{provided that $\ee^2N\ge c$.}
\end{align}
Here $\mu\ge0$ is the boson mass, 
$E_{\mathrm{P}}<0$ is the Pekar energy, and $c,c',c''>0$ are universal
constants. (With $\ee_{\mathrm{N}}$ denoting the coupling constant in Nelson's articles
\cite{Nelson1964proc,Nelson1964}, we have the relation 
$2^\eh(2\pi)^{\nf{3}{2}}\ee=\ee_{\mathrm{N}}$.)
The leading behavior $\propto-\ee^4N^3$ in \eqref{GSEHintro} and \eqref{asympintro} is familiar from 
the closely related Fr\"{o}hlich polaron model \cite{BleyThesis,Bley2016,LiebThomas1997}, which can 
be renormalized as in \cite{Nelson1964} even without introducing energy counter terms. 
(If a sufficiently strong electrostatic Coulomb repulsion between the matter particles is taken into 
account, then one actually observes thermodynamic stability, i.e., a behavior of the minimal energy 
proportional to $-N$ in the Fr\"{o}hlich polaron model without restriction to symmetry subspaces
\cite{FLST2011}. For sufficiently weak electrostatic repulsion, the minimal energy of
{\em fermionic} multi-polaron systems behaves like $-N^{\nf{7}{3}}$, \cite{GriesemerMoeller2010}.)
The work on the polaron model
\cite{LiebThomas1997} suggests that $8\pi^4E_{\mathrm{P}}$ should in fact be the correct 
leading coefficient in \eqref{asympintro}. Numerics shows that $E_{\mathrm{P}}=-0.10851 \ldots$,
\cite{GerlachLoewen1991}, whence the leading coefficient in the lower bound in \eqref{asympintro}
is presumably too large by the factor $32/\pi^2|E_{\mathrm{P}}|<30$. Getting rid of this artifact
is, however, beyond the scope of this article. 

Finally, we present two applications of the new formula \eqref{FKintro}. First, we shall fill a gap left
open in the earlier literature by proving that the semi-groups generated by the renormalized Nelson
Hamiltonian and its non-Fock version are positivity improving at positive times with respect to a 
natural convex cone. In the non-Fock case this result was explicitly mentioned as an open problem
in \cite[\textsection10]{HHS2005} and it entails uniqueness and strict positivity of the ground state 
eigenvector found there.
As already observed in \cite{Matte2016} the ergodicity of the semi-groups follows easily from the 
structure of the integrand in \eqref{FKintro} and standard tools associated with Perron-Frobenius type 
arguments in quantum field theory; see, e.g., \cite{Faris1972,Simon1974}. In the second application 
we employ some results of \cite{Matte2016} on ultra-violet regularized operators to discuss the 
continuous dependence of the right hand side in the first line of \eqref{FKintro} 
on $\ul{\V{x}}$, $\ee$, and $V$.


\subsection*{Organization and general notation}

\noindent
{\em The remaining part of this article is structured as follows:}

In Sect.~\ref{secmodel} we introduce some basic notation and give a precise definition of Nelson's
model. In Sect.~\ref{secBP} we shall analyze certain $\ul{\V{x}}$-independent one-particle versions of
$U_{\infty,t}^{N,\pm}(\ul{\V{x}})$ and eventually define the latter two processes. Sect.~\ref{secphase}
is devoted to the complex action $u_{\infty,t}^N(\ul{\V{x}})$. In Sect.~\ref{secFK} we work out the 
semi-group properties between Fock space-valued $L^p$-spaces including the norm convergence
of semi-groups, as the ultra-violet cut-off is removed. At the end of Sect.~\ref{secFK} we
establish the above Feynman-Kac formula and (re-)define the renormalized Nelson Hamiltonian;
see Thm.~\ref{thmconvTkappa} and Def.~\ref{defHinfty}. (Our version of Nelson's theorem is also
anticipated in Thm.~\ref{thmNelsonHam}.)
The lower bound \eqref{GSEHintro} is obtained in Cor.~\ref{corlbspecH}.
The Feynman-Kac formulas in the non-Fock representation and for the fiber 
Hamiltonians are derived in Sect.~\ref{secnonFock} and Sect.~\ref{secfiber}, respectively.
The positivity improving and continuity properties alluded to above as well as the bounds in 
\eqref{asympintro} are proved in Sect.~\ref{secappl}. The main text is followed by three appendices
presenting well-known material on the Kolmogorov test lemma, exponential moment bounds for sums
of pair potentials (see also \cite{Bley2016}), and a general formula for the infimum of a spectrum.

\subsubsection*{Some general notation.} 

\noindent
The characteristic function of a set $A$ is denoted by $1_A$ and
we abbreviate $a\wedge b:=\min\{a,b\}$ and $a\vee b:=\max\{a,b\}$, for $a,b\in\RR$.

The Borel $\sigma$-algebra of a topological space
$\sT$ is denoted by $\fB(\sT)$. The Lebesgue-Borel measure on $\RR^n$ is denoted
by $\lambda^n$ and, as usual, we shall write $\Id t:=\Id\lambda^1(t)$,
$\Id\V{x}:=\Id\lambda^3(\V{x})$, etc., if a symbol $t$, $\V{x}$, etc., for
the integration variable is specified. 

The set of bounded operators on a Banach space $\sX$ is denoted by $\LO(\sX)$.
The symbols $\dom(T)$ and $\fdom(T)$ stand for the domain and form
domain, respectively, of a suitable linear operator $T$. 
The spectrum of a self-adjoint operator $T$ in a Hilbert space is denoted by $\sigma(T)$.

The symbols $c_{a,b,\ldots},c_{a,b,\ldots}',\ldots\,$ denote non-negative constants that depend
solely on the quantities displayed in their subscripts (if any). Their values might change from one
estimate to another.


\section{Definition of the Nelson model}\label{secmodel}

\noindent
As mentioned earlier, Nelson's model describes a system of a fixed number of non-relativistic 
matter particles interacting with a quantized radiation field comprised of bosons.
The one-boson Hilbert space, i.e., the state space for a single boson, is
\begin{align*}
\HP&:=L^2(\RR^3)=L^2(\RR^3,\lambda^3).
\end{align*}
The bosons are described in momentum space, whence we use the letter $\V{k}\in\RR^3$ 
to denote the variables of elements of $\HP$. The state space of the full radiation field is the
bosonic Fock space modeled over $\HP$ given by the orthogonal direct sum
\begin{align}\label{defFockspace}
\sF&:=\CC\oplus\bigoplus_{n=1}^\infty L^2_{\mathrm{sym}}(\RR^{3n},\lambda^{3n}).
\end{align}
Here $L^2_{\mathrm{sym}}(\RR^{3n},\lambda^{3n})$ is the closed subspace of all
$\psi_n\in L^2(\RR^{3n},\lambda^{3n})$ satisfying
$\psi_n(\V{k}_{\pi(1)},\ldots,\V{k}_{\pi(n)})=\psi_n(\V{k}_1,\ldots,\V{k}_n)$, $\lambda^{3n}$-a.e.
for all permutations $\pi$ of $\{1,\ldots,n\}$, where
$\V{k}_1,\ldots,\V{k}_n\in\RR^3$. As usual we write $\vp(f)$ for the self-adjoint field
operator in $\sF$ corresponding to some $f\in\HP$. If $\vo$ is a multiplication operator 
in $\HP$ with a real-valued function, then its self-adjoint (differential) second quantization acting 
in $\sF$ is denoted by $\Id\Gamma(\vo)$. 
In Subsect.~\ref{ssecFock} we shall recall some facts on the Weyl representation on $\sF$ 
and in particular we shall recall the precise meaning of the symbols $\vp(f)$ and
$\Id\Gamma(\vo)$.

The Hilbert space for the interacting matter-radiation system is now given by 
the vector-valued $L^2$-space $L^2(\RR^3,\sF)$. Before we
define Nelson's Hamiltonian acting in it, we introduce some assumptions and some notation.

\begin{hyp}\label{hypNelson}
Throughout the whole article we shall work with the following standing hypotheses:
\begin{enumerate}[leftmargin=*]
\item[{\rm(1)}] The boson dispersion relation and momentum are, respectively, given by
\begin{align*}
\omega(\V{k}):=(\V{k}^2+\mu^2)^\eh,\quad\V{m}(\V{k}):=\V{k},\quad\V{k}\in\RR^3.
\end{align*}
Here the boson mass is non-negative and possibly zero, $\mu\ge0$.
\item[{\rm(2)}]
The cut-off function $\chi:\RR^3\to[0,1]$ is measurable, even, i.e.,
$\chi(-\V{k})=\chi(\V{k})$, for all $\V{k}\in\RR^3$, and continuously differentiable
on the open unit ball about the origin in $\RR^3$. Furthermore, 
\begin{align*}
\chi(0)=1,\quad\sup_{|\V{k}|<1}|\nabla\chi(\V{k})|\le1,\quad\chi\in\HP. 
\end{align*}
\item[{\rm(3)}] The measurable function $\eta:\RR^3\to[0,1]$ is even.
\item[{\rm(4)}] The coupling constant $\ee$ may be any real number.
\item[{\rm(5)}] The number of matter particles and the dimension of the position space for the matter 
particles are, respectively, given by
\begin{align*}
N\in\NN,\qquad\nu:=3N.
\end{align*}
\item[{\rm(6)}] The external potential
$V:\RR^{\nu}\to\RR$ is Kato decomposable, i.e., it can be written as $V=V_+-V_-$ with 
non-negative functions $V_+$ and $V_-$ such that $V_-$ is in the Kato class $K_\nu$ and
$V_+$ is in the local Kato class $K_\nu^{\loc}$.
\end{enumerate}
\end{hyp}

We refer to \cite{AizenmanSimon1982,Simon1982} for the definition of and information on the
classes $K_\nu$ and $K_\nu^{\loc}$. The function $\eta$ is introduced in Hyp.~\ref{hypNelson}(3) 
to cover infra-red regularized versions of Nelson's model as well; in Nelson's original model we have 
$\eta=1$. We also absorbed some common normalization constants in $\ee$; recall the remarks
following \eqref{asympintro}.

\begin{notation} Throughout the article we use of the following abbreviations: 
\begin{enumerate}[leftmargin=*]
\item[{\rm(1)}] The dispersion relation in polar coordinates is denoted by
\begin{align*}
\omega_\rho&:=(\rho^2+\mu^2)^\eh,\quad\rho\ge0.
\end{align*}
\item[{\rm(2)}] We set
\begin{align*}
\chi_\kappa(\V{k}):=\chi(\V{k}/\kappa),\;\kappa\in\NN,\quad\chi_\infty(\V{k}):=1,\quad\V{k}\in\RR^3.
\end{align*}
With this we write, for all $\kappa\in\NN\cup\{\infty\}$ and $\Lambda\ge0$,
\begin{align}\label{deffbeta}
f_\kappa&:=\ee\eta\omega^\mh\chi_\kappa,\qquad\;\:
\beta_\kappa:=(\omega+\V{m}^2/2)^{-1}f_\kappa,
\\\label{deffsigmakappa}
f_{\Lambda,\kappa}&:=1_{\{|\V{m}|\ge\Lambda\}}f_\kappa,\quad
\beta_{\Lambda,\kappa}:=1_{\{|\V{m}|\ge\Lambda\}}\beta_\kappa.
\end{align}
Furthermore,
\begin{align}
E_\kappa^{\ren}&:=\int_{\RR^3}{f_{\kappa}}{\beta_{\kappa}}\Id\lambda^3,\quad\kappa\in\NN.
\end{align}
\item[{\rm(3)}] 
The coordinates in $\RR^\nu$ are always split into $N$ groups of three and denoted 
$\ul{\V{x}}=(\V{x}_1,\ldots,\V{x}_N)\in\RR^\nu$ with $\V{x}_1,\ldots,\V{x}_N\in\RR^3$. 
We abbreviate
\begin{align}\label{deffNkappa}
f_{\kappa}^N(\ul{\V{x}})&:=\sum_{\ell=1}^Ne^{-i\V{m}\cdot\V{x}_\ell}f_\kappa,
\quad
\beta^N_{\Lambda,\kappa}(\ul{\V{x}}):=\sum_{\ell=1}^Ne^{-i\V{m}\cdot{\V{x}_\ell}}
\beta_{\Lambda,\kappa},
\end{align}
for all $\ul{\V{x}}\in\RR^\nu$ and $\kappa\in\NN\cup\{\infty\}$, and we set 
$\beta^N_{\kappa}(\ul{\V{x}}):=\beta^N_{0,\kappa}(\ul{\V{x}})$.
\item[{\rm(4)}] Some one-boson processes will a priori be defined in the auxiliary Hilbert space
\begin{align}\label{deffmh}
\mathfrak{d}_{\mh}&:=L^2\big(\RR^3,(1+\omega)^{-1}\lambda^3\big).
\end{align}
In the discussion of our Feynman-Kac integrands we shall often work with the auxiliary Hilbert space
\begin{align}\label{deffrk}
\mathfrak{k}&:=L^2\big(\RR^3,(1+\omega^{-1})\lambda^3\big),
\end{align}
and the time-dependent norms
\begin{align}\label{defnormt}
\|f\|_t:=\big(\|f\|_\HP^2+\|(t\omega)^\mh f\|_\HP^2\big)^\eh,\quad f\in\mathfrak{k},\,t>0.
\end{align}
\end{enumerate}
\end{notation}

We never refer to the dependence on $\mu$, $\eta$, or $\ee$ in our notation as their specific
choices do not affect the validity of any argument used in this paper. The whole
renormalization procedure carried through below is necessary only because $f_\infty\notin\HP$, 
for non-zero $\ee$ and $\eta=1$ near infinity, due to its slow decay. Furthermore, it is important
to notice that $\beta_{\Lambda,\infty}\in\HP$, for all $\Lambda>0$, 
while $\beta_\kappa\notin\HP$, $\kappa\in\NN\cup\{\infty\}$, for non-zero $\ee$, $\eta=1$ near zero, 
and $\mu=0$, because $\beta_\kappa$ is too singular at $0$ in this case. 

Next, we introduce the Nelson Hamiltonian. In its construction and sometimes later on we shall employ 
the $\sF$-valued Fourier transform on $\RR^\nu$ defined by the Bochner-Lebesgue integrals
\begin{align*}
(\cF\Psi)(\ul{\V{\xi}})&:=\hat{\Psi}(\ul{\V{\xi}}):=\frac{1}{(2\pi)^{\nf{\nu}{2}}}
\int_{\RR^{\nu}}e^{-i\ul{\V{\xi}}\cdot\ul{\V{x}}}\Psi(\ul{\V{x}})\Id\ul{\V{x}},\quad\ul{\V{\xi}}\in\RR^{\nu},
\end{align*} 
for $\Psi\in L^1(\RR^\nu,\sF)\cap L^2(\RR^\nu,\sF)$, and by
isometric extension to a unitary map $\cF$ on $L^2(\RR^\nu,\sF)$.
In complete analogy to the scalar case we define the first order Sobolev space
$H^1(\RR^\nu,\sF):=\{\Psi\in L^2(\RR^\nu,\sF):|\ul{\V{\xi}}|\hat{\Psi}\in L^2(\RR^\nu,\sF)\}$ and the
generalized gradient $\nabla\Psi:=i\cF^*\ul{\V{\xi}}\hat{\Psi}$, for all $\Psi\in H^1(\RR^\nu,\sF)$.

Since $V_-$ is infinitesimally form bounded with respect to the Laplacian \cite{AizenmanSimon1982}, 
there is a unique self-adjoint operator, denoted $H_{N,0}^V$, 
representing the closed, semi-bounded quadratic form given by 
$$
\dom(\mathfrak{q}_{N,0}^V):=H^1(\RR^\nu,\sF)\cap
\fdom(V_+\id_{\sF})\cap L^2(\RR^\nu,\fdom(\Id\Gamma(\omega)))
$$ 
and
$$
\mathfrak{q}_{N,0}^V[\Psi]:=
\frac{1}{2}\|\nabla\Psi\|^2+\int_{\RR^\nu}V(\ul{\V{x}})\|\Psi(\ul{\V{x}})\|^2\Id\ul{\V{x}}+
\int_{\RR^\nu}\|\Id\Gamma(\omega)^\eh\Psi(\ul{\V{x}})\|^2\Id\ul{\V{x}},
$$
for all $\Psi\in\dom(\mathfrak{q}_{N,0}^V)$.

\begin{defn}\label{defNelsonHamreg}
The {\em ultra-violet regularized Nelson Hamiltonians} are defined by
\begin{align}\label{defNelsonHam}
H_{N,\kappa}^V:=H_{N,0}^V+\int_{\RR^\nu}^\oplus\vp(f_{\kappa}^N(\ul{\V{x}}))\Id\ul{\V{x}},
\quad\kappa\in\NN.
\end{align}
The quadratic form associated with $H_{N,\kappa}^V$
is denoted by $\mathfrak{q}_{N,\kappa}^V$, i.e.,
\begin{align*}
\mathfrak{q}_{N,\kappa}^V[\Psi]&:=\mathfrak{q}_{N,0}^V[\Psi]+\int_{\RR^\nu}\SPn{\Psi(\ul{\V{x}})}{
\vp(f_{\kappa}^N(\ul{\V{x}}))\Psi(\ul{\V{x}})}\Id\ul{\V{x}},\quad\Psi\in\dom(\mathfrak{q}_{N,0}^V),\,
\kappa\in\NN.
\end{align*}
\end{defn}

In fact, each $H_{N,\kappa}^V$, $\kappa\in\NN$, is well-defined, self-adjoint, and semi-bounded on 
$\dom(H_{N,0}^V)$ since the direct integral in \eqref{defNelsonHam} is infinitesimally 
$H_{N,0}^V$-bounded. The latter fact in turn is a consequence of the 
well-known relative bounds
\begin{align*}
\|\vp(e^{-i\V{m}\cdot\V{x}}f_\kappa)\psi\|&\le2^\eh\|(1\vee\omega^\mh)f_\kappa\|
\|(1+\Id\Gamma(\omega))^\eh\psi\|,\quad\psi\in\fdom(\Id\Gamma(\omega)),
\end{align*}
valid for all $\V{x}\in\RR^3$ and $\kappa\in\NN$. 

The infima of the spectra of $H^V_{N,\kappa}$ diverge to $-\infty$ as $\kappa$ goes to inifinity.
Adding suitable counter-terms we can, however, achieve the following:

\begin{thm}\label{thmNelsonHam}
The sequence $\{H_{N,\kappa}^V+NE_{\kappa}^{\ren}\}_{\kappa\in\NN}$ converges in the norm 
resolvent sense to a self-adjoint operator $H_{N,\infty}^V$, which is bounded from below.
\end{thm}

Nelson \cite{Nelson1964} actually proved Thm.~\ref{thmNelsonHam} for massive bosons, with the 
norm resolvent convergence replaced by strong resolvent convergence, and in the case where 
$\chi$ is the characteristic function of the open unit ball. For strictly positive boson masses,
Ammari \cite[Thm.~3.8 \& Prop.~3.9]{Ammari2000} observed that the convergence actually takes 
place in the norm resolvent sense. He also verified that the construction of the 
{\em (ultra-violet renormalized) Nelson Hamiltonian} $H_{N,\infty}^V$ is, up to finite energy shifts,
independent of the particular choice of $\chi$ within the class of smooth cut-offs he was considering. 
Massless renormalized Nelson Hamiltonians appeared in \cite{Froehlich1974,HHS2005}. While the 
Gross transformed operators considered there converge in norm resolvent sense, the arguments of 
\cite{Froehlich1974,HHS2005} imply strong resolvent convergence of the original operators.

Our Feynman-Kac formulas will enable us to give an independent proof of Thm.~\ref{thmNelsonHam} 
at the end of Subsect.~\ref{ssecFKren}. In Def.~\ref{defHinfty} of that subsection we shall re-introduce 
$H_{N,\infty}^V$ as the generator of an explicitly given Feynman-Kac semi-group (recall
\eqref{FKintro}) that does not depend on any cut-off function.


\section{Basic processes}\label{secBP}

\noindent
In this section we study two $\HP$-valued, ``basic'' stochastic processes entering into our 
Feynman-Kac formulas. 
More precisely, we shall introduce approximating sequences for these processes indexed 
by $\kappa$ and analyze their limiting behavior as $\kappa$ goes to infinity. 
The two processes and their approximations are defined in Subsect.~\ref{ssecBP} and studied
in more detail separately in Subsects.~\ref{ssecUminus} and~\ref{ssecUplus}. As all processes
appearing in these three subsections correspond to one matter particle only, we shall extend our
definitions and results to the case of $N$ matter particles in Subsect.~\ref{ssecUN}. In 
Subsect.~\ref{sseckey} we shall finally prove a technical key lemma on a process for $N$ matter
particles that will permit us to derive an exponential moment bound on the complex action in
Sect.~\ref{secphase}.

First we will, however, introduce some notation for probabilistic objects used throughout
this article.


\subsection{Probabilistic preliminaries and notation}\label{ssecprobprel}

\begin{notation}{\em Stochastic bases and Brownian motion.}
Throughout the paper we fix a stochastic basis $\BB:=(\Omega,\fF,(\fF_t)_{t\ge0},\PP)$ 
satisfying the usual assumptions, i.e., the probability space $(\Omega,\fF,\PP)$ is complete,
the filtration $(\fF_t)_{t\ge0}$ is right continuous, and $\fF_0$ contains all $\PP$-zero sets. 
We use the letter $\gamma$ to denote the elements of $\Omega$.

In Subsects.~\ref{ssecUminus} and~\ref{ssecUplus}, and also later on in Sect.~\ref{secfiber} and
Subsect.~\ref{ssecerg}, the symbol $\V{B}=(\V{B}_t)_{t\ge0}$ denotes a three-dimensional 
Brownian motion on $\BB$ with covariance matrix $\id_{\RR^3}$. 

In Subsect.~\ref{ssecUN} and the
later Sects.~\ref{secphase}, \ref{secFK}, and~\ref{secappl}, the symbol $\ul{\V{b}}$ denotes
a $\nu$-dimensional $\BB$-Brownian motion with covariance matrix $\id_{\RR^\nu}$, that we think of 
as being split into $N$ independent three-dimensional Brownian motions $\V{b}_\ell$ describing the 
paths of the $N$ matter particles, i.e., $\ul{\V{b}}=(\V{b}_1,\ldots,\V{b}_N)$.

For all $t\ge0$, we further set 
\begin{align}\label{deftb}
{}^t\!\V{B}_s:=\V{B}_{t+s}-\V{B}_t,\quad{}^t\ul{\V{b}}_s:=\ul{\V{b}}_{t+s}-\ul{\V{b}}_t,\quad s\ge0, 
\end{align}
so that ${}^t\!\V{B}$ and ${}^t\ul{\V{b}}$ are Brownian motions with respect to the time-shifted 
stochastic basis $\BB_t:=(\Omega,\fF,(\fF_{t+s})_{s\ge0},\PP)$, which again satisfies the usual 
assumptions.
\end{notation}

\begin{ex}\label{exWiener}
Let $d\in\NN$, $\Omega_{\mathrm{W}}^d:=C([0,\infty),\RR^d)$ and let $\PP_{\mathrm{W}}^d$
denote the completion of the Wiener measure on $\Omega_{\mathrm{W}}^d$. 
The $\sigma$-algebra $\fF_{\mathrm{W}}^d$ is the corresponding domain of 
$\PP_{\mathrm{W}}^d$. Furthermore, we let
$(\fF_{\mathrm{W},t}^d)_{t\ge0}$ denote the $\PP_{\mathrm{W}}^d$-completion of the filtration
generated by the evaluation maps $\pr^d_t(\gamma):=\gamma(t)$, $t\ge0$,
$\gamma\in\Omega_{\mathrm{W}}^d$. Then $(\fF_{\mathrm{W},t}^d)_{t\ge0}$ is automatically
right continuous. Hence, the stochastic basis
$\BB_{\mathrm{W}}^d:=(\Omega_{\mathrm{W}}^d,\fF_{\mathrm{W}}^d,
(\fF_{\mathrm{W},t}^d)_{t\ge0},\PP_{\mathrm{W}}^d)$ satisfies the usual assumptions.
By construction, $\mathrm{pr}^d$ is a Brownian motion on $\BB_{\mathrm{W}}^d$ with
covariance matrix $\id_{\RR^d}$.
\end{ex}

We recall that two stochastic processes $X$ and $Y$ on $[0,\infty)$
with values in some measurable space are called indistinguishable, iff there is a $\PP$-zero set 
$\sN\in\fF$ such that $X_t(\gamma)=Y_t(\gamma)$, for all $t\ge0$ and 
$\gamma\in\Omega\setminus\sN$. 

We call a Hilbert space-valued process $X$ on $[0,\infty)$ continuous, iff all its paths 
$X(\gamma)$, $\gamma\in\Omega$, are continuous (and not just almost every path).
The Brownian motions $\V{B}$ and $\ul{\V{b}}$ are continuous in this sense.

To bound the expectation of exponentials of real-valued martingales we shall repeatedly
employ the following remark:

\begin{rem}\label{remstochGronwall}
Let $d\in\NN$ and $\V{z}$ be a predictable $\RR^d$-valued process on $[0,\infty)$ 
such that $\int_0^t\EE[\V{z}_s^2]\Id s<\infty$ holds for all $t\ge0$. Let $\V{b}$ be a
$d$-dimensional Brownian motion with respect to $\BB$ and
define the real-valued continuous $L^2$-martingale $M$ up to indistinguishability by
$$
M_t:=\int_0^t\V{z}_s\Id\V{b}_s,\;\;t\ge0,\quad\text{so that}\quad
\llbracket M\rrbracket_t=\int_0^t\V{z}_s^2\Id s,\;\; t\ge0.
$$
Then the following estimate follows from \cite[Lem.~4.1]{BleyThomas2015},
\begin{align}\label{BleyThomasmart}
\EE\big[e^{M_t}\big]&\le\EE\big[e^{(pp'/2)\llbracket M\rrbracket_t}\big]^{\nf{1}{p}},\quad p>1,
\end{align}
where $p'$ is the exponent conjugate to $p$. With the help of Scheutzow's stochastic Gronwall lemma 
\cite{Scheutzow2013} we shall derive an analogue of this bound for the running
supremum of $e^M$. (This will be convenient later on in proving the strong continuity of the ultra-violet 
renormalized semi-group.) In fact, the It\={o} formula
\begin{align*}
e^{M_t}=1+\int_0^te^{M_s}\Id M_s
+\frac{1}{2}\int_0^te^{M_s}\Id\llbracket M\rrbracket_s\,\quad t\ge0,\;\text{$\PP$-a.s.},
\end{align*}
where $(\int_0^te^{2M_s}\Id M_s)_{t\ge0}$ is a continuous
local martingale, and \cite[Thm.~4]{Scheutzow2013} directly imply
$\EE[\sup_{s\le t}e^{\delta M_s}]\le c_{\delta,p}
\EE[e^{(\delta p/2)\llbracket M\rrbracket_t}\big]^{\nf{1}{p}}$,
for all $p>1$ and $\delta\in(0,1)$ such that $\delta p'<1$. Replacing $\V{z}$ by
$\V{z}/\delta$ and writing $1/\delta=p'q$ we arrive at
\begin{align}\label{expbdmart}
\EE\big[\sup_{s\le t}e^{M_s}\big]\le c_{p,q}\EE\big[e^{(pp'q/2)\llbracket M\rrbracket_t}\big]^{\nf{1}{p}},
\quad t\ge0,\,p,q>1.
\end{align} 
The constant is given by $c_{p,q}=[1+(4\wedge q)(\pi/q)/\sin(\pi/q)]^{\nf{1}{p'}}$, \cite{Scheutzow2013}.
\end{rem}

For later reference we further recall that, if $d\in\NN$ and $\V{b}$ is a $d$-dimensional 
$\BB$-Brownian motion, then
\begin{align}\label{bdskuno}
\EE\big[\sup_{s\le t}e^{\delta|{\V{b}}_s|^2/2t}\big]&\le e(1-\delta)^{\nf{d}{2}},\;\;\delta\in(0,1),\quad
\EE\big[\sup_{s\le t}e^{a|\V{b}_s|}\big]\le 2^{\nf{d}{2}}e^{a^2t},\;\;a>0.
\end{align}
Here the second bound follows from the first one, which in turn is a consequence of Doob's
maximal inequality applied to the submartingale $e^{\delta|\V{b}|^2}$. In fact, the latter implies 
$\EE[(\sup_{s\le t}e^{\delta|{\V{b}}_s|^2/2pt})^p]\le(p/(p-1))^p\EE[e^{\delta|{\V{b}}_t|^2/2t}]$,
for all $p>1$.



\subsection{Basic one-boson path integrals}\label{ssecBP}

\noindent
In the next definition we introduce some basic, ultra-violet regularized 
$\HP$-valued functionals on $C([0,\infty),\RR^3)$. Their compositions with Brownian motion
will be studied in the succeeding two subsections.

\begin{defn}\label{defbasicproc}
Let $\kappa\in\NN$, $t\ge0$, $\V{\alpha}\in C([0,\infty),\RR^3)$, 
and write $\V{\alpha}_s:=\V{\alpha}(s)$, $s\ge0$.
Then we introduce the following $\HP$-valued Bochner-Lebesgue integrals,
\begin{align}\label{defUpm}
U_{\kappa,t}^-[\V{\alpha}]&:=\int_0^te^{-s\omega-i\V{m}\cdot\V{\alpha}_s}f_\kappa{\Id} s,
\quad
U_{\kappa,t}^+[\V{\alpha}]:=\int_0^te^{-(t-s)\omega-i\V{m}\cdot\V{\alpha}_s}f_\kappa{\Id} s.
\end{align}
\end{defn}

\begin{lem}\label{lemUpmcontHP}
Let $\kappa\in\NN$ and $\V{\alpha}\in C([0,\infty),\RR^3)$. Then the maps
$[0,\infty)\ni t\mapsto U_{\kappa,t}^\pm[\V{\alpha}]\in\HP$ are continuously differentiable with
\begin{align}\label{diffUpm}
\frac{\Id}{\Id t}U^-_{\kappa,t}[\V{\alpha}]&=e^{-t\omega-i\V{m}\cdot\V{\alpha}_t}f_\kappa,\quad
\frac{\Id}{\Id t}U^+_{\kappa,t}[\V{\alpha}]=-\omega U^+_{\kappa,t}[\V{\alpha}]+
e^{-i\V{m}\cdot\V{\alpha}_t}f_\kappa,\quad t\ge0.
\end{align}
\end{lem}

\begin{proof}
Of course, the continuous differentiability of $U_{\kappa}^-[\V{\alpha}]$ is just an instance of the 
fundamental theorem of calculus for $\HP$-valued integrals. 

Furthermore, let $\epsilon\in(\mh,\eh)$. Then
the bound $\omega^{1+2\epsilon} e^{-2s\omega}\le c_\epsilon/s^{1+2\epsilon}$, $s>0$, implies 
\begin{align}\nonumber
\int_t^\tau\|\omega^{1+\epsilon}e^{-(\tau-s)\omega} f_\kappa\|_\HP\Id s&=
\int_0^{\tau-t}\|\omega^{1+\epsilon}e^{-s\omega} f_\kappa\|_\HP\Id s
\\\label{mads0}
&\le c_\epsilon^\eh|\ee|\|\chi_\kappa\|_\HP\int_0^{\tau-t}\frac{\Id s}{s^{\epsilon+\eh}}
=c_{\epsilon,\ee,\kappa}|\tau-t|^{\eh-\epsilon},
\end{align}
for all $\tau\ge t\ge0$. Hence, $U^\pm_{\kappa,t}[\V{\alpha}]\in\dom(\omega^{1+\epsilon})$ and we 
further observe that
\begin{align*}
&\|\omega U^+_{\kappa,\tau}[\V{\alpha}]-\omega U^+_{\kappa,t}[\V{\alpha}]\|_\HP
\\
&\le\Big\|\int_0^t\omega(e^{-(\tau-t)\omega}-1)e^{-(t-s)\omega-i\V{m}\cdot\V{\alpha}_s}
f_\kappa\Id s\Big\|_\HP
+\Big\|\int_t^\tau\omega e^{-(\tau-s)\omega-i\V{m}\cdot\V{\alpha}_s}
f_\kappa\Id s\Big\|_\HP
\\
&\le|\tau-t|^\epsilon\int_0^t\|\omega^{1+\epsilon}e^{-(t-s)\omega} f_\kappa\|_\HP\Id s
+c_{0,\ee,\kappa}|\tau-t|^{\eh},
\end{align*}
for all $\tau\ge t\ge0$ and $\epsilon\in[0,\eh)$, which proves that 
$\omega U^+_\kappa[\V{\alpha}]$ is (locally H\"{o}lder) continuous on $[0,\infty)$. Likewise,
\begin{align}\nonumber
\big\|&\delta^{-1}(U_{\kappa,t+\delta}^+[\V{\alpha}]-U_{\kappa,t}^+[\V{\alpha}])
+\omega U_{\kappa,t}^+[\V{\alpha}]-e^{-i\V{m}\cdot\V{\alpha}_t}f_\kappa\big\|_\HP
\\\nonumber
&\le \int_0^t\big\|\{\delta^{-1}(e^{-\delta\omega}-1)+\omega\}e^{-s\omega}f_\kappa\big\|_\HP\Id s
\\\label{mads1}
&\quad+\frac{1}{\delta}\int_t^{t+\delta}\big\|(e^{-(t+\delta-s)\omega-i\V{m}\cdot\V{\alpha}_s}
-e^{-i\V{m}\cdot\V{\alpha}_t})f_\kappa\big\|_\HP\Id s,
\end{align}
for all $t\ge0$ and $\delta>0$. By the dominated convergence theorem and \eqref{mads0},
the integral in the second line of \eqref{mads1} goes to zero, as $\delta\downarrow0$, 
because the term in the curly brackets $\{\cdots\}$ is bounded from above by $2\omega$.
For every $\ve>0$, we further find some $\delta_0>0$ such that the norm under the integral
in the third line of \eqref{mads1} is $\le\ve$, provided that $0<\delta\le\delta_0$ and
$s\in[t,t+\delta]$. These remarks prove that $U_\kappa^+[\V{\alpha}]$ is differentiable
from the right on $[0,\infty)$ with a continuous right derivative 
$-\omega U^+_{\kappa}[\V{\alpha}]+e^{-i\V{m}\cdot\V{\alpha}}f_\kappa$. 
This finally implies that $U_\kappa^+[\V{\alpha}]$ is continuously differentiable on $[0,\infty)$
and \eqref{diffUpm} is satisfied.
\end{proof}

In the case $\kappa=\infty$ the Bochner-Lebesgue integrals in \eqref{defUpm} still make sense
provided that they are constructed in the auxiliary Hilbert space $\mathfrak{d}_{\mh}$ defined
in \eqref{deffmh}.

\begin{lem}\label{lemUinftyint}
For all $t\ge0$ and $\V{\alpha}\in C([0,\infty),\RR^3)$,
the {\em $\mathfrak{d}_{\mh}$-valued} Bochner-Lebesgue integrals
\begin{align}\label{Uinftyint}
U_{\infty,t}^-[\V{\alpha}]&:=\int_0^te^{-s\omega-i\V{m}\cdot\V{\alpha}_s}f_\infty\Id s,
\quad U_{\infty,t}^+[\V{\alpha}]:=\int_0^te^{-(t-s)\omega-i\V{m}\cdot\V{\alpha}_s}f_\infty\Id s,
\end{align}
are well-defined. In fact, they are well-defined in the Hilbert space 
$L^2(\RR^3,\omega^a\lambda^3)$ with $a\in(-2,0)$ as well. 
Furthermore, the following statements hold true:
\begin{enumerate}[leftmargin=*]
\item[{\rm(1)}] For every $\epsilon\in(\mh,\eh)$, there exists $c_\epsilon>0$ such that
\begin{align}\label{mona0}
\sup_{\V{\alpha}\in C([0,\infty),\RR^3)}\|\omega^{\epsilon-\eh} U_{\kappa,t}^\pm[\V{\alpha}]\|_{\HP}^2
\le c_\epsilon^2\ee^2t^{1-2\epsilon},\quad t\ge0,\,\kappa\in\NN\cup\{\infty\}.
\end{align}
\item[{\rm(2)}] With the same constants as in {\rm(1)} we have 
\begin{align*}
\|U_{\kappa,\tau}^\pm[\V{\alpha}]-U_{\kappa,t}^\pm[\V{\alpha}]\|_{\mathfrak{d}_{\mh}}
&\le\|\omega^\mh\{U_{\kappa,\tau}^\pm[\V{\alpha}]-U_{\kappa,t}^\pm[\V{\alpha}]\}\|_{\HP}
\\
&\le c_\epsilon|\ee|(\tau\vee t)^{\eh-\epsilon}|\tau-t|^\epsilon+c_0|\ee||\tau-t|^\eh,
\end{align*}
for all $\epsilon\in(0,\eh)$, $\tau,t\ge0$, $\V{\alpha}\in C([0,\infty),\RR^3)$, 
and $\kappa\in\NN\cup\{\infty\}$.
\item[{\rm(3)}]
For all $\tau>0$, 
\begin{align}\label{conv2000}
\sup_{t\ge\tau}\sup_{\V{\alpha}\in C([0,\infty),\RR^3)}\|(t\omega)^\mh\{U_{\kappa,t}^\pm[\V{\alpha}]
-U_{\infty,t}^\pm[\V{\alpha}]\}\|_{\HP}\xrightarrow{\;\;\kappa\to\infty\;\;}0.
\end{align}
\end{enumerate}
\end{lem}

\begin{proof}
For all $t\ge0$ and $\epsilon\in(\mh,\eh)$, we observe that 
\begin{align}\nonumber
\int_0^t\|\omega^{\epsilon-1}e^{-(t-s)\omega}\|_\HP\Id s
&=\int_0^t\|\omega^{\epsilon-1}e^{-s\omega}\|_\HP\Id s
\le(4\pi)^\eh\int_0^t\Big(\int_0^\infty\frac{e^{-2s\rho}\rho^2}{\omega_\rho^{2-2\epsilon}}
\Id\rho\Big)^\eh\Id s
\\\label{anne1}
&\le c_\epsilon'\int_0^ts^{\mh-\epsilon}\Id s=c_\epsilon t^{\eh-\epsilon}.
\end{align}
Hence, the $\mathfrak{d}_{\mh}$-valued Bochner-Lebesgue integrals in \eqref{Uinftyint} 
exist and the bounds in \eqref{mona0} hold true.

For $\tau\ge t\ge0$, $\epsilon\in(0,\eh)$, and $\kappa\in\NN\cup\{\infty\}$, we may further deduce that
\begin{align*}
\|&\omega^\mh\{U^\pm_{\kappa,\tau}[\V{\alpha}]-U^\pm_{\kappa,t}[\V{\alpha}]\}\|_{\HP}
\\
&\le(\tau-t)^\epsilon|\ee|\int_0^t\|\omega^{\epsilon-1}e^{-(t-s)\omega}\|_\HP\Id s
+|\ee|\int_0^{\tau-t}\|\omega^{-1}e^{-s\omega}\|_\HP\Id s,
\end{align*}
which together with \eqref{anne1} proves Part~(2).

Now let $\ve>0$. By Hyp.~\ref{hypNelson}(2) we then find some radius $\vr_\ve>0$ such that
\begin{align*}
\forall\,\V{k}\in\RR^3:\quad
|\V{k}|\le\rho_\ve\;\;\Rightarrow\;\;|1-\chi(\V{k})|\le\ve.
\end{align*}
For all $\kappa\in\NN$, $t\ge0$, and $\V{\alpha}\in C([0,\infty),\RR^3)$, this permits to get
\begin{align}\nonumber
\|U_{\kappa,t}^\pm[\V{\alpha}]-U_{\infty,t}^\pm[\V{\alpha}]\|_{\mathfrak{d}_{\mh}}
&\le
\|\omega^\mh(U_{\kappa,t}^\pm[\V{\alpha}]-U_{\infty,t}^\pm[\V{\alpha}])\|_{\HP}
\\\nonumber
&\le|\ee|\int_0^t\|e^{-s\omega}(1-\chi_\kappa)/\omega\|_\HP\Id s
\\\nonumber
&\le(4\pi)^\eh|\ee| \int_0^t\Big(\ve^2\int_0^{\infty} e^{-2s\rho}\Id\rho
+\int_{\kappa\rho_\ve}^\infty{e^{-2s\rho}}\Id\rho\Big)^\eh\Id s
\\\label{conv1}
&\le(8\pi)^\eh|\ee|\ve t^\eh+(2\pi)^\eh|\ee| t^\eh\int_0^1e^{-st\kappa\rho_\ve}\frac{\Id s}{{s}^\eh}.
\end{align}
Since $\ve>0$ is arbitrary, we conclude that \eqref{conv2000} is valid, for all $\tau>0$.
\end{proof}

\begin{rem}\label{remguenther}
Let $a\in(0,\eh)$ and $\kappa\in\NN\cup\{\infty\}$. Since
$(0,t]\times\Omega\ni(s,\gamma)\mapsto e^{-s\omega-i\V{m}\cdot\V{B}_s(\gamma)}f_\kappa
\in L^2(\RR^3,\omega^{-1}\lambda^3)$ is $\fB((0,t])\otimes\fF_t$-measurable, for every $t>0$, 
it follows from Lem.~\ref{lemUinftyint} that $U^\pm_\kappa[\V{B}]$, seen as 
$L^2(\RR^3,\omega^{-1}\lambda^3)$-valued processes, are adapted with locally
$a$-H\"{o}lder continuous paths.
\end{rem}

\begin{rem}\label{remUshift}
Let $\V{\alpha}\in C([0,\infty),\RR^3)$, $t\ge0$, and define 
${}^t\V{\alpha}\in C([0,\infty),\RR^3)$ and $\V{\alpha}_{t-\bullet}\in C([0,\infty),\RR^3)$ by
\begin{align}\label{alphashiftrev}
{}^t\V{\alpha}_s&:=\V{\alpha}_{t+s}-\V{\alpha}_t,\quad
(\V{\alpha}_{t-\bullet})_s:=\V{\alpha}_{(t-s)\vee0},\qquad s\ge0.
\end{align}
Then the following relations hold for all $\kappa\in\NN\cup\{\infty\}$ and $s\ge0$,
\begin{align}\label{shift1}
e^{-t\omega-i\V{m}\cdot\V{\alpha}_t}U_{\kappa,s}^-[{}^t\V{\alpha}]+U_{\kappa,t}^-[\V{\alpha}]
&=U_{\kappa,s+t}^-[\V{\alpha}],
\\\label{shift2}
e^{-i\V{m}\cdot\V{\alpha}_t}U_{\kappa,s}^+[{}^t\V{\alpha}]+e^{-s\omega}U_{\kappa,t}^+[\V{\alpha}]
&=U_{\kappa,s+t}^+[\V{\alpha}].
\end{align}
This follows by shifting the integration variable $r\to r-t$ in the respective first terms.
We emphasize once more that these are identities in $\mathfrak{d}_{\mh}$, if $\kappa=\infty$,
while they hold in $\HP$ as well, if $\kappa$ is finite. The same holds for the transformation rules
\begin{align}\label{revUpm}
U_{\kappa,t}^\pm[\V{\alpha}_{t-\bullet}]&=U^\mp_{\kappa,t}[\V{\alpha}],
\quad\kappa\in\NN\cup\{\infty\}.
\end{align}
\end{rem}

\begin{rem}\label{rem-real}
The integrands in \eqref{defUpm} and in the stochastic integral processes introduced
in \eqref{defMminus} and \eqref{defMtau} below are left invariant under 
the action of the conjugation $C:\HP\to\HP$ given by
$(Cf)(\V{k}):=\ol{f(-\V{k})}$, a.e. $\V{k}$, $f\in\HP$.
Therefore, all the stochastic calculus used below actually takes place in the 
{\em real} Hilbert space
\begin{align*}
\HP_{\RR}&:=\{f\in\HP|\,Cf=f\}. 
\end{align*}
\end{rem}


\subsection{Discussion of $\boldsymbol{U_{\kappa}^-}$}\label{ssecUminus}

\noindent
Next, we study the behavior of the path integral functional $U_\kappa^-$ when it is composed
with the three-dimensional Brownian motion $\V{B}$ introduced in Subsect.~\ref{ssecprobprel}. 
It will turn out that,
even in the case $\kappa=\infty$, the composition $U_\kappa^-[\V{B}]$ is indistinguishable
from a {\em $\HP_{\RR}$-valued} stochastic process. The latter process involves the 
martingale introduced in the following lemma:

\begin{lem}\label{lemMminus}
Let $\kappa\in\NN\cup\{\infty\}$ and $a\in[\mh,\eh)$. Up to indistinguishability,
we define the $\HP_{\RR}$-valued stochastic integral process
\begin{align}\label{defMminus}
M_{\kappa,t}^-[\V{B}]&:=\int_0^te^{-s\omega-i\V{m}\cdot\V{B}_s}i\V{m}\beta_\kappa\Id\V{B}_s,
\quad t\ge0.
\end{align}
Then we $\PP$-a.s. have that $M_{\kappa,t}^-[\V{B}]\in\dom(\omega^a)$, $t\ge0$, and
$\omega^a M_\kappa^-[\V{B}]$ is a continuous
$\HP_{\RR}$-valued $L^2$-martingale. The quadratic variation of $\omega^aM_\kappa^-[\V{B}]$ 
satisfies $\llbracket\omega^aM_\kappa^-[\V{B}]\rrbracket\le c_a\ee^2$, if $|a|<1/2$, and
$\llbracket\omega^\mh M_\kappa^-[\V{B}]\rrbracket_t\le c\ee^2(1+\ln(1\vee t))$, $t\ge0$.
Furthermore,
\begin{align}\label{convMminus}
\sup_{t\ge0}\EE\Big[\sup_{s\le t}\|\omega^a M_{\kappa,s}^-[\V{B}]-
\omega^a M_{\infty,s}^-[\V{B}]\|_\HP^p\Big]\xrightarrow{\;\;\kappa\to\infty\;\;}0,\quad p>0.
\end{align}
\end{lem}

\begin{proof}
First, we notice that 
$(1_{(0,\infty)}(s)e^{-s\omega-i\V{m}\cdot\V{B}_s}\omega^a i\V{m}\beta_\kappa)_{s\ge0}$, 
is a left-continuous, adapted, and in particular predictable $\HP_{\RR}$-valued process.
We next observe that, for all $a\in(-\eh,\eh)$,
\begin{align}\nonumber
\int_0^t\|e^{-s\omega}\omega^a &i\V{m}\beta_{\kappa}\|_\HP^2\Id s
\le\int_0^t\|e^{-s\omega}\omega^a i\V{m}\beta_\infty\|_\HP^2\Id s
\\\nonumber
&\le4\pi\ee^2\int_0^\infty\int_0^t\frac{e^{-2s\omega_\rho}\omega_\rho^{2a}\rho^4}{
\omega_\rho(\omega_\rho+\rho^2/2)^2}\Id s\,\Id\rho
\le2\pi\ee^2\int_0^\infty\frac{\omega_\rho^{2a}\rho^4}{
\omega_\rho^2(\omega_\rho+\rho^2/2)^2}\Id\rho
\\\label{lisa1}
&\le2\pi\ee^2\int_0^{2}\rho^{2a}\Id\rho
+2\pi\ee^2\int_{2}^\infty
\frac{4}{\omega_\rho^{2-2a}}\,\Id\rho\le c_a\ee^2,\quad t\ge0.
\end{align}
In the case $a=-1/2$,
\begin{align}\nonumber
\int_0^t\|e^{-s\omega}&\omega^\mh i\V{m}\beta_{\kappa}\|_\HP^2\Id s
\le2\pi\ee^2\int_0^\infty\frac{(1-e^{-2t\omega_\rho})\rho^4}{
\omega_\rho^3(\omega_\rho+\rho^2/2)^2}\Id\rho
\\\nonumber
&\le4\pi\ee^2\int_0^{1\wedge\nf{1}{t}}\frac{t\rho^4}{\omega_\rho^4}\Id\rho
+2\pi\ee^2\int_{1\wedge\nf{1}{t}}^1\frac{\Id \rho}{\rho}
+2\pi\ee^2\int_1^\infty\frac{2}{\rho^2}\Id\rho
\\\label{lisamh}
&=2\pi\ee^2(4+\ln(1\vee t)),\quad t>0.
\end{align}
These remarks imply that, up to indistinguishability,
\begin{align*}
\wt{M}_{\kappa,t}^{[a]}&:=\int_0^te^{-s\omega-i\V{m}\cdot\V{B}_s}\omega^a 
i\V{m}\beta_\kappa\Id\V{B}_s,\quad t\ge0,
\end{align*}
defines a continuous $\HP_{\RR}$-valued $L^2$-martingale with 
$\llbracket\wt{M}_\kappa^{[a]}\rrbracket\le c_a\ee^2$, if $|a|<1/2$,
and $\llbracket\wt{M}_\kappa^{[\mh]}\rrbracket\le c\ee^2(1+\ln(1\vee t))$, $t\ge0$.
Since $\omega^a$ is a closed operator in $\HP$,
we may further conclude that $M_{\kappa,t}^-[\V{B}]\in\dom(\omega^a)$, $t\ge0$, $\PP$-a.s.,
and that $\omega^a M_\kappa^-[\V{B}]$ and $\wt{M}_\kappa^{[a]}$ are indistinguishable.

Furthermore, Hyp.~\ref{hypNelson}(2) entails
$0\le1-\chi_\kappa(\V{k})\le(|\V{k}|/\kappa)^\ve$, for all $\V{k}\in\RR^3$ with
$|\V{k}|<\kappa$ and $\ve\in[0,1]$.
We set $\ve:=(\eh-a)/2$ in what follows. Similarly as above we then deduce that
\begin{align}\nonumber
\int_0^t&\|e^{-s\omega}\omega^a i\V{m}(\beta_{\kappa}-\beta_\infty)\|_\HP^2\Id s
\\\nonumber
&\le\frac{2\pi\ee^2}{\kappa^{2\ve}}
\Big(\int_0^{1}\rho^{2a+2\ve}\Id\rho
+\int_{1}^{\kappa}\frac{4\Id\rho}{\omega_\rho^{2-2a-2\ve}}\Big)
+8\pi\ee^2\int_{\kappa}^\infty\frac{\Id\rho}{\omega_\rho^{2-2a}}
\\\label{lisa1999}
&\le c_{a}'\ee^2/\kappa^{2\ve}+c_a'\ee^2/\kappa^{1-2a}
\le c_{a}''\ee^2/\kappa^{\eh-a},\quad t\ge0,
\end{align}
which together with a well-known inequality (see, e.g., \cite[Thm.~4.36]{daPrZa2014}) implies
\begin{align*}
\limsup_{\kappa\to\infty}\sup_{t\ge0}
\EE\Big[\sup_{s\le t}\Big\|\int_0^se^{-r\omega-i\V{m}\cdot\V{B}_r}
\omega^a i\V{m}(\beta_{\kappa}-\beta_\infty)\Id\V{B}_r\Big\|_\HP^p&\Big]
\\
\le c_p\limsup_{\kappa\to\infty}\sup_{t\ge0}\EE\Big[\int_0^t\|
e^{-s\omega}\omega^a i\V{m}(\beta_{\kappa}-\beta_\infty)\|_\HP^2\Id s&\Big]^{\nf{p}{2}}=0,
\end{align*}
for every $p>0$.
\end{proof}

Recall the definitions of the auxiliary Hilbert space $\mathfrak{k}$ and the time-dependent norms
$\|\cdot\|_t$ in \eqref{deffrk} and \eqref{defnormt}, respectively.

\begin{lem}\label{lemUminus}
Let $U_\kappa^-[\V{B}]$ be given by \eqref{defUpm} and \eqref{Uinftyint} and $M^-_\kappa[\V{B}]$
by \eqref{defMminus}. Then the following holds:
\begin{enumerate}[leftmargin=*]
\item[{\rm(1)}] There exists a $\PP$-zero set $\sN_-\in\fF$ such that, for all
$t\ge0$ and $\kappa\in\NN\cup\{\infty\}$,
\begin{align}\label{defUminusinfty}
U_{\kappa,t}^-[\V{B}]&=(1-e^{-t\omega-i\V{m}\cdot\V{B}_t})\beta_\kappa-M_{\kappa,t}^-
[\V{B}]\quad\text{on $\Omega\setminus\sN_-$.}
\end{align}
In particular, $1_{\Omega\setminus\sN_-}U_{\infty}^-[\V{B}]$ 
is a continuous, adapted $\HP_\RR$-valued process.
\item[{\rm(2)}] For all $\tau,p>0$,
\begin{align*}
\sup_{t\ge\tau}
\EE\Big[\sup_{s\in[\tau,t]}\big\|U_{\kappa,s}^-[\V{B}]-1_{\Omega\setminus\sN_-}U_{\infty,s}^-[\V{B}]
\big\|_s^p\Big]\xrightarrow{\;\;\kappa\to\infty\;\;}0.
\end{align*}
\end{enumerate}
\end{lem}

\begin{proof}
If $\kappa\in\NN$ and $h\in\HP_{\RR}$, then we may apply It\={o}'s formula to the function
$F_{\kappa,h}\in C^2([0,\infty)\times\RR^3,\RR)$ given by
$F_{\kappa,h}(t,\V{x}):=\SPn{h}{(e^{-t\omega -i\V{m}\cdot\V{x}}-1)\beta_\kappa}_\HP$.
(In the case $\mu=0$ we 
notice that the term $e^{-t\omega -i\V{m}\cdot\V{x}}-1$ in the right entry of the scalar product
defining $F_{\kappa,h}$ compensates for the singularity of $\beta_\kappa$ at zero, so that 
its product with $\beta_\kappa$
defines an element of $C^2([0,\infty)\times\RR^3,\HP_{\RR})$.)
Since we may do this for all $h$ in a countable dense subset of $\HP_{\RR}$ and since
stochastic integrals commute up to indistinguishability with taking scalar products,
we $\PP$-a.s. obtain
\begin{align}\nonumber
(e^{-t\omega -i\V{m}\cdot\V{B}_t}&-1)\beta_\kappa+
\int_0^te^{-s\omega-i\V{m}\cdot\V{B}_s}i\V{m}\beta_\kappa\Id\V{B}_s
\\\label{forUminus}
&=-\int_0^te^{-s\omega-i\V{m}\cdot\V{B}_s}(\omega+\V{m}^2/2)\beta_\kappa\Id s
=-U_{\kappa,t}^-[\V{B}],\quad t\ge0,\;\kappa\in\NN,
\end{align}
which proves \eqref{defUminusinfty} for all finite $\kappa$.
Let $\wt{U}^-_{\infty,t}[\V{B}]$ denote the $\HP$-valued process on the
right hand side of \eqref{defUminusinfty} in the case $\kappa=\infty$.
(Notice again that the map $(t,\V{x})\mapsto(e^{-t\omega -i\V{m}\cdot\V{x}}-1)\beta_\infty\in\HP$
is continuous on $[0,\infty)\times\RR^3$, even in the case $\mu=0$.) Then the convergence 
\begin{align}\label{conv2}
\sup_{t\ge0}
\EE\big[\sup_{s\le t}\|U_{\kappa,s}^-[\V{B}]-\wt{U}_{\infty,s}^-[\V{B}]\|_\HP^p\big]
\xrightarrow{\;\;\kappa\to\infty\;\;}0,\quad p>0,
\end{align}
follows immediately from \eqref{convMminus}, \eqref{forUminus}, and the fact that
$\beta_\kappa-\beta_\infty\to0$ in $\HP$ as $\kappa\to\infty$.
(While $\beta_\kappa$ and $\beta_\infty$ might not belong to $\HP$, their difference
always does by Hyp.~\ref{hypNelson}(2).)
Since the canonical embedding $\HP\to\mathfrak{d}_{\mh}$ is continuous, 
we thus find a $\PP$-zero set $\sN$ 
and integers $1\le\kappa_1<\kappa_2<\ldots\,$ such that on $\Omega\setminus\sN$,
\begin{align*}
U_{\kappa_\ell,t}^-[\V{B}]\xrightarrow{\;\;\ell\to\infty\;\;}\wt{U}_{\infty,t}^-[\V{B}]\;\;
\text{in $\mathfrak{d}_{\mh}$},\quad t\ge0.
\end{align*}
Together with \eqref{conv1} this shows that, in the limiting case $\kappa=\infty$, 
\eqref{defUminusinfty} holds for all $t\ge0$ on $\Omega\setminus\sN$.
Together with \eqref{conv2000} and \eqref{conv2} this concludes the proof of (2) as well.
\end{proof}

\begin{lem}\label{lemexpMminus}
There exist universal constants $b,c>0$ such that, 
for all $p,t\ge0$ and $\kappa\in\NN\cup\{\infty\}$,
\begin{align}\label{expMminus}
\EE\big[\sup_{s\le t} e^{p\|M_{\kappa,s}^-[\V{B}]\|_{\HP}^2}\big]
&\le(1+\pi)e^{bp\ee^2+bp^2\ee^4},
\\\label{expbdUminus}
\EE\big[\sup_{s\le t} e^{p\|1_{\Omega\setminus\sN_-}U_{\kappa,s}^-[\V{B}]\|_s^2}\big]
&\le\sqrt{1+\pi}e^{cp\ee^2(1+\ln(1\vee t))+cp^2\ee^4}.
\end{align}
Here $\sN_-$ is the $\PP$-zero set found in Lem.~\ref{lemUminus}(1).
\end{lem}

\begin{proof}
Recall the notation \eqref{deffsigmakappa}. On account of \eqref{mona0} 
and \eqref{defUminusinfty}, 
the latter identity multiplied with the cut-off function $1_{\{|\V{m}|\ge1\}}$, 
\begin{align}\nonumber
\EE\big[&\sup_{s\le t} e^{p\|1_{\Omega\setminus\sN_-}U_{\kappa,s}^-[\V{B}]\|_s^2}\big]
\\\label{karl1}
&\le e^{cp\ee^2+8p\|\beta_{1,\infty}\|_\HP^2}
\EE\big[\sup_{s\le t}e^{2p\|1_{\{|\V{m}|<1\}}U_{\kappa,s}^-[\V{B}]\|_\HP^2}\big]^\eh
\EE\big[\sup_{s\le t}e^{4p\|M_{\kappa,s}^-[\V{B}]\|_{\HP}^2}\big]^\eh,
\end{align}
for all $p,t\ge0$ and $\kappa\in\NN\cup\{\infty\}$, where
\begin{align}\nonumber
\sup_{s\le t}\|1_{\{|\V{m}|<1\}}U_{\kappa,s}^-[\V{B}]\|_\HP^2
&\le4\pi\ee^2\int_0^1\int_0^t\int_0^t\frac{e^{-(r+s)\omega_\rho}\rho^2}{\omega_\rho}
\Id r\,\Id s\,\Id\rho
\\\nonumber
&=4\pi\ee^2\int_0^1\frac{(1-e^{-t\omega_\rho})^2\rho^2}{\omega_\rho^3}
\Id\rho
\\\label{karl2}
&\le4\pi\ee^2\int_0^{\nf{1}{t}}t\Id \rho+4\pi\ee^2\int_{1\wedge\nf{1}{t}}^1\frac{\Id\rho}{\rho}
=4\pi\ee^2(1+\ln(1\vee t)).
\end{align}
Therefore, it only remains to prove the asserted bound on $M_\kappa^-[\V{B}]$.

An application of It\={o}'s formula for Hilbert space-valued processes
\cite[Thm.~4.32]{daPrZa2014} $\PP$-a.s. yields, for all $t\ge0$ and $\kappa\in\NN\cup\{\infty\}$,
\begin{align}\label{karl2b}
\|M_{\kappa,t}^-[\V{B}]\|^2_\HP&=2N_{\kappa,t}^-
+\int_0^t\|e^{-s\omega}i\V{m}\beta_{\kappa}\|_\HP^2\Id s\le2N_{\kappa,t}^-+c\ee^2,
\end{align}
where we also took \eqref{lisa1} 
into account in the second step and abbreviated
\begin{align*}
N_{\kappa,t}^-&:=\int_0^t\SPn{M_{\kappa,s}^-[\V{B}]}{e^{-\omega s-i\V{m}\cdot\V{B}_s}i\V{m}
\beta_{\kappa}}_\HP\Id\V{B}_s.
\end{align*}
Every process $N_{\kappa}^-$ with $\kappa\in\NN\cup\{\infty\}$
is a local martingale with quadratic variation 
\begin{align*}
\llbracket N_{\kappa}^-\rrbracket_t
&=\int_0^t\SPn{e^{-s\omega/2}M_{\kappa,s}^-[\V{B}]}{e^{-s\omega/2-i\V{m}\cdot\V{B}_s}i\V{m}
\beta_{\kappa}}_\HP^2\Id s,\quad t\ge0.
\end{align*}
Solving the relation \eqref{defUminusinfty} for $M_\kappa^-[\V{B}]$ and plugging the result into 
the previous formula for $\llbracket N_\kappa^-\rrbracket$, we $\PP$-a.s. obtain
\begin{align}\nonumber
\llbracket N_{\kappa}^-\rrbracket_t
&\le2\int_0^t\SPn{(e^{-s\omega-i\V{m}\cdot\V{B}_s}-1)\beta_{\kappa}}{
e^{-s\omega-i\V{m}\cdot\V{B}_s}i\V{m}\beta_{\kappa}}_\HP^2\Id s
\\\nonumber
&\quad+2\int_0^t\SPn{e^{-s\omega/2}U_{\kappa,s}^-[\V{B}]}{
e^{-s\omega/2-i\V{m}\cdot\V{B}_s}i\V{m}\beta_{\kappa}}_\HP^2\Id s
\\\nonumber
&\le8\int_0^t\big\|e^{-s\omega/2}|\V{m}|^\eh\beta_{\kappa}\big\|_\HP^4\Id s
\\\label{lisa77}
&\quad+2\int_0^t\bigg(\ee^2\int_0^t\int_{\RR^3}\frac{e^{-(r+s)\omega}|\V{m}|}{
\omega(\omega+\V{m}^2/2)}\Id\lambda^3\,\Id r\bigg)^2\Id s,
\end{align}
for all $t\ge0$ and $\kappa\in\NN\cup\{\infty\}$. In the second step we used that
$e^{-s\omega/2}U_{\infty,s}^-[\V{B}]=\int_0^se^{-(r+s/2)\omega-i\V{m}\cdot\V{B}_r}f_\infty\Id r$,
$s\ge0$, where the integral on the right hand side is now a {\em $\HP$-valued} Bochner-Lebesgue
integral, that commutes with computing the scalar product.
The integral in the last line of \eqref{lisa77} is bounded from above by $(4\pi\ee^2)^2$ times
\begin{align}\nonumber
\int_0^t\bigg(&\int_0^\infty\frac{e^{-s\omega_\rho}(1-e^{-t\omega_\rho})\rho^3}{\omega_
\rho^2(\omega_\rho+\rho^2/2)}\Id\rho\bigg)^2\Id s
\le\int_0^t\bigg(\int_0^\infty
\frac{e^{-s\omega_\rho}\rho^3}{\omega_\rho^2(\omega_\rho+\rho^2/2)}\Id\rho\bigg)^2\Id s
\\\nonumber
&=\int_0^\infty\int_0^\infty\int_0^t\frac{e^{-s(\omega_r+\omega_\rho)}r^3\rho^3}{
\omega_r^2\omega_\rho^2(\omega_r+r^2/2)(\omega_\rho+\rho^2/2)}\Id s\,\Id r\,\Id\rho
\\\nonumber
&\le\int_0^\infty\int_0^\infty
\frac{(1-e^{-2t\omega_r^\eh\omega_\rho^\eh})r^3\rho^3}{2\omega_r^{\nf{5}{2}}
\omega_\rho^{\nf{5}{2}}(\omega_r+r^2/2)(\omega_\rho+\rho^2/2)}\Id r\,\Id\rho
\\\nonumber
&\le\frac{1}{2}\Big(\int_0^\infty\frac{\rho^3}{\omega_\rho^{\nf{5}{2}}
(\omega_\rho+\rho^2/2)}\Id\rho\Big)^2
\\\label{lisa789}
&\le\frac{1}{2}\bigg(\int_0^1\frac{\Id\rho}{\rho^\eh}+\int_1^\infty\frac{2}{\rho^{\nf{3}{2}}}\Id\rho\bigg)^2
=18.
\end{align}
Here we estimated $e^{-s(\omega_r+\omega_\rho)}\le e^{-2s\omega_r^\eh\omega_\rho^\eh}$
before we performed the $\Id s$-integration. Furthermore,
\begin{align*}
\int_0^t\big\|e^{-s\omega/2}|\V{m}|^\eh\beta_{\kappa}\big\|^4\Id s
&\le(4\pi\ee^2)^2\int_0^t\Big(\int_0^\infty\frac{e^{-s\omega_\rho}\rho^3}{
\omega_\rho(\omega_\rho+\rho^2/2)^2}\Id \rho\Big)^2\Id s\le18(4\pi\ee^2)^2,
\end{align*}
because the integral in the middle is bounded from above by the second integral in the first line
of \eqref{lisa789}. Altogether we find some universal constant $c>0$ such that, $\PP$-a.s., 
$\llbracket N_{\kappa}^-\rrbracket\le c\ee^4$. 
By virtue of \eqref{karl2b} and Rem.~\ref{remstochGronwall} we may now conclude that,
for all $p,t\ge0$ and $\kappa\in\NN\cup\{\infty\}$,
\begin{align*}
e^{-cp\ee^2}\EE\big[\sup_{s\le t} e^{p\|M_{\kappa,s}^-[\V{B}]\|^2}\big]
&\le \EE\big[\sup_{s\le t}e^{2p N_{\kappa,s}^-}\big]
\\
&\le \limsup_{q\to\infty}c_{q,2}\EE\big[e^{qq'(2p)^2\llbracket 
N_{\kappa}^-\rrbracket_t}\big]^{\nf{1}{q}}\le(1+\pi)e^{c'p^2\ee^4}.
\end{align*}
\end{proof}


\subsection{Discussion of $\boldsymbol{U_{\kappa}^+}$}\label{ssecUplus}

\noindent
In this subsection we analyze $U^+_\kappa[\V{B}]$ in a similar fashion as we treated 
$U^-_\kappa[\V{B}]$ in the preceding one. The explicit $t$-dependence of the integrands
in the formulas \eqref{defUpm} and \eqref{Uinftyint} for $U^+_\kappa$ causes, however, some
additional technical difficulties. Roughly speaking we shall first introduce an additional variable
$\tau\ge0$ parametrizing the integrands and later restrict our results to the diagonal
$(t,\tau)=(t,t)$ with the help of Kolmogorov's test lemma.

\begin{lem}\label{lemMtau}
For all $\kappa\in\NN\cup\{\infty\}$ and $\tau\ge0$, we define, a priori only up to indistinguishability,
\begin{align}\label{defMtau}
M_{\kappa,t}^{[\tau]}[\V{B}]
:=\int_0^t1_{(0,\tau)}(s)e^{-(\tau-s)\omega-i\V{m}\cdot\V{B}_s}i\V{m}\beta_{\kappa}\Id\V{B}_s,
\quad t\ge0.
\end{align}
Then the following two statements hold for all $\kappa\in\NN\cup\{\infty\}$:
\begin{enumerate}[leftmargin=*]
\item[{\rm(1)}]
For all $a\in[\mh,\eh)$ and $\tau\ge0$, 
we $\PP$-a.s. have $M_{\kappa,t}^{[\tau]}[\V{B}]\in\dom(\omega^a)$, $t\ge0$, and the process 
$\omega^a M_{\kappa}^{[\tau]}[\V{B}]$ is a continuous square-integrable
$\HP_{\RR}$-valued martingale. The quadratic variation of $\omega^aM_\kappa^{[\tau]}[\V{B}]$ 
satisfies $\llbracket \omega^aM_\kappa^{[\tau]}[\V{B}]\rrbracket\le c_a\ee^2$, if $|a|<1/2$, and
$\llbracket\omega^\mh M_{\kappa}^{[\tau]}[\V{B}]\rrbracket\le c\ee^2(1+\ln(1\vee\tau))$.
\item[{\rm(2)}]
We can choose $M_{\kappa}^{[\tau]}[\V{B}]$ in the equivalence class
modulo indistinguishability defined by the stochastic integrals \eqref{defMtau} for each fixed 
$\tau\ge0$ in such a way that, for all $a\in[\mh,\eh)$,
\begin{align}\label{tizian0}
[0,\infty)^2\ni(t,\tau)\longmapsto\big(\omega^a M_{t,\kappa}^{[\tau]}[\V{B}]\big)(\gamma)\in\HP\;\;
\text{is continuous, for all $\gamma\in\Omega$.}
\end{align}
\end{enumerate}
Furthermore, consider choices of $M_{\kappa}^{[\tau]}[\V{B}]$,
$\kappa\in\NN\cup\{\infty\}$, $\tau\ge0$, as in {\rm(2)}. Let $a\in[\mh,\eh)$ and $p>0$. Then
\begin{align}\label{polbdMtt}
\EE\big[\sup_{s,\tau\le t}\|\omega^a M_{\kappa,s}^{[\tau]}[\V{B}]\|_\HP^p\big]&\le
c_{p,a}(1+t)^{p}|\ee|^p,\quad t\ge0,\,\kappa\in\NN\cup\{\infty\},
\end{align}
and
\begin{align}\label{convMtt}
\EE\Big[\sup_{s,\tau\le t}\big\|\omega^a M_{\kappa,s}^{[\tau]}[\V{B}]
-\omega^a M_{\infty,s}^{[\tau]}[\V{B}]\big\|_\HP^p\Big]
&\le\frac{c_{p,a}'(1+t)^{p}|\ee|^p}{\kappa^{p(1-2a)/8}},\quad t\ge0,\,\kappa\in\NN.
\end{align}
\end{lem}

\begin{proof}
To reduce clutter we drop all arguments $[\V{B}]$ in this proof.

\smallskip

\noindent
{\em Step~1.}
Let $\kappa\in\NN\cup\{\infty\}$ and $a\in[\mh,\eh)$. Then we first observe that the process
$(1_{(0,\tau)}(s)e^{-(\tau-s)\omega-i\V{m}\cdot\V{B}_s}\omega^a i\V{m}\beta_{\kappa})_{s\ge0}$ 
is predictable as the pointwise limit on $[0,\infty)\times\Omega$ of the left-continuous, adapted, and 
thus predictable $\HP$-valued processes 
$$
\big(1_{(0,\tau-\nf{1}{n}]}(s)e^{-(\tau-s)\omega-i\V{m}\cdot\V{B}_s}
\omega^a i\V{m}\beta_{\kappa}\big)_{s\ge0},\quad n\in\NN.
$$ 
Furthermore,
\begin{align}\nonumber
\int_0^t1_{(0,\tau)}(s)\|e^{-(\tau-s)\omega}\omega^{a}i\V{m}\beta_\kappa\|_\HP^2\Id s
&\le \int_0^\tau\|e^{-s\omega}\omega^a i\V{m}\beta_\kappa\|_\HP^2\Id s
\\\label{polbdMtaua}
&\le \left\{\begin{array}{ll}c_a\ee^2,&\text{if $|a|<1/2$},
\\ c\ee^2(1+\ln(1\vee\tau)),&\text{if $a=-1/2$},\end{array}\right.
\end{align}
for all $t\ge0$, according to \eqref{lisa1} and \eqref{lisamh}. 
These remarks ensure that every process given by
$$
\wt{M}_{\kappa,t}^{[\tau,a]}:=\int_0^t1_{(0,\tau)}(s)e^{-(\tau-s)\omega-i\V{m}\cdot\V{B}_s}
\omega^ai\V{m}\beta_{\kappa}\Id\V{B}_s,\quad t\ge0,
$$
with $\tau\ge0$ is a well-defined continuous square-integrable $\HP_{\RR}$-valued martingale, 
whose quadratic variation is given by the left hand side of \eqref{polbdMtaua}.
We further conclude that $\wt{M}_\kappa^{[\tau,a]}$ is indistinguishable from 
$\omega^a M_{\kappa}^{[\tau]}$, whence
\begin{align}\label{polbdMtaub}
\EE\big[\sup_{s\ge0}\|\omega^{a}M_{\kappa,s}^{[\tau]}\|_\HP^p\big]
&\le c_p\EE\Big[\int_0^{\tau}\|e^{-(\tau-s)\omega}\omega^{a}
i\V{m}\beta_\kappa\|_\HP^2\Id s\Big]^{\nf{p}{2}},\quad p>0.
\end{align}
{\em Step~2.} Let $\kappa\in\NN\cup\{\infty\}$, $n\in\NN$, and $a_n:=\eh-\nf{1}{2n}$.
We next employ the Kolmogorov test to show that we can modify
the martingales $M_\kappa^{[\tau]}$ such that \eqref{tizian0} holds true, for every $a\in[\mh,a_n]$.

To this end we pick some $a\in[\mh,a_n]$ and
observe that, for all $p>0$, $\tau>\sigma\ge0$, and $\epsilon\in(0,\eh-a)$, 
\begin{align*}
&\EE\big[\sup_{s\ge0}\| \omega^a M_{\kappa,s}^{[\tau]}-
\omega^a M_{\kappa,s}^{[\sigma]}\|_\HP^p\big]
\\
&\le c_p\EE\bigg[\sup_{s\le \sigma}\Big\| \int_0^s1_{(0,\sigma)}(r)(e^{-(\tau-\sigma)\omega}-1)
e^{-(\sigma-r)\omega-i\V{m}\cdot\V{B}_r}\omega^{a}i\V{m}\beta_\kappa\Id\V{B}_r\Big\|_\HP^p\bigg]
\\
&\quad+c_p\EE\bigg[\sup_{s\le \tau}\Big\| \int_0^s1_{(\sigma,\tau)}(r)
e^{-(\tau-r)\omega-i\V{m}\cdot\V{B}_r}\omega^{a}i\V{m}\beta_\kappa\Id\V{B}_r\Big\|_\HP^p\bigg]
\\
&\le c_p'(\tau-\sigma)^{p\epsilon}\EE\Big[\int_0^{\sigma}\|\omega^{a+\epsilon}
e^{-(\sigma-s)\omega}i\V{m}\beta_\kappa\|^2_\HP\Id s\Big]^{\nf{p}{2}}
\\
&\quad+c_p'\EE\Big[\int_{\sigma}^{\tau}
\|e^{-(\tau-s)\omega}\omega^{a}i\V{m}\beta_\kappa\|^2_\HP\Id s\Big]^{\nf{p}{2}}.
\end{align*}
Since $\delta:=a+\epsilon\in(\mh,\eh)$, the integrands in the last two lines satisfy
\begin{align}\nonumber
\int_0^{\sigma}\|\omega^\delta
e^{-(\sigma-s)\omega}i\V{m}\beta_\kappa\|^2_\HP\Id s&\le  c_\delta\ee^2,
\end{align}
by \eqref{lisa1}, as well as
\begin{align*}
\int_{\sigma}^{\tau}\|e^{-(\tau-s)\omega}\omega^{a}i\V{m}\beta_\kappa\|^2_\HP\Id s
&\le
4\pi\ee^2\int_0^\infty\frac{(1-e^{-2(\tau-\sigma)\omega_\rho})\omega_\rho^{2a}
\rho^4}{2\omega_\rho^2(\omega_\rho+\rho^2/2)^2}\Id\rho
\\
&\le c_{\epsilon,\delta}\ee^2(\tau-\sigma)^{2\epsilon}.
\end{align*}
Here we used 
$1-e^{-2(\tau-\sigma)\omega_\rho}\le 4^\epsilon(\tau-\sigma)^{2\epsilon}\omega_\rho^{2\epsilon}$
and \eqref{lisa1} in the last step. We thus arrive at
\begin{align}\label{KolMtau}
\EE\big[\sup_{s\ge0}\| \omega^a M_{\kappa,s}^{[\tau]}-\omega^{a}M_{\kappa,s}^{[\sigma]}\|_\HP^p\big]
&\le c_{p,a,\epsilon}|\ee|^{p}|\tau-\sigma|^{p\epsilon},
\end{align}
for all $\sigma,\tau\ge0$, $0<\epsilon<1/2-a$, and $a\in[\mh,a_n]$.
Now we fix $p>0$ such that $p(1-2a_n)>2$. Then we may pick some
$\epsilon\in(0,\eh-a_n)$ with $p\epsilon>1$ and the Kolmogorov test lemma (a suitable version is 
stated, e.g., in App.~\ref{appKol}) ensures the existence of continuous, adapted 
$\HP$-valued processes $(X_{\kappa,n,t}^{[\tau]})_{t\ge0}$ parametrized by $\tau\ge0$
such that $[0,\infty)^2\ni(t,\tau)\mapsto X_{\kappa,n,t}^{[\tau]}(\gamma)\in\HP$ is continuous, for all
$\gamma\in\Omega$, and such that, for every $\tau\ge0$, $X_{\kappa,n}^{[\tau]}$ is indistinguishable
from $(\omega^\mh+\omega^{a_n})M_\kappa^{[\tau]}$. Then
$\hat{M}_{\kappa,n}^{[\tau]}:=(\omega^\mh+\omega^{a_n})^{-1}X_\kappa^{[\tau]}$ is indistinguishable
from $M_\kappa^{[\tau]}$ and satisfies \eqref{tizian0}, for all $a\in[\mh,a_n]$.

Next, consider $n\in\NN$ with $n>1$. For every $\tau\in[0,\infty)\cap\QQ$ we then find some
$\PP$-zero set $\sN_{n}^\tau\in\fF$ such that 
$\hat{M}_{\kappa,1,t}^{[\tau]}=\hat{M}_{\kappa,n,t}^{[\tau]}$, $t\ge0$, on 
$\Omega\setminus\sN_{n}^\tau$. By the continuity in $\tau$ it then follows that
$\hat{M}_{\kappa,1,t}^{[\tau]}=\hat{M}_{\kappa,n,t}^{[\tau]}$, $\tau,t\ge0$, on
$\Omega\setminus\sN_{n}$, where $\sN_{n}:=\bigcup_{\tau\in[0,\infty)\cap\QQ}\sN_{n}^\tau$
has $\PP$-measure zero. We now define another $\PP$-zero set 
$\sN:=\bigcup_{n>1}\sN_{n}$ and set $\hat{M}_{\kappa,t}^{[\tau]}:=0$ on $\sN$ and 
$\hat{M}_{\kappa,t}^{[\tau]}:=\hat{M}_{\kappa,1,t}^{[\tau]}$ on $\Omega\setminus\sN$, for all 
$t,\tau\ge0$. Then each $\hat{M}_{\kappa,t}^{[\tau]}$ with $\tau\ge0$ is indistinguishable from
$M_\kappa^{[\tau]}$ and satisfies \eqref{tizian0}, for all $a\in[\mh,\eh)$.

\smallskip

\noindent
{\em Step~3.} Let $a\in[\mh,\eh)$.
By virtue of H\"{o}lder's inequality it suffices to derive the moment bound \eqref{polbdMtt} 
for all $p>0$ satisfying $p(1-2a)>2$. In this case \eqref{polbdMtt} follows, however, from
\eqref{polbdMtaua}, \eqref{polbdMtaub}, \eqref{KolMtau}, and Rem.~\ref{remKol}.
Strictly speaking, since we apply the Kolmogorov lemma another time (with some new choice of $p$), 
we first obtain \eqref{polbdMtt} for possibly different versions of 
$\omega^a\hat{M}_\kappa^{[\tau]}$, call them $\check{M}_{\kappa}^{[\tau,a,p]}$, such that 
$[0,\infty)^2\ni(t,\tau)\mapsto\check{M}_{\kappa,t}^{[\tau,a,p]}(\gamma)\in\HP$ is continuous, for every
$\gamma\in\Omega$. For every $\tau\in[0,\infty)\cap\QQ$, we again find, however, some $\PP$-zero 
set $\sN_\tau$ such that $\check{M}_{\kappa,t}^{[\tau,a,p]}=\omega^a\hat{M}_{\kappa,t}^{[\tau]}$, 
$t\ge0$, on $\Omega\setminus\sN_\tau$. By continuity we then conclude that, for all 
$(t,\tau)\in[0,\infty)^2$, the identity 
$\check{M}_{\kappa,t}^{[\tau,a,p]}=\omega^a\hat{M}_{\kappa,t}^{[\tau]}$ holds true on the $\PP$-zero
set $\bigcup_{\tau\in[0,\infty)\cap\QQ}\sN_\tau$.

\smallskip

\noindent
{\em Step~4.} 
Let $a\in[\mh,\eh)$. On account of H\"{o}lder's inequality it then suffices to prove \eqref{convMtt}
for every $p>0$ with $p(1-2a)/4>1$. Applying \eqref{lisa1999} we first observe that
\begin{align}\nonumber
\EE\big[\sup_{t\ge0}\|\omega^{a}M_{\kappa,t}^{[\tau]}-\omega^{a}M_{\infty,t}^{[\tau]}\|_\HP^p\big]
&\le c_p\EE\Big[\int_0^\tau\|e^{-(\tau-s)\omega}\omega^{a}
i\V{m}(\beta_\kappa-\beta_\infty)\|_\HP^2\Id s\Big]^{\nf{p}{2}}
\\\label{convMtau}
&\le c_{p,a}|\ee|^p\big/\kappa^{p(1-2a)/4},
\end{align}
which holds for all $\tau\ge0$ and $p>0$. Again let $\epsilon\in(0,\eh-a)$,
$\delta:=a+\epsilon$, and set $\iota:=(\eh-\delta)/2$. Then $2(\delta+\iota)<1$. Employing similar 
estimates as in Step~2 and using that $|1-\chi_\kappa(\V{k})|\le(|\V{k}|/\kappa)^\iota$, 
for all $\V{k}\in\RR^3$ with $|\V{k}|<\kappa$, we further obtain
\begin{align}\nonumber
&\EE\Big[\sup_{s\ge0}\big\| \omega^a(M_{\kappa,s}^{[\tau]}-M_{\infty,s}^{[\tau]})-
\omega^{a}(M_{\kappa,s}^{[\sigma]}-M_{\infty,s}^{[\sigma]})\big\|_\HP^p\Big]
\\\nonumber
&\le c_{p,\epsilon}|\ee|^p|\tau-\sigma|^{p\epsilon}
\bigg(\frac{1}{\kappa^{2\iota}}\int_0^\kappa\frac{\omega_\rho^{2\delta}
\rho^{4+2\iota}}{2\omega_\rho^2(\omega_\rho+\rho^2/2)^2}\Id\rho
+\int_\kappa^\infty\frac{\omega_\rho^{2\delta}
\rho^4}{2\omega_\rho^2(\omega_\rho+\rho^2/2)^2}\Id\rho\bigg)^{\nf{p}{2}}
\\\nonumber
&\le c_{p,\epsilon}|\ee|^p|\tau-\sigma|^{p\epsilon}
\bigg(\frac{1}{\kappa^{2\iota}}\int_0^1\frac{\rho^{2(\delta+\iota)}}{2}\Id\rho
+\frac{1}{\kappa^{2\iota}}\int_1^\kappa\frac{2}{\rho^{2-2\delta-2\iota}}\Id\rho
+\int_\kappa^\infty\frac{2}{\rho^{2-2\delta}}\Id\rho\bigg)^{\nf{p}{2}}
\\\nonumber
&\le c_{p,\epsilon}|\ee|^p|\tau-\sigma|^{p\epsilon}
\big(c_{\delta+\iota}/\kappa^{2\iota}+c_\delta/\kappa^{1-2\delta}\big)^{\nf{p}{2}}
\\\nonumber
&\le c_{p,a,\epsilon}|\ee|^p|\tau-\sigma|^{p\epsilon}\big/\kappa^{p(\eh-\delta)/2},
\end{align}
for all $\sigma,\tau\ge0$. Now we choose $\epsilon:=(1-2a)/4$, so that $0<\epsilon<1/2-a$, 
$(1/2-\delta)/2=(1-2a)/8>0$, and our present assumption on $p$ ensures that $p\epsilon>1$.
Then \eqref{convMtt} follows from \eqref{convMtau}, the previous estimation, Rem.~\ref{remKol} on
Kolmogorov's lemma, and an argument similar to the one in Step~3.
\end{proof}

\begin{lem}\label{lemI}
Let $\kappa\in\NN\cup\{\infty\}$ and $\V{\alpha}\in C([0,\infty),\RR^3)$. 
Then the $\HP$-valued Bochner-Lebesgue integrals
\begin{align}\label{defIalpha}
I_{\kappa,t}[\V{\alpha}]&:=\int_0^te^{-(t-s)\omega-i\V{m}\cdot\V{\alpha}_s}
2\omega\beta_\kappa\Id s,\quad t\ge0,
\end{align}
converge absolutely and $I_{\kappa,t}[\V{\alpha}]\in\dom(\omega^a)$, for all $a\in(-1,1)$ and $t\ge0$.
For all $a\in(-1,1)$, the map $[0,\infty)\ni t\mapsto \omega^aI_{\kappa,t}[\V{\alpha}]\in\HP$ is locally
H\"{o}lder continuous. If $a\in(-1,0)$, then $\omega^a I_{\kappa}[\V{\alpha}]$ is continuously 
differentiable on $[0,\infty)$ as an $\HP$-valued function with
\begin{align}\label{derI}
\frac{\Id}{\Id t}\omega^a I_{\kappa,t}[\V{\alpha}]=-\omega^{1+a}I_{\kappa,t}[\V{\alpha}]
+e^{-i\V{m}\cdot\V{\alpha}_t}2\omega^{1+a}\beta_\kappa,\quad t\ge0.
\end{align}
\end{lem}

\begin{proof}
To prove the first claim, we write $a=\delta+\epsilon$ with $\delta,\epsilon\in(\mh,\eh)$ and replace
$\chi_\kappa$ by $j_{\delta,\kappa}:=\omega^\delta\chi_\kappa/(\omega+\V{m}^2/2)$ 
in the estimate \eqref{mads0}, observing that 
$\|j_{\delta,\kappa}\|_{\HP}\le c_\delta(1\vee\mu)^{0\vee\delta}$, for all $\kappa\in\NN\cup\{\infty\}$.
With the so-obtained modification of \eqref{mads0} we can mimic the remaining parts of the
proof of Lem.~\ref{lemUpmcontHP} to verify the other two statements on $I_\kappa[\V{\alpha}]$.
\end{proof}

\begin{lem}\label{lemUplus}
For all $\kappa\in\NN\cup\{\infty\}$ and $\tau\ge0$, let $M_\kappa^{[\tau]}[\V{B}]$ denote a 
particular choice of the martingale introduced in Lem.~\ref{lemMtau} such that \eqref{tizian0} 
is satisfied, for all $a\in[\mh,a_0]$ and some $a_0\in[0,\eh)$.
Then the following holds:
\begin{enumerate}[leftmargin=*]
\item[{\rm(1)}] There exists a $\PP$-zero set $\sN_+\in\fF$ such that, for all $t\ge0$ and
$\kappa\in\NN\cup\{\infty\}$,
\begin{align}\label{relUplus}
U^+_{\kappa,t}[\V{B}]&=(e^{-t\omega}-e^{-i\V{m}\cdot\V{B}_t})\beta_\kappa-M_{\kappa,t}^{[t]}[\V{B}]
+I_{\kappa,t}[\V{B}]\quad\text{on $\Omega\setminus\sN_+$.}
\end{align}
In particular, $1_{\Omega\setminus\sN_+}U_{\infty}^+[\V{B}]$ is a continuous, adapted
$\HP_{\RR}$-valued process.
\item[{\rm(2)}] For all $p>0$,
\begin{align}\label{convUplusHP}
\EE\Big[\sup_{s\le t}\big\|U_{\kappa,s}^+[\V{B}]-1_{\Omega\setminus\sN_+}
U^+_{\infty,s}[\V{B}]\big\|_{\HP}^p\Big]&\xrightarrow{\;\;\kappa\to\infty\;\;}0,\quad t\ge0,
\\\label{convUpluss}
\EE\Big[\sup_{\tau\le s\le t}\big\|U_{\kappa,s}^+[\V{B}]-1_{\Omega\setminus\sN_+}
U^+_{\infty,s}[\V{B}]\big\|_s^p\Big]&\xrightarrow{\;\;\kappa\to\infty\;\;}0,\quad t\ge\tau>0.
\end{align}
\end{enumerate}
\end{lem}

\begin{proof}
{\em Step~1.} Let $\tau>0$ and $\kappa\in\NN\cup\{\infty\}$. 
For every $h$ in a countable dense subset of $\HP_{\RR}$, we now apply
It\={o}'s formula to the function $G_{\kappa,\tau,h}\in C^2([0,\tau)\times\RR^3,\RR)$ given by
$G_{\kappa,\tau,h}(t,\V{x}):=\SPn{h}{(e^{-(\tau-t)\omega-i\V{m}\cdot\V{x}}-e^{-\tau\omega})
\beta_\kappa}_\HP$, $t\in[0,\tau)$, $\V{x}\in\RR^3$. 
(Again the singularity of $\beta_\kappa$ at $0$ is compensated for by the difference of 
exponentials in the formula for $G_{\kappa,\tau,h}$.) We then find a $\PP$-zero set
$\sN_\tau\in\fF$ such that the following identity holds on $\Omega\setminus\sN_\tau$
and for all $t\in[0,\tau)$,
\begin{align}\nonumber
(e^{-(\tau-t)\omega-i\V{m}\cdot\V{B}_t}-e^{-\tau\omega})\beta_{\kappa}
&=-\int_{0}^te^{-(\tau-s)\omega-i\V{m}\cdot\V{B}_s}f_{\kappa}\Id s
\\\label{lisa2}
&\quad+\int_{0}^te^{-(\tau-s)\omega-i\V{m}\cdot\V{B}_s}2\omega\beta_{\kappa}\Id s
-M_{\kappa,t}^{[\tau]}[\V{B}].
\end{align}
On account of Lem.~\ref{lemI} and the strong continuity of 
$[0,\tau]\ni t\mapsto e^{-(\tau-t)\omega}\in\LO(\HP)$
all trajectories of the processes on the left hand side and in the
second line of \eqref{lisa2} are continuous in $t\in[0,\tau]$. Therefore, the limit
$\lim_{t\uparrow\tau}\int_{0}^te^{-(\tau-s)\omega-i\V{m}\cdot\V{B}_s}f_{\kappa}\Id s$
with respect to the topology on $\HP$ exists on $\Omega\setminus\sN_\tau$.
Since the canonical embedding $\HP\to\mathfrak{d}_{\mh}$ is continuous,
the latter limit must, however, agree with $U_{\kappa,t}^+[\V{B}]$. Therefore, the
identity in \eqref{relUplus} is valid, for every fixed $t>0$, on the complement of $\sN_t$.

We now set $\sN_+:=\bigcup_{t\in\QQ:t>0}\sN_t$. Since all terms in \eqref{relUplus} are
continuous in $t\in[0,\infty)$ with respect to the topology on $\mathfrak{d}_{\mh}$ and on
all of $\Omega$, we may then conclude that the identity in \eqref{relUplus} holds on 
$\Omega\setminus\sN_+$ and for all $t\ge0$,
if we consider it as an identity in $\mathfrak{d}_{\mh}$. We already know, however, that the right
hand side of \eqref{relUplus} is a well-defined $\HP$-valued process on $[0,\infty)$. 
Altogether this concludes the proof of Part~(1).

\smallskip

\noindent{\em Step~2.}
Next, we use that $|1-\chi_\kappa(\V{k})|^2\le|\V{k}|/\kappa$,
for all $\V{k}\in\RR^3$ with $|\V{k}|<\kappa$ and $\kappa\in\NN$, which permits to get
\begin{align}\nonumber
\Big\|&\int_0^t e^{-(t-s)\omega-i\V{m}\cdot\V{B}_s}\omega(\beta_\kappa-\beta_\infty)
\Id s\Big\|_\HP^2
\\\nonumber
&=\int_0^t\int_0^t\SPb{e^{-s\omega-i\V{m}\cdot\V{B}_{t-s}}\omega(\beta_\kappa-\beta_\infty)}{
e^{-r\omega-i\V{m}\cdot\V{B}_{t-r}}\omega(\beta_\kappa-\beta_\infty)}_\HP\Id r\,\Id s
\\\nonumber
&\le4\pi\ee^2\int_0^\infty\big\{(\rho/\kappa)1_{[0,\kappa]}(\rho)+1_{(\kappa,\infty)}(\rho)
\big\}\int_0^t\int_0^t\frac{e^{-(r+s)\omega_\rho}\omega_\rho^2\rho^2}{
\omega_\rho(\omega_\rho+\rho^2/2)^2}\Id r\,\Id s\,\Id\rho
\\\nonumber
&\le4\pi\ee^2\int_0^\infty\big\{(\rho/\kappa)1_{[0,\kappa]}(\rho)+1_{(\kappa,\infty)}(\rho)
\big\}\frac{(1-e^{-t\omega_\rho})^2\rho^2}{\omega_\rho(\omega_\rho+\rho^2/2)^2}\Id\rho
\\\label{conv74}
&\le\frac{4\pi\ee^2}{\kappa}\Big(\int_0^21\Id\rho+\int_2^\infty\frac{4}{\rho^2}{\Id\rho}\Big)
+4\pi\ee^2\int_{\kappa}^\infty\frac{2}{\rho^2}\Id\rho
=\frac{24\pi\ee^2}{\kappa}.
\end{align}
\noindent{\em Step~3.} Finally, we observe that \eqref{convUplusHP} follows from
\eqref{convMtt}, \eqref{relUplus}, \eqref{conv74}, and the fact that $\beta_\kappa-\beta_\infty\to0$ 
in $\HP$ as $\kappa\to\infty$. The limit relation \eqref{convUpluss} is a consequence of 
\eqref{conv2000} and \eqref{convUplusHP}.
\end{proof}

\begin{lem}\label{lemexpplus}
For all $\kappa\in\NN\cup\{\infty\}$ and $\tau\ge0$, we introduce the $\HP$-valued semi-martingales
\begin{align}\label{defStt}
S_{\kappa,t}^{[\tau]}[\V{B}]&:=M_{\kappa,t}^{[\tau]}[\V{B}]-\int_0^t1_{(0,\tau)}(s)
e^{-(\tau-s)\omega-i\V{m}\cdot\V{B}_s}2\omega\beta_\kappa\Id s,\quad t\ge0.
\end{align}
Then there exist universal constants $b,c>0$ such that, for all $\kappa\in\NN\cup\{\infty\}$
and $p,\tau\ge0$,
\begin{align}\label{expbdStau}
\EE\big[\sup_{s\le t}e^{p\|S_{\kappa,s}^{[\tau]}[\V{B}]\|_\HP^2}\big]&\le 
(1+\pi)e^{cp\ee^2(1+\ln[1\vee(\tau\wedge t)])+cp^2\ee^4},
\\\label{expbdMtau}
\EE\big[\sup_{s\le t}e^{p\|M_{\kappa,s}^{[\tau]}[\V{B}]\|_\HP^2}\big]&\le 
(1+\pi)e^{bp\ee^2(1+\ln[1\vee(\tau\wedge t)])+bp^2\ee^4}.
\end{align}
\end{lem}

\begin{proof}
Again the arguments $[\V{B}]$ are dropped in the notation in this proof;
$\kappa\in\NN\cup\{\infty\}$ and $\tau>0$ are fixed throughout the proof.

It\={o}'s formula $\PP$-a.s. implies
\begin{align}\nonumber
\|S_{\kappa,t}^{[\tau]}\|_\HP^2
&=2N_{\kappa,t}^{[\tau]}+4L_{\kappa,t}^{[\tau]}
+\int_0^{t\wedge\tau}\|e^{-(\tau-s)\omega}i\V{m}\beta_{\kappa}\|_\HP^2\Id s
\\\nonumber
&\le2N_{\kappa,t}^{[\tau]}+4L_{\kappa,t}^{[\tau]}+c\ee^2,\quad t\ge0,
\end{align}
where we also made use of \eqref{lisa1} and abbreviated
\begin{align}\nonumber
N_{\kappa,t}^{[\tau]}&:=\int_0^t1_{(0,\tau)}(s)\SPn{S_{\kappa,s}^{[\tau]}}{
e^{-(\tau-s)\omega-i\V{m}\cdot\V{B}_s}i\V{m}\beta_{\kappa}}_\HP\Id\V{B}_s,
\\\nonumber
L_{\kappa,t}^{[\tau]}&:=\int_0^{t\wedge\tau}\SPn{S_{\kappa,s}^{[\tau]}}{
e^{-(\tau-s)\omega-i\V{m}\cdot\V{B}_s}\omega\beta_{\kappa}}_\HP\Id s.
\end{align}
We shall see that
$(N_{\kappa,t}^{[\tau]})_{t\ge0}$ is an $L^2$-martingale with quadratic variation 
\begin{align*}
\llbracket N_{\kappa}^{[\tau]}\rrbracket_t
&=\int_0^t1_{(0,\tau)}(s)\SPn{S_{\kappa,s}^{[\tau]}}{e^{-(\tau-s)\omega-i\V{m}\cdot\V{B}_s}i\V{m}
\beta_{\kappa}}_\HP^2\Id s,\quad t\ge0.
\end{align*}
In fact, we may re-express $S_{\kappa,s}^{[\tau]}$, $s\in[0,\tau)$, by means of the terms in the first 
line of \eqref{lisa2} to $\PP$-a.s. obtain
\begin{align*}
\int_0^{t}&1_{(0,\tau)}(s)\SPn{S_{\kappa,s}^{[\tau]}}{
e^{-(\tau-s)\omega -i\V{m}\cdot\V{B}_s}i\V{m}\beta_{\kappa}}_\HP^2\Id s
\\
&\le2\int_0^{t\wedge\tau}
\SPb{(e^{-(\tau-s)\omega-i\V{m}\cdot\V{B}_s}-e^{-\tau\omega})\beta_{\kappa}}{
e^{-(\tau-s)\omega-i\V{m}\cdot\V{B}_s}i\V{m}\beta_{\kappa}}_\HP^2\Id s
\\
&\quad+2\int_0^{t\wedge\tau}\SPB{\int_{0}^se^{-(\tau-r)\omega-i\V{m}\cdot\V{B}_r}f_{\kappa}\Id r}{
e^{-(\tau-s)\omega-i\V{m}\cdot\V{B}_s}i\V{m}\beta_{\kappa}}_\HP^2\Id s
\\
&\le8\int_0^\tau\big\|e^{-(\tau-s)\omega/2}|\V{m}|^\eh\beta_{\kappa}\big\|_\HP^4\Id s
\\
&\quad+2\int_0^\tau\bigg(\ee^2\int_0^\tau\int_{\RR^3}
\frac{e^{-(2\tau-r-s)\omega}|\V{m}|}{
\omega(\omega+\V{m}^2/2)}\Id\lambda^3\Id r\bigg)^2\Id s,\quad t\ge0.
\end{align*}
Simple substitutions show that the terms in the last two lines of the previous estimation are 
identical to the terms in the last two lines of \eqref{lisa77}, if we insert $\tau$ for $t$ in \eqref{lisa77}.
In particular, we may conclude that 
$\llbracket N_{\kappa}^{[\tau]}\rrbracket\le c\ee^4$, for some universal constant $c>0$.

In a similar fashion we deduce that
\begin{align}\nonumber
|L_{\kappa,t}^{[\tau]}|&\le2\int_0^{t\wedge\tau}
\|e^{-(\tau-s)\omega/2}\omega^\eh\beta_\kappa\|_\HP^2\Id s
\\\nonumber
&\quad+4\pi\ee^2\int_0^\infty\int_0^{t\wedge\tau}\int_0^{t\wedge\tau}
\frac{e^{-(2\tau-r-s)\omega_\rho}\omega_\rho\rho^2}{\omega_\rho
(\omega_\rho+\rho^2/2)}\Id r\,\Id s\,\Id\rho
\\\nonumber
&\le4\pi\ee^2\int_0^\infty\int_0^{t\wedge\tau}\frac{e^{-(t\wedge\tau-s)\omega_\rho}\rho^2}{(\omega_\rho
+\rho^2/2)^2}\Id s\,\Id\rho
\\\nonumber
&\quad+4\pi\ee^2\int_0^\infty\int_0^{t\wedge\tau}\int_0^{t\wedge\tau}
\frac{e^{-(2(t\wedge\tau)-r-s)\omega_\rho}\rho^2}{
(\omega_\rho+\rho^2/2)}\Id r\,\Id s\,\Id\rho
\\\nonumber
&\le4\pi\ee^2\int_0^\infty\frac{(1-e^{-(t\wedge\tau)\omega_\rho})\rho^2}{\omega_\rho(\omega_\rho
+\rho^2/2)^2}\Id\rho
+4\pi\ee^2\int_0^\infty
\frac{(1-e^{-(t\wedge\tau)\omega_\rho})\rho^2}{
\omega_\rho^2(\omega_\rho+\rho^2/2)}\,\Id\rho
\\\nonumber
&\le2(4\pi)\ee^2\int_0^{1\wedge\frac{1}{t\wedge\tau}}(t\wedge\tau)\Id\rho
+2(4\pi)\ee^2\int_{1\wedge\frac{1}{t\wedge\tau}}^1
\frac{\Id\rho}{\rho}+2(4\pi)\ee^2\int_1^\infty\frac{2}{\rho^2}\Id\rho
\\\label{julia1}
&=8\pi\ee^2\big(3+\ln[1\vee(t\wedge\tau)]\big),
\end{align}
for all $t\ge0$. We may now conclude the proof of \eqref{expbdStau}
as in the end of the proof of Lem.~\ref{lemexpMminus}.

Finally, we observe similarly as in \eqref{julia1} that
\begin{align}\nonumber
\Big\|\int_0^t&1_{(0,\tau)}(s)e^{-(\tau-s)\omega-i\V{m}\cdot\V{B}_s}
\omega\beta_\kappa\Id s\Big\|_\HP^2
\\\nonumber
&\le4\pi\ee^2\int_0^\infty\int_0^{t\wedge\tau}\int_0^{t\wedge\tau}
\frac{e^{-(2\tau-r-s)\omega_\rho}\omega_\rho^2\rho^2}{
\omega_\rho(\omega_\rho+\rho^2/2)^2}\Id r\,\Id s\,\Id\rho
\\\label{mikkel}
&\le4\pi\ee^2\big(3+\ln[1\vee(t\wedge\tau)]\big),\quad t\ge0.
\end{align}
In fact, the triple integral in the second line of \eqref{mikkel} is bounded from above by
the triple integral in the fourth line of \eqref{julia1}. Thus, 
$\|M_{\kappa,t}^{[\tau]}\|_\HP^2\le2\|S_{\kappa,t}^{[\tau]}\|_\HP^2
+32\pi\ee^2(3+\ln[1\vee(t\wedge\tau)])$, 
which together with \eqref{expbdStau} implies \eqref{expbdMtau}.
\end{proof}

\begin{lem}\label{lembdMtt}
There exist universal constants $b,b',c,c'>0$ 
such that, for all $\kappa\in\NN\cup\{\infty\}$, $p>0$, and $t\ge0$,
\begin{align}\label{expbdMtt}
\EE\Big[\sup_{s,\tau\le t}e^{p\|M_{\kappa,s}^{[\tau]}[\V{B}]\|_\HP^2}\Big]
&\le b'\big(1+p\ee^2(1\vee t)\big)^4e^{bp\ee^2(1+\ln(1\vee t))+bp^2\ee^4},
\\\label{expbdUplus}
\EE\Big[\sup_{s\le t}e^{p\|1_{\Omega\setminus\sN_+}U_{\kappa,s}^{+}[\V{B}]\|_s^2}\Big]
&\le c'\big(1+p\ee^2(1\vee t)\big)e^{cp\ee^2(1+\ln(1\vee t))+cp^2\ee^4}.
\end{align}
\end{lem}

\begin{proof}
To prove \eqref{expbdMtt} we may assume without loss of generality that $p=4$. 
(Otherwise replace $\ee$ by $p^\eh\ee/2$.) For all $\tau,t,T\ge0$,
\begin{align*}
&\EE\Big[\sup_{s\le T}\big|e^{\|M_{\kappa,s}^{[\tau]}\|_\HP^2}
-e^{\|M_{\kappa,s}^{[t]}\|_\HP^2}\big|^4\Big]
\\
&\le\sup_{\sigma\ge0}\EE\Big[\sup_{s\le T}e^{12\|M_{\kappa,s}^{[\sigma]}
\|_\HP^2}\Big]^{\nf{2}{3}}\EE\Big[\sup_{s\le T}\big|\|M_{\kappa,s}^{[\tau]}\|_\HP^2
-\|M_{\kappa,s}^{[t]}\|_\HP^2\big|^{12}\Big]^{\nf{1}{3}}
\\
&\le ce^{c\ee^2(1+\ln(1\vee T))+c\ee^4}
\sup_{\sigma\ge0}\EE\Big[\sup_{s\le T}\|M_{\kappa,s}^{[\sigma]}\|_\HP^{24}\Big]^{\nf{1}{6}}
\EE\Big[\sup_{s\le T}\|M_{\kappa,s}^{[\tau]}-M_{\kappa,s}^{[t]}\|_\HP^{24}\Big]^{\nf{1}{6}}
\\
&\le c'|\ee|^{8}e^{c\ee^2(1+\ln(1\vee T))+c\ee^4}|\tau-t|^{\nf{4}{3}},
\end{align*}
where we dropped all arguments $[\V{B}]$, employed \eqref{expbdMtau} in the second step,
and made use of \eqref{polbdMtaua}, \eqref{polbdMtaub}, and \eqref{KolMtau} in the third one. 
The bound \eqref{expbdMtt} (with $p=4$) then follows from \eqref{expbdMtau}, 
Rem.~\ref{remKol} on Kolmogorov's lemma, and an argument similar to the one in
Step~3 of the proof of Lem.~\ref{lemMtau}.

Next, we employ \eqref{mona0} and \eqref{relUplus}
(multiplied with $1_{\{|\V{m}|\ge1\}}$) to derive the following analog of \eqref{karl1},
\begin{align}\nonumber
\EE\big[&\sup_{s\le t} e^{p\|1_{\Omega\setminus\sN_+}U_{\kappa,s}^+[\V{B}]\|_s^2}\big]
\le e^{cp\ee^2+8p\|\beta_{1,\infty}\|_\HP^2}
\EE\big[\sup_{s\le t}e^{2p\|1_{\{|\V{m}|<1\}}U_{\kappa,s}^+[\V{B}]\|_\HP^2}\big]^\eh
\\\nonumber
&\cdot\EE\big[\sup_{s\le t}e^{16p\|M_{\kappa,s}^{[s]}[\V{B}]\|_{\HP}^2}\big]^{\nf{1}{4}}
\EE\bigg[\sup_{s\le t}\exp\bigg(16p\Big\|\int_0^se^{-(s-r)\omega-i\V{m}\cdot\V{B}_r}
2\omega\beta_\kappa\Id r\Big\|_\HP^2\bigg)\bigg]^{\nf{1}{4}},
\end{align}
for all $p,t\ge0$ and $\kappa\in\NN\cup\{\infty\}$. Here
\begin{align}\label{karl3}
\sup_{s\le t}\|1_{\{|\V{m}|<1\}}U_{\kappa,s}^+[\V{B}]\|_\HP^2&\le4\pi\ee^2(1+\ln(1\vee t)),
\end{align}
since the left hand side of \eqref{karl3} is bounded from above by the triple integral in the first
line of \eqref{karl2}. Combining these remarks with \eqref{mikkel} and
\eqref{expbdMtt} we finally arrive at \eqref{expbdUplus}.
\end{proof}


\subsection{Basic $\HP$-valued processes for $N$ matter particles}\label{ssecUN}

As explained in Subsect.~\ref{ssecprobprel}, $\ul{\V{b}}$ denotes a $\nu$-dimensional Brownian 
motion with respect to $\BB$ with covariance matrix $\id_{\RR^\nu}$, and we think of its components 
as split into $N$ independent three-dimensional Brownian motions $\V{b}_1,\ldots,\V{b}_N$, so that
$\ul{\V{b}}=(\V{b}_1,\ldots,\V{b}_N)$. 

To start with we apply Lem.~\ref{lemUminus} and Lem.~\ref{lemUplus} to the stochastic basis
$\BB_{\mathrm{W}}^3$ and the three-dimensional canonical 
Brownian motion $\mathrm{pr}^3$ on it; recall Ex.~\ref{exWiener}. Then we obtain two 
$\PP_{\mathrm{W}}^3$-zero 
sets $\sN_{-,\mathrm{W}}$ and $\sN_{+,\mathrm{W}}$ on the complement of which we have the 
identities \eqref{defUminusinfty} and \eqref{relUplus}, respectively, with $\V{B}=\mathrm{pr}^3$.
With this we set
\begin{align*}
\sN_\infty&:=\sN_{+,\mathrm{W}}\cup\sN_{-,\mathrm{W}},
\qquad\sN_\kappa:=\emptyset,\;\;\kappa\in\NN.
\end{align*}

\begin{defn}\label{defUN} Let $\kappa\in\NN\cup\{\infty\}$.
\begin{enumerate}[leftmargin=*]
\item[{\rm(1)}] For all $\ul{\V{x}}=(\V{x}_1,\ldots,\V{x}_N)\in\RR^\nu$ and
$\ul{\V{\alpha}}=(\V{\alpha}_1,\ldots,\V{\alpha}_N)\in C([0,\infty),\RR^\nu)$, we set
\begin{align*}
U_{\kappa,t}^{N,\pm}[\ul{\V{x}},\ul{\V{\alpha}}]&:=\Big(\prod_{\ell=1}^N
1_{\Omega_{\mathrm{W}}\setminus\sN_\kappa}(\V{\alpha}_\ell)\Big)
\sum_{\ell=1}^Ne^{-i\V{m}\cdot\V{x}_\ell}U^\pm_{\kappa,t}[\V{\alpha}_\ell],\quad t\ge0.
\end{align*}
\item[{\rm(2)}]
For every $\fF_0$-measurable $\ul{\V{q}}=(\V{q}_1,\ldots,\V{q}_N):\Omega\to\RR^\nu$, 
we define two continuous, $(\fF_t)_{t\ge0}$-adapted, $\HP$-valued processes 
$U_{\kappa}^{N,-}(\ul{\V{q}})$ and $U_{\kappa}^{N,+}(\ul{\V{q}})$ by setting
\begin{align}\label{defUpmN}
U_{\kappa,t}^{N,\pm}(\ul{\V{q}})&:=U_{\kappa,t}^{N,\pm}[\ul{\V{q}},\ul{\V{b}}]=
1_{\ul{\V{b}}_\bullet^{-1}((\Omega_{\mathrm{W}}^3\setminus\sN_{\kappa})^N)}
\sum_{\ell=1}^Ne^{-i\V{m}\cdot\V{q}_\ell}U^\pm_{\kappa,t}[\V{b}_\ell],\quad t\ge0.
\end{align}
\end{enumerate}
\end{defn}

The notation introduced in the first part of the preceding definition is 
convenient for stating the transformation properties of the following remark, where we use
the notation for time-shifted objects introduced in \eqref{deftb} and \eqref{alphashiftrev}.

\begin{rem}\label{remrevUpmN}
\begin{enumerate}[leftmargin=*]
\item[{\rm(1)}]
For all $\kappa\in\NN$, $\ul{\V{x}}\in\RR^\nu$, $\ul{\V{\alpha}}\in C([0,\infty),\RR^\nu)$,
and $s,t\ge0$,
\begin{align}\label{revUpmN}
U_{\kappa,t}^{N,\pm}[\ul{\V{x}}+\ul{\V{\alpha}}_t,\ul{\V{\alpha}}_{t-\bullet}-\ul{\V{\alpha}}_t]
&=U_{\kappa,t}^{N,\mp}[\ul{\V{x}},\ul{\V{\alpha}}],
\\\label{shift1N}
e^{-t\omega}U_{\kappa,s}^{N,-}[\ul{\V{x}}+\ul{\V{\alpha}}_t,{}^t\ul{\V{\alpha}}]
+U_{\kappa,t}^{N,-}[\ul{\V{x}},\ul{\V{\alpha}}]
&=U_{\kappa,s+t}^{N,-}[\ul{\V{x}},\ul{\V{\alpha}}],
\\\label{shift2N}
U_{\kappa,s}^{N,+}[\ul{\V{x}}+\ul{\V{\alpha}}_t,{}^t\ul{\V{\alpha}}]
+e^{-s\omega}U_{\kappa,t}^{N,+}[\ul{\V{x}},\ul{\V{\alpha}}]
&=U_{\kappa,s+t}^{N,+}[\ul{\V{x}},\ul{\V{\alpha}}].
\end{align}
Here \eqref{revUpmN} follows from \eqref{revUpm}, while \eqref{shift1N} and \eqref{shift2N} are
consequences of \eqref{shift1} and \eqref{shift2}, respectively.
\item[{\rm(2)}] Let $t\ge0$. Employing \eqref{shift1}, \eqref{shift2}, and \eqref{defUpmN} we then
find some $\PP$-zero set $\sN$ such that the following identities
hold on $\Omega\setminus\sN$, for all $\fF_0$-measurabe 
$\ul{\V{q}}:\Omega\to\RR^\nu$ and $s\ge0$,
\begin{align}\label{shift1Ninfty}
e^{-t\omega}U_{\infty,s}^{N,-}[\ul{\V{q}}+\ul{\V{b}}_t,{}^t\ul{\V{b}}]
+U_{\infty,t}^{N,-}[\ul{\V{q}},\ul{\V{b}}]
&=U_{\infty,s+t}^{N,-}[\ul{\V{q}},\ul{\V{b}}],
\\\label{shift2Ninfty}
U_{\infty,s}^{N,+}[\ul{\V{q}}+\ul{\V{b}}_t,{}^t\ul{\V{b}}]
+e^{-s\omega}U_{\infty,t}^{N,+}[\ul{\V{q}},\ul{\V{b}}]
&=U_{\infty,s+t}^{N,+}[\ul{\V{q}},\ul{\V{b}}].
\end{align}
\end{enumerate}
\end{rem}

\begin{rem}\label{remUpmNcont}
Let $\kappa\in\NN\cup\{\infty\}$, $\ul{\V{\alpha}}\in C([0,\infty),\RR^\nu)$, $\gamma\in\Omega$,
and $\ul{\V{q}}:\Omega\to\RR^\nu$ be $\fF_0$-measurable. Recall the definition of $\mathfrak{k}$
in \eqref{deffrk} and that the canonical embedding $\mathfrak{k}\to\HP$ is continuous.
Then, by the above constructions, Rem.~\ref{remguenther}, Lem.~\ref{lemUminus}(1), and
Lem.~\ref{lemUplus}(1) the following maps are continuous,
\begin{align*}
[0,\infty)\times\RR^\nu\ni(t,\ul{\V{x}})\mapsto U_{\kappa,t}^{N,\pm}[\ul{\V{x}},\ul{\V{\alpha}}]
\in\mathfrak{k},\quad[0,\infty)\ni t\mapsto(U_{\kappa,t}^{N,\pm}(\ul{\V{q}}))(\gamma)\in\mathfrak{k},
\end{align*}
and $U_{\kappa}^{N,\pm}(\ul{\V{q}})$ is adapted as a $\mathfrak{k}$-valued process.
Taking \eqref{mona0} into account we further see that, for all $t>0$,
\begin{align}\label{tnbUpmN}
\sup_{\ul{\V{x}}\in\RR^\nu}\sup_{0<s\le t}\|U_{\kappa,s}^{N,\pm}[\ul{\V{x}},\ul{\V{\alpha}}]\|_s<\infty,
\quad\sup_{0<s\le t}\|(U_{\kappa,s}^{N,\pm}(\ul{\V{q}}))(\gamma)\|_s<\infty.
\end{align}
\end{rem}

\begin{cor}\label{corUpmN}
There exist universal constants $b,c>0$ such that
\begin{align}\nonumber
\sup_{\ul{\V{q}}}\EE\Big[&\sup_{s\le t}e^{p\|U_{\kappa,s}^{N,\pm}(\ul{\V{q}})\|_s^2}\Big]
\\\label{expbdUpmN}
&\le b^N\big(1+p\ee^2N(1\vee t)\big)^{N}e^{cp\ee^2N^2(1+\ln(1\vee t))+cp^2\ee^4N^3},
\end{align}
for all $p,t\ge0$ and $\kappa\in\NN\cup\{\infty\}$. Furthermore,
\begin{align}\label{convUpmN}
\sup_{\ul{\V{q}}}\EE\Big[\sup_{\tau\le s\le t}
\big\|U_{\kappa,s}^{N,\pm}(\ul{\V{q}})-U_{\infty,s}^{N,\pm}(\ul{\V{q}})\big\|_{s}^p\Big]
\xrightarrow{\;\;\kappa\to\infty\;\;}0,\quad p,\tau>0,\,t\ge\tau.
\end{align}
Here the suprema $\sup_{\ul{\V{q}}}$ are taken over all $\fF_0$-measurable
$\ul{\V{q}}:\Omega\to\RR^\nu$.
\end{cor}

\begin{proof}
In view of the bound
$\|U_{\kappa,s}^{N,\pm}(\ul{\V{q}})\|_s^2\le N\sum_{\ell=1}^N
\|1_{\Omega_{\mathrm{W}}^3\setminus\sN_\kappa}(\V{b}_\ell)U^\pm_{\kappa,s}[\V{b}_\ell]\|_s^2$
and the independence of the random variables
$\sup_{s\le t}e^{pN\|1_{\Omega_{\mathrm{W}}^3\setminus\sN_\kappa}(\V{b}_\ell) 
U^\pm_{\kappa,s}[\V{b}_\ell]\|_s^2}$ for $\ell\in\{1,\ldots,N\}$, we observe that
\begin{align*}
\EE\Big[\sup_{s\le t}e^{p\|U_{\kappa,s}^{N,\pm}(\ul{\V{q}})\|_s^2}\Big]
&\le\EE\bigg[\prod_{\ell=1}^N\sup_{s\le t}e^{pN\|U_{\kappa,s}^{\pm}[\V{b}_\ell]\|_s^2}\bigg]
=\EE\big[\sup_{s\le t}e^{pN\|U_{\kappa,s}^{\pm}[\V{b}_1]\|_s^2}\big]^N.
\end{align*}
Hence, \eqref{expbdUpmN} is a consequence of \eqref{expbdUminus} and \eqref{expbdUplus}.
On account of Minkowski's inequality, \eqref{convUpmN} follows from
Lem.~\ref{lemUminus}(2) and Lem.~\ref{lemUplus}(2).
\end{proof}


\subsection{A useful lemma}\label{sseckey}

\noindent
In this subsection we prove a technical lemma that will be crucially used in the proof of 
Lem.~\ref{lemkeith} below, which is the key step in the derivation of our exponential moment bound
on the complex action. The lemma deals with the following process: 

\begin{defn}\label{defnthalia}
For all $\kappa\in\NN\cup\{\infty\}$ and $\fF_0$-measurable $\ul{\V{q}}:\Omega\to\RR^\nu$, we set
\begin{align}\label{thalia2}
S_{\kappa,t}^{N}(\ul{\V{q}})&:=\sum_{j=1}^Ne^{-i\V{m}\cdot\V{q}_j}S_{\kappa,t}^{[t]}[\V{b}_j]=
\sum_{j=1}^Ne^{-i\V{m}\cdot\V{q}_j}(M_{\kappa,t}^{[t]}[\V{b}_j]
-I_{\kappa,t}[\V{b}_j]),\quad t\ge0.
\end{align}
\end{defn}

In \eqref{thalia2} we used notation introduced in \eqref{defMtau}, \eqref{defIalpha}, and
\eqref{defStt}. For all $\Lambda\ge0$ and $a\in (\mh,0)$, we further abbreviate
\begin{align*}
\mho_\Lambda(a)&:=4\pi\bigg(\frac{6(1-({1\wedge\Lambda})^{1-2|a|})}{1-2|a|}
+\frac{8}{(1+2|a|)(1\vee\Lambda)^{1+2|a|}}+\frac{1}{|a|(1\vee\Lambda)^{2|a|}}\bigg).
\end{align*}

The inequality \eqref{keith0} in the next lemma will be very useful later on because the power
$\omega^{a+\eh}$ in front of $S_{\kappa}^{N}(\ul{\V{q}})$ on its left hand side is replaced by the 
smaller power $\omega^a$ in the martingale $\ell_{\Lambda,\kappa}(\ul{\V{q}})$.

\begin{lem}\label{lemthalia}
Let $a\in(\mh,0)$, $\kappa\in\NN\cup\{\infty\}$, and $\ul{\V{q}}:\Omega\to\RR^\nu$ 
be $\fF_0$-measurable.
Let $\Lambda\ge0$ and $\vr_\Lambda:\RR^3\to[0,1]$ be the characteristic function of the set
$\{|\V{m}|\ge\Lambda\}$. Then, $\PP$-a.s.,
\begin{align}\nonumber
\|\vr_{\Lambda}\omega^{a}{S}_{\kappa,t}^{N}(\ul{\V{q}})\|_{\HP}^2
&=-2\int_0^t\|\vr_{\Lambda}\omega^{a+\eh}{S}_{\kappa,s}^{N}(\ul{\V{q}})\|_{\HP}^2\Id s
+2a_{\Lambda,\kappa,t}(\ul{\V{q}})
\\\label{kathrin1b}
&\quad+2\ell_{\Lambda,\kappa,t}(\ul{\V{q}})+Nt\|\vr_\Lambda
\omega^{a}i\V{m}\beta_\kappa\|_{\HP}^2,\quad t\ge0,
\end{align}
with
\begin{align}\label{kathrin99b}
a_{\Lambda,\kappa,t}(\ul{\V{q}})&
:=\sum_{j=1}^N\int_0^t\SPn{\vr_{\Lambda}\omega^{a}S_{\kappa,s}^{N}(\ul{\V{q}})}{
\vr_{\Lambda}e^{-i\V{m}\cdot(\V{q}_j+\V{b}_{j,s})}2\omega^{1+a}\beta_\kappa}_{\HP}\Id s,
\\\label{kathrin100}
\ell_{\Lambda,\kappa,t}(\ul{\V{q}})
&:=\sum_{j=1}^N\int_0^t\SPn{\vr_{\Lambda}\omega^{a}S_{\kappa,s}^{N}(\ul{\V{q}})}{\vr_{\Lambda}
e^{-i\V{m}\cdot(\V{q}_j+\V{b}_{j,s})}\omega^{a}i\V{m}\beta_\kappa}_{\HP}\Id\V{b}_{j,s}.
\end{align}
Moreover, the following inequality $\PP$-a.s. holds for all $t\ge0$,
\begin{align}\label{keith0}
\int_0^t\|\vr_{\Lambda}\omega^{a+\eh}{S}_{\kappa,s}^{N}(\ul{\V{q}})\|_{\HP}^2\Id s
&\le \ell_{\Lambda,\kappa,t}(\ul{\V{q}})+2\mho_{\Lambda}(a)\ee^2N^2t+\frac{Nt}{2}\|\vr_\Lambda
\omega^{a}i\V{m}\beta_\kappa\|_{\HP}^2.
\end{align}
Since $\ell_{\Lambda,\kappa}(\ul{\V{q}})$ is a continuous $L^2$-martingale, it further entails, 
for all $t\ge0$,
\begin{align}\label{keithE}
\EE\Big[\int_0^t\|\vr_{\Lambda}\omega^{a+\eh}{S}_{\kappa,s}^{N}(\ul{\V{q}})\|_{\HP}^2\Id s\Big]
&\le 2\mho_{\Lambda}(a)\ee^2N^2t+\frac{Nt}{2}\|\vr_\Lambda\omega^{a}i\V{m}\beta_\kappa\|_{\HP}^2.
\end{align}
\end{lem}

\begin{proof}
{\em Step~1.} Employing \eqref{relUplus} to re-write the integrands
in \eqref{kathrin99b} and \eqref{kathrin100} and taking Lem.~\ref{lemUinftyint} into account,
we see that $\ell_{\Lambda,\kappa}(\ul{\V{q}})$ is indeed a continuous $L^2$-martingale 
$\PP$-a.s. satisfying
\begin{align}\label{thalia55}
\llbracket\ell_{\Lambda,\kappa}(\ul{\V{q}})\rrbracket_t
&=\sum_{j=1}^N
\int_0^t\SPb{\omega^{a}S_{\kappa,s}^{N}(\ul{\V{q}})}{
\vr_{\Lambda} e^{-i\V{m}\cdot(\V{q}_j+\V{b}_{j,s})}\omega^{a}i\V{m}\beta_\kappa}^2\Id s,
\end{align}
and we $\PP$-a.s. find the bounds
\begin{align}\label{thalia56}
|a_{\Lambda,\kappa,t}(\ul{\V{q}})|&\le 2G_{\Lambda,t}(a)\ee^2N^2t,\quad
\llbracket\ell_{\Lambda,\kappa}(\ul{\V{q}})\rrbracket_t\le G_{\Lambda,t}(a)^2\ee^4N^3t,
\quad t\ge0,
\end{align}
where
\begin{align}\nonumber
G_{\Lambda,t}(a)&:=2\|\vr_\Lambda\omega^{a+\nf{1}{2}}\beta_\infty\|_{\HP}^2
+\int_{\RR^3}\int_0^t\frac{e^{-(t-s)\omega}}{\omega^{2|a|}(\omega+\V{m}^2/2)}\Id s\Id\lambda^3
\\\label{thalia57}
&\le4\pi\ee^2\bigg(3\int_{1\wedge\Lambda}^1\frac{\Id\rho}{\rho^{2|a|}}+8\int_{1\vee\Lambda}^\infty
\frac{\Id\rho}{\rho^{2+2|a|}}+2\int_{1\vee\Lambda}^\infty\frac{\Id\rho}{\rho^{1+2|a|}}\bigg)
=\mho_\Lambda(a).
\end{align}
In particular, we see that \eqref{kathrin1b} implies \eqref{keith0}.

\smallskip

\noindent
{\em Step~2.}
Let $\theta_n$ denote the characteristic function of the open ball of radius $n$ about $0$ in $\RR^3$
and abbreviate
\begin{align*}
{Z}_{\kappa,t}^{[n]}[\V{b}_j]
&:=\int_0^te^{s\omega-i\V{m}\cdot\V{b}_{j,s}}\theta_n\omega^{a}
i\V{m}\beta_\kappa\Id\V{b}_{j,s},\quad t\ge0,\,n\in\NN,\,j\in\{1,\ldots,N\}.
\end{align*}
Then we find a $\PP$-zero set $\sN\in\fF$ such that, for all 
$(t,\tau)\in[0,\infty)\times(\QQ\cap[0,\infty))$ with $t<\tau$, $n\in\NN$,
and $j\in\{1,\ldots,N\}$,
\begin{align}\label{felix3}
\theta_n\omega^{a}{M}_{\kappa,t}^{[\tau]}[\V{b}_j]
&=e^{-\tau\omega}{Z}_{\kappa,t}^{[n]}[\V{b}_j],\quad\text{on $\Omega\setminus\sN$.}
\end{align}
Here both sides are continuous in $(t,\tau)\in[0,\infty)^2$ at every point of
$\Omega$, whence the statement \eqref{felix3} is actually valid for all $0\le t\le\tau$, $n\in\NN$,
and $j\in\{1,\ldots,N\}$.

For every $h\in\dom(\omega^2)$, the function 
$[0,\infty)\times\HP_{\RR}\ni(t,g)\mapsto\SPn{e^{-t\omega}h}{g}_{\HP}$ is twice continuously 
differentiable and its partial derivatives of order $\le2$ are bounded on bounded subsets of
$[0,\infty)\times\HP_{\RR}$, whence we may apply the It\={o} formula of 
\cite[Thm.~4.32]{daPrZa2014} to the process
$(\SPn{e^{-t\omega}h}{{Z}_{\kappa,t}^{[n]}[\V{b}_j]}_\HP)_{t\ge0}$.
Doing this for every $h$ in a countable dense
subset of $\dom(\omega^2)$, we $\PP$-a.s. deduce that
\begin{align*}
e^{-t\omega}{Z}_{\kappa,t}^{[n]}[\V{b}_j]
&=-\int_0^t\theta_n\omega e^{-s\omega}{Z}_{\kappa,s}^{[n]}[\V{b}_j]\Id s
+L_{\kappa,t}^{[n]}[\V{b}_j],\quad t\ge0,
\end{align*}
with the continuous $\HP_{\RR}$-valued $L^2$-martingale
\begin{align*}
L_{\kappa,t}^{[n]}[\V{b}_j]&:=\int_0^t e^{-i\V{m}\cdot\V{b}_{j,s}}\theta_n
\omega^{a}i\V{m}\beta_\kappa\Id\V{b}_{j,s},\quad t\ge0.
\end{align*}
In combination with \eqref{felix3} this shows that, $\PP$-a.s., 
\begin{align*}
\theta_n\omega^{a}{M}_{\kappa,t}^{[t]}[\V{b}_j]
&=-\int_0^t\omega\theta_n\omega^{a}{M}_{\kappa,s}^{[s]}[\V{b}_j]\Id s
+L_{\kappa,t}^{[n]}[\V{b}_j],\quad t\ge0.
\end{align*}
Taking also \eqref{derI} into account we $\PP$-a.s. arrive at
\begin{align}\nonumber
\theta_n\omega^{a}{S}_{\kappa,t}^{[t]}[\V{b}_j]
&=-\int_0^t\omega\theta_n\omega^{a}{S}_{\kappa,s}^{[s]}[\V{b}_j]\Id s
\\\label{thalia1}
&\quad-\theta_n\int_0^t2e^{-i\V{m}\cdot\V{b}_{j,s}}\omega^{1+a}\beta_\kappa\Id s
+L_{\kappa,t}^{[n]}[\V{b}_j],\quad t\ge0.
\end{align}
for all $n\in\NN$ and $j\in\{1,\ldots,N\}$. 
Set $\vr_{\Lambda,n}:=\vr_\Lambda\theta_n$ in what follows. Since the strongly $\fF_0$-measurable 
$\LO(\HP_{\RR})$-valued function $e^{-i\V{m}\cdot\V{q}_j}$ (i.e., 
$\Omega\ni\gamma\mapsto e^{-i\V{m}\cdot\V{q}_j(\gamma)}h$
is $\fF_0$-measurable, for every $h\in\HP_{\RR}$) commutes up to indistinguishability with the
stochastic integrations, we $\PP$-a.s. have
\begin{align*}
\sum_{j=1}^N
e^{-i\V{m}\cdot\V{q}_j}\vr_\Lambda L_{\kappa,t}^{[n]}[\V{b}_j]
&=L_{\Lambda,\kappa,t}^{N,[n]}(\ul{\V{q}}):=
\sum_{j=1}^N\int_0^te^{-i\V{m}\cdot(\V{q}_j+\V{b}_{j,s})}\vr_{\Lambda,n}
\omega^{a}i\V{m}\beta_\kappa\Id\V{b}_{j,s},
\end{align*}
for all $t\ge0$. From \eqref{thalia1} and the latter remark we infer that every process 
$(\vr_{\Lambda,n}\omega^{-\nf{1}{4}}{S}_{\kappa,t}^{N,+}(\ul{\V{q}}))_{t\ge0}$ 
with $n\in\NN$ is a continuous $\HP_{\RR}$-valued semi-martingale. 

Define $a_{\Lambda,\kappa,t}^{[n]}(\ul{\V{q}})$ and
$\ell_{\Lambda,\kappa,t}^{[n]}(\ul{\V{q}})$ upon replacing $\vr_\Lambda$ by 
$\vr_{\Lambda,n}$ in \eqref{kathrin99b} and \eqref{kathrin100}, respectively.
Then we may employ It\={o}'s formula in
combination with \eqref{thalia1} and \eqref{thalia2} to $\PP$-a.s. get
\begin{align}\nonumber
\|\vr_{\Lambda,n}\omega^{a}S_{\kappa,t}^{N,+}(\ul{\V{q}})\|_{\HP}^2
&=-2\int_0^t\|\vr_{\Lambda,n}\omega^{a+\eh}S_{\kappa,s}^{N,+}(\ul{\V{q}})\|_{\HP}^2\Id s
\\\label{kathrin1}
&\quad+2a_{\Lambda,\kappa,t}^{[n]}(\ul{\V{q}})+2\ell_{\Lambda,\kappa,t}^{[n]}(\ul{\V{q}})
+\llbracket L_{\Lambda,\kappa}^{N,[n]}(\ul{\V{q}})\rrbracket_t,
\end{align}
for all $t\ge0$ and $n\in\NN$. Here we further observe that 
\begin{align}\label{kathrin101}
\llbracket L_{\Lambda,\kappa}^{N,[n]}(\ul{\V{q}})\rrbracket_t
&=N\int_0^t\|\vr_{\Lambda,n}\omega^{a}i\V{m}\beta_\kappa\|^2\Id s.
\end{align}
By virtue of Lem.~\ref{lemMtau} and \eqref{kathrin101} we may 
employ the dominated convergence theorem to show that,
pointwise on $[0,\infty)\times\Omega$, all terms in \eqref{kathrin1} except for the local martingale
$\ell_{\Lambda,\kappa}^{[n]}$ converge, as $n\to\infty$, to the respective terms in \eqref{kathrin1b}.
Furthermore, the absolute values of the integrands in $\ell_{\Lambda,\kappa}^{[n]}(\ul{\V{q}})$ 
are dominated by a constant times the continuous,
adapted process $(\|\vr_\Lambda\omega^{a}{S}_{\kappa,s}^{N,+}(\ul{\V{q}})\|)_{s\ge0}$. 
According to the dominated convergence theorem for stochastic integrals \cite[Thm.~24.2]{Me1982} 
we thus find integers $1\le n_1<n_2<\ldots\,$ such that, $\PP$-a.s.,
\begin{align*}
\sup_{s\le t}|\ell_{\Lambda,\kappa,s}^{[n_j]}(\ul{\V{q}})-\ell_{\Lambda,\kappa,s}(\ul{\V{q}})|
\xrightarrow{\;\;j\to\infty\;\;}0,\quad t\ge0.
\end{align*}
Putting these remarks together we see that \eqref{kathrin1b} is $\PP$-a.s. satisfied.
\end{proof}


\section{The complex action in the Feynman-Kac formula}\label{secphase}

\noindent
The objective of this section is to study Feynman's {\em complex action} \cite{Feynman1949} in the
Nelson model, i.e., the stochastic process  given by the logarithm of the vacuum expectation values 
of the Feynman-Kac integrands introduced later on. For finite $\kappa$, the complex action is given as 
a well-defined triple Lebesgue integral, which becomes ill-defined when the ultra-violet cut-off is 
dropped. It is, however, possible to exploit the presence of oscillating terms in its integrand by means 
of non-stationary phase type arguments involving repeated applications of It\={o}'s formula. In this way 
the difference of the complex action for finite $\kappa$ and the renormalization energy 
$tNE_\kappa^\ren$ can be written as a sum of terms that are well-defined even for $\kappa=\infty$
and thus allow for a definition of the limiting complex action. 
(We shall subtract the renormalization term in our definition of the complex action right away.) 
As a tradeoff we encounter simple and double {\em stochastic} integrals among those terms.

These general ideas appeared already in Nelson's original, probabilistic approach to the 
renormalization of his operator \cite{Nelson1964proc} and have been revisited in \cite{GHL2014}. 
We shall present a somewhat novel implementation of the non-stationary phase expansions
in Subsect.~\ref{ssecukappa}, pointing out similarities and differences to the earlier work along the 
way. In Subsect.~\ref{ssecuconv} we prove a simple preliminary lemma on the convergence
of the complex action as $\kappa$ goes to infinity. Our main new result on the complex action is
an exponential moment bound with an improved right hand side in comparison to the earlier
literature \cite{BleyThesis,Bley2016,BleyThomas2015,GHL2014,Nelson1964proc}, which
eventually will lead to the lower bound on the spectrum in \eqref{GSEHintro}. 
It is derived in Subsect.~\ref{ssecuge}, where appropriate remarks on 
the earlier literature are given as well. In Subsect.~\ref{ssecuinfty} we finally discuss a few 
additional properties of the complex action needed to analyze our Feynman-Kac semi-groups.


\subsection{Definition of the complex action}\label{ssecukappa}

\noindent
The formula for the complex action $u_{\kappa,t}^N$ with a finite $\kappa$ introduced in the next 
definition goes back to Feynman's original work \cite{Feynman1949}, modulo obvious 
modifications due to the fact that Feynman used his formal path integrals to represented unitary 
groups instead of semi-groups. Nelson re-derived the formula in \cite{Nelson1964proc} by ``expressing
[Feynman's] result in the language of Markov processes''. A more recent and concise textbook 
presentation as well as many related remarks and references can be found in \cite{LHB2011}; see
also \cite[App.~1]{GMM2016} for a derivation in a more general setting. 

\begin{defn}\label{defu}
Let $\kappa\in\NN$ and $t\ge0$. For all $\ul{\V{x}}\in\RR^\nu$ and
$\ul{\V{\alpha}}\in C([0,\infty),\RR^\nu)$, we define
\begin{align*}
u_{\kappa,t}^N[\ul{\V{x}},\ul{\V{\alpha}}]
&:=\sum_{j,\ell=1}^N\int_0^t\SPb{e^{-i\V{m}\cdot\V{x}_j}U_{\kappa,s}^+[\V{\alpha}_j]}{
e^{-i\V{m}\cdot(\V{x}_\ell+\V{\alpha}_{\ell,s})}f_\kappa}_{\HP}\Id s-tNE_\kappa^{\ren}.
\end{align*}
For every $\fF_0$-measurable $\ul{\V{q}}:\Omega\to\RR^\nu$, we further set
\begin{align*}
u_{\kappa,t}^N(\ul{\V{q}})&:=u_{\kappa,t}^N[\ul{\V{q}},\ul{\V{b}}].
\end{align*}
\end{defn}

\begin{rem}\label{remrevu}
Let $\kappa\in\NN$, $t>0$, $\ul{\V{x}}\in\RR^\nu$, and $\ul{\V{\alpha}}\in C([0,\infty),\RR^\nu)$. 
Then elementary substitutions and Fubini's theorem imply that
\begin{align*}
\int_0^t\SPb{e^{-i\V{m}\cdot(\V{x}_j+\V{\alpha}_{j,t})}U_{\kappa,s}^+
[\V{\alpha}_{j,t-\bullet}-\V{\alpha}_{j,t}]}{e^{-i\V{m}\cdot(\V{x}_\ell
+\V{\alpha}_{\ell,t}+\V{\alpha}_{\ell,t-s}-\V{\alpha}_{\ell,t})}f_\kappa}_{\HP}\Id s&
\\
=\int_0^t\SPb{e^{-i\V{m}\cdot(\V{x}_j+\V{\alpha}_{j,r})}f_\kappa}{
e^{-i\V{m}\cdot\V{x}_\ell}U_{\kappa,r}^+[\V{\alpha}_\ell]}_{\HP}\Id r&,
\end{align*}
for all $j,\ell\in\{1,\ldots,N\}$. Since the scalar products under the integrals in the previous
relation are real by Rem.~\ref{rem-real}, we obtain
\begin{align}\label{revu}
u_{\kappa,t}^N[\ul{\V{x}}+\ul{\V{\alpha}}_t,\ul{\V{\alpha}}_{t-\bullet}-\ul{\V{\alpha}}_t]
=u_{\kappa,t}^N[\ul{\V{x}},\ul{\V{\alpha}}].
\end{align}
\end{rem}

Our next goal is to define the complex action for $\kappa=\infty$. As mentioned earlier, we have to 
exploit the oscillations in its integrand by repeated applications of It\={o}'s formula to find an expression 
for $u_\kappa^N(\ul{\V{q}})$ that is meaningful for $\kappa=\infty$ as well.
The first step is taken in the next lemma, where we essentially re-derive
a representation of the complex action employed by  Gubinelli et al. \cite{GHL2014}.
(In fact, setting $\tau=T=t$ in Eqns.~(2.24) and~(2.33) of \cite{GHL2014} one obtains an analogue of 
\eqref{foru} below.) While our proof uses It\={o}'s formula for a scalar-product of
$\HP_{\RR}$-valued semi-martingales, Gubinelli et al. apply the ordinary It\={o} formula and the 
stochastic Fubini theorem. 

Notice that, for $j=\ell$, the term in the second line of \eqref{foru} equals $tE_\kappa^{\ren}$.

\begin{lem}\label{lemYu}
Let $\kappa\in\NN$, $j,\ell\in\{1,\ldots,N\}$, and $\ul{\V{q}}:\Omega\to\RR^\nu$
be $\fF_0$-measurable. Then, $\PP$-a.s., the following identity holds, for all $t\ge0$,
\begin{align}\nonumber
\int_0^t&\SPn{e^{-i\V{m}\cdot\V{q}_j}U_{\kappa,s}^+[\V{b}_j]}{
e^{-i\V{m}\cdot(\V{q}_\ell+\V{b}_{\ell,s})}f_\kappa}_{\HP}\Id s
\\\nonumber
&=\int_0^t\SPn{e^{i\V{m}\cdot(\V{q}_\ell-\V{q}_j+\V{b}_{\ell,s}-\V{b}_{j,s})}
\omega^{-\nf{1}{2}}f_{\kappa}}{\omega^{\nf{1}{2}}\beta_{\kappa}}_{\HP}\Id s
\\\nonumber
&\quad-\SPb{\omega^{-\nf{1}{2}}e^{-i\V{m}\cdot\V{q}_j}U_{\kappa,t}^+[\V{b}_j]}{
\omega^{\nf{1}{2}}e^{-i\V{m}\cdot(\V{q}_\ell+\V{b}_{\ell,t})}\beta_{\kappa}}_{\HP}
\\\label{foru}
&\quad-\int_0^t\SPn{U_{\kappa,s}^+[\V{b}_j]}{e^{i\V{m}\cdot(\V{q}_j-\V{q}_\ell-\V{b}_{\ell,s})}
i\V{m}\beta_{\kappa}}_{\HP}\Id\V{b}_{\ell,s}.
\end{align}
\end{lem}

\begin{proof}
Since $\kappa$ is finite, $\omega^\eh\V{m}^2\beta_\kappa\in\HP$. Hence, by It\={o}'s 
formula, $e^{i\V{m}\cdot(\V{q}_j-\V{q}_\ell-\V{b}_{\ell})}\omega^\eh\beta_\kappa$ 
is a continuous $\HP_{\RR}$-valued semi-martingale $\PP$-a.s. satisfying
\begin{align}\nonumber
e^{i\V{m}\cdot(\V{q}_j-\V{q}_\ell-\V{b}_{\ell,t})}\omega^\eh\beta_\kappa
&=e^{i\V{m}\cdot(\V{q}_j-\V{q}_\ell)}\omega^\eh\beta_\kappa
-\frac{1}{2}\int_0^te^{i\V{m}\cdot(\V{q}_j-\V{q}_\ell-\V{b}_{\ell,s})}\omega^\eh\V{m}^2\beta_\kappa\Id s
\\\label{zacharias1}
&\quad
-\int_0^te^{i\V{m}\cdot(\V{q}_j-\V{q}_\ell-\V{b}_{\ell,s})}\omega^\eh i\V{m}\beta_\kappa\Id\V{b}_{\ell,s},
\quad t\ge0.
\end{align}
From Lem.~\ref{lemUpmcontHP} we further infer that
\begin{align}\label{zacharias2}
\omega^\mh U_{\kappa,t}^+[\V{b}_j]&=-\int_0^t(\omega^\eh U^+_{\kappa,s}[\V{b}_j]
-e^{-i\V{m}\cdot\V{b}_{j,s}}\omega^\mh f_\kappa)\Id s,\quad t\ge0,\;\text{on $\Omega$,}
\end{align}
which defines a continuously differentiable $\HP_{\RR}$-valued integral process. 
The claim now follows upon applying
It\={o}'s formula to $\SPn{\omega^\mh U_{\kappa}^+[\V{b}_j]}{
e^{i\V{m}\cdot(\V{q}_j-\V{q}_\ell-\V{b}_{\ell})}\omega^\eh\beta_\kappa}_{\HP}$ and taking into account
\eqref{zacharias1}, \eqref{zacharias2}, and $(\omega+\V{m}^2/2)\beta_\kappa=f_\kappa$.
\end{proof}

Interpreted in a suitable way, Formula \eqref{foru} can in principle be used to define 
$u_{\infty,t}^N(\ul{\V{q}})$. 
In fact, the term in the third line of \eqref{foru} makes sense immediately 
for $\kappa=\infty$ by Lem.~\ref{lemUinftyint}.
Moreover, $f_\infty\beta_\infty\in L^p(\RR^3)$, for every $p\in(1,\nf{3}{2})$, whence its
Fourier transform is in every $L^{q}(\RR^3)$ with $q\in(3,\infty)$. Hence, for $\kappa=\infty$,
we could replace the $\Id s$-integral on the right hand side of \eqref{foru} by 
the expression
$(2\pi)^{\nf{3}{2}}\int_0^t(f_\infty\beta_\infty)^\wedge(\V{q}_\ell-\V{q}_j+\V{b}_{\ell,s}-\V{b}_{j,s})\Id s$, 
which yields, however, a well-defined process only outside a $\ul{\V{q}}$-dependent $\PP$-zero set.
Likewise, to replace the stochastic integral in \eqref{foru}, we could try to define a stochastic integral
of $(2\pi)^{\nf{3}{2}}\int_0^s(e^{-(s-r)\omega}i\V{m}f_\infty\beta_\infty)^\wedge
(\V{q}_\ell-\V{q}_j+\V{b}_{\ell,s}-\V{b}_{j,r})\Id r$ with respect to $\V{b}_\ell$.
Analogous formulas are used indeed in \cite{GHL2014} to define (a diagonal part 
of) the limiting complex action.

For a definition of $u_{\infty,t}^N(\ul{\V{q}})$ we shall, however, pass to another representation of
$u_{\kappa,t}^N(\ul{\V{q}})$ with a closer resemblance to the one given by Nelson 
\cite[pp.~107/8]{Nelson1964proc}. The integrands appearing in it are more 
regular than the ones proposed in the preceding paragraph. It is derived as follows: 

First, we employ our formula \eqref{relUplus} and the definition \eqref{defStt} of
$S_{\kappa,t}^{[t]}$, which $\PP$-a.s. permit to write the stochastic integral in \eqref{foru} as
\begin{align}\nonumber
\int_0^t&\SPn{U_{\kappa,s}^+[\V{b}_j]}{e^{i\V{m}\cdot(\V{q}_j-\V{q}_\ell-\V{b}_{\ell,s})}
i\V{m}\beta_{\kappa}}_{\HP}\Id\V{b}_{\ell,s}
\\\nonumber
&=\int_0^t\SPn{e^{-s\omega}\omega^{\nf{1}{4}}\beta_\kappa}{
e^{i\V{m}\cdot(\V{q}_j-\V{q}_\ell-\V{b}_{\ell,s})}\omega^{-\nf{1}{4}}i\V{m}\beta_{\kappa}}_{\HP}
\Id\V{b}_{\ell,s}
\\\nonumber
&\quad-\int_0^t\SPn{\omega^{\nf{1}{4}}\beta_\kappa}{
e^{i\V{m}\cdot(\V{q}_j+\V{b}_{j,s}-\V{q}_\ell-\V{b}_{\ell,s})}\omega^{-\nf{1}{4}}
i\V{m}\beta_{\kappa}}_{\HP}\Id\V{b}_{\ell,s}
\\\label{alfie1}
&\quad-\int_0^t\SPn{\omega^{\nf{1}{4}}S_{\kappa,s}^{[s]}[\V{b}_j]}{
e^{i\V{m}\cdot(\V{q}_j-\V{q}_\ell-\V{b}_{\ell,s})}\omega^{-\nf{1}{4}}i\V{m}\beta_{\kappa}}_{\HP}
\Id\V{b}_{\ell,s},\quad t\ge0.
\end{align}
In view of Lem.~\ref{lemMtau}, Lem.~\ref{lemI}, and \eqref{defStt}, the three integrands on the
right hand side are adapted and continuous with locally bounded moments of any order, even when
$\kappa=\infty$. In the next step we will apply the identity \eqref{defUminusinfty} and It\={o}'s formula 
to the integrals in the second and third line of \eqref{alfie1}, respectively.

Before we do this we shall, however, introduce the notation employed in our final formula for
$u_{\kappa,t}^N(\ul{\V{q}})$. Let us first introduce the shorthand
\begin{align*}
\Theta^N_t[\ul{\V{x}},\ul{\V{\alpha}}]&:=
\Big|\sum_{\ell=1}^Ne^{-i\V{m}\cdot\V{x}_\ell}\Big|^2
+\Big|\sum_{\ell=1}^Ne^{-i\V{m}\cdot(\V{x}_\ell+\V{\alpha}_{\ell,t})}\Big|^2
-2e^{-t\omega}\Re\!\sum_{j,\ell=1}^N\!e^{i\V{m}\cdot(\V{x}_\ell-\V{x}_j-\V{\alpha}_{j,t})}.
\end{align*}
Notice that $\Theta^N_t[\ul{\V{x}},\ul{\V{\alpha}}]\ge0$ and 
$\Theta^N_t[\ul{\V{x}},\ul{\V{\alpha}}]\beta_\infty^2\in L^1(\RR^3)$, for all
$\ul{\V{x}}\in\RR^\nu$ and $\ul{\V{\alpha}}\in C([0,\infty),\RR^\nu)$, because
$\Theta^N_t[\ul{\V{x}},\ul{\V{\alpha}}]$ compensates for the infra-red singularity of $\beta_\infty$
in the case $\mu=0$.

For later use we also include the infra-red cut-off parameter $\Lambda\ge0$ in the next definition;
recall the notation $\beta_{\Lambda,\kappa}$, $\beta_{\Lambda,\kappa}^N(\ul{\V{x}})$, 
$U_{\kappa}^{N,\pm}[\ul{\V{x}},\ul{\V{\alpha}}]$, and $S_\kappa^N(\ul{\V{q}})$ introduced in
\eqref{deffsigmakappa}, \eqref{deffNkappa}, Def.~\ref{defUN}, and \eqref{thalia2}, respectively.

\begin{defn}\label{defnbcvm}
Let $\kappa\in\NN\cup\{\infty\}$, $\Lambda\ge0$, and $\ul{\V{q}}:\Omega\to\RR^\nu$ be
$\fF_0$-measurable. 
\begin{enumerate}[leftmargin=*]
\item[{\rm(1)}] Let $t\ge0$, $\ul{\V{x}}\in\RR^\nu$, and $\ul{\V{\alpha}}\in C([0,\infty),\RR^\nu)$. 
Then we define
\begin{align}\label{forbTheta}
b_{\Lambda,\kappa,t}^N[\ul{\V{x}},\ul{\V{\alpha}}]&:=\frac{1}{2}\int_{\RR^3}
\Theta_t^N[\ul{\V{x}},\ul{\V{\alpha}}]\beta_{\Lambda,\kappa}^2\Id\lambda^3,
\\\label{defcplus}
c_{\Lambda,\kappa,t}^{N,+}[\ul{\V{x}},\ul{\V{\alpha}}]
&:=\SPn{\omega^{-\nf{1}{2}}U_{\kappa,t}^{N,+}[\ul{\V{x}},\ul{\V{\alpha}}]}{
\omega^{\nf{1}{2}}\beta_{\Lambda,\kappa}^N(\ul{\V{x}}+\ul{\V{\alpha}}_t)}_{\HP},
\\\label{defcminus}
c_{\Lambda,\kappa,t}^{N,-}[\ul{\V{x}},\ul{\V{\alpha}}]
&:=\SPn{\omega^{\nf{1}{2}}\beta_{\Lambda,\kappa}^N(\ul{\V{x}})}{
\omega^{-\nf{1}{2}}U_{\kappa,t}^{N,-}[\ul{\V{x}},\ul{\V{\alpha}}]}_{\HP}.
\end{align}
Abbreviating
\begin{align*}
w_{\Lambda,\kappa}(\V{x})&:=\SPn{e^{i\V{m}\cdot\V{x}}\omega^\eh\beta_{\Lambda,\kappa}}{
\omega^\eh\beta_{\Lambda,\kappa}},\quad\V{x}\in\RR^3,
\\
v_{\Lambda,\kappa,t}^{[j,\ell]}[\ul{\V{x}},\ul{\V{\alpha}}]
&:=\int_0^tw_{\Lambda,\kappa}(\V{x}_\ell-\V{x}_j+\V{\alpha}_{\ell,s}-\V{\alpha}_{j,s})\Id s,
\quad j,\ell\in\{1,\ldots,N\}.
\end{align*}
we further define
\begin{align}\label{john1}
v_{\Lambda,\kappa,t}^N[\ul{\V{x}},\ul{\V{\alpha}}]
&:=2{\sum_{1\le j<\ell\le N}}v_{\Lambda,\kappa,t}^{[j,\ell]}[\ul{\V{x}},\ul{\V{\alpha}}].
\end{align}
If the so-obtained maps are composed with $(\ul{\V{q}},\ul{\V{b}}_{\bullet})$, then, as usual, we write
\begin{align*}
c_{\Lambda,\kappa,t}^{N,\pm}(\ul{\V{q}})&:=
c_{\Lambda,\kappa,t}^{N,\pm}[\ul{\V{q}},\ul{\V{b}}_{\bullet}],\;\;\text{and}\;\;
y_{\Lambda,\kappa}^N(\ul{\V{q}}):=c_{\Lambda,\kappa,t}^{N,\pm}[\ul{\V{q}},\ul{\V{b}}_{\bullet}],
\;\;\text{if $y$ is $b$ or $v$.}
\end{align*}
\item[{\rm(2)}] In view of Lem.~\ref{lemMtau}, Lem.~\ref{lemI}, and \eqref{defStt}, the formula
\begin{align}\label{forvecd}
\V{d}_{\Lambda,\kappa,t}^{[\ell]}(\ul{\V{q}})&:=
\SPn{\omega^{\nf{1}{4}}S_{\kappa,t}^{N}(\ul{\V{q}})}{e^{-i\V{m}\cdot(\V{q}_\ell+\V{b}_{\ell,t})}
\omega^{-\nf{1}{4}}i\V{m}\beta_{\Lambda,\kappa}},\quad t\ge0,
\end{align}
defines a continuous adapted $\RR^3$-valued process with locally bounded moments of any order, 
for every $\ell\in\{1,\ldots,N\}$. Therefore, up to indistinguishability,
\begin{align}\label{defm}
m_{\Lambda,\kappa,t}^N(\ul{\V{q}})
&:=\sum_{\ell=1}^N\int_0^t\V{d}_{\Lambda,\kappa,s}^{[\ell]}(\ul{\V{q}})\Id\V{b}_{\ell,s},\quad t\ge0.
\end{align}
defines a continuous, real-valued $L^2$-martingale $m_{\Lambda,\kappa}^N(\ul{\V{q}})$.
\item[{\rm(3)}] In the case $\Lambda=0$ we will drop the corresponding subscript in the notation for
all processes introduced in the previous two items, i.e., 
\begin{align*}
c_{\kappa,t}^{N,\pm}(\ul{\V{q}}):=c_{0,\kappa,t}^{N,\pm}(\ul{\V{q}}),\;\;\text{and}\;\;
y_{\kappa,t}^N(\ul{\V{q}})&:=y_{0,\kappa,t}^N(\ul{\V{q}}),\;\;\text{if $y$ is $b$, $v$, or $m$.}
\end{align*}
\end{enumerate}
\end{defn}

\begin{thm}\label{thm2ndIto}
Let $\kappa\in\NN$ and $\ul{\V{q}}:\Omega\to\RR^\nu$ be $\fF_0$-measurable. Then, $\PP$-a.s.,
\begin{align}\label{defuinfty}
u_{\kappa,t}^N(\ul{\V{q}})
&=-b^N_{\kappa,t}(\ul{\V{q}})+c_{\kappa}^{N,-}(\ul{\V{q}})-c_{\kappa}^{N,+}(\ul{\V{q}})
+v_{\kappa,t}^N(\ul{\V{q}})+m_{\kappa,t}^N(\ul{\V{q}}),\quad t\ge0.
\end{align}
\end{thm}

\begin{defn}\label{defnuinfty}
Let $\ul{\V{q}}:\Omega\to\RR^\nu$ be $\fF_0$-measurable. Then we define
$u_{\infty}^N(\ul{\V{q}})$ by the formula \eqref{defuinfty} with $\infty$ substituted for $\kappa$.
\end{defn}

\begin{rem}\label{remZerlegungu}
\begin{enumerate}[leftmargin=*]
\item[(1)]
Direct analogues of $b^N_{\kappa,t}(\ul{\V{q}})$, $v_{\kappa,t}^N(\ul{\V{q}})$, and the double
stochastic integral contributing to $m_{\kappa,t}^N(\ul{\V{q}})$ appear in Nelson's expansion of
the complex action \cite[Lem.~10 \& Lem.~14]{Nelson1964proc}. The remaining terms in 
\eqref{defuinfty} are, however, replaced by a double Lebesgue integral in \cite{Nelson1964proc}. 
While Nelson worked with
a slightly different choice of $\beta_\kappa$, this discrepancy might mainly be due to 
cancellations owing to time reversal and space reflection symmetries that Nelson observed
between the terms involving precisely one stochastic integration.  We were, frankly speaking, 
unable to follow his explanations at this point \cite[pp.~106/7]{Nelson1964proc} and did 
not observe analogous cancellations.
\item[(2)] 
Bley \cite{Bley2016} employs the Clark-Ocone formula to expand the complex action for finite 
$\kappa$. Further elaboration on the Clark-Ocone expansion might possibly lead to another 
meaningful expression for the limiting complex action.
\item[(3)]
The formula \eqref{defuinfty} is convenient for comparing $u_\kappa^N(\ul{\V{q}})$ with the
complex action in the non-Fock representation. There, the first two terms on the right hand side
of \eqref{defuinfty} get an extra pre-factor $2$, while the third one disappears, and the last two
remain unchanged.
\end{enumerate}
\end{rem}

\begin{proof}[Proof of Thm.~\ref{thm2ndIto}]
On account of Lem.~\ref{lemMminus} the integral in the second line of \eqref{alfie1} can
$\PP$-a.s. be written as
\begin{align}\nonumber
\int_0^t&\SPn{e^{-s\omega}\omega^{\nf{1}{4}}\beta_\kappa}{
e^{i\V{m}\cdot(\V{q}_j-\V{q}_\ell-\V{b}_{\ell,s})}\omega^{-\nf{1}{4}}i\V{m}\beta_{\kappa}}_{\HP}
\Id\V{b}_{\ell,s}
\\\nonumber
&=\SPn{\omega^{\nf{1}{4}}e^{-i\V{m}\cdot\V{q}_j}\beta_\kappa}{\omega^{-\nf{1}{4}}
e^{-i\V{m}\cdot\V{q}_\ell}M_{\kappa,t}^-[\V{b}_\ell]}_{\HP}
\\\nonumber
&=\SPn{\omega^{\nf{1}{4}}\beta_\kappa}{
e^{i\V{m}\cdot(\V{q}_j-\V{q}_\ell)}(1-e^{-t\omega-i\V{m}\cdot\V{b}_{\ell,t}})
\omega^{-\nf{1}{4}}\beta_\kappa}_{\HP}
\\\label{alfie1999}
&\quad-\SPn{\omega^\eh e^{-i\V{m}\cdot\V{q}_j}\beta_\kappa}{
\omega^{\mh}e^{-i\V{m}\cdot\V{q}_\ell}U_{\kappa,t}^-[\V{b}_\ell]},\quad t\ge0.
\end{align}
Here we also applied \eqref{defUminusinfty} in the second step.
Furthermore, we observe that the integral in the third line of \eqref{alfie1} vanishes for
$j=\ell$, because $\V{m}$ is odd and $\beta_\kappa$ is even.
To treat the off-diagonal terms we exploit that
the scalar products under the stochastic integrations are all real, from which we $\PP$-a.s. infer that
\begin{align}\nonumber
-&\int_0^t\SPn{\omega^{\nf{1}{4}}\beta_\kappa}{
e^{i\V{m}\cdot(\V{q}_j+\V{b}_{j,s}-\V{q}_\ell-\V{b}_{\ell,s})}\omega^{-\nf{1}{4}}
i\V{m}\beta_{\kappa}}_{\HP}\Id\V{b}_{\ell,s}
\\\nonumber
&-\int_0^t\SPn{\omega^{\nf{1}{4}}\beta_\kappa}{
e^{i\V{m}\cdot(\V{q}_\ell+\V{b}_{\ell,s}-\V{q}_j-\V{b}_{j,s})}\omega^{-\nf{1}{4}}
i\V{m}\beta_{\kappa}}_{\HP}\Id\V{b}_{j,s}
\\\nonumber
&=\int_0^t\SPn{\omega^{\nf{1}{4}}\beta_\kappa}{
e^{i\V{m}\cdot(\V{q}_j+\V{b}_{j,s}-\V{q}_\ell-\V{b}_{\ell,s})}\omega^{-\nf{1}{4}}
i\V{m}\beta_{\kappa}}_{\HP}\Id(\V{b}_{j}-\V{b}_{\ell})_s
\\\nonumber
&=\SPn{\omega^{\nf{1}{4}}\beta_\kappa}{e^{i\V{m}\cdot(\V{q}_j-\V{q}_\ell)}
(e^{i\V{m}\cdot(\V{b}_{j,t}-\V{b}_{\ell,t})}-1)\omega^{-\nf{1}{4}}\beta_{\kappa}}_{\HP}
\\\label{alfie2000}
&\quad+\int_0^t\SPn{\omega^{\nf{1}{4}}\beta_\kappa}{
e^{i\V{m}\cdot(\V{q}_j+\V{b}_{j,s}-\V{q}_\ell-\V{b}_{\ell,s})}\omega^{-\nf{1}{4}}\V{m}^2
\beta_{\kappa}}_{\HP}\Id s,\quad t\ge0,\,j<\ell.
\end{align}
In the second step we applied It\={o}'s formula to the twice continuously differentiable function 
$(\V{x},\V{y})\mapsto\SPn{\omega^{\nf{1}{4}}\beta_\kappa}{e^{i\V{m}\cdot\V{x}}
(e^{i\V{m}\cdot\V{y}}-1)\omega^{-\nf{1}{4}}\beta_{\kappa}}_{\HP}$ and the $\RR^6$-valued process
$(\V{q}_j-\V{q}_\ell,\V{b}_{j,t}-\V{b}_{\ell,t})_{t\ge0}$, exploiting that
$|e^{i\V{m}\cdot\V{y}}-1|\omega^{-\nf{1}{4}}\beta_\kappa\in\HP$.
Next, we substitute $\V{m}^2\beta_\kappa=2f_\kappa-2\omega\beta_\kappa$ in the last
line of \eqref{alfie2000}; the contribution with $2f_\kappa$ will eventually cancel the term in the
second line of \eqref{foru}. In fact, 
after combining \eqref{alfie1999} and \eqref{alfie2000} with \eqref{foru} and 
\eqref{alfie1}, elementary rearrangements yield the asserted identity \eqref{defuinfty}. 
(To replace double summations subject to the condition $j<\ell$ by double summations
subject to $j\not=\ell$, we again exploit that $\omega$, $f_\kappa$, and $\beta_\kappa$ are even
while $\V{m}$ is odd.)
\end{proof}


\subsection{Convergence properties of the complex action}\label{ssecuconv}

The process $u_{\infty}^N(\ul{\V{q}})$ introduced in Def.~\ref{defuinfty} is indeed the limiting
complex action:

\begin{prop}\label{lemuinfty}
Let $\kappa\in\NN$, $p,t>0$, and $\ul{\V{q}}:\Omega\to\RR^\nu$ be $\fF_0$-measurable. Then
\begin{align}\label{felixu}
\EE\Big[\sup_{s\le t}|u_{\kappa,s}^N(\ul{\V{q}})-u_{\infty,s}^N(\ul{\V{q}})|^p\Big]
&\le c_p\frac{(\ee^2N^2(1\vee t))^{p}}{\kappa^{\nf{p}{4}}},
\end{align}
where $c_p>0$ depends only on $p$.
\end{prop}

\begin{proof}
The relation \eqref{felixu} follows upon combining the bounds on the various terms on the right hand
side of \eqref{defuinfty} derived in the following four steps. For a start, let us observe that, by
Hyp.~\ref{hypNelson}(2),
$0\le1-\chi_\kappa^2(\V{k})\le2|\V{k}|/\kappa$, if $\V{k}\in\RR^3$ satisfies $|\V{k}|<\kappa$, whence
\begin{align}\label{wendelin0}
\int_{\RR^3}\frac{\omega^\iota}{2^\iota}(1-\chi_\kappa^2)\beta_{\infty}^2\Id\lambda^3
\le8\pi\ee^2\bigg(
\int_0^1\frac{\Id\rho}{\kappa}+\int_1^\kappa\frac{2\Id\rho}{\kappa\rho}
+\int_\kappa^\infty\frac{\Id\rho}{\rho^2}\bigg),\;\,\iota\in\{0,1\}.
\end{align}
{\em Step~1.} First, we claim that
\begin{align}\label{wendelin1}
|b_{\kappa,t}^N(\ul{\V{q}})-b_{\infty,t}^N(\ul{\V{q}})|&\le64\pi\ee^2N^2({1+\ln(\kappa)})\big/{\kappa}.
\end{align}
In fact, the left hand side of \eqref{wendelin1} equals
$|\int_{\RR^3}\Theta_t^N(\ul{\V{q}})(1-\chi_\kappa^2)\beta_\infty^2\Id\lambda^3|$, where
$\Theta_t^N(\ul{\V{q}})$ is defined in front of Def.~\ref{defnbcvm}.
Hence, \eqref{wendelin1} follows from \eqref{wendelin0} with $\iota=0$.

\smallskip

\noindent
{\em Step~2.} Next, we assert that
\begin{align}\label{wendelin1b}
|c_{\kappa,t}^{N,\pm}(\ul{\V{q}})-c_{\infty,t}^{N,\pm}(\ul{\V{q}})|
&\le c\ee^2N^2t^\eh(1+\ln(\kappa))^\eh\big/\kappa^\eh.
\end{align}
These bounds follow from the identities
$\omega^\mh U_{\kappa,t}^{N,\pm}(\ul{\V{q}})=\chi_\kappa\omega^\mh U_{\infty,t}^{N,\pm}(\ul{\V{q}})$,
the bound \eqref{mona0}, and the fact that $\|\omega^\eh(1-\chi_\kappa^2)\beta_\infty\|_\HP^2/2$ 
is bounded from above by the left hand side of \eqref{wendelin0} with $\iota=1$.

\smallskip

\noindent
{\em Step~3.}
By estimating all terms on its left hand side trivially and taking \eqref{wendelin0} with $\iota=1$
into account we further obtain
\begin{align*}
|v_{\kappa,t}^N(\ul{\V{q}})-v_{\infty,t}^N(\ul{\V{q}})|&\le32\pi\ee^2N(N-1)t(1+\ln(\kappa))\big/\kappa.
\end{align*}
{\em Step~4.}
Finally, we shall derive the bound
\begin{align}\label{achim1}
\EE\Big[\sup_{s\le t}|m^N_{\kappa,s}(\ul{\V{q}})-m^N_{\infty,s}(\ul{\V{q}})|^p\Big]
&\le c_p(\ee^4N^3t)^{\nf{p}{2}}\big/\kappa^{\nf{p}{4}}.
\end{align}
We find some $\PP$-zero set $\sN$ such that,
for all $j\in\{1,\ldots,N\}$, $t\in[0,\infty)\cap\QQ$, and $\gamma\in\Omega\setminus\sN$ the relation
$S_{\kappa,t}^{[t]}[\V{b}_j](\gamma)=\chi_\kappa S_{\infty,t}^{[t]}[\V{b}_j](\gamma)$ holds true.
By continuity (recall Lem.~\ref{lemMtau}(2)) we conclude that it actually holds, for all
$t\ge0$ and $\gamma\in\Omega\setminus\sN$. Therefore,
\begin{align*}
&\EE\Big[\sup_{s\le t}|m^N_{\kappa,s}(\ul{\V{q}})-m^N_{\infty,s}(\ul{\V{q}})|^p\Big]
\\
&\le c_p\EE\Big[\sum_{\ell=1}^N\int_0^t
\SPb{\omega^{\nf{1}{4}}S_{\infty,s}^{N}(\ul{\V{q}})}{e^{-i\V{m}\cdot(\V{q}_\ell+\V{b}_{\ell,s})}
(1-\chi_\kappa^2)\omega^{-\nf{1}{4}}i\V{m}\beta_\infty}^2\Id s\Big]^{\nf{p}{2}}
\\
&\le c_pN^{\nf{p}{2}}\|\omega^{-\nf{1}{4}}(1-\chi_\kappa^2)i\V{m}\beta_\infty\|_\HP^p
\EE\Big[\int_0^t\|\omega^{\nf{1}{4}}S_{\infty,s}^{N}(\ul{\V{q}})\|_{\HP}^{2}\Id s\Big]^{\nf{p}{2}}
\end{align*}
Inequality \eqref{achim1}
now follows from \eqref{keithE} with $\Lambda=0$ and $a=-\nf{1}{4}$ and from 
\begin{align*}
\|\omega^{-\nf{1}{4}}(1-\chi_\kappa^2)i\V{m}\beta_\infty\|_\HP^2
&\le4\pi\ee^2\bigg(\int_0^\kappa\frac{4^2\rho^\eh\Id\rho}{\kappa^2}+\int_\kappa^\infty
\frac{4\Id\rho}{\rho^{\nf{3}{2}}}\bigg)\le\frac{44\pi\ee^2}{\kappa^\eh}.
\end{align*}
\end{proof}


\subsection{Exponential moment bounds on the complex action}\label{ssecuge}

\noindent
The next theorem is our main result on the complex action. The asserted exponential moment
bound \eqref{expbdu} has an improved right hand side compared to a recent bound by Bley 
\cite{Bley2016}, where $p^2\ee^4N^3$ is replaced by $p^2\ee^4N^3\ln^2(p\ee^2N)$, for sufficiently 
large $p\ee^2$, and by $p^2\ee^4N^3[1\vee\ln^2(N)]$, for sufficiently small $p\ee^2$. Besides,
we add a running supremum in \eqref{expbdu} (which were probably not a good idea if we were
interested in good values for the rate $c$). While Bley derives uniform bounds for finite $\kappa$,
we construct and estimate the limiting action $u_\infty^N(\ul{\V{q}})$ as well.

Earlier, Gubinelli et al. \cite{GHL2014} obtained an exponential moment bound on the complex
action for massive bosons, including the case $\kappa=\infty$. As explained by Bley and Thomas 
\cite[\textsection3.4]{BleyThomas2015}, the arguments of \cite{GHL2014} yield a right hand side that is 
log-linear in $t$ (with a worse dependence on $\mu$, $p\ee^2$, and $N$ than in \cite{Bley2016}), 
when they are combined with a result of \cite{BleyThomas2015}.

Nelson \cite{Nelson1964proc} derived exponential moment bounds for massive bosons
that are locally uniform in $t\ge0$, without paying special attention to the precise dependence of their 
right hand sides on $\ee$, $N$, $t$, and $\mu$. The exponentially
damping factors $e^{-t\omega}$, etc., present in the complex actions studied here and in
\cite{BleyThesis,BleyThomas2015,GHL2014} are, however, replaced by oscillating terms 
$e^{it\omega}$ in Nelson's \cite{Nelson1964proc}, because he was eventually interested in the 
unitary group. Hence, one has to be careful in comparing estimates. 

The key idea leading to our improved right hand side in \eqref{expbdu} is to benefit from the
bound \eqref{keith0} after applying Rem.~\ref{remstochGronwall} to the martingale
$m_{\Lambda,\kappa}^N(\ul{\V{q}})$, for some large $\Lambda$; it is worked out in 
Lem.~\ref{lemkeith} below. The crucial point about
\eqref{keith0} is that the weight $\omega^{a+\eh}$ in front of the process $S_\kappa^N(\ul{\V{q}})$
on its left hand side is traded for $\omega^a$ in the martingale $\ell_{\Lambda,\kappa}(\ul{\V{q}})$
on its right hand side. In particular,  $\ell_{\Lambda,\kappa}(\ul{\V{q}})$ has a less singular integrand 
than $m_{\Lambda,\kappa}^N(\ul{\V{q}})$, whence we can control its exponential moment by
applying Rem.~\ref{remstochGronwall} a second time and concluding by elementary estimations.
This effect actually is somewhat reminiscent of the general strategy in \cite{Bley2016},
where repeated applications of the martingale estimate \eqref{BleyThomasmart} are combined with 
successive Clark-Ocone expansions of the exponents in its right hand side, until eventually the 
right hand side becomes sufficiently tractable.

\begin{thm}\label{lemuexp}
There exist universal constants $b,c>0$ such that, for all $\kappa\in\NN\cup\{\infty\}$, 
$\fF_0$-measurable $\ul{\V{q}}:\Omega\to\RR^\nu$, $t\ge0$, and $p>0$, 
\begin{align}\label{expbdu}
\EE\big[\sup_{s\le t}e^{pu_{\kappa,s}^N(\ul{\V{q}})}\big]&\le b^Ne^{cp^2\ee^4N^3t}.
\end{align}
Moreover, 
\begin{align}\label{convu}
\sup_{\ul{\V{q}}}\EE\big[\sup_{s\le t}|e^{u_{\kappa,s}^N(\ul{\V{q}})}-e^{u_{\infty,s}^N(\ul{\V{q}})}|^{p}\big]
&\xrightarrow{\;\;\kappa\to\infty\;\;}0,\quad t\ge0,\,p>0,
\end{align}
where the supremum in front of the expectation is taken over all $\fF_0$-measurable 
functions $\ul{\V{q}}:\Omega\to\RR^\nu$.
\end{thm}


\begin{proof}
To start with we observe that \eqref{convu} is implied by \eqref{felixu} and \eqref{expbdu}.

In the derivation of \eqref{expbdu}, the most cumbersome term is the martingale 
$m_\kappa^N(\ul{\V{q}})$, for which we only find a suitable bound, if it is restricted to 
sufficiently large boson momenta. Therefore, we pick some $\Lambda>0$ and split off the
following infra-red part from the complex action,
\begin{align}\label{foruIR}
u_{\Lambda,\kappa,t}^{N,<}(\ul{\V{q}})
&=\sum_{j,\ell=1}^N\int_0^t\int_0^s\int_{\{|\V{m}|<\Lambda\}}{e^{-(s-r)\omega-i\V{m}\cdot
(\V{q}_\ell+\V{b}_{\ell,s}-\V{q}_j-\V{b}_{j,r})}}f_\kappa^2\Id\lambda^3\Id r\,\Id s.
\end{align}
Replacing its integrand by its modulus and computing the $\Id r$-integral first, we find the
elementary bound $|u_{\Lambda,\kappa,t}^{N,<}(\ul{\V{q}})|\le 4\pi\ee^2N^2t\Lambda$.
Next, we notice that Thm.~\ref{thm2ndIto} can be applied with $1_{\{|\V{m}|<\Lambda\}}\eta$ or
$1_{\{|\V{m}|\ge\Lambda\}}\eta$ in place of $\eta$, and in the former case it still holds true for
$\kappa=\infty$. Hence, $\PP$-a.s.,
\begin{align}\nonumber
u_{\kappa,t}^N(\ul{\V{q}})&=u_{\Lambda,\kappa,t}^{N,<}(\ul{\V{q}})-tNE_{\Lambda,\kappa}
-b^N_{\Lambda,\kappa,t}(\ul{\V{q}})+c_{\Lambda,\kappa}^{N,-}(\ul{\V{q}})
\\\label{morten1a}
&\quad-c_{\Lambda,\kappa}^{N,+}(\ul{\V{q}})+v_{\Lambda,\kappa,t}^N(\ul{\V{q}})
+m_{\Lambda,\kappa,t}^N(\ul{\V{q}}),\quad t\ge0,
\end{align}
where $\kappa$ can be finite or infinite. Here $b^N_{\Lambda,\kappa,t}(\ul{\V{q}})$
and the energy correction 
$E_{\Lambda,\kappa}:=\int_{\{|\V{m}|<\Lambda\}}f_\kappa\beta_\kappa\Id\lambda^3$ are
non-negative. Furthermore,
\begin{align}\label{morten1}
|c_{\Lambda,\kappa,t}^{N,\pm}(\ul{\V{q}})|&\le8\pi\ee^2N^2/\Lambda.
\end{align}
In fact, on account of \eqref{anne1}, $\omega^\mh U_{\kappa,t}^\pm[\V{b}_j]$ is given by an
{\em $\HP$-valued} Bochner-Lebesgue integral, which can be interchanged
with the scalar products in \eqref{defcplus} and \eqref{defcminus}. Hence,
both terms on the left hand side of \eqref{morten1} are dominated by
\begin{align*}
4\pi\ee^2N^2\int_\Lambda^\infty\int_0^t\frac{e^{-(t-s)\omega_\rho}\rho^2}{
\omega_\rho(\omega_\rho+\rho^2/2)}\Id s\Id\rho
&\le8\pi\ee^2N^2\int_{\Lambda}^\infty\frac{\Id\rho}{\rho^2}.
\end{align*}
Altogether we conclude that
\begin{align*}
\EE\big[\sup_{s\le t}e^{pu_{\kappa,s}^N(\ul{\V{q}})}\big]
&\le e^{cp\ee^2N^2/\Lambda+4\pi p\ee^2N^2t\Lambda}
\EE\Big[\sup_{s\le t}e^{rpm_{\Lambda,\kappa,s}^N(\ul{\V{q}})}\Big]^{\nf{1}{r}}
\EE\Big[\sup_{s\le t}e^{r'pv_{\Lambda,\kappa,s}^N(\ul{\V{q}})}\Big]^{\nf{1}{r'}}\!,
\end{align*}
for every $r\in(1,\infty)$ with  $r'$ denoting its conjugated exponent.
Now we choose $\Lambda:=64\pi pqr\ee^2N$, with an arbitrary $q>1$. With this choice,
suitable bounds on the latter two expectations are given in Lem.~\ref{lemfrkv} and
Lem.~\ref{lemkeith} below.
\end{proof}

\begin{lem}\label{lemfrkv}
Let $\kappa\in\NN\cup\{\infty\}$ and $\ul{\V{q}}:\Omega\to\RR^\nu$ be $\fF_0$-measurable.
Put $\Lambda(\tau):=64\pi\tau$, $\tau>0$.
Then there exists a universal constant $c>0$ such that
\begin{align}\label{expbdfrkv}
\EE\Big[\sup_{s\le t}e^{p|{v}^N_{\Lambda(q\ee^2N),\kappa,s}(\ul{\V{q}})|}\Big]
&\le \{e^{pNt/4q}\}\wedge\{e^{c(p^4/q^2)\ee^4N^2(N-1)t}\},\quad q,p,t>0.
\end{align}
\end{lem}

\begin{proof}
A trivial estimate yields
$|v_{\Lambda,\kappa}^N(\ul{\V{q}})|\le N^2t\|\omega^\eh\beta_{\Lambda,\kappa}\|^2
\le16\pi\ee^2N^2t/\Lambda$.

To derive the non-trivial bound in \eqref{expbdfrkv}
let $p,t,\Lambda>0$ and let $\cL_{\ul{\V{q}}}$ denote the law of $\ul{\V{q}}$. 
Since $\ul{\V{q}}:\Omega\to\RR^\nu$ is $\fF_0$-measurable, we then have
\begin{align*}
\EE\Big[\sup_{s\le t}e^{p|{v}_{\Lambda,\kappa,s}^N(\ul{\V{q}})|}\Big]
&=\int_{\RR^\nu}\EE\Big[\sup_{s\le t}e^{p|{v}_{\Lambda,\kappa,s}^N(\ul{\V{x}})|}\Big]
\Id\cL_{\ul{\V{q}}}(\ul{\V{x}});
\end{align*}
see, e.g., \cite[Prop.~1.12]{daPrZa2014}.
Hence, it suffices to treat only constant $\ul{\V{q}}=\ul{\V{x}}\in\RR^\nu$. 

For disjoint pairs $\{j,\ell\}$ and $\{j',\ell'\}$ of indices in $\{1,\ldots,N\}$, the random variables 
$\sup_{s\le t}e^{p|{v}_{\Lambda,\kappa,s}^{[j,\ell]}[\ul{\V{x}},\ul{\V{b}}_{\bullet}]|}$ and
$\sup_{s\le t}e^{p|{v}_{\Lambda,\kappa,s}^{[j',\ell']}[\ul{\V{x}},\ul{\V{b}}_{\bullet}]|}$ 
are independent. Similarly as in \cite[Eqn.~(3.4)\&(3.5)]{Bley2016} we shall use these 
independences together with H\"{o}lder's inequality in the next estimate.
For the sake of completeness we provide a proof (slightly different from \cite{Bley2016}) of the 
elementary inequality employed in the second step of the following estimation in App.~\ref{apppair},
\begin{align}\nonumber
\EE\Big[\sup_{s\le t}e^{p|{v}_{\Lambda,\kappa,s}^N(\ul{\V{x}})|}\Big]
&\le\EE\Big[\mathop{\prod_{j,\ell=1}^N}_{j<\ell}\sup_{s\le t}
e^{2p|{v}_{\Lambda,\kappa,s}^{[j,\ell]}[\ul{\V{x}},\ul{\V{b}}_{\bullet}]|}\Big]
\le\mathop{\max_{j,\ell=1}^N}_{j<\ell}\EE\Big[\sup_{s\le t}
e^{2pN|{v}_{\Lambda,\kappa,s}^{[j,\ell]}[\ul{\V{x}},\ul{\V{b}}_{\bullet}]|}\Big]^{\frac{N-1}{2}}
\\\label{Gonzalo}
&\le\sup_{\V{y}\in\RR^3}\EE\Big[e^{2pN
\int_0^t|w_{\Lambda,\kappa}(2^\eh(\V{y}+\V{B}_{s}))|\Id s}\Big]^{\frac{N-1}{2}}.
\end{align}
Here we also used the fact that $2^\mh(\V{b}_{j}-\V{b}_{\ell})$ with $j\not=\ell$ is a three-dimensional
Brownian motion having the same distribution as $\V{B}$ in the last step. 

Our plan is now to
employ the following special case of an inequality due to Carmona \cite[Rem.~3.1]{Carmona1979},
\begin{align}\label{Carmonaineq}
\sup_{\V{y}\in\RR^3}\EE\Big[e^{\int_0^t|v(\V{y}+\V{B}_s)|\Id s}\Big]
&\le ce^{c'\|v\|_2^{4}t}, \quad v\in L^2(\RR^3).
\end{align}
Since $v_{\Lambda,\kappa}$ is the Fourier transform of a function in
$L^1(\RR^3)\cap L^2(\RR^3)$, we may exploit the isometry of the Fourier transform to get
\begin{align}\label{samy}
\|w_{\Lambda,\kappa}(2^\eh(\cdot))\|_2^2&\le\frac{4\pi\ee^4}{2^{\nf{3}{2}}}
\int_{\Lambda}^\infty\frac{\rho^2}{\omega_\rho^4}\Id\rho
\le\frac{c\ee^4}{\Lambda}.
\end{align}
Combining \eqref{Gonzalo}, \eqref{Carmonaineq}, and \eqref{samy} we arrive at
\begin{align}\nonumber
\sup_{\ul{\V{x}}\in\RR^\nu}\EE\Big[\sup_{s\le t}e^{p|{v}_{\Lambda,\kappa,s}^N(\ul{\V{x}})|}\Big]
&\le e^{c'(N-1)(p\ee^2N)^4t/\Lambda^2}.
\end{align}
\end{proof}

\begin{lem}\label{lemkeith}
Let $\kappa\in\NN\cup\{\infty\}$ and $\ul{\V{q}}:\Omega\to\RR^\nu$ be $\fF_0$-measurable.
Put $\Lambda(\tau):=64\pi\tau$, $\tau>0$. Then there exists a universal constant $c>0$ and, 
for every $q>1$, some $c_q>0$ depending only on $q$ such that
\begin{align}\label{keith2000}
\EE\Big[\sup_{s\le t}e^{\pm pm_{\Lambda(qp\ee^2N),\kappa,s}(\ul{\V{q}})}\Big]
&\le c_qe^{c([p\ee^2N^2]\wedge[qp^2\ee^4N^3])t},\quad p,t>0.
\end{align}
\end{lem}

\begin{proof}
Let $\Lambda>0$. The key ingredient in this proof is Lem.~\ref{lemthalia} that we apply with
$a=-\nf{1}{4}$, writing $\sigma_\Lambda:=\mho_\Lambda(-\nf{1}{4})$, 
where the function $\mho_\Lambda$ is defined in front of Lem.~\ref{lemthalia}. 
Recall also the notation $\vr_\Lambda=1_{\{|\V{m}|\ge\Lambda\}}$. According to \eqref{keith0}
the quadratic variation of $m_{\Lambda,\kappa}^N(\ul{\V{q}})$ satisfies, $\PP$-a.s.,
\begin{align}\nonumber
\llbracket m_{\Lambda,\kappa}^{N}(\ul{\V{q}})\rrbracket_t
&=\sum_{\ell=1}^N\int_0^t\SPb{\vr_\Lambda\omega^{\nf{1}{4}}S_{\kappa,s}^{N}(\ul{\V{q}})}{
 e^{-i\V{m}\cdot(\V{q}_\ell+\V{b}_{\ell,s})}
\omega^{-\nf{1}{4}}i\V{m}\beta_{\Lambda,\kappa}}_{\HP}^2\Id s 
\\\nonumber
&\le N\|\omega^{-\nf{1}{4}}i\V{m}\beta_{\Lambda,\kappa}\|_{\HP}^2
\int_0^t\|\vr_\Lambda\omega^{\nf{1}{4}}S_{\kappa,s}^{N}(\ul{\V{q}})\|_{\HP}^2\Id s
\\\label{kathrin3}
&\le N\|\omega^{-\nf{1}{4}}i\V{m}\beta_{\Lambda,\kappa}\|_{\HP}^2\ell_{\Lambda,\kappa,t}(\ul{\V{q}})
+6\sigma_\Lambda^2\ee^4N^3t.
\end{align}
In the second step we also observed that
\begin{align}\nonumber
\|\omega^{-\nf{1}{4}}i\V{m}\beta_{\Lambda,\kappa}\|^2
&\le4\pi\ee^2\bigg(\int_{1\wedge\Lambda}^1\rho^\eh\Id\rho+
4\int_{1\vee\Lambda}^\infty\frac{\Id\rho}{\rho^{\nf{3}{2}}}\bigg)
\\\label{kathrin2}
&=8\pi\ee^2\bigg(\frac{1-({1\wedge\Lambda})^{\nf{3}{2}}}{3}+\frac{4}{(1\vee\Lambda)^\eh}\bigg)
\le2\sigma_\Lambda\ee^2.
\end{align}
Let $q>1$. Combining \eqref{kathrin3} with Rem.~\ref{remstochGronwall} (which we apply with
$p=p'=2$) we then conclude that
\begin{align}\nonumber
\EE\Big[\sup_{s\le t}e^{\pm pm_{\Lambda,\kappa,s}^N}\Big]
&\le c_q\EE\Big[e^{2qp^2\llbracket m_{\Lambda,\kappa}^{N}(\ul{\V{q}})\rrbracket_t}\Big]^{\nf{1}{2}}
\\\nonumber
&\le c_q\EE\Big[e^{2qp^2N\|\omega^{-\nf{1}{4}}i\V{m}\beta_{\Lambda,\kappa}\|_{\HP}^2
\int_0^t\|\vr_\Lambda\omega^{\nf{1}{4}}S_{\kappa,s}^{N}(\ul{\V{q}})\|_{\HP}^2\Id s}\Big]^{\nf{1}{2}}
\\\nonumber
&\le c_qe^{6\sigma_\Lambda^2qp^2\ee^4N^3t}
\EE\Big[e^{2qp^2N\|\omega^{-\nf{1}{4}}i\V{m}\beta_{\Lambda,\kappa}\|_{\HP}^2
\ell_{\Lambda,\kappa,t}(\ul{\V{q}})}\Big]^{\nf{1}{2}}
\\\nonumber
&\le c_{q}'e^{6\sigma_\Lambda^2qp^2\ee^4N^3t}
\EE\Big[e^{8q^3p^4N^2\|\omega^{-\nf{1}{4}}i\V{m}\beta_{\Lambda,\kappa}\|_{\HP}^4
\llbracket\ell_{\Lambda,\kappa}(\ul{\V{q}})\rrbracket_t}\Big]^{\nf{1}{4}}
\\\label{keith1983}
&\le c_{q}'e^{6\sigma_\Lambda^2qp^2\ee^4N^3t}
e^{2q^3p^4\|\omega^{-\nf{1}{4}}i\V{m}\beta_{\Lambda,\kappa}\|_{\HP}^4\sigma_\Lambda^2\ee^4N^5t}.
\end{align}
Here we applied \eqref{thalia56} and \eqref{thalia57} in the last step.
This proves the theorem in the case $qp\ee^2N<1$, where 
$\|\omega^{-\nf{1}{4}}i\V{m}\beta_{\Lambda(qp\ee^2N),\kappa}\|_{\HP}^2\le35\pi\ee^2$ and
$\sigma_{\Lambda(p\ee^2N)}<86\pi$. 
In the case $qp\ee^2N\ge1$, we could insert the bounds \eqref{kathrin2} and
$\sigma_{\Lambda(qp\ee^2N)}^2\le c/qp\ee^2N$ into \eqref{keith1983}.

It is, however, possible to improve the right hand side of the above bound
for large $qp\ee^2N$. In fact, if $\Lambda\ge1$, then 
$\|\omega^{-\nf{1}{4}}i\V{m}\beta_{\Lambda,\kappa}\|_{\HP}^2\le32\pi\ee^2/\Lambda^\eh$
by \eqref{kathrin2}. Combining the third and fourth inequalities in \eqref{keith1983} with
\eqref{thalia55} we thus obtain
\begin{align}\nonumber
\EE\Big[e^{2qp^2N\|\omega^{-\nf{1}{4}}i\V{m}\beta_{\Lambda,\kappa}\|_{\HP}^2
\int_0^t\|\vr_\Lambda\omega^{\nf{1}{4}}S_{\kappa,s}^{N}(\ul{\V{q}})\|_{\HP}^2\Id s}\Big]
\le c_qe^{12\sigma_\Lambda^2qp^2\ee^4N^3t}\qquad&
\\\label{keith1984}
\cdot\,\EE\Big[e^{8q^3p^4N^3((32\pi)^2\ee^4/\Lambda)
\|\omega^{-\nf{1}{4}}i\V{m}\beta_{\Lambda,\kappa}\|_{\HP}^2
\int_0^t\|\vr_{\Lambda}\omega^{-\nf{1}{4}}S_{\kappa,s}^{N}(\ul{\V{q}})\|_{\HP}^2\Id s}\Big]^{\nf{1}{2}}&.
\end{align}
By definition of $\Lambda(qp\ee^2N)$ we now have
\begin{align*}
(4q^2p^2(32\pi)^2\ee^4N^2/\Lambda(qp\ee^2N))\vr_{\Lambda(qp\ee^2N)}
=\Lambda(qp\ee^2N)\vr_{\Lambda(qp\ee^2N)}\le\omega\vr_{\Lambda(qp\ee^2N)}.
\end{align*}
Hence, if we choose $\Lambda=\Lambda(qp\ee^2N)$ 
and consider the case $qp\ee^2N\ge1$, then we may 
bound the expectation $\EE[\cdots]$ in the second line of \eqref{keith1984}
from above by the expectation in the first line of \eqref{keith1984}.
Thanks to \eqref{keith1983} we know already that the latter is finite, whence we may solve
the resulting bound for the left hand side of \eqref{keith1984}, obtaining
\begin{align*}
\EE\Big[e^{2qp^2N\|\omega^{-\nf{1}{4}}i\V{m}\beta_{\Lambda(p\ee^2N),\kappa}\|_{\HP}^2
\int_0^t\|\vr_{\Lambda(p\ee^2N)}\omega^{\nf{1}{4}}S_{\kappa,s}^{N}(\ul{\V{q}})\|_{\HP}^2\Id s}\Big]
&\le c_qe^{24\sigma_{\Lambda(qp\ee^2N)}^2qp^2\ee^4N^3t}.
\end{align*}
Since $qp\ee^2N\ge1$ entails $\sigma_{\Lambda(qp\ee^2N)}^2\le c/qp\ee^2N$, this concludes
the proof of \eqref{keith2000}.
\end{proof}

Later on, we shall use the next remark to derive the lower bound in \eqref{asympintro}.

\begin{rem}\label{rem256}
Since the parameters $r,q>1$ can be chosen as close to $1$ as we please in the end of
the proof of Thm.~\ref{lemuexp}, the last estimate in that proof and the bounds \eqref{expbdfrkv}
and \eqref{keith2000} actually imply
\begin{align*}
\EE\Big[\sup_{s\le t}e^{pu_{\kappa,s}^N(\ul{\V{x}})}\Big]&\le c_q
e^{256\pi^2qp^2\ee^4N^3t+cp\ee^2N^2t+cN(1\vee t)},\quad\text{if $p\ee^2N\ge1$,}
\end{align*}
for all $p,t>0$, $q>1$, $\ul{\V{x}}\in\RR^\nu$, $\kappa\in\NN\cup\{\infty\}$,
and some universal constant $c>0$.
\end{rem}

The last remark of this subsection will be used to show the upper bound in \eqref{asympintro}.

\begin{rem}\label{remuge}
Let $\kappa\in\NN\cup\{\infty\}$ and $\ul{\V{q}}:\Omega\to\RR^\nu$ be $\fF_0$-measurable.
Recall the notation \eqref{foruIR} and define 
$u_{\Lambda,\kappa}^{N,>}(\ul{\V{q}}):=u_{\kappa}^N(\ul{\V{q}})
-u_{\Lambda,\kappa}^{N,<}(\ul{\V{q}})$, for every $\Lambda>0$.
A trivial estimation yields $0\le b_{\Lambda,\kappa,t}^N(\ul{\V{q}})\le16\pi\ee^2N^2/\Lambda$.
Pick $q,r>1$ and put $\Lambda(pqr\ee^2N):=64\pi pqr\ee^2N$.
Applying \eqref{morten1}, \eqref{expbdfrkv}, and \eqref{keith2000} to the remaining terms
in the decomposition \eqref{morten1a} and using H\"{o}lder's inequality as in the end of the proof
of Thm.~\ref{lemuexp} we then obtain
\begin{align*}
\EE\Big[e^{\pm pu_{\Lambda(pqr\ee^2N^2),\kappa}^{N,>}(\ul{\V{q}})}\Big]
&\le c_q
e^{pNE_{\Lambda(pqr\ee^2N^2),\kappa}t+cp\ee^2N^2t+cN(1\vee t)},\quad\text{if $p\ee^2N\ge1$,}
\end{align*}
for all $p,t>0$, with a universal constant $c>0$.
\end{rem}


\subsection{Further properties of the complex action}\label{ssecuinfty}

\noindent
In this subsection we provide some technical results on the complex action that are
relevant for discussing the Markov property of our Feynman-Kac integrands and for
proving semi-group properties. We start by considering the dependence of the martingales
$m_{\kappa}^N(\ul{\V{q}})$ on the starting points for the Brownian motions.

\begin{lem}\label{lemmartm}
Let $\kappa\in\NN\cup\{\infty\}$. Then the following holds:
\begin{enumerate}[leftmargin=*]
\item[{\rm(1)}] 
If $\ul{\V{q}},\smash{\ul{\V{q}}}_n:\Omega\to\RR^\nu$, $n\in\NN$, are $\fF_0$-measurable
such that $\smash{\ul{\V{q}}}_n\to\ul{\V{q}}$, $n\to\infty$, in probability, then
\begin{align}\label{limprobm}
\mathop{\mathrm{lim\,prob}}_{n\to\infty}\sup_{s\le t}
\big|m_{\kappa,s}^N(\smash{\ul{\V{q}}}_n)-m_{\kappa,s}^N(\ul{\V{q}})\big|=0,
\quad t\ge0.
\end{align}
\item[{\rm(2)}] The martingales $m_{\kappa}^N(\ul{\V{x}})$, defined up to indistinguishability
for each $\ul{\V{x}}\in\RR^\nu$ by \eqref{defm}, can be chosen such that, for all
$\gamma\in\Omega$, the following map is continuous,
\begin{align*}
[0,\infty)\times\RR^\nu\ni(t,\ul{\V{x}})\longmapsto(m_{\kappa,t}^N(\ul{\V{x}}))(\gamma)\in\RR.
\end{align*}
\end{enumerate}
\end{lem}

\begin{proof}
\noindent{\em Step~1.} We fix $\kappa\in\NN\cup\{\infty\}$ and
pick some $\epsilon\in(0,\nf{1}{4})$ as well as another $\fF_0$-measurable
$\tilde{\ul{\V{q}}}:\Omega\to\RR^\nu$. Then
\begin{align}\label{Hoelderd}
\big|\V{d}_{\kappa,t}^{[\ell]}(\tilde{\ul{\V{q}}})-\V{d}_{\kappa,t}^{[\ell]}({\ul{\V{q}}})\big|
&\le2^\epsilon|\ul{\V{q}}-\tilde{\ul{\V{q}}}|^\epsilon\sum_{j=1}^N\|\omega^{\nf{1}{4}+\epsilon}
S_{\kappa,t}^{[t]}[\V{b}_j]\|_\HP\|\omega^{-\nf{1}{4}}i\V{m}\beta_{\Lambda,\kappa}\|_\HP,
\end{align}
for all $t\ge0$. If $\smash{\ul{\V{q}}}_n\to\ul{\V{q}}$ in probability and $\delta>0$, we thus obtain
\begin{align*}
\limsup_{n\to\infty}&\,\EE\bigg[\int_0^t\big(\V{d}_{\Lambda,\kappa,s}^{[\ell]}(\smash{\ul{\V{q}}}_n)
-\V{d}_{\Lambda,\kappa,s}^{[\ell]}({\ul{\V{q}}})\big)^2\Id s\bigg]^\eh
\\
&\le\limsup_{n\to\infty} 
Nt^\eh\max_{j=1}^N\sup_{s\le t}\EE\Big[1_{\{|\smash{\ul{\V{q}}}_n-\ul{\V{q}}|\ge\delta\}}
\|\omega^{\nf{1}{4}}S_{\kappa,s}^{[s]}[\ul{\V{b}}_j]\|_\HP^2\|\omega^{-\nf{1}{4}}i\V{m}
\beta_{\Lambda,\kappa}\|_\HP^2\Big]^\eh
\\
&\quad+Nt^\eh(2\delta)^\epsilon\max_{j=1}^N\sup_{s\le t}\EE\Big[\|\omega^{\nf{1}{4}+\epsilon}
S_{\kappa,t}^{[t]}[\ul{\V{b}}_j]\|_\HP^2\|\omega^{-\nf{1}{4}}
i\V{m}\beta_{\Lambda,\kappa}\|_\HP^2\Big]^\eh
\\
&\le c_{\epsilon}\ee^2Nt^\eh\delta^\epsilon,
\end{align*}
where we used \eqref{polbdMtaua}, \eqref{polbdMtaub}, and the Cauchy-Schwarz inequality 
in the last step. Since
\begin{align}\nonumber
\PP\Big\{\sup_{s\le t}|m_{\kappa,s}^N(\tilde{\ul{\V{q}}})
-m_{\kappa,s}^N(\ul{\V{q}})|\ge\vs\Big\}
&\le\frac{1}{\vs^p}\EE\Big[\sup_{s\le t}|m_{\kappa,s}^N(\tilde{\ul{\V{q}}})
-m_{\kappa,s}^N(\ul{\V{q}})|^p\Big]
\\\label{ansgar1}
&\le\frac{c_p}{\vs^p}\bigg(\sum_{\ell=1}^N\int_0^t\EE\Big[
\big({\V{d}}_{\kappa,s}^{[\ell]}(\tilde{\ul{\V{q}}})
-{\V{d}}_{\kappa,s}^{[\ell]}(\ul{\V{q}})\big)^2\Big]\Id s\bigg)^{\nf{p}{2}},
\end{align}
for all $\vs,p>0$, this proves \eqref{limprobm}.

\smallskip

\noindent{\em Step~2.} Combining \eqref{Hoelderd}, the second inequality in \eqref{ansgar1}, and 
\eqref{keithE} (choosing $a=-\nf{1}{8}$, $\Lambda=0$, and $N=1$ in the latter bound), we find
\begin{align*}
\EE\Big[\sup_{s\le t}|m_{\kappa,s}^N({\ul{\V{x}}})-m_{\kappa,s}^N(\ul{\V{y}})|^p\Big]
&\le c_p|{\ul{\V{x}}}-{\ul{\V{y}}}|^{\nf{p}{8}}|\ee|^pN^{3p}t^{\nf{p}{2}},\quad
{\ul{\V{x}}},{\ul{\V{y}}}\in\RR^\nu,\,p,t>0.
\end{align*}
Part~(2) thus follows from Kolmogorov's test lemma.
\end{proof}

Let $\ul{\V{q}}:\Omega\to\RR^\nu$ be $\fF_0$-measurable in what follows.
Recall from Def.~\ref{defu} that, for finite $\kappa$, the random variables 
$u_{\kappa,t}^N(\ul{\V{q}})=u_{\kappa,t}^N[\ul{\V{q}},\ul{\V{b}}]$ 
are compositions of Borel measurable functions on $\RR^\nu\times C([0,\infty),\RR^\nu)$ with the map
$\Omega\ni\gamma\mapsto(\ul{\V{q}}(\gamma),\ul{\V{b}}_\bullet(\gamma))$. The same holds true
for all contributions to $u_{\infty}^N(\ul{\V{q}})$ in \eqref{defuinfty} except for the martingale 
$m_{\infty}^N(\ul{\V{q}})$, which depends on in a slightly more subtle fashion. In our formulation of the 
Markov properties later on, it is thus convenient to introduce a certain standard realization of 
$m_{\infty}^N$ and, hence, of $u_\infty^N$ such that all other realizations can be obtained from them 
by plugging in $(\ul{\V{q}},\ul{\V{b}})$:

Recall the notation for objects related to the Wiener measure introduced in Ex.~\ref{exWiener}.

\begin{defn}\label{defnstandardu}
The symbols $m_{\infty}^N[\ul{\V{x}},\cdot]$, $\ul{\V{x}}\in\RR^\nu$, denote choices of martingales 
as in Lem.~\ref{lemmartm}(2), when this lemma is applied with $\BB_{\mathrm{W}}^\nu$ as underlying 
stochastic basis and the evaluation process $\pr^\nu$ as Brownian motion, so that the 
real-valued map
\begin{align}\label{contwm}
[0,\infty)\times\RR^\nu\ni(t,\ul{\V{x}})\mapsto m_{\infty,t}^N[\ul{\V{x}},\ul{\V{\alpha}}]
\quad\text{is continuous,}
\end{align}
for every $\ul{\V{\alpha}}\in\Omega_{\mathrm{W}}^\nu$. With this we set
\begin{align*}
u_{\infty,t}^N[\ul{\V{x}},\ul{\V{\alpha}}]&:=-b_{\infty,t}^N[\ul{\V{x}},\ul{\V{\alpha}}]
+c_{\infty,t}^{N,-}[\ul{\V{x}},\ul{\V{\alpha}}]-c_{\infty,t}^{N,+}[\ul{\V{x}},\ul{\V{\alpha}}]
+v_{\infty,t}^N[\ul{\V{x}},\ul{\V{\alpha}}]+m_{\infty,t}^N[\ul{\V{x}},\ul{\V{\alpha}}],
\end{align*}
for all $t\ge0$, $\ul{\V{x}}\in\RR^\nu$, and $\ul{\V{\alpha}}\in\Omega_{\mathrm{W}}^\nu$. 
\end{defn}

Let summarize some earlier observations and a consequence of Def.~\ref{defnstandardu}:

\begin{lem}\label{lemucont}
The real-valued map
$[0,\infty)\times\RR^\nu\ni(t,\ul{\V{x}})\mapsto u_{\kappa,t}^N[\ul{\V{x}},\ul{\V{\alpha}}]$
is continuous, for all $\ul{\V{\alpha}}\in\Omega_{\mathrm{W}}^\nu$ and $\kappa\in\NN\cup\{\infty\}$.
\end{lem}

Standard arguments further yield the following:

\begin{lem}\label{lemstrsol}
Let $u_{\infty}^N(\ul{\V{q}})$ be the process introduced in Def.~\ref{defnuinfty}, for any stochastic 
basis $\BB=(\Omega,\fF,(\fF_t)_{t\ge0},\PP)$, $\nu$-dimensional $\BB$-Brownian motion 
$\ul{\V{b}}$, and $\fF_0$-measurable $\ul{\V{q}}:\Omega\to\RR^\nu$. Then there exists a 
$\PP$-zero set $\sN\in\fF$ such that
\begin{align*}
(u_{\infty,t}^N(\ul{\V{q}}))(\gamma)&=u_{\infty,t}^N[\ul{\V{q}}(\gamma),\ul{\V{b}}(\gamma)],
\quad t\ge0,\,\gamma\in\Omega\setminus\sN.
\end{align*}
\end{lem}

\begin{proof}
Since all other terms in the complex action are defined pathwise, it suffices to find a $\PP$-zero set
$\sN\in\fF$ such that
\begin{align}\label{wwmm}
(m_{\infty,t}^N(\ul{\V{q}}))(\gamma)=m_{\infty,t}^N[\ul{\V{q}}(\gamma),\ul{\V{b}}(\gamma)],
\end{align}
for all $t\ge0$ and $\gamma\in\Omega\setminus\sN$. If $\ul{\V{q}}$ is constant and
$t\ge0$ is fixed, then a well-known
transformation argument (see, e.g., \cite[Lem.~6.27]{HackenbrochThalmaier1994}) for stochastic
integrals shows that \eqref{wwmm} holds for $\PP$-a.e. $\gamma$.
Since all processes in \eqref{wwmm} are continuous, this already implies that \eqref{wwmm} holds
for all $t\ge0$ on the complement of some $t$-independent $\PP$-zero set, 
still provided that $\ul{\V{q}}$ is constant.

If $\ul{\V{q}}$ is an arbitrary $\fF_0$-measurable $\RR^\nu$-valued function, then we set
$\smash{\ul{\V{q}}}_n:=I_n(\ul{\V{q}})$, $n\in\NN$, where $I_n:\RR^\nu\to\RR^\nu$ is given by
$I_n(\ul{\V{x}}):=\ul{\V{y}}/2^n$ with $\ul{\V{y}}\in\ZZ^\nu$ such that
$y_j/2^n\le x_j<(y_j+1)/2^n$, for all $j\in\{1,\ldots,\nu\}$. Then the pathwise uniqueness property
of stochastic integrals (see, e.g., \cite[Kor.~1 on p.~188]{HackenbrochThalmaier1994}) implies that
\begin{align}\label{wwmmn}
(m_{\infty,t}^N(\smash{\ul{\V{q}}}_n))(\gamma)
=m_{\infty,t}^N[\smash{\ul{\V{q}}}_n(\gamma),\ul{\V{b}}(\gamma)],
\end{align}
for all $t\ge0$, $n\in\NN$, $\gamma\in\Omega\setminus\sN_1$, and some $\PP$-zero set
$\sN_1\in\fF$. By the continuity of the maps in \eqref{contwm}, the right hand sides of 
\eqref{wwmmn} converge to the respective right hand sides of \eqref{wwmm}, for every 
$\gamma\in\Omega$. On account of \eqref{limprobm} the left hand sides of 
\eqref{wwmmn} converge to the respective left hand sides of \eqref{wwmm}
locally uniformly (with respect to $t$) in probability.
\end{proof}

If $t\ge0$, then all results obtained in this section so far apply in particular to the 
time-shifted stochastic basis $\BB_t$ introduced in Subsect.~\ref{ssecprobprel} and the time shifted 
$\nu$-dimensional Brownian motion ${}^t\ul{\V{b}}$ given by \eqref{deftb}. 
This observation is used to obtain the following:

\begin{lem}\label{lemushift}
Let $t\ge0$ and $\ul{\V{q}}:\Omega\to\RR^\nu$ be $\fF_0$-measurable. Then, 
for every $\kappa\in\NN$, the following holds true on $\Omega$,
\begin{align}\nonumber
u_{\kappa,s}^N[\ul{\V{q}}+\ul{\V{b}}_t,{}^t\ul{\V{b}}]+u_{\kappa,t}^N[\ul{\V{q}},\ul{\V{b}}]
+\SPb{U^{N,-}_{\kappa,s}[\ul{\V{q}}+\ul{\V{b}}_t,{}^t\ul{\V{b}}]}{
U^{N,+}_{\kappa,t}[\ul{\V{q}},\ul{\V{b}}]}_{\HP}&
\\\label{shift3}
=u_{\kappa,s+t}^N[\ul{\V{q}},\ul{\V{b}}]&,\quad s\ge0.
\end{align}
Moreover, there is a $\PP$-zero set $\sN_t$ such that \eqref{shift3} is valid
for $\kappa=\infty$ on $\Omega\setminus\sN_t$.
\end{lem}

\begin{proof}
For $\kappa\in\NN$, the asserted identity follows from the computation
\begin{align*}
u_{\kappa,s}^N&[\ul{\V{x}}+\ul{\V{\alpha}}_t,{}^t\ul{\V{\alpha}}]+sNE_\kappa^\ren
\\
&=\sum_{j,\ell=1}^N\int_0^s\SPb{e^{-i\V{m}\cdot(\V{x}_j+\V{\alpha}_{j,t})}
U^+_{\kappa,\tau}[^t\V{\alpha}_j]}{
e^{-i\V{m}\cdot(\V{x}_\ell+\V{\alpha}_{\ell,t+\tau})}f_\kappa}_{\HP}\Id\tau
\\
&=\sum_{j,\ell=1}^N\int_0^s\SPb{e^{-i\V{m}\cdot\V{x}_j}U^+_{\kappa,\tau+t}[\V{\alpha}_j]}{
e^{-i\V{m}\cdot(\V{x}_\ell+\V{\alpha}_{\ell,t+\tau})}f_\kappa}_{\HP}\Id\tau
\\
&\quad-\sum_{j,\ell=1}^N\int_0^s\SPb{e^{-i\V{m}\cdot\V{x}_j}U^+_{\kappa,t}[\V{\alpha}_j]}{
e^{-\tau\omega-i\V{m}\cdot(\V{x}_\ell+\V{\alpha}_{\ell,t+\tau})}f_\kappa}_{\HP}\Id\tau
\\
&=\sum_{j,\ell=1}^N\int_t^{s+t}\SPb{e^{-i\V{m}\cdot\V{x}_j}U^+_{\kappa,\tau}[\V{\alpha}_j]}{
e^{-i\V{m}\cdot(\V{x}_\ell+\V{\alpha}_{\ell,\tau})}f_\kappa}_{\HP}\Id\tau
\\
&\quad-\SPb{U^{N,+}_{\kappa,t}[\ul{\V{x}},\ul{\V{\alpha}}]}{
U^{N,-}_{\kappa,s}[\ul{\V{x}}+\ul{\V{\alpha}}_t,{}^t\V{\alpha}]}_{\HP},\quad s,t\ge0,
\end{align*}
valid for all $\ul{\V{x}}\in\RR^\nu$ and $\ul{\V{\alpha}}\in C([0,\infty),\RR^\nu)$,
where we used \eqref{shift2} in the first step.
The second statement now follows by approximation from Lem.~\ref{lemUminus}(2), 
Lem.~\ref{lemUplus}(2), Lem.~\ref{lemuinfty}, and Lem.~\ref{lemstrsol}, applied both to the 
original data $(\BB,\ul{\V{b}})$ as well as to the time shifted data $(\BB_t,{}^t\ul{\V{b}})$.
\end{proof}


\section{Feynman-Kac formulas}\label{secFK}

\noindent
In this section we derive Feynman-Kac formulas for the Nelson model. The section is divided into
four subsections. In the first one we recall some Fock space calculus and introduce a certain
family of operators on $\sF$ (see \eqref{michi-1}) that appear as building blocks in our
Feynman-Kac integrands. The latter integrands are defined and discussed in 
Subsect.~\ref{ssecFKint}. These integrands in turn give rise to Feyman-Kac semi-groups acting
between Fock space valued $L^p$-spaces whose analysis is the objective of 
Subsect.~\ref{ssecFKSG}. We shall encounter the first Feynman-Kac formulas in 
Subsect.~\ref{ssecFKreg}. There we consider finite $\kappa$ and verify that the corresponding
Feyman-Kac semi-group on $L^2(\RR^\nu,\sF)$ is equal to the semi-group generated by
$H_{N,\kappa}^V$. The final Subsect.~\ref{ssecFKren} consists of hardly more than our
definition of the ultra-violet renormalized Nelson Hamiltonian, which is introduced as
the generator of the Feyman-Kac semi-group on $L^2(\RR^\nu,\sF)$ in the case $\kappa=\infty$.
The lower bound on its spectrum and our formal proof of Nelson's Thm.~\ref{thmNelsonHam} will
then be immediate consequences of the earlier subsections.


\subsection{Some Fock space calculus}\label{ssecFock}

\noindent
We start by recalling some Fock space calculus and in particular the construction of the
Weyl representation on $\sF$; see, e.g., \cite{Parthasarathy1992} for more details. 

For any Hilbert space $\sK$, let $\sU(\sK)$ be the set of unitary operators on $\sK$
equipped with the topology corresponding to the strong convergence of bounded operators.
The cartesian product $E:=\HP\times\sU(\HP)$ equipped with the product topology
and the semi-direct product law $(f,Q)(g,R):=(f+Qg,QR)$ is called the Euclidean group over $\HP$.
Then the Weyl representation is a strongly continuous projective representation
$\sW:E\to\sU(\sF)$ satisfying
\begin{align}\label{Weyl1}
\sW(f,Q)\sW(g,R)&=e^{-i\Im\SPn{f}{Qg}_{\HP}}\sW\big((f,Q)(g,R)\big),
\end{align}
for all $(f,Q),(g,R)\in E$.
To recall its construction it is very convenient to work with exponential vectors in $\sF$,
which are defined by
\begin{align}\label{defexpv}
\zeta(h)&:=\big(1,h,2^\mh h^{\otimes_2},\ldots,(n!)^\mh h^{\otimes_n},\ldots\:\big)\in\sF, 
\quad h\in\HP,
\end{align}
with $h^{\otimes_n}(\V{k}_1,\ldots,\V{k}_n):=h(\V{k}_1)\dots h(\V{k}_n)$, $\lambda^{3n}$-a.e. for
every $n\in\NN$. They satisfy
\begin{align}\label{SPexpv}
\SPn{\zeta(g)}{\zeta(h)}&=e^{\SPn{g}{h}_{\HP}},\quad g,h\in\HP.
\end{align}
Since the set of all exponential vectors $\{\zeta(h):h\in\HP\}$ is linearly
independent, the prescription
\begin{align}\label{defWeylexpv}
\sW(f,Q)\zeta(h)&:=e^{-\|f\|_{\HP}^2/2-\SPn{f}{Qh}_{\HP}}\zeta(f+Qh),\quad f,h\in\HP,\,Q\in\sU(\HP),
\end{align}
uniquely defines bijective linear maps $\sW(f,Q):\sE\to\sE$, where 
$\sE:=\mathrm{span}_{\CC}\{\zeta(h):h\in\HP\}$ is dense in $\sF$.
These maps turn out to be isometric as well. Hence they have unique extensions to unitary
operators on $\sF$, which are again denoted by $\sW(f,Q)$. The semi-direct product rule
on $E$ eventually leads to the Weyl relations \eqref{Weyl1}. The following abbreviations are
customary,
\begin{align}\label{abbrevWeyl}
\sW(f):=\sW(f,\id),\;\;f\in\HP,\qquad\Gamma(Q):=\sW(0,Q),\;\;Q\in\sU(\HP).
\end{align}

Let $f\in\HP$ and $\vo$ be a self-adjoint multiplication operator in $\HP$. 
Then \eqref{Weyl1} and the strong continuity of $\sW$ imply that the maps
\begin{align}
\RR\ni t\longmapsto\sW(-itf),\quad\RR\ni t\longmapsto\Gamma(e^{-it\vo}),
\end{align}
are strongly continuous unitary groups. Their self-adjoint generators are denoted by
$\vp(f)$ and $\Id\Gamma(\vo)$, respectively. For instance, we then have
\begin{align}\label{SGdGamma}
e^{-t\Id\Gamma(\vo)}\zeta(h)&=\zeta(e^{-t\vo}h),\quad h\in\HP,\,t\ge0,\quad
\text{if $\vo\ge0$.}
\end{align}
Furthermore, it turns out that 
\begin{align}
\label{vpaad}
\vp(f)&=\ol{\ad(f)+a(f)},\quad f\in\HP.
\end{align}
Here, for every $f\in\HP$, the corresponding creation operator $\ad(f)$ and annihilation operator 
$a(f)$ are closed operators in $\sF$, which are mutually adjoint to each other.
The subspace $\sE$ is a core consisting only of analytic vectors for both of them and 
\begin{align}\label{adexpv}
\ad(f)\zeta(h)&=\frac{\Id}{\Id t}\zeta(h+tf)\Big|_{t=0},\quad\exp\{\ad(f)\}\zeta(h)=\zeta(h+f),
\\\label{aexpv}
a(f)\zeta(h)&=\SPn{f}{h}_{\HP}\zeta(h),\qquad\qquad\;\exp\{a(f)\}\zeta(h)=e^{\SPn{f}{h}_{\HP}}\zeta(h),
\end{align}
for all $f,h\in\HP$. Furthermore, 
\begin{align}\label{dGammaexpv}
\Id\Gamma(\vo)\zeta(h)&=\ad(\vo h)\zeta(h),\quad h\in\dom(\vo),
\end{align}
for every self-adjoint multiplication operator $\vo$ in $\HP$. 

In what follows we shall again work with the auxiliary one-boson Hilbert space $\mathfrak{k}$ and the 
norms $\|\cdot\|_t$, $t>0$, defined in \eqref{deffrk} and \eqref{defnormt}, respectively. The subsequent 
Lem.~\ref{lemmichael} actually is the reason why we introduced the norms $\|\cdot\|_t$.

One can show that, if $f\in\mathfrak{k}$, then $\dom(\Id\Gamma(\omega)^\eh)$
is contained in $\dom(\ad(f))\cap\dom(a(f))$ and both $\ad(f)$ and $a(f)$ map
$\dom(\Id\Gamma(\omega)^{s+\eh})$ into $\dom(\Id\Gamma(\omega)^{s})$, for every $s\ge0$.

\begin{lem}\label{lemmichael}
Let $t>0$, $m\in\NN_0$, and define $F_{m,t}:\mathfrak{k}^{m+1}\to\LO(\sF)$ by
\begin{align}\label{michi-1}
F_{m,t}(f_1,\ldots,f_m,g)
&:=\sum_{n=0}^\infty\frac{1}{n!}\ad(f_m)\ldots\ad(f_1)\ad(g)^ne^{-t\Id\Gamma(\omega)}.
\end{align}
Then the following holds true:
\begin{enumerate}[leftmargin=*]
\item[{\rm(1)}] $F_{m,t}$ is well-defined and analytic on $\mathfrak{k}^{m+1}$. Moreover, there 
exist $c_m>0$, depending only on $m$, and a universal constant $c>0$ such that
\begin{align}\label{bdmichi}
\|F_{m,t}(f_1,\ldots,f_m,g)\|&\le c_m e^{4\|g\|_{t}^2}\prod_{j=1}^m{\|f_j\|_{t}}.
\end{align}
\item[{\rm(2)}]
The derivative of $F_{0,t}$ at $g\in\mathfrak{k}$ applied to the tangent 
vector $f_1\in\mathfrak{k}$ is given by
\begin{align*}
d_gF_{0,t}(g)f_1&=F_{1,t}(f_1,g).
\end{align*}
\end{enumerate}
\end{lem}

\begin{proof}
For a proof of Parts~(1) and~(2) we refer to \cite[App.~6]{GMM2016}; more details are given in
\cite[App.~C]{Matte20XX}. In fact, in \cite{GMM2016}
the exponential $e^{4\|g\|_{t}^2}$ is replaced by the series $\sum_{n=0}^\infty(n!)^\mh(2\|g\|_{t})^n$. 
By a weighted Cauchy-Schwarz inequality the latter is, however, dominated by 
$\sqrt{2}e^{4\|g\|_{t}^2}$.
\end{proof}

\begin{rem}\label{remFadj}
In view of $(\ad(g))^*=a(g)$, \eqref{SGdGamma}, and \eqref{aexpv},
\begin{align}\label{Fadjexpv}
F_{0,t}(g)^*\zeta(h)&=e^{\SPn{g}{h}_{\HP}}\zeta(e^{-t\omega}h),\quad g\in\mathfrak{k},\,h\in\HP,\,t>0.
\end{align}
Scalar multiplying the previous identity with an exponential vector and comparing the result
with \eqref{SPexpv} (or using \eqref{SGdGamma} and \eqref{adexpv}), we further see that
\begin{align}\label{Fexpv}
F_{0,t}(g)\zeta(h)&=\zeta(e^{-t\omega}h+g),\quad g\in\mathfrak{k},\,h\in\HP,\,t>0.
\end{align}
\end{rem}


\subsection{Definition and discussion of the Feynman-Kac integrands}\label{ssecFKint}

\noindent
We are now in a position to define the Feynman-Kac integrands for the Nelson model.
We warn the reader that the {\em adjoints} of the operators introduced in the next definition
will appear in our Feynman-Kac formulas.

Recall that, since $V$ is Kato decomposable, the sets
\begin{align*}
\sN_V(\ul{\V{x}})&:=\big\{\gamma\in\Omega\,\big|\,V(\ul{\V{x}}+\smash{\ul{\V{b}}}_\bullet(\gamma))
\notin\cL^1_{\loc}([0,\infty))\big\}
\end{align*}
are measurable with $\PP(\sN_V(\ul{\V{x}}))=0$, for all $\ul{\V{x}}\in\RR^\nu$. Furthermore,
\begin{align}\label{Katobd}
\sup_{\ul{\V{x}}\in\RR^\nu}\EE\big[e^{p\int_0^tV_-(\ul{\V{x}}+\smash{\ul{\V{b}}}_s)\Id s}\big]
&\le e^{tc_{pV_-}},\quad p,t\ge0,
\end{align}
because $V_-$ is in the Kato class $K_\nu$.

\begin{defn}\label{defW}
Let $t\ge0$ and $\kappa\in\NN\cup\{\infty\}$.
\begin{enumerate}[leftmargin=*]
\item[{\rm(1)}]
For all $\ul{\V{x}}\in\RR^\nu$ and $\ul{\V{\alpha}}\in C([0,\infty),\RR^\nu)$, we define
\begin{align}\nonumber
W_{\kappa,t}[\ul{\V{x}},\ul{\V{\alpha}}]&:=e^{u_{\kappa,t}^N[\ul{\V{x}},\ul{\V{\alpha}}]}F_{0,\nf{t}{2}}
(-U_{\kappa,t}^{N,+}[\ul{\V{x}},\ul{\V{\alpha}}])
F_{0,\nf{t}{2}}(-U_{\kappa,t}^{N,-}[\ul{\V{x}},\ul{\V{\alpha}}])^*,
\end{align}
and, in case $V$ is {\em bounded}, we further set
\begin{align}\nonumber
W_{\kappa,t}^V[\ul{\V{x}},\ul{\V{\alpha}}]&:=
e^{-\int_0^tV(\ul{\V{x}}+\smash{\ul{\V{\alpha}}}_s)\Id s}W_{\kappa,t}[\ul{\V{x}},\ul{\V{\alpha}}].
\end{align}
\item[{\rm(2)}] For all $\fF_0$-measurable $\ul{\V{q}}:\Omega\to\RR^\nu$, we define
\begin{align}\nonumber
W_{\kappa,t}(\ul{\V{q}})&:=e^{u_{\kappa,t}^N(\ul{\V{q}})}F_{0,\nf{t}{2}}(-U_{\kappa,t}^{N,+}(\ul{\V{q}}))
F_{0,\nf{t}{2}}(-U_{\kappa,t}^{N,-}(\ul{\V{q}}))^*.
\end{align}
\item[{\rm(3)}] For every $\ul{\V{x}}\in\RR^\nu$, we set
\begin{align}\nonumber
W_{\kappa,t}^V(\ul{\V{x}})&:=1_{\sN_V(\ul{\V{x}})}
e^{-\int_0^tV(\ul{\V{x}}+\smash{\ul{\V{b}}}_s)\Id s}W_{\kappa,t}(\ul{\V{x}}).
\end{align}
\end{enumerate}
\end{defn}

\begin{rem}\label{remW}
Let $\kappa\in\NN\cup\{\infty\}$ and $\ul{\V{q}}:\Omega\to\RR^\nu$
be $\fF_0$-measurable. Then the following holds:
\begin{enumerate}[leftmargin=*]
\item[{\rm(1)}] There exists a $\PP$-zero set $\sN\in\fF$ such that
\begin{align*}
(W_{\kappa,t}(\ul{\V{q}}))(\gamma)&=W_{\kappa,t}[\ul{\V{q}}(\gamma),\ul{\V{b}}(\gamma)],
\quad t\ge0,\,\gamma\in\Omega\setminus\sN.
\end{align*}
This follows from Def.~\ref{defUN} and Lem.~\ref{lemstrsol}. (For finite $\kappa$ we can 
choose $\sN=\emptyset$.)
\item[{\rm(2)}] Let $s,t\ge0$ and $\ul{\V{x}}$. Let $(\fH_s^{[t]})_{s\ge0}$ denote the completion
of the filtration generated by the Brownian motion ${}^t\ul{\V{b}}$. Then 
$W_{\kappa,s}[\ul{\V{x}},{}^t\ul{\V{b}}]:\Omega\to\LO(\sF)$ is 
$\fH_s^{[t]}$-$\fB(\LO(\sF))$-measurable and attains its values in a separable subset of $\LO(\sF)$.
This follows from the continuity of $F_{0,r}:\mathfrak{k}\to\LO(\sF)$, $r>0$, the separability
of $\mathfrak{k}$, and the corresponding measurability properties of
$u_{\kappa}^N[\ul{\V{x}},{}^t\ul{\V{b}}]$ and $U_{\kappa}^{N,\pm}[\ul{\V{x}},{}^t\ul{\V{b}}]$;
recall Rem.~\ref{remUpmNcont}. 
In particular, $W_{\kappa,s}[\ul{\V{x}},{}^t\ul{\V{b}}]$ is $\fF_t$-independent. 
The same remarks hold for its adjoint $W_{\kappa,s}[\ul{\V{x}},{}^t\ul{\V{b}}]^*$.
\item[{\rm(3)}] On account of \eqref{Fadjexpv} and \eqref{Fexpv},
\begin{align}\label{Wexpvec}
W_{\kappa,t}(\ul{\V{q}})\zeta(h)&=e^{u_{\kappa,t}^N(\ul{\V{q}})
-\SPn{U^{N,-}_{\kappa,t}(\ul{\V{q}})}{h}_\HP}
\zeta\big(e^{-t\omega}h-U^{N,+}_{\kappa,t}(\ul{\V{q}})\big),
\\\label{Wadjexpvec}
W_{\kappa,t}(\ul{\V{q}})^*\zeta(h)&=
e^{u_{\kappa,t}^N(\ul{\V{q}})-\SPn{U^{N,+}_{\kappa,t}(\ul{\V{q}})}{h}_\HP}
\zeta\big(e^{-t\omega}h-U^{N,-}_{\kappa,t}(\ul{\V{q}})\big),
\end{align}
for all $h\in\HP$ and $t\ge0$. If also $g\in\HP$, then \eqref{SPexpv} and \eqref{Wexpvec} imply
\begin{align*}
\SPn{\zeta(g)}{W_{\kappa,t}(\ul{\V{q}})\zeta(h)}&=e^{u_{\kappa,t}^N(\ul{\V{q}})-
\SPn{U^{N,-}_{\kappa,t}(\ul{\V{q}})}{h}_\HP-\SPn{g}{U^{N,+}_{\kappa,t}(\ul{\V{q}})}_\HP
+\SPn{g}{e^{-t\omega}h}_\HP}.
\end{align*}
Setting $g=h=0$ we see in particular that the expectation value of $W_{\kappa,t}(\ul{\V{q}})$ 
with respect to the vacuum vector $\zeta(0)$ is just $e^{u_{\kappa,t}(\ul{\V{q}})}$.
\item[{\rm(4)}] In view of \eqref{bdmichi},
\begin{align}\label{pwbdW}
\|W_{\kappa,t}(\ul{\V{q}})\|&\le ce^{u_{\kappa,t}^N(\ul{\V{q}})+c\|U_{\kappa,t}^{N,+}(\ul{\V{q}})\|_t^2
+c\|U_{\kappa,t}^{N,-}(\ul{\V{q}})\|_t^2},\quad t>0,
\end{align}
for some universal constant $c>0$. Combining this with \eqref{expbdUpmN},
\eqref{expbdu}, and \eqref{Katobd} we conclude that
$\sup_{s\le t}\|W_{\kappa,s}^V(\ul{\V{x}})\|\in\cL^p(\Omega,\PP)$, for all $p,t>0$, with
\begin{align}\label{bdWp}
\sup_{\ul{\V{x}}\in\RR^\nu}\EE\big[\sup_{s\le t}\|W_{\kappa,s}^V(\ul{\V{x}})\|^p\big]
&\le e^{cp^2\ee^4N^3(1\vee t)+A(p\ee^2,N,t)+c_{2pV_-}t}.
\end{align}
Here $c>0$ is another universal constant and $A$ is logarithmically bounded in $t$ and contains
$N$-dependent terms of lower order than $N^3$,
\begin{align}\label{defAeps}
A(q,N,t)&:=cN\big(1+\ln[1+qN(1\vee t)]\big)+cqN^2(1+\ln[1\vee t]).
\end{align}
\item[{\rm(5)}] On account of Rem.~\ref{remUpmNcont},
Lem.~\ref{lemucont}, and Lem.~\ref{lemmichael}, the following $\LO(\sF)$-valued maps are
continuous, for every $\gamma\in\Omega$,
\begin{align}\label{Wxbcont}
(0,\infty)\ni t\mapsto(W_{\kappa,t}(\ul{\V{q}}))(\gamma),\quad
(0,\infty)\times\RR^\nu\ni(t,\ul{\V{x}})\mapsto W_{\kappa,t}[\ul{\V{x}},\ul{\V{b}}(\gamma)].
\end{align}
\item[{\rm(6)}] Let $\gamma\in\Omega$. If $\psi$ is a linear combination of exponential vectors, 
then \eqref{Wexpvec} (and an obvious analogue), Rem.~\ref{remUpmNcont}, 
and Lem.~\ref{lemucont} show that 
\begin{align}\label{Wpsixbcont}
[0,\infty)\ni t\mapsto(W_{\kappa,t}(\ul{\V{q}}))(\gamma)\psi,\quad
[0,\infty)\times\RR^\nu\ni(t,\ul{\V{x}})\mapsto W_{\kappa,t}[\ul{\V{x}},\ul{\V{b}}(\gamma)]\psi,
\end{align}
are continuous $\sF$-valued maps.
Employing \eqref{bdmichi}, \eqref{tnbUpmN}, and Lem.~\ref{lemucont}, we further see that
$\sup_{0<t\le n,|\ul{\V{x}}|\le n}\|W_{\kappa,t}[\ul{\V{x}},\ul{\V{b}}(\gamma)]\|<\infty$, $n\in\NN$.
Since the exponential vectors are total in $\sF$,
we conclude that the maps in \eqref{Wpsixbcont} are actually continuous, for all $\psi\in\sF$.
\item[{\rm(7)}] Assume in addition that $V$ is bounded and $\kappa\in\NN$. 
Let $t>0$ and pick some path $\ul{\V{\alpha}}\in C([0,\infty),\RR^\nu)$. Then
\eqref{revUpmN} and \eqref{revu} imply
\begin{align}\label{revW}
W_{\kappa,t}^V[\ul{\V{x}}+\ul{\V{\alpha}}_t,\ul{\V{\alpha}}_{t-\bullet}-\ul{\V{\alpha}}_t]
&=W_{\kappa,t}^V[\ul{\V{x}},\ul{\V{\alpha}}]^*.
\end{align}
If $\ul{\V{\alpha}}_0={0}$ and $\tilde{\ul{\V{\alpha}}}:=\ul{\V{\alpha}}_{t-\bullet}-\ul{\V{\alpha}}_t$, 
then we further deduce that
\begin{align}\nonumber
\int_{\RR^\nu}\SPb{\Phi(\ul{\V{x}})}{
&W_{\kappa,t}^V[\ul{\V{x}},\ul{\V{\alpha}}]^*\Psi(\ul{\V{x}}+\ul{\V{\alpha}}_t)}\Id\ul{\V{x}}
\\\label{symW0}
&=\int_{\RR^\nu}\SPb{W_{\kappa,t}^V[\ul{\V{x}},\tilde{\ul{\V{\alpha}}}]^*
\Phi(\ul{\V{x}}+\tilde{\ul{\V{\alpha}}}_t)}{\Psi(\ul{\V{x}})}\Id\ul{\V{x}},
\end{align}
for all measurable $\Phi,\Psi:\RR^\nu\to\sF$ such that one of the two integrals above exists.
\end{enumerate}
\end{rem}

\begin{prop}\label{propMarkovW}
Assume that $V$ is bounded. Let $\kappa\in\NN\cup\{\infty\}$, $\ul{\V{q}}:\Omega\to\RR^\nu$ 
be $\fF_0$-measurable, and $t\ge0$. Then the statement
\begin{align}\label{Markov0}
W_{\kappa,s}^V[\ul{\V{q}}+\ul{\V{b}}_t,{}^t\ul{\V{b}}]
W_{\kappa,t}^V[\ul{\V{q}},\ul{\V{b}}]&=W_{\kappa,s+t}^V[\ul{\V{q}},\ul{\V{b}}],\quad s\ge0,
\end{align}
holds on $\Omega$, if $\kappa\in\NN$, and on the complement of some $\PP$-zero set, 
if $\kappa=\infty$.
\end{prop}

\begin{proof}
Applying \eqref{Wexpvec} repeatedly we see that 
\begin{align}\label{Markov0b}
W_{\kappa,s}^V[\ul{\V{q}}+\ul{\V{b}}_t,{}^t\ul{\V{b}}]
W_{\kappa,t}^V[\ul{\V{q}},\ul{\V{b}}]\zeta(h)&=W_{\kappa,s+t}^V[\ul{\V{q}},\ul{\V{b}}]
\zeta(h),\quad s\ge0,\,h\in\HP,
\end{align}
is implied by \eqref{shift1N}, \eqref{shift2N}, and \eqref{shift3}, if $\kappa$ is finite,
and on the complement of a $\PP$-zero set by \eqref{shift1Ninfty}, \eqref{shift2Ninfty}, and 
Lem.~\ref{lemushift}, if $\kappa=\infty$.
Since the set of exponential vectors is total in $\sF$, this proves the proposition. 
\end{proof}

\begin{prop}\label{propconvW}
Let $\tau_2\ge\tau_1>0$ and $p>0$. Then
\begin{align}\label{convW}
\sup_{\ul{\V{x}}\in\RR^\nu}\EE\Big[\sup_{t\in[\tau_1,\tau_2]}\|W_{\kappa,t}^{V}(\ul{\V{x}})
-W_{\infty,t}^V(\ul{\V{x}})\|^p\Big]\xrightarrow{\;\;\kappa\to\infty\;\;}0,
\end{align}
where the convergence is uniform as $\ee$ varies in a compact subset of $\RR$ and as $V$ varies
in a set of Kato decomposable potentials satisfying 
\begin{align}\label{Katounif}
\sup_{\ul{\V{x}}\in\RR^\nu}\EE\big[e^{\tilde{p}\int_0^{\tau_2}V_-(\ul{\V{x}}+\smash{\ul{\V{b}}}_s)\Id s}
\big]\le A_{\tilde{p},\tau_2},\quad\tilde{p}>0,
\end{align}
with $V$-independent $A_{\tilde{p},\tau_2}>0$.
\end{prop}

\begin{proof}
Let $\kappa\in\NN$, $\ul{\V{x}}\in\RR^\nu$, $t>0$, and abbreviate
\begin{align*}
U_{\kappa,t}^{N,\pm}(\ul{\V{x}},\tau)&:=
\tau U_{\kappa,t}^{N,\pm}(\ul{\V{x}})+(1-\tau)U_{\infty,t}^{N,\pm}(\ul{\V{x}}),\quad\tau\in[0,1],
\end{align*}
so that $\partial_\tau U_{\kappa,t}^{N,\pm}(\ul{\V{x}},\tau)
=U_{\kappa,t}^{N,\pm}(\ul{\V{x}})-U_{\infty,t}^{N,\pm}(\ul{\V{x}})$. Employing 
Lem.~\ref{lemmichael}(2) in the first step and Lem.~\ref{lemmichael}(1) in the second one,
we then observe that, pointwise on $\Omega$,
\begin{align*}
\big\|F_{0,\nf{t}{2}}&(-U_{\kappa,t}^{N,\pm}(\ul{\V{x}}))
-F_{0,\nf{t}{2}}(-U_{\infty,t}^{N,\pm}(\ul{\V{x}}))\big\|
\\
&\le\int_0^1\big\|F_{1,\nf{t}{2}}\big(-\partial_\tau U_{\kappa,t}^{N,\pm}
(\ul{\V{x}},\tau),-U_{\kappa,t}^{N,\pm}(\ul{\V{x}},\tau)\big)\big\|\Id\tau
\\
&\le c' \|U_{\kappa,t}^{N,\pm}(\ul{\V{x}})-U_{\infty,t}^{N,\pm}(\ul{\V{x}})\|_{t}
e^{c\|U_{\infty,t}^{N,\pm}(\ul{\V{x}})\|_t^2+c\|U_{\kappa,t}^{N,\pm}(\ul{\V{x}})\|_t^2}.
\end{align*}
Therefore, if $0<\tau_1\le1$, $\tau_2\ge\tau_1$, and $p>0$, then
\begin{align*}
\sup_{\ul{\V{x}}\in\RR^\nu}&\EE\Big[\sup_{t\in[\tau_1,\tau_2]}
\big\|F_{0,\nf{t}{2}}(-U_{\kappa,t}^{N,\pm}(\ul{\V{x}}))-F_{0,\nf{t}{2}}(
-U_{\infty,t}^{N,\pm}(\ul{\V{x}}))\big\|^p\Big]
\\
&\le c_p\sup_{\ul{\V{x}}\in\RR^\nu}\EE\Big[\sup_{t\in[\tau_1,\tau_2]}
\|U_{\kappa,t}^{N,\pm}(\ul{\V{x}})-U_{\infty,t}^{N,\pm}(\ul{\V{x}})\|_{t}^{2p}\Big]^\eh
\\
&\qquad\;\cdot
\sup_{\tilde{\kappa}\in\NN\cup\{\infty\}}\sup_{\ul{\V{x}}\in\RR^\nu}\EE\Big[\sup_{t\in[\tau_1,\tau_2]}
e^{4cp\|U_{\tilde{\kappa},t}^{N,\pm}(\ul{\V{x}})\|_{t}^2}\Big]^\eh,
\end{align*}
where the right hand side goes to zero, as $\kappa\to\infty$, according to Cor.~\ref{corUpmN}.
Combining this result with \eqref{expbdUpmN},
\eqref{expbdu}, \eqref{convu}, \eqref{bdmichi}, \eqref{Katounif}, the formula
in Def.~\ref{defW}(2), telescopic summations, and H\"{o}lder's inequality we arrive at \eqref{convW}.
\end{proof}


\subsection{Definition and discussion of the Feynman-Kac semi-group}\label{ssecFKSG}

\begin{defn}\label{defnTV}
Let $\sV$ be the vector space of all measurable functions $\Psi:\RR^\nu\to\sF$ for which we find
$p\in[1,\infty]$ and $a\ge0$ such that $e^{-a|\cdot|}\Psi\in L^p(\RR^\nu,\sF)$. 
Let $\sM$ be the vector space of measurable functions from $\RR^\nu$ to $\sF$. For all 
$\kappa\in\NN\cup\{\infty\}$ and $t\ge0$, we define a linear map $T^V_{\kappa,t}:\sV\to\sM$ 
by setting
\begin{align}\label{defTV}
(T^V_{\kappa,t}\Psi)(\ul{\V{x}})&:=\EE\big[W_{\kappa,t}^V(\ul{\V{x}})^*
\Psi(\ul{\V{x}}+\ul{\V{b}}_t)\big],
\end{align}
for all $\ul{\V{x}}\in\RR^\nu$ for which 
$\|W_{\kappa,t}^V(\ul{\V{x}})^*\Psi(\ul{\V{x}}+\ul{\V{b}}_t)\|_\sF\in L^1(\PP)$,
and $(T^V_{\kappa,t}\Psi)(\ul{\V{x}}):=0$ otherwise.
\end{defn}

Since $\PP\{\ul{\V{b}}_t\in N\}=0$, for every Borel set $N\subset\RR^\nu$ with $\lambda^\nu(N)=0$,
$T^V_{\kappa,t}$ is also well-defined on the usual equivalence classes of functions in $\sV$.

In the following proposition we study the action of $T^V_{\kappa,t}$ in the spaces 
$L^p(\RR^\nu,\sF)$. In what follows we shall write
$\|\cdot\|_p$ both for the norm on $L^p(\RR^\nu,\sF)$ and on $L^p(\RR^\nu)$. Likewise,
$\|\cdot\|_{p,q}$ denotes both the operator norm on $\LO(L^p(\RR^\nu,\sF),L^q(\RR^\nu,\sF))$
and on $\LO(L^p(\RR^\nu),L^q(\RR^\nu))$, if $1\le p\le q\le\infty$. This should not cause any
confusion. As usual, $p'$ is the exponent conjugate to $p$ and $\infty^{-1}:=0$, $\infty/q:=\infty$, 
for $q\in(0,\infty)$. If $1\le p<\infty$, then $\langle\cdot,\cdot\cdot\rangle_{p,p'}$ stands for
the dual pairing between $L^p(\RR^\nu,\sF)$ and $L^{p'}(\RR^\nu,\sF)$.

The singularity at $t=0$ of the right hand side of \eqref{LpLq} in the next proposition is the same as
for Schr\"{o}dinger semi-groups without coupling to quantized fields \cite{Carmona1979,Simon1982}.

\begin{prop}\label{propT}
Let $\kappa\in\NN\cup\{\infty\}$ and $1\le p\le q\le\infty$. Suppose that
$F:\RR^\nu\to\RR$ is Lipschitz continuous with Lipschitz constant $L\ge0$ and
$\Psi:\RR^\nu\to\sF$ is measurable such that $e^F\Psi\in L^p(\RR^\nu,\sF)$.
Then the following holds:
\begin{enumerate}[leftmargin=*]
\item[{\rm(1)}] If $p>1$, then the expectation in \eqref{defTV} is absolutely convergent for all $t\ge0$ 
and $\ul{\V{x}}\in\RR^\nu$. Furthermore, $e^FT^V_t\Psi\in L^q(\RR^\nu,\sF)$, for all $t>0$, and
\begin{align}\label{LpLq}
\|e^FT_{\kappa,t}^V\Psi\|_q&\le c_{\nu,p,q}
\frac{e^{c_{p,q}L^2t+c_{p,q,V_-}t+c_{p,q}\ee^4N^3(1\vee t)
+c_{p,q}A(\ee^2,N,t)}}{t^{\nu(p^{-1}-q^{-1})/2}}\|e^F\Psi\|_p.
\end{align}
Here the constant $c_{p,q,V_-}\ge0$ satisfies $c_{p,q,0}=0$ and 
$A$ is defined in \eqref{defAeps}.
\item[{\rm(2)}] If $p=1$ and $t>0$, then the expectation in \eqref{defTV} is absolutely convergent for
a.e. $\ul{\V{x}}\in\RR^\nu$, we again have $e^FT^V_t\Psi\in L^q(\RR^\nu,\sF)$, and \eqref{LpLq}
still holds true.
\item[{\rm(3)}] For $\tau_2\ge\tau_1>0$,
\begin{align}\label{convT}
\sup_{t\in[\tau_1,\tau_2]}\|T_{\kappa,t}^V-T_{\infty,t}^V\|_{p,q}\xrightarrow{\;\;\kappa\to\infty\;\;}0.
\end{align}
The convergence is uniform as $\ee$ varies in a compact subset of $\RR$ and as $V$ varies
in a set of Kato decomposable potentials satisfying \eqref{Katounif} with
fixed $A_{\tilde{p},\tau_2}>0$.
\item[{\rm(4)}] If $p>1$ and $\Psi\in L^p(\RR^\nu,\sF)$, then the following Markov property holds:
For fixed $s,t\ge0$, and $\ul{\V{x}}\in\RR^\nu$, 
\begin{align}\label{Markov1}
\EE^{\fF_t}\big[W_{\kappa,s+t}^V(\ul{\V{x}})^*\Psi(\ul{\V{x}}+\ul{\V{b}}_{s+t})\big]
&=W_{\kappa,t}^V(\ul{\V{x}})^*(T_{\kappa,s}^V\Psi)(\ul{\V{x}}+\ul{\V{b}}_t),\quad\text{$\PP$-a.s.}
\end{align}
\item[{\rm(5)}] $(T_{\kappa,t}^V)_{t\ge0}$ defines a semi-group on $L^p(\RR^\nu,\sF)$.
\item[{\rm(6)}] If $p<\infty$, then the semi-group $(T_{\kappa,t}^V)_{t\ge0}$ is strongly continuous
in $L^p(\RR^\nu,\sF)$ and the following self-adjointness relation is satisfied,
\begin{align}\label{TVsad}
\langle T_{\kappa,t}^V\Phi,\Psi\rangle_{p,p'}&=\langle \Phi,T_{\kappa,t}^V\Psi\rangle_{p,p'},\quad
\Phi\in L^p(\RR^\nu,\sF),\,\Psi\in L^{p'}(\RR^\nu,\sF).
\end{align}
\end{enumerate}
\end{prop}

\begin{proof}
{\em Step~1.} Let $1<p\le q\le\infty$ and $\Psi_F:=e^F\Psi\in L^p(\RR^\nu,\sF)$. 
Suppose first that $p<\infty$ in addition. Then
\begin{align}\nonumber
e^{F(\ul{\V{x}})}&\EE\Big[\|W_{\kappa,t}^V(\ul{\V{x}})^*\|\|\Psi(\ul{\V{x}}+\ul{\V{b}}_t)\|_\sF\Big]
\\\label{anna1}
&\le\EE\big[e^{2Lp'|\ul{\V{b}}_t|}\big]^{\nf{1}{2p'}}\sup_{\ul{\V{y}}\in\RR^\nu}
\EE\big[\|W_{\kappa,t}^V(\ul{\V{y}})\|^{2p'}\big]^{\nf{1}{2p'}}
\EE\big[\|\Psi_F(\ul{\V{x}}+\ul{\V{b}}_t)\|_\sF^p\big]^{\nf{1}{p}},
\end{align}
for all $\ul{\V{x}}\in\RR^\nu$. On account of Rem.~\ref{remW}(4), the $L^q(\RR^\nu)$-norm of the
right hand side of \eqref{anna1} is less than or equal to some 
$(t,p',\ee,V_-,N)$-dependent constant (having the form of the numerator in \eqref{LpLq}) times
$$
c_{\nu,p}\big\|e^{t\Delta/2}\|\Psi_F(\cdot)\|_\sF^p\big\|_{\nf{q}{p}}^{\nf{1}{p}}
\le c_{\nu,p,q}t^{-\nu(p^{-1}-q^{-1})/2} \big\|\|\Psi_F(\cdot)\|_\sF^p\big\|_{1}^{\nf{1}{p}}.
$$
This proves Part~(1) for $p\in(1,\infty)$. For $p=q=\infty$, \eqref{anna1} is still valid if the last 
expectation in the second line is replaced by $\|\Psi_F\|_\infty$, which again leads to \eqref{LpLq}.

\smallskip

\noindent{\em Step~2.} Next, we consider the case $1=p\le q<\infty$, assuming that $V$ is bounded
for a start. Let $\kappa\in\NN$, $t>0$, and $f\in L^{q'}(\RR^3)$ be non-negative. 
If also $q>1$, then \eqref{revW} implies
\begin{align*}
&\int_{\RR^\nu} f(\ul{\V{x}})e^{F(\ul{\V{x}})}
\EE\big[\|W_{\kappa,t}^V(\ul{\V{x}})^*\|\|\Psi(\ul{\V{x}}+\ul{\V{b}}_t)\|_\sF\big]\Id\ul{\V{x}}
\\
&\le\EE\bigg[e^{L|\ul{\V{b}}_t|}\int_{\RR^\nu}f(\ul{\V{x}})
\big\|W_{\kappa,t}^V[\ul{\V{x}}+\ul{\V{b}}_t,\ul{\V{b}}_{t-\bullet}-\ul{\V{b}}_t]\big\|
\|\Psi_F(\ul{\V{x}}+\ul{\V{b}}_t)\|_\sF\Id\ul{\V{x}}\bigg]
\\
&=\EE\bigg[e^{L|\ul{\V{b}}_t|}\int_{\RR^\nu}f(\ul{\V{x}}-\ul{\V{b}}_t)
\big\|W_{\kappa,t}^V[\ul{\V{x}},\ul{\V{b}}_{t-\bullet}-\ul{\V{b}}_t]\big\|\|\Psi_F(\ul{\V{x}})\|_\sF
\Id\ul{\V{x}}\bigg]
\\
&\le\EE\big[e^{2qL|\ul{\V{b}}_t|}\big]^{\nf{1}{2q}}
\sup_{\ul{\V{y}}\in\RR^\nu}\EE\big[\|W_{\kappa,t}^V[\ul{\V{y}},\ul{\V{b}}_{t-\bullet}
-\ul{\V{b}}_t]\|^{2q}\big]^{\nf{1}{2q}}
\Big\{\sup_{\ul{\V{z}}\in\RR^\nu}\EE\big[f(\ul{\V{z}}-\ul{\V{b}}_t)^{q'}\big]^{\nf{1}{q'}}\Big\}\|\Psi_F\|_1
\\
&\le c_{\nu,q}e^{2qL^2t}\sup_{\ul{\V{y}}\in\RR^\nu}
\EE\big[\|W_{\kappa,t}^V[\ul{\V{y}},\ul{\V{b}}]\|^{2q}\big]^{\nf{1}{2q}}
\Big\{\|e^{t\Delta/2}\|_{1,\infty}^{\nf{1}{q'}}\|f^{q'}\|_1^{\nf{1}{q'}}\Big\}\|\Psi_F\|_1
\\
&\le c_{\nu,q}e^{2qL^2t+c_q\ee^4N^3(1\vee t)+c_qA(\ee^2,N,t)+c_{q,V_-}t}
t^{-\nu(1-q^{-1})/2}\|f\|_{q'}\|e^F\Psi\|_1=:C(f,\Psi),
\end{align*}
where we used the fact that the processes $(\ul{\V{b}}_s)_{s\in[0,t]}$ and 
$(\ul{\V{b}}_{t-s}-\ul{\V{b}}_t)_{s\in[0,t]}$ have the same distribution in the penultimate step. 
In the last step we applied \eqref{bdWp}.
If we replace the two terms in the big curly brackets $\{\cdots\}$ by $\|f\|_\infty$, 
then the above estimation is valid in the case $1=p=q$ as well. To include the case $\kappa=\infty$
we put $f_n:=1_{B_n}(n\wedge f)$, where $B_n$ is the open ball of radius $n\in\NN$ about
$0$ in $\RR^\nu$, and define random variables 
$g_n(\ul{\V{x}}):=n\wedge\|\Psi(\ul{\V{x}}+\ul{\V{b}}_t)\|_\sF$,
for all $\ul{\V{x}}\in\RR^\nu$ and $n\in\NN$. Then \eqref{convW} implies
\begin{align*}
\int_{\RR^\nu}&f_n(\ul{\V{x}})e^{F(\ul{\V{x}})}
\EE\big[\|W_{\infty,t}^V(\ul{\V{x}})^*\|g_n(\ul{\V{x}})\big]\Id\ul{\V{x}}
\\
&=\lim_{\kappa\to\infty}\int_{\RR^\nu}f_n(\ul{\V{x}})e^{F(\ul{\V{x}})}
\EE\big[\|W_{\kappa,t}^V(\ul{\V{x}})^*\|g_n(\ul{\V{x}})\big]\Id\ul{\V{x}}
\le C(f,\Psi),\quad n\in\NN.
\end{align*}
Hence, by monotone convergence,
\begin{align}\label{rudi13}
\int_{\RR^\nu} f(\ul{\V{x}})e^{F(\ul{\V{x}})}
\EE\big[\|W_{\kappa,t}^V(\ul{\V{x}})^*\|\|\Psi(\ul{\V{x}}+\ul{\V{b}}_t)\|_\sF\big]\Id\ul{\V{x}}\le C(f,\Psi),
\quad\kappa\in\NN\cup\{\infty\}.
\end{align}
If $V$ is unbounded, then we apply \eqref{rudi13} to $V_n^m:=(m\wedge V_+)-(n\wedge V_-)$,
observing that $C(f,\Psi)$ can be chosen independently of $m,n\in\NN$.
Then we pass to the limit $m\to\infty$ by dominated convergence, and to the limit $n\to\infty$
with the help of the monotone convergence theorem.
Altogether this shows that the integral in \eqref{defTV} is absolutely convergent 
for a.e. $\ul{\V{x}}\in\RR^\nu$ and proves \eqref{LpLq} for $1=p\le q<\infty$.

\smallskip

\noindent{\em Step~3.} Next, we prove Part~(3), first under the extra condition $p>1$. To this end
we just have to replace $W_{\kappa}^V$ by the difference $W_{\kappa}^V-W_\infty^V$ in
Step~1 and apply \eqref{convW}.

If instead $1=p\le q<\infty$ and if $V$ is bounded, then we obtain, as in the beginning of Step~2,
\begin{align}\nonumber
\int_{\RR^\nu} &f_n(\ul{\V{x}})\EE\big[\|W_{\kappa,t}^V(\ul{\V{x}})^*
-W_{\tilde{\kappa},t}^V(\ul{\V{x}})^*\|g_n(\ul{\V{x}})\big]\Id\ul{\V{x}}
\\\label{anne2}
&\le c_{\nu,q}t^{-\nu(1-q^{-1})/2}\sup_{\ul{\V{y}}\in\RR^\nu}
\EE\big[\|W_{\kappa,t}^V[\ul{\V{y}},\ul{\V{b}}]-
W_{\tilde{\kappa},t}^V[\ul{\V{y}},\ul{\V{b}}]\|^{2q}\big]^{\nf{1}{2q}}
\|f\|_{q'}\|\Psi\|_1,
\end{align}
for all $n,\kappa,\tilde{\kappa}\in\NN$, $\Psi\in L^1(\RR^\nu,\sF)$, and non-negative 
$f\in L^{q'}(\RR^\nu)$. By virtue of \eqref{convW} we first conclude that \eqref{anne2} is
available for $\tilde{\kappa}=\infty$ and $n,\kappa\in\NN$, too. After that we employ the 
monotone convergence theorem to pass to the limit $n\to\infty$ in \eqref{anne2} with
$\tilde{\kappa}=\infty$ and $\kappa\in\NN$. If $V$ is possibly unbounded and we apply this
procedure to every $V_n^m$ defined as in Step~2, then this results in
\begin{align}\nonumber
\sup_{t\in[\tau_1,\tau_2]}\int_{\RR^\nu} f(\ul{\V{x}})\EE\Big[\|W_{\kappa,t}^{V_n^m}(\ul{\V{x}})^*
-W_{\infty,t}^{V_n^m}(\ul{\V{x}})^*\|&\|\Psi(\ul{\V{x}}+\ul{\V{b}}_t)\|_\sF\Big]\Id\ul{\V{x}}
\\\label{anne278}
&\le o(\kappa)
\|f\|_{q'}\|\Psi\|_1,\quad\kappa\to\infty,
\end{align}
for all $m,n\in\NN$.
Here the little-$o$ symbol depends on $\tau_2\ge\tau_1>0$, $\nu$, and $q$. According to
Prop.~\ref{propconvW} it is, however, independent of $\ee$, when $\ee$ varies in a compact set, 
and it is independent of $V$, when $V$ varies in a set of Kato decomposable potentials as described 
in the statement of Prop.~\ref{propconvW}. In particular, $o(\kappa)$ is independent of $m$ and $n$.
Therefore, we may first pass to the limit $m\to\infty$ and after that to the $n\to\infty$ on the left hand
side of \eqref{anne278} by the same arguments as in the end of Step~2.
The resulting bound permits to get \eqref{convT} for $1=p\le q<\infty$.

We postpone the case $p=1$, $q=\infty$ to Step~6.

\smallskip

\noindent{\em Step~4.} For $\kappa\in\NN$ and bounded $V$, the self-adjontness relation 
\eqref{TVsad} follows upon substituting $\ul{\V{\alpha}}:=\ul{\V{b}}(\gamma)$ in \eqref{symW0}, 
for all $\gamma\in\Omega$, and taking the expectation of the so-obtained formula. 
It can be extended to unbounded $V$ by inserting the potentials $V_n^n$, $n\in\NN$, defined as 
in Step~2 and employing the dominated convergence theorem.
All necessary integrability properties are assured by Steps~1 and~2. 
Thanks to the by now available special case $p=q$ of Part~(3) we may then pass
to the limit $\kappa\to\infty$ in \eqref{TVsad}.

\smallskip

\noindent{\em Step~5.} If $p>1$ and $\Psi\in L^p(\RR^\nu,\sF)$, then \eqref{Markov1} 
with $V$ replaced by $V_n^n$, $n\in\NN$, defined as in Step~2 follows from
Prop.~\ref{propMarkovW} and the Markov property of Brownian motion.
In view of Def.~\ref{defW}(3), 
$W_{\kappa,t}^{V_n^n}(\ul{\V{x}})^*\to W_{\kappa,t}^{V}(\ul{\V{x}})^*$ in $\LO(\sF)$ and on
$\Omega$, as $n\to\infty$. The estimates of Step~1 and the dominated convergence theorem 
further ensure that $(T_{\kappa,t}^{V_n^n}\Psi)(\ul{\V{z}})\to(T_{\kappa,t}^{V}\Psi)(\ul{\V{z}})$, for
all $\ul{\V{z}}\in\RR^\nu$. Since 
$W_{\kappa,s+t}^{V_n^n}(\ul{\V{x}})^*\Psi(\ul{\V{x}}+\ul{\V{b}}_{s+t})
\to W_{\kappa,s+t}^V(\ul{\V{x}})^*\Psi(\ul{\V{x}}+\ul{\V{b}}_{s+t})$ in $L^1(\Omega,\sF;\PP)$
(by Step~1 and dominated convergence) and since the vector-valued conditional expectation
$\EE^{\fF_t}$ is a contractive projection on $L^1(\Omega,\sF;\PP)$, it finally follows that
$\EE^{\fF_t}[W_{\kappa,s+t}^{V_n^n}(\ul{\V{x}})^*\Psi(\ul{\V{x}}+\ul{\V{b}}_{s+t})]
\to\EE^{\fF_t}[W_{\kappa,s+t}^V(\ul{\V{x}})^*\Psi(\ul{\V{x}}+\ul{\V{b}}_{s+t})]$, $\PP$-a.s. along
a subsequence. Altogether this proves \eqref{Markov1}.

Taking the expectation of \eqref{Markov1} we see that 
$(T^V_{\kappa,t}(T_{\kappa,s}^V\Psi))(\ul{\V{x}})=(T_{\kappa,s+t}^V\Psi)(\ul{\V{x}})$, for all
$s,t\ge0$ and $\ul{\V{x}}\in\RR^\nu$. In particular, $(T^V_{\kappa,t})_{t\ge0}$ is a semi-group in
$L^p(\RR^\nu,\sF)$. By the duality relation \eqref{TVsad} it is a semi-group
in $L^1(\RR^\nu,\sF)$, too.

\smallskip

\noindent{\em Step~6.} We can now prove \eqref{LpLq} and Part~(3) in the case $p=1$, $q=\infty$,
not yet covered so far. Let $\kappa\in\NN\cup\{\infty\}$, $t>0$, and $s:=t/2$. Then we write
$T_{\kappa,t}^V=T_{\kappa,s}^VT_{\kappa,s}^V$ and apply Step~1 (with
$p=2$, $q=\infty$) to the left and Step~2 (with $p=1$, $q=2$) to the right factor 
$T_{\kappa,s}^V$, which completes the proof of Part~(2). It is now clear that
\begin{align*}
\|T_{\kappa,t}^V-T_{\infty,t}^V\|_{1,\infty}&\le\|T_{\kappa,s}^V-T_{\infty,s}^V\|_{2,\infty}
\|T_{\kappa,s}^V\|_{1,2}+\|T_{\infty,s}^V\|_{2,\infty}\|T_{\kappa,s}^V-T_{\infty,s}^V\|_{1,2},
\end{align*}
where the left hand side goes to zero, as $\kappa\to\infty$, by Steps~1, 2, and~3.

\smallskip

\noindent{\em Step~7.} On account of the semi-group properties, it only remains to show the 
asserted strong continuity at $t=0$. So assume
that $p\in[1,\infty)$. Let $\Psi:\RR^\nu\to\sF$ be a linear combination of vectors of the form
$f\zeta(h)$ with $h\in\HP$ and $f\in C_0^\infty(\RR^\nu)$. Then it follows from
the continuity of the processes $u_\kappa^N(\ul{\V{x}})$ and $U^{N,\pm}_{\kappa}(\ul{\V{x}})$ 
and from the formula \eqref{Wadjexpvec} that 
$W_{\kappa,t}^V(\ul{\V{x}})^*\Psi(\ul{\V{x}}+\ul{\V{b}}_t)\to\Psi(\ul{\V{x}})$, $t\downarrow0$,
on $\Omega$. At the same time, if $F(\ul{\V{x}}):=|\ul{\V{x}}|$, then 
$\ul{\V{x}}\mapsto e^{-F(\ul{\V{x}})}\sup_{s\le 1}(e^{|\ul{\V{b}}_s|}
\|W_{\kappa,s}^V(\ul{\V{x}})\|+1)\|e^F\Psi\|_\infty$ 
is a dominating function for every map
$\ul{\V{x}}\mapsto W_{\kappa,t}^V(\ul{\V{x}})^*\Psi(\ul{\V{x}}+\ul{\V{b}}_t)-\Psi(\ul{\V{x}})$ 
with $t\in(0,1]$, that belongs to $L^p(\RR^\nu\times\Omega,\lambda^\nu\otimes\PP)$ as a
consequence of \eqref{bdskuno} and \eqref{bdWp}. Since 
$$
\int_{\RR^\nu}\big\|\EE[\Phi(\ul{\V{x}})]\big\|^p\Id\ul{\V{x}}
\le\int_{\RR^\nu}\EE\big[\|\Phi(\ul{\V{x}})\|^p\big]\Id\ul{\V{x}},
$$
for every measurable $\Phi:\RR^\nu\to\sF$ with $\|\Phi(\ul{\V{x}})\|\in L^p(\PP)$, a.e. $\ul{\V{x}}$,
these remarks imply that $T^V_{\kappa,t}\Psi\to\Psi$, $t\downarrow0$, in $L^p(\RR^\nu,\sF)$.
Since $\Psi$ can be chosen in a dense subset of $L^p(\RR^\nu,\sF)$ and since
$\sup_{t\in[0,1]}\|T^V_{\kappa,t}\|_{p,p}<\infty$ in view of \eqref{LpLq}, this proves that
$(T^V_{\kappa,t})_{t\ge0}$ is strongly continuous (at $t=0$) in $L^p(\RR^\nu,\sF)$.
\end{proof}

The next remark turns out to be convenient at the end of the proof of Thm.~\ref{thmnonFockIRconv}.

\begin{rem}\label{remTVsadpq}
The relation \eqref{TVsad} can be generalized as follows:
Let $t>0$, $\kappa\in\NN\cup\{\infty\}$, $p\in[1,\infty]$, $q\in[p,\infty]$,
$\Phi\in L^{p}(\RR^\nu,\sF)$, and $\Psi\in L^{q'}(\RR^\nu,\sF)$. Then
\begin{align}\label{TVsadpq}
\langle T_{\kappa,t}^V\Phi,\Psi\rangle_{q,q'}&=\langle\Phi,T_{\kappa,t}^V\Psi\rangle_{p,p'}.
\end{align}

In fact, in view of \eqref{TVsad} it only remains to verify \eqref{TVsadpq} for $q'<\infty$.
In this case we set $\Psi_n:=1_{\{\|\Psi(\cdot)\|_{\sF}\le n\}}\Psi_n$, so that
$\Psi_n\in L^{q'}(\RR^\nu,\sF)\cap L^\infty(\RR^\nu,\sF)$, thus
$\Psi_n\in L^{p'}(\RR^\nu,\sF)$, for all $n\in\NN$. Then \eqref{TVsad} implies
$\langle T_{\kappa,t}^V\Phi,\Psi_n\rangle_{p,p'}=\langle\Phi,T_{\kappa,t}^V\Psi_n\rangle_{p,p'}$
However, since $T_{\kappa,t}^V\Phi\in L^q(\RR^\nu,\sF)$, the relation
$\langle T_{\kappa,t}^V\Phi,\Psi_n\rangle_{p,p'}=\langle T_{\kappa,t}^V\Phi,\Psi_n\rangle_{q,q'}$
holds by definition of the dual pairing. Now it suffices to observe that, as $n\to\infty$,
$\Psi_n\to\Psi$ in $L^{q'}(\RR^\nu,\sF)$ and $T_{\kappa,t}^V\Psi_n\to T_{\kappa,t}^V\Psi$
in $L^{p'}(\RR^\nu,\sF)$ by Prop.~\ref{propT}(1).
\end{rem}

\begin{cor}\label{corLpspec}
Let $p\in[1,\infty)$. Then we find constants $c_p>0$, $\tilde{c}_{p,V_-}\ge0$,
depending only on the quantities displayed in their subscripts and with $\tilde{c}_{p,0}=0$, 
such that, for every $\kappa\in\NN\cup\{\infty\}$, the resolvent set of the generator of 
$(T_{\kappa,t}^V)_{t\ge0}$, considered as a $C_0$-semi-group on $L^p(\RR^\nu,\sF)$, 
contains the interval
$$\big(-\infty,-c_p\ee^4N^3-\tilde{c}_{p,V_-}\big).$$
\end{cor}

\begin{proof}
This is a consequence of the Hille-Yosida theorem, Prop.~\ref{propT}(5)\&(6), and the bound
\eqref{LpLq}.
\end{proof}

As in the theory of Schr\"{o}dinger semi-groups with Kato decomposable potentials \cite{Simon1982}
we can actually show that the infima of the $L^p$-spectra are $p$-independent.


\subsection{Feynman-Kac formula for the ultra-violet regularized Hamiltonian}\label{ssecFKreg}

The Feynman-Kac formula for finite $\kappa$ asserted in the next theorem is
actually a special case of \cite[Thm.~11.3]{GMM2016}. Since the article loc. cit. also
covers the case of matter particles with spin that are {\em minimally} coupled to a quantized  
radiation field, the proof given there is, however, way more complicated than necessary for the 
Nelson model. For this reason we include the fairly short and simple proof of the following theorem.

\begin{thm}\label{thmFKreg}
Let $\kappa\in\NN$, $\Psi\in L^2(\RR^\nu,\sF)$, and $t\ge0$. Then
\begin{align}\label{FKreg}
(e^{-tH_{N,\kappa}^V-tNE_\kappa^\ren}\Psi)(\ul{\V{x}})&=(T_{\kappa,t}^V\Psi)(\ul{\V{x}}),
\quad\text{a.e. $\ul{\V{x}}\in\RR^\nu$}.
\end{align}
\end{thm}

\begin{proof}
{\em Step~1.}
Let $\ul{\V{x}}\in\RR^\nu$ and $h\in\dom(\omega)$. 
In view of \eqref{adexpv} and \eqref{Wexpvec} we then have
\begin{align}\nonumber
\frac{\Id}{\Id t}W_{\kappa,t}(\ul{\V{x}})\zeta(h)
&=\ad\big(-\omega e^{-t\omega}h-\tfrac{\Id}{\Id t}U^{N,+}_{\kappa,t}(\ul{\V{x}})\big)
W_{\kappa,t}(\ul{\V{x}})\zeta(h)
\\\label{sigrid1}
&\quad+\big(\tfrac{\Id}{\Id t}u_{\kappa,t}(\ul{\V{x}})
-\SPn{\tfrac{\Id}{\Id t}U^{N,-}_{\kappa,t}(\ul{\V{x}})}{h}\big)W_{\kappa,t}(\ul{\V{x}})\zeta(h),
\quad t\ge0,
\end{align}
where, according to Lem.~\ref{lemUpmcontHP}, \eqref{defUpmN}, Def.~\ref{defu},
and the definition of $f_{\kappa}^N(\ul{\V{x}})$ in \eqref{deffNkappa},
\begin{align}\label{sigrid2}
\frac{\Id}{\Id t}U_{\kappa,t}^{N,-}(\ul{\V{x}})&=e^{-t\omega}f^N_{\kappa}(\ul{\V{x}}+\ul{\V{b}}_{t}),
\\\label{sigrid3}
\frac{\Id}{\Id t}U_{\kappa,t}^{N,+}(\ul{\V{x}})&=-\omega U_{\kappa,t}^{N,+}(\ul{\V{x}})
+f^N_{\kappa}(\ul{\V{x}}+\ul{\V{b}}_{t}),
\\\label{sigrid4}
\frac{\Id}{\Id t}u_{\kappa,t}^N(\ul{\V{x}})&=
\SPn{f^N_{\kappa}(\ul{\V{x}}+\ul{\V{b}}_{t})}{U^{N,+}_{\kappa,t}(\ul{\V{x}})}-NE_\kappa^\ren.
\end{align}
For every $\ul{\V{y}}\in\RR^\nu$, we define a self-adjoint operator 
(with domain $\dom(\Id\Gamma(\omega))$) by
\begin{align*}
\wh{H}_\kappa(\ul{\V{y}})&:=\Id\Gamma(\omega)+\vp(f^N_{\kappa}(\ul{\V{y}})).
\end{align*}
Then \eqref{vpaad}, \eqref{aexpv}, and \eqref{dGammaexpv} imply
\begin{align*}
\wh{H}_\kappa(\ul{\V{x}}+\ul{\V{b}}_t)W_{\kappa,t}(\ul{\V{x}})\zeta(h)
&=\ad(\omega e^{-t\omega}h-\omega U^{N,+}_{\kappa,t}(\ul{\V{x}})
+f^N_{\kappa}(\ul{\V{x}}+\ul{\V{b}}_{t}))W_{\kappa,t}(\ul{\V{x}})\zeta(h)
\\
&\quad+\SPn{f^N_{\kappa}(\ul{\V{x}}+\ul{\V{b}}_{t})}{
e^{-t\omega}h-U^{N,+}_{\kappa,t}(\ul{\V{x}})}W_{\kappa,t}(\ul{\V{x}})\zeta(h),
\end{align*}
for all $t\ge0$. Comparing the previous formula with \eqref{sigrid1}--\eqref{sigrid4} we find
\begin{align}\label{DGLW}
\frac{\Id}{\Id t}W_{\kappa,t}(\ul{\V{x}})\zeta(h)
&=-(\wh{H}_\kappa(\ul{\V{x}}+\ul{\V{b}}_t)+NE_\kappa^\ren)
W_{\kappa,t}(\ul{\V{x}})\zeta(h),\quad t\ge0,\;\;\text{on $\Omega$.}
\end{align}
{\em Step~2.} In this step we assume in addition that $V$ is bounded. Then the process
$(e^{-\int_0^tV(\ul{\V{x}}+\smash{\ul{\V{b}}}_s)\Id s})_{t\ge0}$ 
has absolutely continuous paths whose
derivatives exist at a.e. $t\ge0$ and are given by the obvious formula.
Let $\ul{\V{x}}\in\RR^\nu$, $g\in C_0^\infty(\RR^3,\RR)$, and 
$\phi,\psi\in\tilde{\sE}_{\RR}:=\mathrm{span}_\RR\{\zeta(h):h\in\dom(\omega)\cap\HP_{\RR}\}$.
Then It\={o}'s formula, the symmetry of $\wh{H}_\kappa(\ul{\V{x}}+\ul{\V{b}}_s)$ on 
$\tilde{\sE}_\RR$, and \eqref{defNelsonHam} $\PP$-a.s. imply
\begin{align}\nonumber
\SPn{W_{\kappa,t}^V(\ul{\V{x}})\phi}{g(\ul{\V{x}}+\ul{\V{b}}_t)\psi}&=\SPn{\phi}{g(\ul{\V{x}})\psi}
\\\nonumber
&\quad-\int_0^t\SPb{W_{\kappa,s}^V(\ul{\V{x}})\phi}{((H_{N,\kappa}^V+NE_\kappa^\ren)g\psi)
(\ul{\V{x}}+\ul{\V{b}}_s)}\Id s
\\\label{sigrid5}
&\quad+\int_0^t\SPn{W_{\kappa,s}^V(\ul{\V{x}})\phi}{(\nabla g\psi)(\ul{\V{x}}
+\ul{\V{b}}_s)}\Id\ul{\V{b}}_s,\;\; t\ge0.
\end{align}
On account of \eqref{bdWp} the stochastic integral in the third line of \eqref{sigrid5} is an 
$L^2$-martingale, which permits to get
\begin{align}\label{diff1}
(T_{\kappa,t}^V\Psi)(\ul{\V{x}})-\Psi(\ul{\V{x}})
&=-\int_0^t\big(T_{\kappa,s}^V(H_{N,\kappa}^V+NE_\kappa^\ren)\Psi\big)(\ul{\V{x}})\Id s,
\quad t\ge0,\,\ul{\V{x}}\in\RR^\nu,
\end{align}
for every 
$\Psi\in \sD:=\mathrm{span}_\CC\{g\psi:g\in C_0^\infty(\RR^\nu,\RR),\psi\in\tilde{\sE}_\RR\}$. 
Since, for fixed $t\ge0$,
\begin{align*}
\int_0^t\big(T_{\kappa,s}^V(H_{N,\kappa}^V+NE_\kappa^\ren)\Psi\big)(\ul{\V{x}})\Id s&=
\Big(\int_0^tT_{\kappa,s}^V(H_{N,\kappa}^V+NE_\kappa^\ren)\Psi\Id s\Big)(\ul{\V{x}}),
\quad\text{a.e. $\ul{\V{x}}$,}
\end{align*}
where the integral on the right hand side is a Bochner-Lebesgue integral 
constructed in $L^2(\RR^\nu,\sF)$, we readily infer from \eqref{diff1} and the strong
coninuity of $(T^V_{\kappa,t})_{t\ge0}$ on $L^2(\RR^\nu,\sF)$ that
\begin{align*}
\big\|\tfrac{1}{t}(T_{\kappa,t}^V\Psi-\Psi)-(H_{N,\kappa}^V+NE_\kappa^\ren)\Psi\big\|
&\le\sup_{s\le t}\|(T_{\kappa,s}^V-\id)(H_{N,\kappa}^V+NE_\kappa^\ren)\Psi\|
\xrightarrow{\,t\downarrow0\,}0,
\end{align*}
for every $\Psi\in\sD$. This shows that $H^V_\kappa+NE_\kappa^\ren$ agrees
with the self-adjoint generator of the semi-group $(T_{\kappa,t}^V)_{t\ge0}$ 
on the domain $\sD$. Since $V$ is assumed to be bounded, 
we also know, however, that $H_{N,\kappa}^V+NE_\kappa^\ren$ is essentially self-adjoint on 
$\sD$. Therefore, $H_{N,\kappa}^V+NE_\kappa^\ren$ {\em is} the generator of 
$(T_{\kappa,t}^V)_{t\ge0}$ in $L^2(\RR^\nu,\sF)$, i.e., 
\eqref{FKreg} holds for all $\Psi\in L^2(\RR^\nu,\sF)$.

\smallskip

\noindent{\em Step~3.} Finally, we extend the Feynman-Kac formula \eqref{FKreg} from bounded
measurable $V$ to the general case of Kato decomposable $V$, following a standard procedure. 

Let $\Psi\in L^2(\RR^\nu,\sF)$ and $t>0$ be fixed in the rest of the proof. 

First, we assume in addition that $V$ is bounded from 
below and set $V_n:=n\wedge V$, $n\in\NN$. Then
$\EE[W_{\kappa,t}^{V_n}(\ul{\V{x}})^*\Psi(\ul{\V{x}}+\ul{\V{b}}_t)]
\to\EE[W_{\kappa,t}^{V}(\ul{\V{x}})^*\Psi(\ul{\V{x}}+\ul{\V{b}}_t)]$, $n\to\infty$, 
for all $\ul{\V{x}}\in\RR^\nu$, by dominated convergence and Prop.~\ref{propT}(1), 
while the monotone convergence of the quadratic forms 
$\mathfrak{q}_{N,\kappa}^{V_n}\uparrow\mathfrak{q}_{N,\kappa}^V$ on the 
domain $\dom(\mathfrak{q}_{N,\kappa}^V)$, which is dense in $L^2(\RR^\nu,\sF)$,
implies that $e^{-tH_{N,\kappa}^{V_n}}\Psi\to e^{-tH^V_{N,\kappa}}\Psi$ a.e. along a subsequence; 
see, e.g., \cite[Thm. VIII.20(b) and Thm.~S.14]{ReedSimonI}.
This proves \eqref{FKreg} in the case $\inf V>-\infty$.
 
Finally, we consider a general Kato-decomposable $V$ and set $V_n:=(-n)\vee V$, $n\in\NN$.
For every $\ul{\V{x}}\in\RR^\nu$, we have the domination
\begin{align}\label{dominique}
\|W_{\kappa,t}^{V_n}(\ul{\V{x}})^*\Psi(\ul{\V{x}}+\ul{\V{b}}_t)\|\le 
\|W_{\kappa,t}^V(\ul{\V{x}})\|\|\Psi(\ul{\V{x}}+\ul{\V{b}}_t)\|.
\end{align}
Here $\|\Psi(\ul{\V{x}}+\ul{\V{b}}_t)\|\in L^2(\PP)$ and $\|W_{\kappa,t}^V(\ul{\V{x}})\|\in L^2(\PP)$
by \eqref{bdWp}, so that the right hand side of \eqref{dominique}
is actually $\PP$-integrable. Therefore, 
$\EE[W_{\kappa,t}^{V_n}(\ul{\V{x}})^*\Psi(\ul{\V{x}}+\ul{\V{b}}_t)]
\to\EE[W_{\kappa,t}^{V}(\ul{\V{x}})^*\Psi(\ul{\V{x}}+\ul{\V{b}}_t)]$, $n\to\infty$, 
for every $\ul{\V{x}}\in\RR^\nu$. The monotone convergence of the quadratic forms 
$\mathfrak{q}_{N,\kappa}^{V_n}\downarrow\mathfrak{q}_{N,\kappa}^V$ on 
$\dom(\mathfrak{q}_{N,\kappa}^{V})=\dom(\mathfrak{q}_{N,\kappa}^{V_+})
=\bigcup_{n\in\NN}\dom(\mathfrak{q}_{N,\kappa}^{V_n})$
implies, however, that $e^{-tH_{N,\kappa}^{V_n}}\Psi\to e^{-tH^V_{N,\kappa}}\Psi$ 
a.e. along a subsequence; see, e.g., \cite[Thm. VIII.20(b) and Thm.~S.16]{ReedSimonI}.
\end{proof}

\begin{rem}
For finite $\kappa\in\NN$, \eqref{DGLW} actually implies the pointwise operator norm bound
$\ln\|W_{\kappa,t}(\ul{\V{x}})\|\le {\|\omega^\mh f^N_{\kappa}(\ul{\V{x}})\|_{\HP}^2t}-tNE_\kappa^\ren$
on $\Omega$, which is non-uniform in $\kappa$; see \cite[Thm.~5.3]{GMM2016}. 
Thanks to this it is possible to extend Thm.~\ref{thmFKreg} to
a larger class of potentials; see \cite[Thm.~11.3]{GMM2016}. We restrict ourselves to Kato
decomposable potentials, because our analysis requires bounds like \eqref{bdWp} which is
uniform in $\kappa$ and holds for $\kappa=\infty$ as well.
\end{rem}


\subsection{Feynman-Kac formula without ultra-violet cut-off}\label{ssecFKren}

In this short subsection we complete our independent construction of the Nelson Hamiltonian 
without ultra-violet cut-off. The Feynman-Kac formula for it will actually hold by definition.

\begin{thm}\label{thmconvTkappa}
For every $t\ge0$, the sequence 
$\{e^{-tH_{N,\kappa}^V-tNE_\kappa^{\ren}}\}_{\kappa\in\NN}$ converges in operator norm
to $T_{\infty,t}^V$. The convergence is uniform as $t$ varies in a compact subset of $(0,\infty)$.
\end{thm}

\begin{proof}
Combine Prop.~\ref{propT}(3) (with $p=q=2$) and Thm.~\ref{thmFKreg}.
\end{proof}

In particular, the following definition makes sense:

\begin{defn}\label{defHinfty}
The self-adjoint generator of the semi-group $(T_{\infty,t}^V)_{t\ge0}$ acting in 
$L^2(\RR^\nu,\sF)$ is called
the (ultra-violet renormalized) Nelson Hamiltonian. It is denoted by $H_{N,\infty}^V$.
\end{defn}

\begin{rem}
As the semi-group $(T^V_{\infty,t})_{t\ge0}$ is given by explicit formulas, and not just as an abstract
limiting object, our definition of the ultra-violet renormalized Nelson Hamiltonian $H_{N,\infty}^V$
does not depend on the choice of any cutoff function. The independence of $H_{N,\infty}^V$ 
(up to finite energy shifts) on the choice of cutoff functions (in a certain class at least) 
has been observed earlier in \cite[Prop.~3.9]{Ammari2000}.
\end{rem}

\begin{cor}\label{corlbspecH}
We find a universal constant $c>0$ and some $\tilde{c}_{V_-}\ge0$, depending only on $V_-$
with $\tilde{c}_0=0$, such that
\begin{align}\label{lbspecH}
\inf\sigma(H_{N,\infty}^V)\ge-c\ee^4N^3-\tilde{c}_{V_-}.
\end{align}
The number on the right hand side is also a lower bound on the spectra of 
all operators $H_{N,\kappa}^V+NE_\kappa^{\ren}$ with $\kappa\in\NN$.
\end{cor}

\begin{proof}
Combine Cor.~\ref{corLpspec}, Thm.~\ref{thmFKreg}, and Def.~\ref{defHinfty}.
\end{proof}

Now a standard argument finishes our independent proof of Thm.~\ref{thmNelsonHam}:

\begin{proof}[Proof of Thm.~\ref{thmNelsonHam}]
Since the formula $(A+1)^{-1}=\int_0^\infty e^{-t(A+1)}\Id t$ is valid for every non-negative
self-adjoint operator $A$ in some Hilbert space, Thm.~\ref{thmconvTkappa} and 
Cor.~\ref{corlbspecH} entail the convergence $H_{N,\kappa}^V+NE_\kappa^{\ren}\to H_{N,\infty}^V$, 
$\kappa\to\infty$, in norm resolvent sense.
\end{proof}


\section{The renormalized Nelson model in the non-Fock representation}\label{secnonFock}

\noindent
In this section we provide the first non-perturbative construction of the renormalized Nelson
Hamiltonian in a non-Fock representation; see \cite{Arai2001} for a discussion of ultra-violet
regularized Nelson Hamiltonians in non-Fock representations. In a perturbative setting, a
renormalized Nelson Hamiltonian in the non-Fock representation has been constructed in
\cite{HHS2005}; see Rem.~\ref{remHHS}. As in the previous section we shall first
define and analyze the corresponding semi-group and define the renormalized Hamiltonian as its
generator. This procedure does not necessitate any smallness assumptions on $|\ee|$ like the 
KLMN theorem used in \cite{HHS2005}.

The non-Fock representation is obtained in two steps. First, we associate a
unitary transformation to an infra-red cut-off parameter $\Lambda>0$, which by now is commonly
called Gross transformation. Nelson employed this transformation in \cite{Nelson1964} and called
it approximate dressing transformation. We learned from \cite{HHS2005} that it essentially 
goes back to Tomonaga. After that we remove the infra-red cut-off in the transformed
semi-group. While the Gross transformations themselves do not have a limit as 
$\Lambda\downarrow0$, we shall find a well-defined limiting semi-group. (The term ``non-Fock'' 
actually originates in the effect that, in the limit $\Lambda\downarrow0$, the Gross transformations
give rise to a new representation of the Weyl relations inequivalent to the one induced by $\sW$; 
see \cite{Arai2001}. Despite of this nomenclature, all semi-groups 
constructed below still act on Fock space-valued $L^p$-spaces.)

Although they are not unitarily equivalent,
the spectra of the renormalized Hamiltonians in the original and the non-Fock representation agree;
see Rem.~\ref{remnonFock}(2) below. Hence, if one is interested in a certain property of the spectrum
as a set, then one can equally well work in the non-Fock representation. The latter has the pleasant 
feature that, even without any infra-red regularizations, the massless Nelson model can have ground 
states in the non-Fock representation 
\cite{Arai2001,HHS2005,LorincziMinlosSpohn2002b,Panati2009}, while 
this is not the case for the original massless Nelson model 
\cite{LHB2011,LorincziMinlosSpohn2002,Panati2009}. 
The existence of a ground state can, for instance,
facilitate the analysis of binding energies \cite{HainzlHirokawaSpohn2005}.

Let us start our constructions by defining the Gross transformations in a slightly more general setting;
recall the definition of $\beta^N_{\Lambda,\kappa}(\ul{\V{x}})$ and the Weyl representation $\sW$
in \eqref{deffNkappa} and Subsect.~\ref{ssecFock}, respectively.

\begin{defn}\label{defnGross}
Let $\kappa\in\NN\cup\{\infty\}$ and $\Lambda\ge0$. In the case $\Lambda=0$ we assume in addition
that $\eta$ is chosen such that $\omega^{-3}\eta^2$ is integrable in a neighborhood of $0$.
For all $p\in[1,\infty]$, we then define a Gross 
transformation $G_{\Lambda,\kappa}$ on $L^p(\RR^\nu,\sF)$ by
\begin{align}\label{defGrosstrafo}
(G_{\Lambda,\kappa}\Psi)(\ul{\V{x}})&:=\sW(\beta^N_{\Lambda,\kappa}(\ul{\V{x}}))\Psi(\ul{\V{x}}),
\quad\Psi\in L^p(\RR^\nu,\sF),\;\text{a.e. $\ul{\V{x}}\in\RR^\nu$.}
\end{align}
\end{defn}

In view of \eqref{Weyl1}, $G_{\Lambda,\kappa}:L^p(\RR^\nu,\sF)\to L^p(\RR^\nu,\sF)$ is isometric 
and surjective, for all $\kappa\in\NN\cup\{\infty\}$, $\Lambda\ge0$, and $p\in[1,\infty]$, with
\begin{align*}
(G_{\Lambda,\kappa}^{-1}\Psi)(\ul{\V{x}})&=\sW(-\beta^N_{\Lambda,\kappa}(\ul{\V{x}}))\Psi(\ul{\V{x}}),
\quad\Psi\in L^p(\RR^\nu,\sF),\;\text{a.e. $\ul{\V{x}}\in\RR^\nu$.}
\end{align*}
For $p\in[1,\infty)$, the adjoint of $G_{\Lambda,\kappa}\!\!\upharpoonright_{L^p(\RR^\nu,\sF)}$ is given 
by $G_{\Lambda,\kappa}^{-1}\!\!\upharpoonright_{L^{p'}(\RR^\nu,\sF)}$, 
where $p'$ is the exponent conjugate to $p$. 

Next, we introduce the stochastic processes encountered in the transformed semi-group and, 
after that, the transformed semi-group itself. We shall see in Prop.~\ref{propnonFock} below that
the following formulas are actually the correct ones; recall the notation 
\eqref{forbTheta}--\eqref{defcminus} and \eqref{defuinfty}.

\begin{defn}\label{defnnonFockuU}
For all $\kappa\in\NN\cup\{\infty\}$, $t,\Lambda\ge0$, and $\ul{\V{x}}\in\RR^\nu$, we abbreviate
\begin{align}\label{defnonFockUminus}
\wt{U}^{N,-}_{\Lambda,\kappa,t}(\ul{\V{x}})&:=\{\beta_{\Lambda,\kappa}^N(\ul{\V{x}})-
e^{-t\omega}\beta_{\Lambda,\kappa}^N(\ul{\V{x}}+\ul{\V{b}}_t)\}-U_{\kappa,t}^{N,-}(\ul{\V{x}}),
\\\label{defnonFockUplus}
\wt{U}^{N,+}_{\Lambda,\kappa,t}(\ul{\V{x}})&:=\{\beta_{\Lambda,\kappa}^N(\ul{\V{x}}+\ul{\V{b}}_t)
-e^{-t\omega}\beta_{\Lambda,\kappa}^N(\ul{\V{x}})\}-U_{\kappa,t}^{N,+}(\ul{\V{x}}),
\\\label{defnonFocku}
\tilde{u}_{\Lambda,\kappa,t}^N(\ul{\V{x}})&:={u}_{\kappa,t}^N(\ul{\V{x}})
-b_{\Lambda,\kappa,t}^N(\ul{\V{x}})
+c_{\Lambda,\kappa,t}^{N,-}(\ul{\V{x}})+c_{\Lambda,\kappa,t}^{N,+}(\ul{\V{x}}).
\end{align}
In the case $\Lambda=0$ these objects will be denoted as
\begin{align}\nonumber
\tilde{u}_{\kappa,t}^N(\ul{\V{x}}):=\tilde{u}_{0,\kappa,t}^N(\ul{\V{x}}),\quad
\quad\wt{U}^{N,\pm}_{\kappa,t}(\ul{\V{x}}):=\wt{U}^{N,\pm}_{0,\kappa,t}(\ul{\V{x}}).
\end{align}
\end{defn}

Again we notice that the whole terms inside the curly brackets in \eqref{defnonFockUminus} and 
\eqref{defnonFockUplus} are continuous adapted $\mathfrak{k}$-valued processes, as the
phase differences $e^{-i\V{m}\cdot\V{y}_\ell}-e^{-t\omega-i\V{m}\cdot\V{z}_\ell}$ compensate 
for the infra-red singularity of $\beta_\kappa$. 
The separate terms of the differences inside $\{\cdots\}$ do in general not even belong to $\HP$. 

Finally, we observe that $\wt{U}^{N,-}_{\kappa,t}(\ul{\V{x}})$ is a square-integrable 
$\HP$-valued martingale,
\begin{align*}
\wt{U}^{N,-}_{\kappa,t}(\ul{\V{x}})&=\sum_{\ell=1}^N
e^{-i\V{m}\cdot\V{x}_\ell}M_{\kappa,t}^-[\V{b}_\ell].
\end{align*}

\begin{defn}\label{defnnonFockSG}
Let $\kappa\in\NN\cup\{\infty\}$, $t,\Lambda\ge0$, and $\ul{\V{x}}\in\RR^\nu$. Then we set
\begin{align*}
\wt{W}_{\Lambda,\kappa,t}(\ul{\V{x}})&:=e^{\tilde{u}_{\Lambda,\kappa,t}^N(\ul{\V{x}})}
F_{0,\nf{t}{2}}(\wt{U}^{N,+}_{\Lambda,\kappa,t}(\ul{\V{x}}))
F_{0,\nf{t}{2}}(\wt{U}^{N,-}_{\Lambda,\kappa,t}(\ul{\V{x}}))^*.
\end{align*}
For every element $\Psi$ of the vector space $\sV$ introduced in Def.~\ref{defnTV}, we put
\begin{align}\label{defwtTV}
(\wt{T}_{\Lambda,\kappa,t}^V\Psi)(\ul{\V{x}})&:=
\EE\big[e^{-\int_0^tV(\ul{\V{x}}+\ul{\V{b}}_s)\Id s}\wt{W}_{\Lambda,\kappa,t}(\ul{\V{x}})^*
\Psi(\ul{\V{x}}+\ul{\V{b}}_t)\big],
\end{align}
provided that the expectation converges absolutely, and 
$(\wt{T}_{\Lambda,\kappa,t}^V\Psi)(\ul{\V{x}}):=0$ otherwise. We further abbreviate
\begin{align*}
\wt{W}_{\kappa,t}(\ul{\V{x}})&:=\wt{W}_{0,\kappa,t}(\ul{\V{x}}),\quad
\wt{T}_{\kappa,t}^V:=\wt{T}_{0,\kappa,t}^V.
\end{align*}
\end{defn}

We shall see in Prop.~\ref{propnonFock} and Thm.~\ref{thmnonFockIRconv}
that the expectation in \eqref{defwtTV} converges absolutely for at least a.e. $\ul{\V{x}}$
and that the restriction of $\wt{T}_{\Lambda,\kappa,t}^V$ to $L^p(\RR^\nu,\sF)$, $p\in[1,\infty]$, is a
bounded linear $L^q(\RR^\nu,\sF)$-valued map, for every $q\in[p,\infty]$.

Let us, however, first derive a general transformation formula:

\begin{lem}\label{lemtrafo}
Let $f^+,f^-\in\mathfrak{k}$, $g,\tilde{g}\in\HP$, and $R$, $Q$, $\tilde{Q}$ be unitary operators 
on $\HP$ commuting with every $e^{-s\omega}$, $s\ge0$. Let $t>0$, assume that
\begin{align}\label{trafoWhyp}
g-e^{-t\omega}R^*\tilde{Q}^*\tilde{g}\in\mathfrak{k},\quad
\tilde{g}-e^{-t\omega}\tilde{Q}Rg\in\mathfrak{k},
\end{align} 
and abbreviate
\begin{align*}
\alpha&:=-\|g\|_{\HP}^2/2-\|\tilde{g}\|_{\HP}^2/2+\SPn{f^-}{g}_{\HP}+\SPn{\tilde{g}}{\tilde{Q}f^+}_{\HP}
+\SPn{\tilde{g}}{e^{-t\omega}\tilde{Q}Rg}_{\HP}.
\end{align*}
Then $Q^*$ and $\tilde{Q}$ map $\mathfrak{k}$ into itself and
the following operator identity holds true,
\begin{align}\nonumber
\sW&(\tilde{g},\tilde{Q})F_{0,\nf{t}{2}}(-f^+)\Gamma(R)F_{0,\nf{t}{2}}(-f^-)^*\sW(-g,Q)
\\\nonumber
&=e^\alpha F_{0,\nf{t}{2}}(\tilde{g}-e^{-t\omega}\tilde{Q}Rg-\tilde{Q}f^+)
\Gamma(\tilde{Q}RQ)
\\\label{trafoW}
&\quad\times F_{0,\nf{t}{2}}(Q^*g-e^{-t\omega}Q^*R^*\tilde{Q}^*\tilde{g}-Q^*f^-)^*.
\end{align}
\end{lem}

\begin{proof}
That $Q^*\mathfrak{k},\tilde{Q}\mathfrak{k}\subset\mathfrak{k}$ follows easily from
$[e^{-s\omega},\tilde{Q}]=[e^{-s\omega},Q^*]=0$, $s\ge0$, and the relations
$\|\omega^\mh f\|_{\HP}=\lim_{\ve\downarrow0}\|(\omega+\ve)^\mh f\|_{\HP}$, $f\in\mathfrak{k}$, and
$(\omega+\ve)^\mh h=\int_0^\infty e^{-s\ve-s\omega}h\Id s/\sqrt{\pi s}$, $h\in\HP$.

To prove the asserted operator identity
we pick an exponential vector $\zeta(h)$ with $h\in\HP$. Then the defining relation 
\eqref{defWeylexpv} for the Weyl representation and \eqref{Fadjexpv} entail
\begin{align*}
F_{0,\nf{t}{2}}(-f^-)^*\sW(-g,Q)\zeta(h)&=e^{-\|g\|_{\HP}^2/2+\SPn{g}{Qh}_{\HP}-\SPn{f^-}{Qh-g}_{\HP}}
\zeta(e^{-t\omega/2}Qh-e^{-t\omega/2}g),
\end{align*}
whence \eqref{Fexpv} permits to get
\begin{align*}
F_{0,\nf{t}{2}}&(-f^+)\Gamma(R)F_{0,\nf{t}{2}}(-f^-)^*\sW(-g,Q)\zeta(h)
\\
&=e^{-\|g\|_{\HP}^2/2+\SPn{g}{Qh}_{\HP}-\SPn{f^-}{Qh-g}_{\HP}}
\zeta(e^{-t\omega}RQh-e^{-t\omega}Rg-f^+).
\end{align*}
Applying \eqref{defWeylexpv} once more we arrive at
\begin{align*}
\sW&(\tilde{g},\tilde{Q})
F_{0,\nf{t}{2}}(-f^+)\Gamma(R)F_{0,\nf{t}{2}}(-f^-)^*\sW(-g,Q)\zeta(h)
\\
&=e^{-\|g\|_{\HP}^2/2-\|\tilde{g}\|_{\HP}^2/2+\SPn{g}{Qh}_{\HP}-\SPn{f^-}{Qh-g}_{\HP}
-\SPn{\tilde{g}}{\tilde{Q}e^{-t\omega}RQh-\tilde{Q}e^{-t\omega}Rg-\tilde{Q}f^+}_{\HP}}
\\
&\quad\times
\zeta(\tilde{Q}e^{-t\omega}RQh+\tilde{g}-\tilde{Q}e^{-t\omega}Rg-\tilde{Q}f^+)
\\
&=e^{\alpha+\SPn{Q^*g-e^{-t\omega}Q^*R^*\tilde{Q}^*\tilde{g}-Q^*f^-}{h}_{\HP}}
\\
&\quad\times
\zeta(e^{-t\omega}\tilde{Q}RQh+\tilde{g}-e^{-t\omega}\tilde{Q}Rg-\tilde{Q}f^+).
\end{align*}
On account of \eqref{Fadjexpv}, \eqref{Fexpv}, and the condition \eqref{trafoWhyp} the previous 
identity shows that \eqref{trafoW} holds true on the linear hull generated by all
exponential vectors. By continuity it then extends to an identity in $\LO(\sF)$.
\end{proof}

\begin{prop}\label{propnonFock}
Let $\kappa\in\NN\cup\{\infty\}$, $\Lambda\ge0$, $t\ge0$, and $\ul{\V{x}}\in\RR^3$. In the case
$\mu=\Lambda=0$ assume in addition that $\eta$ is chosen such that 
$\omega^{-3}\eta^2$ is integrable 
in a neighborhood of $0$. Then $\beta_{\Lambda,\kappa}\in\HP$ and
the following identity holds on $\Omega$,
\begin{align}\label{idGross}
\wt{W}_{\Lambda,\kappa,t}(\ul{\V{x}})&=\sW(\beta_{\Lambda,\kappa}^N(\ul{\V{x}}+\ul{\V{b}}_t))
{W}_{\kappa,t}(\ul{\V{x}})\sW(-\beta_{\Lambda,\kappa}^N(\ul{\V{x}})).
\end{align}
If $1\le p\le q\le\infty$, then $\wt{T}_{\Lambda,\kappa,t}^V$ is a well-defined element of
$\LO(L^p(\RR^\nu,\sF),L^q(\RR^\nu,\sF))$ satisfying
\begin{align}\label{idGrossT}
\wt{T}_{\Lambda,\kappa,t}^V=G_{\Lambda,\kappa} T_{\kappa,t}^VG_{\Lambda,\kappa}^{-1}.
\end{align}
\end{prop}

\begin{proof}
The relation \eqref{idGross} follows from Def.~\ref{defW}, Def.~\ref{defnnonFockSG}, and
Lem.~\ref{lemtrafo} with $R=Q=\tilde{Q}=\id$ and obvious choices of $f^\pm$. 
Notice that, if $g=\beta_{\Lambda,\kappa}^N(\ul{\V{x}})$ and
$\tilde{g}=\beta_{\Lambda,\kappa}^N(\ul{\V{x}}+\ul{\V{b}}_t)$, then \eqref{trafoWhyp} is satisfied
according to the remarks in the paragraph after Def.~\ref{defnnonFockuU}.

Employing \eqref{idGross} and  Prop.~\ref{propT}(1)\&(2) we first conclude that the expectation in
\eqref{defwtTV} with $\Psi\in L^p(\RR^\nu,\sF)$
is absolutely convergent for every $\ul{\V{x}}\in\RR^\nu$, if $p\in(1,\infty]$, and for a.e. $\ul{\V{x}}$, if
$p=1$. After that we readily observe the validity of \eqref{idGrossT}.
\end{proof}

\begin{thm}\label{thmnonFockIRconv}
The following assertions hold true:
\begin{enumerate}[leftmargin=*]
\item[{\rm(1)}] Let $t>0$, $\kappa\in\NN\cup\{\infty\}$,
$F:\RR^\nu\to\RR$ be Lipschitz continuous with Lipschitz constant $L\ge0$, and
$\Psi:\RR^\nu\to\sF$ be measurable with $e^F\Psi\in L^p(\RR^\nu,\sF)$, for some $p\in[1,\infty]$.
Then the expectation in \eqref{defwtTV} converges absolutely for at least 
a.e. $\ul{\V{x}}\in\RR^\nu$, even in the case $\mu=\Lambda=0$ without any additional assumption 
on $\eta$. It converges absolutely for {\em all} $\ul{\V{x}}\in\RR^\nu$ in case $p>1$.
Furthermore, $e^F\wt{T}_{\Lambda,\kappa,t}^V\Psi\in L^q(\RR^\nu,\sF)$ and
the bound \eqref{LpLq} holds with $T_{\kappa,t}^V$ replaced by
$\wt{T}_{\Lambda,\kappa,t}^V$, for all $\Lambda\ge0$ and $q\in[p,\infty]$.
\item[{\rm(2)}]
For all $\Lambda\ge0$, $\kappa\in\NN\cup\{\infty\}$, and $t>0$, the restriction of 
$\wt{T}_{\Lambda,\kappa,t}^V$ to $L^p(\RR^\nu,\sF)$ with $p\in[1,\infty]$ belongs to
$\LO(L^p(\RR^\nu,\sF),L^q(\RR^\nu,\sF))$, for every $q\in[p,\infty]$, and
\begin{align}\label{normTwtT}
\|\wt{T}_{\Lambda,\kappa,t}^V\|_{p,q}=\|{T}_{\kappa,t}^V\|_{p,q}.
\end{align}
\item[{\rm(3)}] For all $1\le p\le q\le\infty$,
\begin{align*}
\sup_{\kappa\in\NN\cup\{\infty\}}\sup_{t\in[\tau,T]}
\|\wt{T}_{\Lambda,\kappa,t}^{V}-\wt{T}_{\kappa,t}^V\|_{p,q}
\xrightarrow{\;\;\Lambda\downarrow0\;\;}0,\quad0<\tau\le T.
\end{align*}
\item[{}\rm(4)]
For all $p\in[1,\infty)$, $\Psi\in L^p(\RR^\nu,\sF)$, and $\tau>0$,
\begin{align*}
\lim_{\Lambda\downarrow0}\sup_{\kappa\in\NN\cup\{\infty\}}\sup_{t\in[0,\tau]}
\|\wt{T}_{\Lambda,\kappa,t}^{V}\Psi-\wt{T}_{\kappa,t}^V\Psi\|_p&=0.
\end{align*}
\item[{\rm(5)}] 
For all $\Lambda\ge0$ and $\kappa\in\NN\cup\{\infty\}$, $(\wt{T}_{\Lambda,\kappa,t}^V)_{t\ge0}$ 
is a semi-group on every $L^p(\RR^\nu,\sF)$, $p\in[1,\infty]$. For $p\in[1,\infty)$, it is 
strongly continuous. 
\item[{\rm(6)}] For all $t>0$, $\Lambda\ge0$, $\kappa\in\NN\cup\{\infty\}$, and $1\le p\le q\le\infty$,
\begin{align}\label{wtTVsadpq}
\langle\wt{T}_{\Lambda,\kappa,t}^V\Phi,\Psi\rangle_{q,q'}
&=\langle\Phi,\wt{T}_{\Lambda,\kappa,t}^V\Psi\rangle_{p,p'},
\quad\Phi\in L^{p}(\RR^\nu,\sF),\,\Psi\in L^{q'}(\RR^\nu,\sF).
\end{align}
\item[{\rm(7)}]
For all $1\le p\le q\le\infty$, 
\begin{align*}
\sup_{\Lambda\ge0}\sup_{t\in[\tau,T]}\|\wt{T}_{\Lambda,\kappa,t}^V-\wt{T}_{\Lambda,\infty,t}^V\|_{p,q}
\xrightarrow{\;\;\kappa\to\infty\;\;}0,\quad0<\tau\le T.
\end{align*}
\end{enumerate}
\end{thm}

Before we prove this theorem we give the formal definition of the renormalized Nelson Hamiltonian 
in the non-Fock representation and make two remarks.

\begin{defn}\label{defnHinftynonFock}
For all $\kappa\in\NN\cup\{\infty\}$ and $\Lambda\ge0$, 
the self-adjoint generator of the semi-group $(\wt{T}_{\Lambda,\kappa,t}^V)_{t\ge0}$ 
on the Hilbert space $L^2(\RR^\nu,\sF)$ is denoted by $\wt{H}_{N,\Lambda,\kappa}^V$ and we write
$\wt{H}_{N,\kappa}^V:=\wt{H}_{N,0,\kappa}^V$ for short.
If $\mu=0$, $\eta=1$, and $\ee\not=0$, then $\wt{H}_{N,\infty}^V$ is called the 
renormalized Nelson Hamiltonian in the non-Fock representation.
\end{defn}

\begin{rem}\label{remnonFock}
\begin{enumerate}[leftmargin=*]
\item[{(1)}]
Combining Parts~(3) and~(6) of Thm.~\ref{thmnonFockIRconv}, we see that
\begin{align*}
\wt{T}_{\infty,t}^V&=\lim_{\Lambda\downarrow0}\lim_{\kappa\to\infty}
\wt{T}_{\Lambda,\kappa,t}^{V}=\lim_{\kappa\to\infty}\lim_{\Lambda\downarrow0}
\wt{T}_{\Lambda,\kappa,t}^{V},
\quad t>0,
\end{align*}
in $\LO(L^p(\RR^\nu,\sF),L^q(\RR^\nu,\sF))$ with $1\le p\le q\le\infty$. 
\item[(2)] Let $\kappa\in\NN\cup\{\infty\}$. If $\Lambda>0$, then
$\wt{H}_{N,\Lambda,\kappa}^{V}$ and $H_{N,\kappa}^V$ are unitarily equivalent, because their 
semi-groups are intertwined by the Gross transformation which is unitary on $L^2(\RR^\nu,\sF)$.
Therefore, it follows from Prop.~\ref{propT}(3), Thm.~\ref{thmnonFockIRconv}(3),
and general principles \cite[Thm.~VIII.23(a) \& Thm.~VIII.24(a)]{ReedSimonI} that
\begin{align*}
\sigma(\wt{H}_{N,\kappa}^V)&=\sigma(H_{N,\kappa}^V).
\end{align*}
\end{enumerate}
\end{rem}

\begin{rem}\label{remHHS}
For $N=1$, $V(\V{x})=-c\ee/|\V{x}|$, and sufficiently small $\ee>0$, Hirokawa et al. \cite{HHS2005}
proved that the limit of $G_{\Lambda,\kappa}H_{N,\kappa}^VG_{\Lambda,\kappa}^*$, 
as $\kappa\to\infty$ and $\Lambda\downarrow0$, exists in the norm resolvent sense. In view of
Rem.~\ref{remnonFock}(1) their limit operator agrees with $\wt{H}_{1,\infty}^V$ in this case.
\end{rem}

\begin{proof}[Proof of Thm.~\ref{thmnonFockIRconv}]
Throughout the proof, $\kappa\in\NN\cup\{\infty\}$ and $\ul{\V{x}}\in\RR^\nu$ are arbitrary and
no constant will depend on these quantities.

\smallskip

\noindent 
{\em Step~1.}
Let $t,\Lambda\ge0$. Then
\begin{align*}
|\SPn{\omega^{-\nf{1}{2}}U^{N,-}_{\kappa,t}(\ul{\V{x}})}{1_{\{|\V{m}|<1\}}\omega^{\nf{1}{2}}
\beta_{\Lambda,\kappa}^N(\ul{\V{x}})}|&\le\ee^2N^2\int_{\{|\V{m}|<1\}}\int_0^t
\frac{e^{-s\omega}}{\omega(\omega+\V{m}^2/2)}\Id s\Id\lambda^3
\\
&\le c\ee^2N^2(1+\ln(1\vee t)),
\end{align*}
which together with \eqref{morten1} (where we choose $\Lambda=1$ and a suitable $\eta$) implies
\begin{align}\label{kuno2}
|c_{\Lambda,\kappa,t}^{N,\pm}(\ul{\V{x}})|&\le c\ee^2N^2(1+\ln(1\vee t)),
\end{align}
if the minus-sign is chosen; the other case is shown analogously. Since 
$b_{\Lambda,\kappa,t}^N(\ul{\V{x}})\ge0$, a combination of \eqref{expbdu} and \eqref{kuno2} yields
\begin{align}\label{kuno4}
\EE\big[\sup_{s\le t}e^{\tau p\tilde{u}_{\Lambda,\kappa,s}^N(\ul{\V{x}})
+(1-\tau)p\tilde{u}_{\kappa,s}^N(\ul{\V{x}})}\big]
\le c^Ne^{c'p\ee^2N^2(1+\ln(1\vee t))+c'p^2\ee^4N^3t},
\end{align}
for all $\tau\in[0,1]$ and $p>0$, with universal constants $c,c'>0$.

\smallskip

\noindent{\em Step~2.} 
Let $\Lambda>0$. We next estimate the difference of $\tilde{u}_{\kappa}^N(\ul{\V{x}})$ and
$\tilde{u}_{\Lambda,\kappa}^N(\ul{\V{x}})$. To this end
we first observe that $b_{\Lambda,\kappa}^N(\ul{\V{x}})$ can be written as
\begin{align}\nonumber
b_{\Lambda,\kappa,s}^N(\ul{\V{x}})
&=\frac{1}{2}\|\beta_{\Lambda,\kappa}^N(\ul{\V{x}})-
\beta_{\Lambda,\kappa}^N(\ul{\V{x}}+\ul{\V{b}}_s)\|_{\HP}^2
\\\label{bfor}
&\quad+\int_{\{|\V{m}|\ge\Lambda\}}(1-e^{-s\omega})\Re\sum_{j,\ell=1}^N
e^{i\V{m}\cdot(\V{x}_\ell-\V{x}_j-\V{b}_{j,s})}\beta_\kappa^2\Id\lambda^3,\quad s\ge0.
\end{align}
Note that this expression is well-defined even for $\Lambda=0$,
because the various differences of exponential functions compensate for the infra-red 
singularity of $\beta_\kappa$. (Only the {\em difference} inside the norm in the first line of
\eqref{bfor} is $\HP$; the individual terms do in general not belong to $\HP$.)
Employing \eqref{bfor} we obtain
\begin{align}\nonumber
|\tilde{u}_{\kappa,s}^N&(\ul{\V{x}})-\tilde{u}_{\Lambda,\kappa,s}^N(\ul{\V{x}})|
\\\nonumber
&\le|c_{\kappa,s}^{N,-}(\ul{\V{x}})-c_{\Lambda,\kappa,s}^{N,-}(\ul{\V{x}})|
+|c_{\kappa,s}^{N,+}(\ul{\V{x}})-c_{\Lambda,\kappa,s}^{N,+}(\ul{\V{x}})|
\\\label{kuno5}
&\quad+\frac{1}{2}\|1_{\{|\V{m}|<\Lambda\}}(\beta_{\kappa}^N(\ul{\V{x}})-\beta_{\kappa}^N
(\ul{\V{x}}+\ul{\V{b}}_s))\|_{\HP}^2
+sN^2\int_{\{|\V{m}|<\Lambda\}}\omega\beta_\infty^2\Id\lambda^3.
\end{align}
Elementary estimations then show that the first, second, and fourth term on the right hand 
side of \eqref{kuno5} are bounded from above by some universal constant times 
$\ee^2N^2\Lambda t$, provided that $s\in[0,t]$. For later use we next consider an expression that is 
slightly more general than the third term on the right hand side of \eqref{kuno5}.
Let $\iota\in\{0,1\}$. For every $s>0$, we then find that
\begin{align}\nonumber
\big\|&1_{\{|\V{m}|<\Lambda\}}(s\omega)^{-\nf{\iota}{2}}(\beta_{\kappa}^N(\ul{\V{x}})-\beta_{\kappa}^N
(\ul{\V{x}}+\ul{\V{b}}_s))\big\|_{\HP}^2
\\\nonumber
&=\int_{\{|\V{m}|<\Lambda\}}\frac{1}{(s\omega)^{\iota}}\Big|\sum_{j=1}^N(e^{-i\V{m}\cdot\V{x}_j}-
e^{-i\V{m}\cdot(\V{x}_j+\V{b}_{j,s})})\beta_{\kappa}\Big|^2\Id\lambda^3
\\\nonumber
&\le\frac{\ee^2N}{s^\iota}\sum_{j=1}^N\int_{\{|\V{m}|<\Lambda\}}\frac{|1-e^{-i\V{m}\cdot\V{b}_{j,s}}|^2}{
\omega^{1+\iota}(\omega+\V{m}^2/2)^2}\Id\lambda^3
\\\nonumber
&\le\frac{\ee^2N}{s^\iota}\sum_{j=1}^N|\V{b}_{j,s}|^{1+\nf{\iota}{2}}\int_{\{|\V{m}|<\Lambda\}}
\frac{|\V{m}|^{\nf{\iota}{2}}}{\omega^\iota}\cdot
\frac{\Id\lambda^3}{(\omega+\V{m}^2/2)^2}
\\\label{kuno2000}
&\le c\ee^2\cdot\left\{\begin{array}{ll}
N^{\nf{3}{2}}\Lambda|\ul{\V{b}}_s|,&\iota=0,
\\
N^{\nf{5}{4}}\Lambda^\eh s^{-1}|\ul{\V{b}}_s|^{\nf{3}{2}},&\iota=1,
\end{array}
\right.
\end{align}
whence the Burkholder inequality 
\begin{align}\label{BurkBM}
\EE\big[\sup_{s\le t}|\ul{\V{b}}_s|^q\big]&\le c_q(Nt)^{\nf{q}{2}},\quad t,q>0,
\end{align}
implies, for all $t\ge\tau>0$ and $p>0$,
\begin{align}\label{kuno2001a}
\EE\Big[\sup_{s\le t}\big\|1_{\{|\V{m}|<\Lambda\}}(\beta_{\kappa}^N(\ul{\V{x}})-\beta_{\kappa}^N
(\ul{\V{x}}+\ul{\V{b}}_s))\big\|_{\HP}^{2p}\Big]
&\le c_p(\ee^2N^2\Lambda t^\eh)^{p},
\\\label{kuno2001b}
\EE\bigg[\sup_{\tau\le s\le t}\Big\|\frac{1_{\{|\V{m}|<\Lambda\}}}{(s\omega)^\eh}
(\beta_{\kappa}^N(\ul{\V{x}})-\beta_{\kappa}^N(\ul{\V{x}}+\ul{\V{b}}_s))\Big\|_{\HP}^{2p}\bigg]
&\le c_p(\ee^2N^2\Lambda^\eh\tau^{-1}t^{\nf{3}{4}})^{p}.
\end{align}
Here \eqref{kuno2001b} is used in Step~4 below. Altogether we conclude that
\begin{align}\label{kuno6a}
\EE\Big[\sup_{s\le t}|\tilde{u}_{\kappa,s}^N(\ul{\V{x}})-\tilde{u}_{\Lambda,\kappa,s}^N(\ul{\V{x}})|^p\Big]
&\le c_p'(\ee^2N^2\Lambda (t\vee t^\eh))^{p},\quad t,p>0.
\end{align}
Combining \eqref{kuno4} and \eqref{kuno6a} we further arrive at
\begin{align}\label{kuno6}
\sup_{\kappa\in\NN\cup\{\infty\}}\sup_{\ul{\V{x}}\in\RR^\nu}
\EE\Big[\sup_{s\le t}|e^{\tilde{u}_{\kappa,s}^N(\ul{\V{x}})}
-e^{\tilde{u}_{\Lambda,\kappa,s}^N(\ul{\V{x}})}|^p\Big]
&\xrightarrow{\;\;\Lambda\downarrow0\;\;}0,\quad t>0.
\end{align}
\noindent{\em Step~3.}
Let $\Lambda\ge0$, $p>0$, and $t\ge\tau>0$.
Employing \eqref{bdmichi}, \eqref{defnonFockUminus}, and \eqref{defnonFockUplus} 
we next observe that
\begin{align}\nonumber
\EE\big[\sup_{\tau\le s\le t}&\|F_{0,\nf{s}{2}}(\wt{U}_{\Lambda,\kappa,s}^{N,\pm}(\ul{\V{x}}))\|^p\big]
\le c^p\EE\big[\sup_{\tau\le s\le t}e^{8p\|\wt{U}_{\Lambda,\kappa,s}^{N,\pm}(\ul{\V{x}})\|^2_s}\big]
\\\nonumber
&\le c^p\Big(\sup_{s\le t}\sup_{\ul{\V{z}}\in\RR^\nu}e^{c'p\|(1-e^{-s\omega})
\beta_{\Lambda,\kappa}^N(\ul{\V{z}})\|^2_s}\Big)
\\\label{kuno8aa}
&\quad\cdot
\EE\big[\sup_{\tau\le s\le t}e^{c'p\|\beta_{\Lambda,\kappa}^N(\ul{\V{x}})
-\beta_{\Lambda,\kappa}^N(\ul{\V{x}}+\ul{\V{b}}_s)\|^2_s}\big]^\eh\EE\big[
\sup_{s\le t}e^{c'p\|{U}_{\kappa,s}^{N,\pm}(\ul{\V{x}})\|^2_s}\big]^\eh.
\end{align}
The usual elementary estimates (similar to the last two steps in \eqref{julia1}) reveal that
\begin{align}\label{kuno6b}
\|(1-e^{-s\omega})\beta_{\Lambda,\kappa}^N(\ul{\V{z}})\|^2_s
&\le c\ee^2N^2(1+\ln(1\vee s)),\quad s>0.
\end{align}
If we set $\Lambda=\infty$ in the first four lines of \eqref{kuno2000}, then the $\lambda^3$-integral
in the fourth line is still finite. Combining the so-obtained bounds for $\iota=0$ and $\iota=1$ we
deduce that, for all $0<\tau\le s\le t$,
\begin{align}\label{kuno7a}
cp\|\beta_{\Lambda,\kappa}^N(\ul{\V{x}})
-\beta_{\Lambda,\kappa}^N(\ul{\V{x}}+\ul{\V{b}}_s)\|^2_{\HP}
&\le c'p\ee^2N^{\nf{3}{2}}|\ul{\V{b}}_s|\le c''p^2\ee^4N^3t+\frac{|\ul{\V{b}}_s|^2}{4t},
\\\nonumber
cp\|\beta_{\Lambda,\kappa}^N(\ul{\V{x}})
-\beta_{\Lambda,\kappa}^N(\ul{\V{x}}+\ul{\V{b}}_s)\|^2_s
&\le c'p\ee^2\big(N^{\nf{3}{2}}|\ul{\V{b}}_s|+N^{\nf{5}{4}}s^{-1}|\ul{\V{b}}_s|^{\nf{3}{2}}\big)
\\\label{kuno7}
&\le c''p^2\ee^4N^3\big(t+p^2\ee^4N^2\tau^{-4}t^3\big)+\frac{|\ul{\V{b}}_s|^2}{4t}.
\end{align}
Here $c',c''>0$ depend only on the constant $c>0$ on the left hand sides. From \eqref{bdskuno}, 
\eqref{expbdUpmN}, \eqref{kuno6b}, \eqref{kuno7a}, and \eqref{kuno7}, we now infer that
\begin{align}\nonumber
\sup_{\Lambda\ge0}&\sup_{\kappa\in\NN\cup\{\infty\}}
\EE\big[\sup_{\tau\le s\le t}e^{p\|\wt{U}_{\Lambda,\kappa,s}^{N,\pm}(\ul{\V{x}})\|^2_{\HP}}\big]
\\\label{kuno8a}
&\le c^N(1+p\ee^2N(1\vee t))^Ne^{c'p\ee^2N^2(1+\ln(1\vee t))+cp^2\ee^4N^3(1\vee t)},
\\\label{kuno8}
\sup_{\Lambda\ge0}&\sup_{\kappa\in\NN\cup\{\infty\}}
\EE\big[\sup_{\tau\le s\le t}e^{p\|\wt{U}_{\Lambda,\kappa,s}^{N,\pm}(\ul{\V{x}})\|^2_s}\big]
\le c_{p,N,\ee,\tau,t}<\infty.
\end{align}
{\em Step~4.} Let $\tau,p,\Lambda>0$ and $t\ge\tau$. Then 
\begin{align}\nonumber
\EE\Big[&\sup_{\tau\le s\le t}\big\|\wt{U}_{\Lambda,\kappa,s}^{N,\pm}(\ul{\V{x}})
-\wt{U}_{\kappa,s}^{N,\pm}(\ul{\V{x}})\big\|_s^p\Big]
\\\nonumber
&\le c_p\EE\Big[\sup_{\tau\le s\le t}\big\|1_{\{|\V{m}|<\Lambda\}}
(\beta_{\kappa}^N(\ul{\V{x}})-\beta_{\kappa}^N(\ul{\V{x}}+\ul{\V{b}}_s))\big\|_s^p\Big]
\\\nonumber
&\quad+c_p\sup_{s\le t}\sup_{\ul{\V{z}}\in\RR^\nu}
\|1_{\{|\V{m}|<\Lambda\}}(1-e^{-s\omega})\beta_{\kappa}^N(\ul{\V{z}})\|_s^p
\\\label{kuno9}
&\le c_p'(\ee^2N^2\Lambda t^\eh)^{\nf{p}{2}}+c_p'(\ee^2N^2\Lambda^\eh\tau^{-1} 
t^{\nf{3}{4}})^{\nf{p}{2}}+c_p'(\ee^2N^2\Lambda t)^{\nf{p}{2}},
\end{align}
by \eqref{kuno2001a}, \eqref{kuno2001b}, and straightforward estimations applied to the term in
the third line. Likewise, by \eqref{kuno2001a},
\begin{align}\label{kuno9b}
\EE\Big[\sup_{s\le t}\big\|\wt{U}_{\Lambda,\kappa,s}^{N,\pm}(\ul{\V{x}})
-\wt{U}_{\kappa,s}^{N,\pm}(\ul{\V{x}})\big\|_{\HP}^p\Big]
&\le c_p(\ee^2N^2\Lambda t^\eh)^{\nf{p}{2}}
+c_p(\ee^2N^2\Lambda t)^{\nf{p}{2}}.
\end{align}
Thanks to Lem.~\ref{lemmichael} we further know that
\begin{align}\nonumber
&\EE\Big[\sup_{\tau\le s\le t}\big\|F_{0,\nf{s}{2}}(\wt{U}_{\Lambda,\kappa,s}^{N,\pm}(\ul{\V{x}}))
-F_{0,\nf{s}{2}}(\wt{U}_{\kappa,s}^{N,\pm}(\ul{\V{x}}))\big\|^p\Big]
\\\nonumber
&=\EE\bigg[\sup_{\tau\le s\le t}\Big\|\int_0^1F_{1,\nf{s}{2}}\big(
\wt{U}_{\Lambda,\kappa,s}^{N,\pm}(\ul{\V{x}})
-\wt{U}_{\kappa,s}^{N,\pm}(\ul{\V{x}}),\tau\wt{U}_{\Lambda,\kappa,s}^{N,\pm}(\ul{\V{x}})
+(1-\tau)\wt{U}_{\kappa,s}^{N,\pm}(\ul{\V{x}})\big)\Id\tau\Big\|^p\bigg]
\\\nonumber
&\le c_p\sup_{\Lambda'\ge0}\EE\big[\sup_{\tau\le s\le t}
e^{cp\|\wt{U}_{\Lambda',\kappa,s}^{N,\pm}(\ul{\V{x}})\|_s^2}\big]^{\nf{2}{3}}
\EE\big[\sup_{\tau\le s\le t}\|\wt{U}_{\Lambda,\kappa,s}^{N,\pm}(\ul{\V{x}})
-\wt{U}_{\kappa,s}^{N,\pm}(\ul{\V{x}})\|_s^{3p}\big]^{\nf{1}{3}}.
\end{align}
If $T\ge\tau>0$, then the previous bound, \eqref{kuno8}, and \eqref{kuno9} imply
\begin{align}\label{kuno10}
\sup_{\kappa\in\NN\cup\{\infty\}}\mathop{\sup_{\ul{\V{x}}\in\RR^\nu}}\limits_{t\in[\tau,T]}
\EE\Big[\sup_{\tau\le s\le t}\big\|F_{0,\nf{s}{2}}(\wt{U}_{\Lambda,\kappa,s}^{N,\pm}(\ul{\V{x}}))
-F_{0,\nf{s}{2}}(\wt{U}_{\kappa,s}^{N,\pm}(\ul{\V{x}}))\big\|^p\Big]
\xrightarrow{\;\;\Lambda\downarrow0\;\;}0.
\end{align}
It is now a straightforward consequence of telescopic summations,
H\"{o}lder's inequality, \eqref{kuno4}, \eqref{kuno6}, the first bound in \eqref{kuno8aa},
\eqref{kuno8}, as well as \eqref{kuno10} that
\begin{align}\label{kuno11}
\sup_{\kappa\in\NN\cup\{\infty\}}\sup_{t\in[\tau,T]}\sup_{\ul{\V{x}}\in\RR^\nu}
\EE\Big[\sup_{\tau\le s\le t}\big\|\wt{W}_{\Lambda,\kappa,s}(\ul{\V{x}})
-\wt{W}_{\kappa,s}(\ul{\V{x}})\big\|^p\Big]\xrightarrow{\;\;\Lambda\downarrow0\;\;}0.
\end{align}
{\em Step~6.} Now we are in a position to prove Part~(1) for $(p,q)\not=(1,\infty)$.

On the one hand,
$\|\wt{W}_{\Lambda,\kappa}(\ul{\V{x}})\|=\|{W}_{\kappa}(\ul{\V{x}})\|$, for every $\Lambda>0$,  
in view of \eqref{idGross}. On the other hand, \eqref{kuno11} implies that, for all $t>0$,
$\ul{\V{x}}\in\RR^\nu$, and $\kappa\in\NN\cup\{\infty\}$, we find $\Lambda_n>0$, $n\in\NN$,
with $\Lambda_n\downarrow0$ and $\wt{W}_{\Lambda_n,\kappa,t}(\ul{\V{x}})\to
\wt{W}_{\kappa,t}(\ul{\V{x}})$, $\PP$-a.s. in $\LO(\sF)$, as $n\to\infty$. In particular,
$\|\wt{W}_{\kappa,t}(\ul{\V{x}})\|=\|{W}_{\kappa,t}(\ul{\V{x}})\|$, $\PP$-a.s., for all $t>0$
and $\ul{\V{x}}\in\RR^\nu$. Therefore, we can replace $\|W_{\kappa,t}^V(\ul{\V{x}})^*\|$ by any
$e^{-\int_0^tV(\ul{\V{x}}+\ul{\V{b}}_s)\Id s}\|\wt{W}_{\Lambda,\kappa,t}(\ul{\V{x}})^*\|$ with
$\Lambda\ge0$ on the left hand sides of \eqref{anna1} and \eqref{rudi13}. Thus, apart from
the case where $p=1$ and $q=\infty$, Part~(1) follows from Steps~1 and~2 of the proof of
Prop.~\ref{propT}. In particular, we see that 
$\wt{T}_{\kappa,t}^V\in\LO(L^p(\RR^\nu,\sF),L^p(\RR^\nu,\sF))$, for all $1\le p\le q\le \infty$
and $\kappa\in\NN\cup\{\infty\}$, with the current exception of the case $p=1$, $q=\infty$.

\smallskip

\noindent
{\em Step~7.}
With the help of \eqref{Katobd} and \eqref{kuno11} we can now prove Part~(3) for 
$1<p\le q\le\infty$ by proceeding along the lines of Step~1 of the proof of Prop.~\ref{propT}. 

To cover the case $1=p\le q<\infty$ of Part~(3) we first observe an analogue of 
\eqref{revW}: If $\Lambda>0$, $\kappa\in\NN$, and $V$ is bounded, then we set
\begin{align*}
\wt{W}_{\Lambda,\kappa,t}^V[\ul{\V{x}},\ul{\V{\alpha}}]&:=
\sW(\beta_{\Lambda,\kappa}^N(\ul{\V{x}}+\ul{\V{\alpha}}_t))
W^V_{\kappa,t}[\ul{\V{x}},\ul{\V{\alpha}}]\sW(-\beta_{\Lambda,\kappa}^N(\ul{\V{x}})),
\end{align*}
for all $t>0$, $\ul{\V{x}}\in\RR^\nu$, and $\ul{\V{\alpha}}\in C([0,\infty),\RR^\nu)$. 
In view of \eqref{idGross}, 
$\wt{W}_{\Lambda,\kappa,t}^V[\ul{\V{x}},\ul{\V{b}}]=e^{-\int_0^tV(\ul{\V{x}}+\ul{\V{b}}_s)\Id s}
\wt{W}_{\Lambda,\kappa,t}(\ul{\V{x}})$ under the above assumptions. 
If also $\ul{\V{\alpha}}_0=0$, then the relation
\begin{align}\label{revwtW}
\wt{W}_{\Lambda,\kappa,t}^V[\ul{\V{x}}+\ul{\V{\alpha}}_t,\ul{\V{\alpha}}_{t-\bullet}-\ul{\V{\alpha}}_t]
&=\wt{W}_{\Lambda,\kappa,t}^V[\ul{\V{x}},\ul{\V{\alpha}}]^*
\end{align}
is a direct consequence of \eqref{revW}. Employing \eqref{Katobd} and \eqref{revwtW} we can mimic 
the first estimation in Step~2 of the proof of Prop.~\ref{propT} to arrive at the bound
\begin{align}\nonumber
\int_{\{|\ul{\V{x}}|<n\}} f_n(\ul{\V{x}})&\EE\big[e^{-\int_0^tV_n^m(\ul{\V{x}}+\ul{\V{b}}_s)\Id s}\|
\wt{W}_{\Lambda,\kappa,t}(\ul{\V{x}})^*-\wt{W}_{\Lambda',\kappa,t}(\ul{\V{x}})^*\|
(n\wedge\|\Psi(\ul{\V{x}}+\ul{\V{b}}_t)\|)\big]
\Id\ul{\V{x}}
\\\nonumber
&\le c_{\nu,q}\tau^{-\nu(1-q^{-1})/2}\|f\|_{q'}\|\Psi\|_1
\sup_{\ul{\V{z}}\in\RR^\nu}\EE\Big[e^{\int_0^TV_-(\ul{\V{z}}+\ul{\V{b}}_s)\Id s}\Big]^{\nf{1}{2q}}
\\\label{gudrun1}
&\qquad\cdot\sup_{\kappa'\in\NN\cup\{\infty\}}\sup_{\ul{\V{x}}\in\RR^\nu}
\EE\Big[\sup_{\tau\le s\le T}\|\wt{W}_{\Lambda,\kappa',s}(\ul{\V{x}})
-\wt{W}_{\Lambda',\kappa',s}(\ul{\V{x}})\|^{2q}\Big]^{\nf{1}{2q}},
\end{align}
valid for all $T\ge t\ge\tau>0$, non-negative $f\in L^{q'}(\RR^\nu)$ with $f_n:=n\wedge f$, 
$\Psi\in L^1(\RR^\nu,\sF)$, $\kappa,\ell,m,n\in\NN$, and $\Lambda,\Lambda'>0$.
$V$ can be any Kato decomposable potential in \eqref{gudrun1}, and the
$V_n^m$ are defined as in Step~2 of the proof of Prop.~\ref{propT}.
Now, we apply \eqref{convW} and \eqref{idGross} to extend \eqref{gudrun1} to the case
$\kappa=\infty$. In the next step we invoke \eqref{kuno11} to extend \eqref{gudrun1}
to $\Lambda'=0$. After that we pass to the limit $n\to\infty$ on its right hand side by monotone 
convergence. Finally, we let $m$ go to infinity with the help of the dominated convergence theorem.
The resulting extension of \eqref{gudrun1} with $\Lambda'=0$ proves Part~(3) for $1=p\le q<\infty$.

\smallskip

\noindent{\em Step~8.} Let us now consider the semi-group properties:
As a consequence of Prop.~\ref{propT}(5)\&(6), the relation \eqref{idGrossT}, and
the remarks on the Gross transformation following Def.~\ref{defnGross}, the families
$(\wt{T}_{\Lambda,\kappa,t}^{V})_{t\ge0}$ with $\Lambda>0$ are semi-groups on every 
$L^p(\RR^\nu,\sF)$, $p\in[1,\infty]$.  Applying the by now available special cases $p=q\in[1,\infty]$ 
of Part~(3), we see that $(\wt{T}_{\kappa,t}^{V})_{t\ge0}$ is a semi-group on every 
$L^p(\RR^\nu,\sF)$, $p\in[1,\infty]$, as well.

\smallskip

\noindent
{\em Step~9.} With the semi-group properties at hand we may now prove the missing
case $p=1$, $q=\infty$ of Parts~(1) and~(3) similarly as in Step~6 of the proof of Prop.~\ref{propT}.

\smallskip

\noindent
{\em Step~10.}
Next, we observe that,  if $\Lambda>0$, then the identity \eqref{normTwtT} follows from 
\eqref{idGrossT}, for all $1\le p\le q\le\infty$ and $\kappa\in\NN\cup\{\infty\}$. 
We can extend it to $\Lambda=0$ with the help of Part~(3), which altogether proves Part~(2).

\smallskip

\noindent
{\em Step~11.} 
Let us now turn to the proof of Part~(4). We pick $p,t>0$ and start by considering the expressions
\begin{align}\label{kuno12}
\wt{W}_{\Lambda,\kappa,t}(\ul{\V{x}})^*\zeta(h)
&=e^{\tilde{u}_{\Lambda,\kappa,t}^N(\ul{\V{x}})+
\SPn{\wt{U}_{\Lambda,\kappa,s}^{N,+}(\ul{\V{x}})}{h}_{\HP}}
\zeta(e^{-t\omega}h+\wt{U}_{\Lambda,\kappa,t}^{N,-}(\ul{\V{x}})).
\end{align}
We first observe that the formula $\|\zeta(f)\|=e^{\|f\|_{\HP}^2/2}$, $f\in\HP$, and \eqref{kuno8a} imply
\begin{align}\nonumber
\sup_{\Lambda\ge0}\EE\Big[&\sup_{s\le t}\big\|\zeta(e^{-s\omega}h
+\wt{U}_{\Lambda,\kappa,s}^{N,-}(\ul{\V{x}}))\big\|^p\Big]
\\\label{kuno13}
&\le c^Ne^{p\|h\|_{\HP}^2}(1+p\ee^2N(1\vee t))^N
e^{c'p\ee^2N^2(1+\ln(1\vee t))+c'p^2\ee^4N^3(1\vee t)}.
\end{align}
Employing the bound
\begin{align}\nonumber
\|\zeta(f)-\zeta(g)\|&\le\|f-g\|_{\HP}e^{2\|f\|_{\HP}^2+2\|g\|_{\HP}^2},\quad f,g\in\HP,
\end{align}
we further deduce that
\begin{align*}
&\EE\Big[\sup_{s\le t}\big\|\zeta(e^{-s\omega}h+\wt{U}_{\Lambda,\kappa,s}^{N,-}(\ul{\V{x}}))-
\zeta(e^{-s\omega}h+\wt{U}_{\kappa,s}^{N,-}(\ul{\V{x}}))\big\|^p\Big]
\\
&\le e^{8p\|h\|_{\HP}^2}\sup_{\Lambda'\ge0}\EE\Big[\sup_{s\le t}
e^{12p\|\wt{U}_{\Lambda',\kappa,s}^{N,-}(\ul{\V{x}})\|_{\HP}^2}\Big]^{\nf{2}{3}}
\EE\Big[\sup_{s\le t}\big\|\wt{U}_{\Lambda,\kappa,s}^{N,-}(\ul{\V{x}})
-\wt{U}_{\kappa,s}^{N,-}(\ul{\V{x}})\big\|_{\HP}^{3p}\Big]^{\nf{1}{3}}.
\end{align*}
By virtue of the previous bound, \eqref{kuno4}, \eqref{kuno6}, \eqref{kuno8a}, \eqref{kuno9b},
\eqref{kuno12}, and \eqref{kuno13} it is now straightforward to verify that
\begin{align}\label{gudrun55}
\sup_{\kappa\in\NN\cup\{\infty\}}\sup_{\ul{\V{x}}\in\RR^\nu}\EE\Big[\sup_{s\le t}\big\|
\wt{W}_{\Lambda,\kappa,s}(\ul{\V{x}})^*\zeta(h)-\wt{W}_{\kappa,s}(\ul{\V{x}})^*\zeta(h)\big\|^p\Big]
\xrightarrow{\;\;\Lambda\downarrow0\;\;}0.
\end{align}

Now, assume that $p\in[1,\infty)$. Then the set 
$\{f\zeta(h):f\in L_0^\infty(\RR^\nu),h\in\HP\}$ is total in $L^p(\RR^\nu,\sF)$, where
$L_0^\infty(\RR^\nu)$ is the vector space of essentially bounded measurable functions 
on $\RR^\nu$with compact support. 
Furthermore, we may infer from \eqref{LpLq} and the by now available \eqref{normTwtT} that
$$
\sup_{s\le t}\sup_{\Lambda\ge0}\sup_{\kappa\in\NN\cup\{\infty\}}
\|\wt{T}_{\Lambda,\kappa,s}^V\|_{q,q}<\infty,\quad t>0,\,q\in[1,\infty].
$$
Hence, it suffices to pick some $f\in L_0^\infty(\RR^\nu)$ and $h\in\HP$ and show that
\begin{align*}
\sup_{s\le t}\sup_{\kappa\in\NN\cup\{\infty\}}\big\|(\wt{T}_{\Lambda,\kappa,s}^V-
\wt{T}_{\kappa,s}^V)(f\zeta(h))\big\|_p\xrightarrow{\;\;\Lambda\downarrow0\;\;}0.
\end{align*}
We put $F(\ul{\V{x}}):=|\ul{\V{x}}|$, $\ul{\V{x}}\in\RR^\nu$.
Then the latter limit relation follows from \eqref{Katobd} and \eqref{gudrun55} because
\begin{align*}
&\sup_{s\le t}\big\|(\wt{T}_{\Lambda,\kappa,s}^V-\wt{T}_{\kappa,s}^V)(f\zeta(h))\big\|_p^p
\\
&\le\sup_{\ul{\V{z}}\in\RR^\nu}\EE\big[e^{3\int_0^tV_-(\ul{\V{z}}+\ul{\V{b}}_s)\Id s}\big]^{\nf{p}{3}}
\sup_{\ul{\V{z}}\in\RR^\nu}\EE\Big[\sup_{s\le t}\big\|(
\wt{W}_{\Lambda,\kappa,s}(\ul{\V{z}})-\wt{W}_{\kappa,s}(\ul{\V{z}}))^*\zeta(h)\big\|^{3}\Big]^{\nf{p}{3}}
\\
&\quad\cdot\|e^Ff\|_\infty^p
\sup_{s\le t}\EE\big[e^{3|\ul{\V{b}}_s|}\big]^{\nf{p}{3}}\int_{\RR^\nu}e^{-pF(\ul{\V{x}})}\Id\ul{\V{x}}.
\end{align*}
\noindent{\em Step~12.}
In view of Prop.~\ref{propT}(6) and \eqref{idGrossT} the semi-groups
$(\wt{T}_{\Lambda,\kappa,t}^{V})_{t\ge0}$ with $\Lambda>0$ are strongly continuous on every 
$L^p(\RR^\nu,\sF)$, $p\in[1,\infty)$. The uniform limit relation of Part~(4) can now be used to transfer 
the strong continuity (at zero) of every $(\wt{T}_{\Lambda,\kappa,t}^{V})_{t\ge0}$, $\Lambda>0$,  
to the strong continuity of $(\wt{T}_{\kappa,t}^V)_{t\ge0}$. This completes the proof of Part~(5).

\smallskip

\noindent
{\em Step~13.} Next, we observe that, for $\Lambda>0$, \eqref{wtTVsadpq} follows from
Rem.~\ref{remTVsadpq}, \eqref{idGrossT}, and the remarks on the Gross transformation following 
Def.~\ref{defnGross}. Employing Part~(3), we can extend \eqref{wtTVsadpq} to the case 
$\Lambda=0$.
 
\smallskip

\noindent
{\em Step~14.} In the case $1<p\le q\le\infty$, Part~(7) can be obtained by the same procedure 
that we used to prove Prop.~\ref{propT}(3) with $p>1$, starting from 
\begin{align}\label{line1}
\sup_{\Lambda\ge0}\sup_{\ul{\V{x}}\in\RR^\nu}
\EE\Big[\sup_{\tau\le s\le t}\big\|\wt{W}_{\Lambda,\kappa,s}(\ul{\V{x}})-
\wt{W}_{\Lambda,\infty,s}(\ul{\V{x}})\big\|^{\tilde{p}}\Big]&\xrightarrow{\;\;\kappa\to\infty\;\;}0,
\quad t\ge\tau>0.
\end{align}
Here and in what follows $\tilde{p}>0$. The relation \eqref{line1} can be derived as in 
Prop.~\ref{propconvW} with the help of \eqref{kuno4}, \eqref{kuno8}, and
\begin{align}\label{line2}
\sup_{\Lambda\ge0}\sup_{\ul{\V{x}}\in\RR^\nu}\EE\Big[
\sup_{s\le t}|\tilde{u}^N_{\Lambda,\kappa,s}(\ul{\V{x}})-\tilde{u}^N_{\Lambda,\infty,s}
(\ul{\V{x}})|^{\tilde{p}}\Big]&\xrightarrow{\;\;\kappa\to\infty\;\;}0,\quad t>0,
\\\label{line3}
\sup_{\Lambda\ge0}\sup_{\ul{\V{x}}\in\RR^\nu}\EE\Big[
\sup_{s\le t}\big\|\wt{U}^{N,\pm}_{\Lambda,\kappa,s}(\ul{\V{x}})-
\wt{U}^{N,\pm}_{\Lambda,\infty,s}(\ul{\V{x}})\big\|_s^{\tilde{p}}\Big]&\xrightarrow{\;\;\kappa\to\infty\;\;}0,
\quad t\ge\tau>0.
\end{align}
Here \eqref{line2} is a consequence of \eqref{felixu}, \eqref{defnonFocku}, and the observation that 
\eqref{wendelin1} and \eqref{wendelin1b} are uniform in the choice of $\eta$, which can in particular 
be relaced by $1_{\{|\V{m}|\ge\Lambda\}}\eta$ in \eqref{wendelin1} and \eqref{wendelin1b}.
Furthermore, \eqref{line3} follows from \eqref{convUpmN}, \eqref{defnonFockUminus},
\eqref{defnonFockUplus}, \eqref{BurkBM}, and the elementary bounds
\begin{align*}
&\Big\|\frac{1}{(s\omega)^{\nf{\iota}{2}}}(1-e^{-s\omega})\big(\beta_{\Lambda,\kappa}(\ul{\V{x}})-
\beta_{\Lambda,\infty}(\ul{\V{x}})\big)\Big\|_{\HP}^2\le \frac{c\ee^2N^2}{\kappa},
\\
&\Big\|\frac{1}{(s\omega)^{\nf{\iota}{2}}}\Big(\big(\beta_{\Lambda,\kappa}(\ul{\V{x}})-
\beta_{\Lambda,\kappa}(\ul{\V{x}}+\ul{\V{b}}_s)\big)-
\big(\beta_{\Lambda,\infty}(\ul{\V{x}})-\beta_{\Lambda,\infty}(\ul{\V{x}}+\ul{\V{b}}_s)\big)
\Big)\Big\|_{\HP}^2
\\&\hspace{7.1cm}\le c(N^\mh s^{-1}|\ul{\V{b}}_s|)^\iota\frac{\ee^2N^2}{\kappa},
\end{align*}
for all $\iota\in\{0,1\}$, $\Lambda\ge0$, $\kappa\in\NN$, $s>0$, and $\ul{\V{x}}\in\RR^\nu$. 
Here used $|\chi_\kappa^2(\V{k})-\chi_\infty^2(\V{k})|\le2|\V{k}|/\kappa$, $|\V{k}|<\kappa$, 
to derive the latter two bounds.

Part~(7) can then be extended to the case $1=p\le q<\infty$ by means of Part~(6). Finally,
the case $p=1$, $q=\infty$ of Part~(7) follows from the already proven cases together with
the semi-group relations and the fact that $\|\wt{T}_{\Lambda,\kappa,t}^V\|_{\tilde{p},\tilde{q}}$
is uniformly bounded in $\Lambda\ge0$, $\kappa\in\NN\cup\{\infty\}$, and $t\in[\tau,T]$,
for fixed $1\le\tilde{p}\le\tilde{q}\le\infty$.
\end{proof}


\section{Feynman-Kac formula for fiber Hamiltonians}\label{secfiber}

\noindent
In this section we consider only one matter particle whose dynamics is not influenced
by any external potential, i.e., we set $N=1$ and $V=0$. Then the Nelson model becomes
translation invariant and the Nelson Hamiltonian unitarily equivalent to a direct
integral of fiber Hamiltonians, each attached to a fixed total momentum of the matter-radiation
system. If $\kappa$ is finite, then it is actually very easy to explicitly realize this fiber decomposition
and find expressions for the fiber Hamiltonians; see Rem.~\ref{remfibdec}. After adding the energy
renormalizations $E_\kappa^\ren$ to the fiber Hamiltonians, we can then try to analyze their limit as
$\kappa$ goes to infinity. For massive bosons, this has been done by Cannon in \cite{Cannon1971} 
by using Gross transformations similarly as in Nelson's article \cite{Nelson1964}. Massless
renormalized fiber Hamiltonians for Nelson's model are studied in \cite{Froehlich1973,Froehlich1974}.

Here we shall give an independent construction of renormalized fiber Hamiltonians for arbitrary
non-negative boson masses, again by showing that suitable Feynman-Kac semi-groups converge, 
as $\kappa\to\infty$, in norm to an explicitly given limiting semi-group. Then the renormalized 
fiber Hamiltonian is the generator of the limiting semi-group by definition. The main new result is
our fairly explicit expression for the Feynman-Kac semi-group for $\kappa=\infty$.

In what follows, $\V{B}$ is again a three-dimensional Brownian motion as explained in the beginning
of Sect.~\ref{secBP}. Recall that the $\PP$-zero sets $\sN_-$ and $\sN_+$ have been introduced
in Lem.~\ref{lemUminus} and Lem.~\ref{lemUplus}, respectively.

We start by observing that, in the case $N=1$, the complex action $u_{\kappa,t}^1[\V{x},\V{B}]$ is
actually $\V{x}$-independent as a direct consequence of Def.~\ref{defu} and Def.~\ref{defnstandardu}. 
Hence, we abbreviate
\begin{align*}
u_{\kappa}[\V{B}]:=u_{\kappa}^1[\V{0},\V{B}],\quad\kappa\in\NN\cup\{\infty\}.
\end{align*}
Given some $\V{x}\in\RR^3$, let us recall the notation 
$\Gamma(e^{i\V{m}\cdot\V{x}}):=\sW(0,e^{i\V{m}\cdot\V{x}})$ for the Weyl operator whose
action on exponential vectors is given by
\begin{align}\label{Weylmexpv}
\Gamma(e^{i\V{m}\cdot\V{x}})\zeta(h)&=\zeta(e^{i\V{m}\cdot\V{x}}h),\quad h\in\HP.
\end{align}
Then $\Gamma(e^{i\V{m}\cdot\V{x}})=e^{i\Id\Gamma(\V{m})\cdot\V{x}}$. Here the $j$-th
component of the formal vector of operators
$\Id\Gamma(\V{m}):=(\Id\Gamma(m_1),\Id\Gamma(m_2),\Id\Gamma(m_3))$ is reduced by the 
subspaces in the decomposition \eqref{defFockspace} of $\sF$. It acts by multiplication with $0$ in 
the vacuum subspace $\CC$,  and by maximal multiplication with the symmetric function 
$(\V{k}_1,\ldots,\V{k}_n)\mapsto{k}_{1,j}+\dots+{k}_{n,j}$ 
in $L^2_{\mathrm{sym}}(\RR^{3n},\lambda^{3n})$.

\begin{defn}\label{defwhW}
Let $\kappa\in\NN\cup\{\infty\}$ and $t\ge0$. Then we define
\begin{align*}
\wh{W}_{\kappa,t}[\V{B}]&:=e^{u_{\kappa,t}[\V{B}]}
F_{0,\nf{t}{2}}(-e^{i\V{m}\cdot\V{B}_t}U_{\kappa,t}^+[\V{B}])\Gamma(e^{i\V{m}\cdot\V{B}_t})
F_{0,\nf{t}{2}}(-U_{\kappa,t}^-[\V{B}])^*
\end{align*}
on $\Omega\setminus(\sN_+\cup\sN_-)$ and $\wh{W}_{\kappa,t}[\V{B}]:=\id$
on $\sN_+\cup\sN_-$. We further set
\begin{align}\label{astrid1}
\wh{T}_{\kappa,t}(\V{\zeta})&:=\EE\big[e^{i\V{\zeta}\cdot\V{B}_t}\wh{W}_{\kappa,t}[\V{B}]^*\big],
\quad\V{\zeta}\in\CC^3.
\end{align}
\end{defn}

The definition \eqref{astrid1}, where the expectation is a $\LO(\sF)$-valued Bochner-Lebesgue
integral, requires some justification, which is given in Rem.~\ref{remwhW}(2) and
Prop.~\ref{propfiber} below.

\begin{rem}\label{remwhW}
Let $\kappa\in\NN\cup\{\infty\}$. Then the following holds:
\begin{enumerate}[leftmargin=*]
\item[{\rm(1)}] For all $h\in\HP$ and $t\ge0$,
\begin{align}\label{whWexpv}
\wh{W}_{\kappa,t}[\V{B}]\zeta(h)
&=e^{u_{\kappa,t}[\V{B}]-\SPn{U_{\kappa,t}^-[\V{B}]}{h}_{\HP}}\zeta(e^{-t\omega+i\V{m}\cdot\V{B}_t}
h-e^{i\V{m}\cdot\V{B}_t}U_{\kappa,t}^+[\V{B}]),
\\\label{whWadjexpv}
\wh{W}_{\kappa,t}[\V{B}]^*\zeta(h)
&=e^{u_{\kappa,t}[\V{B}]-\SPn{e^{i\V{m}\cdot\V{B}_t}U_{\kappa,t}^+[\V{B}]}{h}_{\HP}}
\zeta(e^{-t\omega-i\V{m}\cdot\V{B}_t}h-U_{\kappa,t}^-[\V{B}]).
\end{align}
This follows from \eqref{SGdGamma}, \eqref{Fadjexpv}, \eqref{Fexpv},
\eqref{Weylmexpv}, and $\Gamma(e^{i\V{m}\cdot\V{x}})^*=\Gamma(e^{-i\V{m}\cdot\V{x}})$.
\item[{\rm(2)}] Let $t>0$. Then the map 
$$
\RR^\nu\times\mathfrak{k}^2\ni(\V{x},g,h)\mapsto F_{\nf{t}{3}}(g)
e^{-t\Id\Gamma(\omega)/3+i\Id\Gamma(\V{m})\cdot\V{x}}F_{\nf{t}{3}}(h)^*\in\LO(\sF)
$$
is continuous as a consequence of Lem.~\ref{lemmichael}(1) and the bound $|\V{m}|\le\omega$. Since 
$\RR^\nu\times\mathfrak{k}^2$ is separable, it follows that its range is separable as well. In view of
\begin{align*}
\wh{W}_{\kappa,t}[\V{B}]&=
e^{u_{\kappa,t}[\V{B}]}F_{0,\nf{t}{3}}(-e^{i\V{m}\cdot\V{B}_t}U_{\kappa,t}^+[\V{B}])
e^{-t\Id\Gamma(\omega)/3+i\V{m}\cdot\V{B}_t}F_{0,\nf{t}{3}}(-U_{\kappa,t}^-[\V{B}])^*,
\end{align*}
this implies that $\wh{W}_{\kappa,t}[\V{B}]$ and $\wh{W}_{\kappa,t}[\V{B}]^*$ are 
$\fF_t$-$\fB(\LO(\sF))$-measurable with a separable range.
\end{enumerate}
\end{rem}

In what follows $\cF$ denotes the $\sF$-valued Fourier transformation on $\RR^3$. We also
introduce the unitary operator
\begin{align*}
Q:=\cF\Big(\int_{\RR^3}^\oplus\Gamma(e^{i\V{m}\cdot\V{x}})\Id\V{x}\Big)
\in\sU\big(L^2(\RR^3,\sF)\big).
\end{align*}

\begin{prop}\label{propfiber}
Let $\kappa\in\NN\cup\{\infty\}$ and $t\ge0$. Then the map 
$\CC^3\ni\V{\zeta}\mapsto\wh{T}_{\kappa,t}(\V{\zeta})\in\LO(\sF)$ is well-defined and analytic 
and there exist universal constants $c,c'>0$ such that
\begin{align}\label{bdwhT}
\|\wh{T}_{\kappa,t}(\V{\zeta})\|
&\le c'\big(1+\ee^2(1\vee t)^2\big)e^{c|\V{\eta}|^2t+c\ee^4t+c\ee^2(1+\ln(1\vee t))},
\quad\V{\zeta}\in\CC^3.
\end{align}
Furthermore,
\begin{align}\label{fiberdec1}
Qe^{-tH_{1,\kappa}^0}Q^*&=\int_{\RR^3}^\oplus\wh{T}_{\kappa,t}(\V{\xi})\Id\V{\xi}.
\end{align}
\end{prop}

\begin{proof}
{\em Step~1.} Let $\V{\zeta}\in\CC^3$ and set $\V{\eta}:=\Im\,\V{\zeta}$. Then \ref{bdskuno},
\eqref{expbdUminus}, \eqref{expbdUplus}, \eqref{expbdu} with $N=1$, and \eqref{bdmichi} imply
\begin{align}\nonumber
\EE\big[&\sup_{s\le t}|e^{i\V{\zeta}\cdot\V{B}_s}|\|\wh{W}_{\kappa,s}[\V{B}]\|\big]^4
\\\nonumber
&\le\EE\big[\sup_{s\le t}e^{-4\V{\eta}\cdot\V{B}_s}\big]
\EE\big[\sup_{s\le t}e^{4u_{\kappa,s}[\V{B}]}\big]
\EE\big[\sup_{s\le t}e^{c\|U^-_{\kappa,s}[\V{B}]\|_s^2}\big]
\EE\big[\sup_{s\le t}e^{c\|U^+_{\kappa,s}[\V{B}]\|_s^2}\big]
\\\label{bdwhWsup}
&\le c'\big(1+\ee^2(1\vee t)^2\big)e^{c|\V{\eta}|^2t+c\ee^4t+c\ee^2(1+\ln(1\vee t))},
\end{align}
with universal constants $c,c'>0$. Together with Rem.~\ref{remwhW}(2) this shows that
$\wt{T}_{\kappa,t}(\V{\zeta})$ is well-defined and satisfies \eqref{bdwhT}.
Employing similar estimates it is straightforward to show that 
$\CC^3\ni\V{\zeta}\mapsto\wh{T}_{\kappa,t}(\V{\zeta})\in\LO(\sF)$ 
is complex differentiable, i.e., analytic.

\smallskip

\noindent{\em Step~2.}
Let $\V{x}\in\RR^3$ and write $\V{B}_t^{\V{x}}:=\V{x}+\V{B}_t$ for short. We shall show that 
\begin{align}\label{astrid2000}
W_{\kappa,t}(\V{x})^*=\Gamma(e^{-i\V{m}\cdot\V{x}})\wh{W}_{\kappa,t}[\V{B}]^*
\Gamma(e^{i\V{m}\cdot\V{B}_t^{\V{x}}}),\quad\text{$\PP$-a.s.}
\end{align}
In fact, if $h\in\HP$, then we observe with the help of \eqref{Wexpvec}, \eqref{Weylmexpv}, 
\eqref{whWexpv}, and the relations $U_{\kappa,t}^{1,\pm}(\V{x})=1_{\Omega\setminus\sN_\pm}
e^{-i\V{m}\cdot\V{x}}U_{\kappa,t}^\pm[\V{B}]$, $\PP$-a.s., that
\begin{align*}
&W_{\kappa,t}(\V{x})\zeta(h)
\\
&=\Gamma(e^{-i\V{m}\cdot\V{B}_t^{\V{x}}})e^{u_{\kappa,t}[\V{B}]
-\SPn{U_{\kappa,t}^-[\V{B}]}{e^{i\V{m}\cdot\V{x}}h}_{\HP}}
\zeta\big(e^{-t\omega+i\V{m}\cdot\V{B}_t}e^{i\V{m}\cdot\V{x}}h-e^{i\V{m}\cdot\V{B}_t}
U^+_{\kappa,t}[\V{B}]\big)
\\
&=\Gamma(e^{-i\V{m}\cdot\V{B}_t^{\V{x}}})\wh{W}_{\kappa,t}[\V{B}]\Gamma(e^{i\V{m}\cdot\V{x}})
\zeta(h),\quad\text{$\PP$-a.s.},
\end{align*}
which extends to an operator identity in $\LO(\sF)$. 

\smallskip

\noindent{\em Step~3.}
Next, we prove that the fiber decomposition \eqref{fiberdec1} holds. Let $\Psi$ be a finite linear 
combination of functions of the form $f\zeta(h)$ with $f\in\sS(\RR^3)$, the Schwartz space over 
$\RR^3$, and $h\in\bigcap_{n\in\NN}\dom(|\V{m}|^n)$. Then \eqref{astrid2000} permits to get
\begin{align*}
(T_{\kappa,t}^0\Psi)(\V{x})
&=\Gamma(e^{-i\V{m}\cdot\V{x}})\EE\big[\wh{W}_{\kappa,t}[\V{B}]^*
\Gamma(e^{i\V{m}\cdot\V{B}_t^{\V{x}}})\Psi(\V{B}_t^{\V{x}})\big],\quad\V{x}\in\RR^3.
\end{align*}
Set $\Phi(\V{y}):=\Gamma(e^{i\V{m}\cdot\V{y}})\Psi(\V{y})$, $\V{y}\in\RR^3$,
so that $\hat{\Phi}\in C^\infty(\RR^3,\sF)$. 
Thanks to the condition $|\V{m}|^nh\in\HP$, $n\in\NN$, it is straightforward to verify by partial 
integration that $\hat{\Phi}$ is rapidly decreasing. Hence, we may further deduce that
\begin{align*}
\Gamma(e^{i\V{m}\cdot\V{x}})(T_{\kappa,t}^0\Psi)(\V{x})
&=\frac{1}{(2\pi)^{\nf{3}{2}}}
\EE\Big[\wh{W}_{\kappa,t}[\V{B}]^*\int_{\RR^3}
e^{i\V{\xi}\cdot\V{B}_t^{\V{x}}}\hat{\Phi}(\V{\xi})\Id\V{\xi}\Big]
\\
&=\frac{1}{(2\pi)^{\nf{3}{2}}}\int_{\RR^3}e^{i\V{\xi}\cdot\V{x}}
\EE\big[e^{i\V{\xi}\cdot\V{B}_t}\wh{W}_{\kappa,t}[\V{B}]^*\big]\hat{\Phi}(\V{\xi})\Id\V{\xi},
\quad\V{x}\in\RR^3.
\end{align*}
On account of the Fourier inversion formula and \eqref{bdwhWsup} this implies
\begin{align*}
\big(QT_{\kappa,t}^0\Psi\big)(\V{\xi})&=
\wh{T}_{\kappa,t}(\V{\xi})(Q\Psi)(\V{\xi}),\quad\text{a.e. $\V{\xi}\in\RR^3$}.
\end{align*}
Since $\Psi$ can be chosen in a dense subset of $L^2(\RR^3,\sF)$ and $\|\wh{T}_{\kappa,t}(\V{\xi})\|$
is bounded uniformly in $\V{\xi}\in\RR^3$, the previous relation and the Feynman-Kac formula
$T_{\kappa,t}^0=e^{-tH_{1,\kappa}^0}$ imply \eqref{fiberdec1}.
\end{proof}

\begin{lem}\label{lemconvwhT}
Let $\V{\zeta}\in\CC^3$ and $t\ge0$. Then 
$\wh{T}_{\kappa,t}(\V{\zeta})\to\wh{T}_{\infty,t}(\V{\zeta})$ in the operator norm, as $\kappa$
goes to infinity. The convergence is uniform as $t$ varies in any compact subset of 
the open half-axis $(0,\infty)$.
\end{lem}

\begin{proof}
The proof of this lemma is completely analogous to the one of Prop.~\ref{propconvW}.
\end{proof}

\begin{prop}\label{propSGfiber}
Let $\kappa\in\NN\cup\{\infty\}$ and $\V{\zeta}\in\CC^3$. Then 
$(\wh{T}_{\kappa,t}(\V{\zeta}))_{t\ge0}$ is a $C_0$-semi-group on $\sF$ and
\begin{align}\label{adjwhT}
\wh{T}_{\kappa,t}(\V{\zeta})^*=\wh{T}_{\kappa,t}(\ol{\V{\zeta}}),\quad t\ge0.
\end{align}
\end{prop}

\begin{proof}
Let $\kappa\in\NN\cup\{\infty\}$, $\V{\zeta}\in\CC^3$, and $t\ge0$.

Employing \eqref{whWexpv} repeatedly we verify that, for every $h\in\HP$, the equality
\begin{align*}
\wh{W}_{\kappa,s}[{}^t\!\V{B}]\wh{W}_{\kappa,t}[\V{B}]\zeta(h)=\wh{W}_{\kappa,s+t}[\V{B}]\zeta(h),
\quad s\ge0,\;\;\text{on $\Omega\setminus\sN_t$,}
\end{align*}
is implied by \eqref{shift1}, \eqref{shift2}, and \eqref{shift3}. Here the $\PP$-zero set $\sN_t$ 
is $h$-independent and we also took the identity
$u_{\kappa,s}^1[\V{B}_t,{}^t\!\V{B}]=u_{\kappa,s}^1[\V{0},{}^t\!\V{B}]$ into account, which is true in the
case $N=1$ considered at present. By the totality of the exponential vectors in $\sF$ and since
$\V{B}$ and ${}^t\!\V{B}$ have the same distribution, we obtain the Markov property
\begin{align*}
\EE^{\fF_t}\big[e^{i\V{\zeta}\cdot\V{B}_{s+t}}\wh{W}_{\kappa,s+t}[\V{B}]^*\big]
&=e^{i\V{\zeta}\cdot\V{B}_t}\wh{W}_{\kappa,t}[\V{B}]^*\wh{T}_{\kappa,s}(\V{\zeta}),
\quad\text{$\PP$-a.s.},
\end{align*}
for every $s\ge0$.
It entails $\wh{T}_{s+t}(\V{\zeta})=\wh{T}_{\kappa,t}(\V{\zeta})\wh{T}_{\kappa,s}(\V{\zeta})$.

Furthermore, if $\kappa\in\NN$ and $\wt{\V{B}}_s:=\V{B}_{(t-s)\wedge0}-\V{B}_t$, 
$s\ge0$, then the relations
\begin{align*}
&e^{i\V{\zeta}\cdot\V{B}_t}\wh{W}_{\kappa,t}[\V{B}]^*
\\
&=e^{u_{\kappa,t}[\V{B}]+i{\V{\zeta}}\cdot\V{B}_t}
F_{0,\nf{t}{2}}(-U_{\kappa,t}^-[\V{B}])\Gamma(e^{-i\V{m}\cdot\V{B}_t})
F_{0,\nf{t}{2}}(-e^{i\V{m}\cdot\V{B}_t}U_{\kappa,t}^+[\V{B}])^*
\\
&=e^{u_{\kappa,t}[\wt{\V{B}}]-i{\V{\zeta}}\cdot\wt{\V{B}}_t}
F_{0,\nf{t}{2}}(-e^{i\V{m}\cdot\wt{\V{B}}_t}U_{\kappa,t}^+[\wt{\V{B}}])
\Gamma(e^{i\V{m}\cdot\wt{\V{B}}_t})
F_{0,\nf{t}{2}}(-U_{\kappa,t}^-[\wt{\V{B}}])^*
\\
&=e^{-i{\V{\zeta}}\cdot\wt{\V{B}}_t}\wh{W}_{\kappa,t}[\wt{\V{B}}]
=\big(e^{i\ol{\V{\zeta}}\cdot\wt{\V{B}}_t}\wh{W}_{\kappa,t}[\wt{\V{B}}]^*\big)^*,
\quad\text{on $\Omega$},
\end{align*}
follow from \eqref{revUpm}, \eqref{revu}, and the identity
$u_{\kappa,t}^1[\V{B}_t,\wt{\V{B}}]=u_{\kappa,t}^1[\V{0},\wt{\V{B}}]$ which is valid in the case $N=1$.
Taking expectations we obtain \eqref{adjwhT} for finite $\kappa$, because $(\V{B}_s)_{s\in[0,t]}$
and $(\wt{\V{B}}_s)_{s\in[0,t]}$ have the same distribution. In the case $\kappa=\infty$, 
\eqref{adjwhT} now follows from Lem.~\ref{lemconvwhT}.

It remains to verify the strong continuity of $(\wh{T}_{\kappa,t}(\V{\zeta}))_{t\ge0}$. 
On account of \eqref{bdwhT} it suffices to do this on a total subset of $\sF$.
For all $\kappa\in\NN\cup\{\infty\}$, $t\ge0$, and $h\in\HP$, 
the relation $\lim_{s\to t,s\ge0}\wh{T}_{\kappa,s}\zeta(h)=\wh{T}_{\kappa,t}\zeta(h)$ follows, 
however, from \eqref{whWadjexpv}, the fact that
$\sup_{s\le\tau}\|\wh{W}_{\kappa,s}[\V{B}]^*\zeta(h)\|\in L^1(\PP)$, $\tau\ge0$, by \eqref{bdwhWsup},
and the dominated convergence theorem.
\end{proof}

We now arrive at the main result of this section; the identity \eqref{FKfib} in the next theorem is
the promised Feynman-Kac formula for fiber Hamiltonians. Notice that the next theorem contains
in particular our definition of the {\em ultra-violet renormalized fiber Hamiltonians} 
$\wh{H}_\infty(\V{\xi})$.

For information on analytic families in the sense of Kato we refer to 
\cite[\textsection\textsection VII.1.2]{Kato} and \cite{ReedSimonIV}.

\begin{thm}\label{thmFKfib}
Let $\kappa\in\NN\cup\{\infty\}$. Then there exists a unique analytic family in the sense of Kato,
$\{\wh{H}_{\kappa}(\V{\zeta})\}_{\V{\zeta}\in\CC^3}$, of closed operators in $\sF$ such that
\begin{align}\label{decH1}
QH_{1,\kappa}^0Q^*&=\int_{\RR^3}^\oplus\wh{H}_{\kappa}(\V{\xi})\Id\V{\xi}.
\end{align}
For every $\V{\zeta}\in\CC^3$, the operator $\wh{H}_{\kappa}(\V{\zeta})$ generates the
$C_0$-semi-group $(\wh{T}_{\kappa,t}(\V{\zeta}))_{t\ge0}$ and in particular
\begin{align}\label{FKfib}
e^{-t\wh{H}_{\kappa}(\V{\xi})}&=\wh{T}_{\kappa,t}(\V{\xi}),\quad\V{\xi}\in\RR^3,\,t\ge0.
\end{align}
There is a universal constant $c>0$ such that the resolvent set of every 
$\wh{H}_{\kappa}(\V{\zeta})$
with $\V{\zeta}\in\CC^3$ contains the interval $(-\infty,-c|\Im\,\V{\zeta}|^2-c\ee^4)$.
\end{thm}

\begin{proof}
For every $\V{\zeta}\in\CC^3$, we {\em define} $\wh{H}_{\kappa}(\V{\zeta})$ to be the generator of 
$(\wh{T}_{\kappa,t}(\V{\zeta}))_{t\ge0}$. General principles ensure that $\wh{H}_{\kappa}(\V{\zeta})$
is densely defined and closed, the bound \eqref{bdwhT} and the Hille-Yosida theorem imply the
last statement on its resolvent set, and
the relation \eqref{adjwhT} entails self-adjointness of $\wh{H}_\kappa(\V{\xi})$, if $\V{\xi}\in\RR^3$.
Now \eqref{decH1} is a consequence of \eqref{fiberdec1}.

The uniqueness statement follows from the unique continuation principle of
\cite[Rem.~VII.1.6]{Kato} and the fact that any two strongly resolvent measurable families of 
self-adjoint operators indexed by $\V{\xi}\in\RR^3$, whose direct integral equals 
$QH_{1,\kappa}^0Q^*$, agree almost everywhere on $\RR^3$.
\end{proof}

\begin{rem}\label{remfibdec}
Let $\kappa\in\NN$. Then we can analyze the conjugation of $H_{1,\kappa}^0$ with
the unitary transformation $Q$ directly, of course. It turns out that, for every $\V{\xi}\in\RR^3$,
$\wh{H}_\kappa(\V{\xi})$ is self-adjoint with domain 
$\dom(\Id\Gamma(\omega))\cap\dom(\Id\Gamma(\V{m})^2)$ and given by
\begin{align}\label{forwhHkappa}
\wh{H}_\kappa(\V{\xi})&=\frac{1}{2}(\V{\xi}-\Id\Gamma(\V{m}))^2
+\Id\Gamma(\omega)+\vp(f_\kappa).
\end{align}
Similarly as in Steps~1 and~2 of the proof of Thm.~\ref{thmFKreg} we can also verify directly
that the right hand side of \eqref{forwhHkappa} generates $(\wh{T}_{\kappa,t}(\V{\xi}))_{t\ge0}$
\cite[Thm.~5.3(2) and Thm.~11.1]{GMM2016}. Instead
of \eqref{DGLW} we then $\PP$-a.s. encounter the {\em stochastic} differential equation
\begin{align*}
\wh{W}_{\kappa,t}(\V{\xi})\eta&=\eta-\int_0^t\wh{H}_\kappa(\V{\xi})\wh{W}_{\kappa,s}(\V{\xi})\eta\Id s
-\int_0^ti(\V{\xi}-\Id\Gamma(\V{m}))\wh{W}_{\kappa,s}(\V{\xi})\eta\Id\V{B}_s,\quad t\ge0.
\end{align*}
This equation is one example of a broader class of stochastic differential equations whose solution 
theory for $\fF_0$-measurable 
$\eta:\Omega\to\dom(\Id\Gamma(\omega))\cap\dom(\Id\Gamma(\V{m})^2)$
is studied in \cite{GMM2016}.
\end{rem}


\section{Some applications}\label{secappl}

\noindent
In the following three subsections we provide some examples for the applicability of our results. 
In Subsect.~\ref{ssecerg} we discuss the ergodicity of our semi-groups and prove the lower bound
in \eqref{asympintro}. The upper bound in \eqref{asympintro} is established in Subsect.~\ref{appub}.
Finally, we discuss continuity properties of the semi-group in the Nelson model in
Subsect.~\ref{sseccont}.


\subsection{Ergodicity}\label{ssecerg}

\noindent
In this subsection we employ our Feynman-Kac formulas to verify that
the semi-groups corresponding to the Nelson Hamiltonian $H_{N,\kappa}^V$, its non-Fock
version $\wt{H}_{N,\kappa}^V$, and the fiber
Hamiltonian at zero total momentum $\wh{H}_\kappa(\V{0})$ are positivity improving. Here the
notion of positivity is determined by a $\cQ$-space representation of the Fock space $\sF$ and we
ignore the Pauli principle, i.e., $H_{N,\kappa}^V$ and $\wt{H}_{N,\kappa}^V$ are not restricted to any 
symmetry subspace. 
As corollaries we may verify the usual formulas for the minima of the spectra of these Hamiltonians.
This will also permit us to conclude the proof of the lower bound in \eqref{asympintro}.

The results in this subsection seem to be new in the case $\kappa=\infty$, and the verification of 
Thm.~\ref{thmOscar1} in the non-Fock case has in fact been proposed as an open problem in
\cite{HHS2005}. We should mention that earlier results already imply that the semi-group of 
the renormalized Nelson Hamiltonian is positivity preserving; see \cite{LHB2011} and the references
given there. At the end of this subsection we give some remarks on earlier work on fiber Hamiltonians  
based on a different notion of positivity.

Employing standard tools from \cite{Faris1972,Simon1974},
it has already been observed in \cite{Matte2016} that Feynman-Kac integrands of the
form encountered here are positivity improving, pointwise on $\Omega$. In the proof of
Prop.~\ref{propOscar1} we essentially repeat the argument of \cite{Matte2016} because it is
short and it requires a little complement to deal with the fiber Hamiltonian in Prop.~\ref{propOscar2}.

We shall work with the self-dual convex cone $\sP$ in the Fock space $\sF$ is given by
$$
\sP:=\ol{\mr{\sP}},\quad
\mr{\sP}:=\big\{G(\vp(\V{g}))\zeta(0)\big|\,G\in\sS(\RR^n),\,G\ge0,\,
\V{g}\in\HP_{\RR}^n,\,n\in\NN\big\}.
$$
Here $\sS(\RR^n)$ is the Schwartz space on $\RR^n$ and, for every $G\in\sS(\RR^n)$,
the bounded operator $G(\vp(\V{g})):=G(\vp(g_1),\ldots,\vp(g_n))$ is defined via
the $\sF$-valued Bochner-Lebesgue integrals
\begin{align}\label{kent1}
G(\vp(\V{g}))\psi&:=\frac{1}{(2\pi)^{\nf{n}{2}}}
\int_{\RR^n}\hat{G}(-\V{\xi})\sW(-i\V{\xi}\cdot\V{g})\psi\Id\V{\xi},\quad\psi\in\sF;
\end{align}
recall from \eqref{abbrevWeyl} that $\sW(h):=\sW(h,\id)$, $h\in\HP$.

In fact, there exists a (non-unique) probability space $(\cQ,\fQ,\eta)$ and unitary map
$\cU:\sF\to L^2(\cQ,\eta)$ satisfying $\cU\zeta(0)=1$ such that, for every $g\in\HP_{\RR}$, 
$\hat{\vp}(g):=\cU\vp(g)\cU^*$ is a maximal multiplication operator in $L^2(\cQ,\eta)$ with a
Gaussian random variable, such that $\fQ$ is generated by $\{\hat{\vp}(g):h\in\HP_{\RR}\}$,
and such that $\sP$ is the pre-image under $\cU$ of all non-negative functions in 
$L^2(\cQ,\eta)$; see, e.g., \cite{Simon1974}.
In particular, the vacuum vector $\zeta(0)$ is strictly positive with respect to $\sP$.

\begin{prop}\label{propOscar1}
\begin{enumerate}[leftmargin=*]
\item[{\rm(1)}] Let $t>0$ and $h\in\mathfrak{k}_{\RR}$. Then $F_{0,t}(h)$ and $F_{0,t}(h)^*$ are
positivity improving with respect to $\sP$.
\item[{\rm(2)}]
Let $\kappa\in\NN\cup\{\infty\}$, $t>0$, $\ul{\V{x}}\in\RR^\nu$, and $\gamma\in\Omega$. 
Then $W_{\kappa,t}(\ul{\V{x}},\gamma)$ and $\wt{W}_{\kappa,t}(\ul{\V{x}},\gamma)$
are positivity improving with respect to $\sP$.
\end{enumerate}
\end{prop}

\begin{proof}
{\em Step~1.} Put $s:=\nf{t}{2}>0$. We first show that $F_{0,s}(h)^*$ is positivity
preserving by applying arguments from \cite{Simon1974}. Let also
$\V{g}\in\HP_{\RR}^n$ and $G\in\sS(\RR^n)$. Then a straightforward combination of
\eqref{defWeylexpv}, \eqref{Fadjexpv}, and \eqref{kent1} yields
\begin{align*}
F_{0,s}(&h)^*G(\vp(\V{g}))\zeta(0)
\\
&=\frac{1}{(2\pi)^{\nf{n}{2}}}\int_{\RR^n}
e^{-\|(1-e^{-2s\omega})\V{\xi}\cdot\V{g}\|_{\HP}^2/2-i\SPn{h}{\V{\xi}\cdot\V{g}}_{\HP}}
\hat{G}(-\V{\xi})\sW(-i\V{\xi}\cdot e^{-s\omega}\V{g})\zeta(0){\Id\V{\xi}}.
\end{align*}
Applying \eqref{kent1} once more we obtain
$F_{0,s}(h)^*G(\vp(\V{g}))\zeta(0)=\wt{G}(\vp(e^{-s\omega}\V{g}))\zeta(0)$,
with $\wt{G}$ denoting the inverse Fourier transform of
$\V{\xi}\mapsto e^{-\|(1-e^{-2s\omega})\V{\xi}\cdot\V{g}\|_{\HP}^2/2-i\V{\xi}\cdot\SPn{h}{\V{g}}_{\HP}}
\hat{G}(-\V{\xi})$.
If $G$ is non-negative, then $\wt{G}$ is non-negative as well, because it is the convolution
of a Gaussian and a shifted version of $G$. We conclude that
$F_{0,s}(h)^*$ maps $\mr{\sP}$ into itself. Since it is bounded, it also maps $\sP$ into itself.

\smallskip

\noindent{\em Step~2.} We infer from Step~1 that 
$F_{0,s}(h)=F_{0,s}(h)^{**}$ and  $e^{-s\Id\Gamma(\omega)}=F_{0,s}(0)^*$ are 
positivity preserving. In fact, it is well-known that $e^{-s\Id\Gamma(\omega)}$ improves positivity,
which follows from \cite[Thm.XIII.44(a)$\Rightarrow$(e)]{ReedSimonIV} applied to 
$\cU e^{-s\Id\Gamma(\omega)}\cU^*$ and the fact that $1$ is a non-degenerate eigenvalue of 
$e^{-s\Id\Gamma(\omega)}$ with strictly positive eigenvector $\zeta(0)$.
Since $s=\nf{t}{2}$, we further have $F_{0,t}(h)=F_{0,s}(h)e^{-s\Id\Gamma(\omega)}$ and
$F_{0,t}(h)^*=e^{-s\Id\Gamma(\omega)}F_{0,s}(h)^*$. This implies Part~(1) since a composition
of a positivity preserving and a positivity improving operator improves positivity.
Part~(2) is an immediate consequence of Part~(1), Def.~\ref{defW}, and Def.~\ref{defnnonFockSG}.
\end{proof}

\begin{prop}\label{propOscar2}
Let $\kappa\in\NN\cup\{\infty\}$, $t>0$, and $\gamma\in\Omega$.
Then $\wh{W}_{\kappa,t}(\V{0},\gamma)$ is positivity improving with respect to $\sP$. 
\end{prop}

\begin{proof}
In addition to the arguments of the proof of Prop.~\ref{propOscar1}, it only remains to show that
$\Gamma(e^{i\V{m}\cdot\V{y}})$ is positivity preserving, for every $\V{y}\in\RR^3$. This follows,
however, from the relation
$$
\Gamma(e^{i\V{m}\cdot\V{y}})\sW(-i\V{\xi}\cdot\V{g})\zeta(0)
=\sW(-i\V{\xi}\cdot e^{i\V{m}\cdot\V{y}}\V{g})\zeta(0),
$$
and the fact that $e^{i\V{m}\cdot\V{y}}\V{g}\in\HP_{\RR}^n$, for every $\V{g}\in\HP^n_{\RR}$.
\end{proof}

\begin{thm}\label{thmOscar1}
Let $\kappa\in\NN\cup\{\infty\}$ and $t>0$.
Then $e^{-tH_{N,\kappa}^V}$ and $e^{-t\wt{H}_{N,\kappa}^V}$ are positivity improving with 
respect to the self-dual convex cone in $L^2(\RR^\nu,\sF)$ given by
\begin{equation}\label{def-int-cone}
\int_{\RR^\nu}^\oplus\sP\Id\ul{\V{x}}:=
\big\{\Psi\in L^2(\RR^\nu,\sF):\,\Psi(\ul{\V{x}})\in\sP,\,\text{a.e. $\ul{\V{x}}$}\big\}.
\end{equation}
In particular, if $\inf\sigma(H_{N,\kappa}^V)=\inf\sigma(\wt{H}_{N,\kappa}^V)$ is an eigenvalue of
$H_{N,\kappa}^V$ or $\wt{H}_{N,\kappa}^V$, 
then it is non-degenerate and the corresponding eigenvector can be chosen strictly positive.
\end{thm}

\begin{proof}
The first assertion is an easy consequence of Prop.~\ref{propOscar1}, the Feynman-Kac formula
\eqref{FKreg}, Def.~\ref{defHinfty}, Def.~\ref{defnHinftynonFock}, and well-known properties of 
Brownian motion. The last statement follows from \cite[Cor.~1.2]{Faris1972}.
\end{proof}

\begin{thm}\label{thmOscar2}
Let $\kappa\in\NN\cup\{\infty\}$ and $t>0$. Then $e^{-t\wh{H}_\kappa(\V{0})}$ is positivity
improving with respect to $\sP$. In particular, if $\inf\sigma(\wh{H}_\kappa(\V{0}))$ is an eigenvalue, 
then it is non-degenerate and the corresponding eigenvector can be chosen strictly positive.
\end{thm}

\begin{proof}
The first assertion follows from Prop.~\ref{propOscar2}, the second from \cite[Cor.~1.2]{Faris1972}.
\end{proof}

An analog of the previous theorem for the two-dimensional relativistic Nelson model has been
proved by different methods in \cite{Sloan1974}.

The following two corollaries are new only in the case where $\kappa=\infty$ and the functions
$\Psi$ and $\vr$ in \eqref{poldi1} and \eqref{corinna1}, respectively, are not strictly positive; see
\cite[Chap.~6]{LHB2011} and compare Thm.~\ref{thminfspecapp}.

\begin{cor}\label{corinfspec1}
For all $\kappa\in\NN\cup\{\infty\}$, 
\begin{align}\label{poldi1}
\inf\sigma(H_{N,\kappa}^V)&=-\lim_{t\to\infty}\frac{1}{t}\ln\SPn{\Psi}{e^{-tH_{N,\kappa}^V}\Psi},
\quad0\not=\Psi\in\int_{\RR^\nu}^\oplus\sP\Id\ul{\V{x}},
\end{align}
where $H_{N,\kappa}^V$ can be replaced by $\wt{H}_{N,\kappa}^V$ on one or both sides. 
\end{cor}

\begin{proof}
Let $\cU:\sF\to L^2(\cQ,\eta)$ be a unitary operator as described in front of Prop.~\ref{propOscar1},
and let $\Theta:L^2(\RR^\nu,\sF)\to L^2(\RR^\nu\times\cQ,\lambda^\nu\otimes\eta)$ be the unitary 
operator obtained by composing $\int_{\RR^\nu}^\oplus\cU\Id\ul{\V{x}}$ with the canonical 
isomorphism $L^2(\RR^\nu,L^2(\cQ,\eta))=L^2(\RR^\nu\times\cQ,\lambda^\nu\otimes\eta)$.
Then $\Theta e^{-tH_{N,\kappa}^V}\Theta^*$ is positivity improving with respect to the natural
cone in $L^2(\RR^\nu\times\cQ,\lambda^\nu\otimes\eta)$, for all $t>0$. Therefore, \eqref{poldi1}
follows from the well-known Thm.~\ref{thminfspecapp}.
\end{proof}

The next corollary will be used to verify \eqref{asympintro}.

\begin{cor}\label{corinfspec2}
Let $\kappa\in\NN\cup\{\infty\}$ and let $\vr:\RR^\nu\to\RR$ be 
measurable and non-negative with $\int_{\RR^\nu}\vr^2\Id\lambda^\nu\in(0,\infty)$. Then 
\begin{align}\label{corinna1}
\inf\sigma(H_{N,\kappa}^V)&=-\lim_{t\to\infty}\frac{1}{t}\ln\bigg(\int_{\RR^\nu}\vr(\ul{\V{x}})
\EE\big[e^{u_{\kappa,t}^N(\ul{\V{x}})-\int_0^tV(\ul{\V{x}}+\ul{\V{b}}_s)\Id s}
\vr(\ul{\V{x}}+\ul{\V{b}}_t)\big]\Id\ul{\V{x}}\bigg),
\end{align}
and the same formula holds true with $H_{N,\kappa}^V$ replaced by $\wt{H}_{N,\kappa}^V$
and/or $u_{\kappa,t}^N$ replaced by $\tilde{u}_{\kappa,t}^N$. 
\end{cor}

\begin{proof}
We only have to choose $\vr\zeta(0)$ in \eqref{poldi1}, apply our Feynman-Kac formulas, and take
Rem.~\ref{remW}(3) into account.
\end{proof}

\begin{proof}[Proof of the lower bound in \eqref{asympintro}]
Let $\vr:\RR^\nu\to\RR$ be measurable with $0\le\vr\le1$ and $\int_{\RR^\nu}\vr\Id\lambda^\nu=1$.
Then $\int_{\RR^\nu}\vr^2\Id\lambda^\nu\in(0,\infty)$. Plugging the bound in Rem.~\ref{rem256}
into \eqref{corinna1}, we then find
$\inf\sigma(H_{N,\kappa}^0)\ge-256\pi^2q\ee^4N^3-c\ee^2N^2$, $\kappa\in\NN\cup\{\infty\}$,
provided that $\ee^2N\ge1$.
Since $q>1$ is arbitrary, this proves the first inequality in \eqref{asympintro}.
\end{proof}

Since $\sP\subset\sF$ is the pre-image of the canonical convex cone in the 
$L^2$-space associated with a probability measure
under a unitary map $\cU$ satisfying $\cU^{-1}1=\zeta(0)$, the mere fact that 
every $e^{-t\wh{H}_\kappa(\V{0})}$, $t\ge0$, is positivity preserving with respect to $\sP$ already 
implies the first identity in
\begin{align*}
\inf\sigma(\wh{H}_\kappa(\V{0}))&=-\lim_{t\to\infty}\frac{1}{t}\ln\SPn{\zeta(0)}{
e^{-t\wh{H}_\kappa(\V{0})}\zeta(0)}=-\lim_{t\to\infty}\frac{1}{t}
\ln\Big(\EE\big[e^{u_{\kappa,t}[\V{B}]}\big]\Big),
\end{align*}
which holds for all $\kappa\in\NN\cup\{\infty\}$; see \cite{LorincziMinlosSpohn2002} and 
\cite[Lem.~2.5]{AbdesselamHasler2012} or Thm.~\ref{thminfspecapp}.
The second identity, following from the Feynman-Kac formula, is new for $\kappa=\infty$ and $\mu=0$, 
because the limiting complex action has not been constructed before in this case.

One can define a different notion of positivity in the Fock space $\sF$ by requiring that an element be
positive, iff its components with respect to the direct sum in \eqref{defFockspace} are positive
according to the canonical notion of positivity in $\RR$ and $L^2(\RR^{3n},\RR)$, $n\in\NN$.
If $\ee<0$, then the semi-group of {\em every} fiber Hamiltonian $\wh{H}_\kappa(\V{\xi})$, 
$\V{\xi}\in\RR^3$, is known
to be positivity improving for finite $\kappa$ and positivity preserving for $\kappa=\infty$ with respect
to this alternative positive cone in $\sF$ \cite{Froehlich1973,Froehlich1974}; see also
\cite[\textsection3.3]{Moeller2005}. 
(In \cite{Froehlich1973,Froehlich1974}, ergodicity of $e^{-t\wh{H}_\infty(\V{\xi})}$, $t>0$, is announced,
but a proof seems to be unavailable. Miyao showed that fiber Hamiltonians in the related but simpler 
renormalized Fr\"{o}hlich polaron model generate positivity improving semi-groups \cite{Miyao2010}.) 
If $\ee>0$, then one can either change its sign by a unitary transformation or change the sign of the 
elements called positive by definition in the subspaces in \eqref{defFockspace} with an odd $n$; see, 
e.g., \cite{AbdesselamHasler2012,Moeller2005}. In any case \cite{AbdesselamHasler2012}, 
\begin{align*}
\inf\sigma(\wh{H}_\kappa(\V{\xi}))&=-\lim_{t\to\infty}\frac{1}{t}\ln\Big(\EE\big[
e^{u_{\kappa,t}[\V{B}]+i\V{\xi}\cdot\V{B}_t}\big]\Big),\quad\V{\xi}\in\RR^3,\,\kappa\in\NN.
\end{align*}


\subsection{Upper bound on the minimal energy}\label{appub}

\noindent
Next, we complement our lower bound on the spectrum of the renormalized Nelson
Hamiltonian by deriving a corresponding upper bound. We restrict ourselves to the case $V=0$,
$\eta=1$, $\kappa=\infty$, for sufficiently large $\ee^2N>4$. 

The coefficient for the leading $\ee^4N^3$-term in our bound will be the Pekar energy
$E_{\mathrm{P}}:=\inf\cE_{\mathrm{P}}$, which is the {\em minimum} of the Pekar
functional given by $\dom(\cE_{\mathrm{P}}):=\{g\in W^{1,2}(\RR^3):\|g\|=1\}$ and
\begin{align*}
\sE_{\mathrm{P}}(g)&:=\frac{1}{2}\|\nabla g\|^2-
\frac{4\pi}{\sqrt{2}}\int_{\RR^3}\frac{|\hat{\rho}|^2}{\V{m}^2}\Id\lambda^3,
\quad g\in\dom(\sE_{\mathrm{P}}),\;\text{with $\rho:=|g|^2$.}
\end{align*}
This definition makes sense, as $g\in\dom(\cE_{\mathrm{P}})$ and $\rho:=|g|^2$ entails
$\rho\in W^{1,1}(\RR^3)$ and $\|\hat{\rho}\|_\infty\le\|\rho\|_1=1$, so that $|\hat{\rho}|^2/\V{m}^2$ is at 
least locally integrable. From the Sobolev inequality we then infer that $\rho\in L^{\nf{3}{2}}(\RR^3)$, 
whence $\hat{\rho}\in L^3(\RR^3)$ by the Hausdorff-Young inequality. Hence, by H\"{o}lder's inequality, 
\begin{align}\label{carlo1}
\int_{\{|\V{m}|\ge R\}}\frac{|\hat{\rho}|^2}{\V{m}^2}\Id\lambda^3&\le
\|\hat{\rho}\|_3^{\nf{2}{3}}\Big(\int_{\{|\V{m}|\ge R\}}\frac{\Id\lambda^3}{|\V{m}|^6}\Big)^{\nf{1}{3}}
=\|\hat{\rho}\|_3^{\nf{2}{3}}\frac{(4\pi/3)^{\nf{1}{3}}}{R},
\end{align}
for all $R>0$, which finally shows that $|\hat{\rho}|^2/\V{m}^2\in L^1(\RR^3)$.

The existence of an up to translations unique minimizer of $\cE_{\mathrm{P}}$ 
has been shown in \cite{Lieb1977} and numerics reveals that
$E_{\mathrm{P}}=-0.10851 \ldots\,$, \cite{GerlachLoewen1991}.

\begin{thm}\label{thmPekarub}
Let $\eta=1$ and let $E^0_{N,\infty}(\ee)$ denote the infimum of the spectrum of 
$H_{N,\infty}^0$ with coupling constant $\ee$. Then there exists a universal constant $c>0$ 
such that
\begin{align*}
E^0_{N,\infty}(\ee)&\le8\pi^4\ee^4N^3E_{\mathrm{P}}+c(1+\mu+\ln(\ee^2N))\ee^2N^2,\quad
\text{so long as $\ee^2N>4$.}
\end{align*}
\end{thm}

\begin{proof}
{\em Step~1.} In this step we only assume that $\ee^2N\ge1$.

We pick some $q\in(1,\infty)$, denote its conjugate exponent by $q'$, and
put $\Lambda:=70\pi q'\ee^2N$. We shall use some notation introduced in the proof of 
Thm.~\ref{lemuexp} and Rem.~\ref{remuge}. Pick some measurable
$\vr:\RR^\nu\to\RR$ with $0\le\vr\le1$ and $\int_{\RR^\nu}\vr\Id\lambda^\nu=1$, so that also
$\int_{\RR^\nu}\vr^2\Id\lambda^\nu\in(0,1]$.
Then the inequality derived in Rem.~\ref{remuge} implies
\begin{align*}
\EE&\Big[e^{u_{\Lambda,\infty,t}^{N,<}(\ul{\V{x}})}\vr(\ul{\V{x}}+\ul{\V{b}}_t)\Big]
\\
&\le
\EE\Big[e^{qu_{\Lambda,\infty,t}^{N,<}(\ul{\V{x}})+qu_{\Lambda,\infty,t}^{N,>}(\ul{\V{x}})}
\vr(\ul{\V{x}}+\ul{\V{b}}_t)\Big]^{\nf{1}{q}}
\EE\Big[e^{-q'u_{\Lambda,\infty,t}^{N,>}(\ul{\V{x}})}\vr(\ul{\V{x}}+\ul{\V{b}}_t)\Big]^{\nf{1}{q'}}
\\
&\le{c^{\nf{1}{q'}}}{e^{c'\ee^2N^2t+c'N(1\vee t)/q'+NE_{\Lambda,\infty}t}}
\EE\Big[e^{qu_{\infty,t}^N(\ul{\V{x}})}\vr(\ul{\V{x}}+\ul{\V{b}}_t)\Big]^{\nf{1}{q}}
,\quad t\ge0,\,\ul{\V{x}}\in\RR^\nu,
\end{align*}
which on account of H\"{o}lder's inequality, $0\le\vr\le1$, and $\|\vr\|_1=1$ permits to get
\begin{align*}
\frac{1}{c^{\nf{1}{q'}}}{e^{-c''\ee^2N^2t-NE_{\Lambda,\infty}t}}&
\int_{\RR^\nu}\vr(\ul{\V{x}})\EE\Big[e^{u_{\Lambda,\infty,t}^{N,<}(\ul{\V{x}})}
\vr(\ul{\V{x}}+\ul{\V{b}}_t)\Big]\Id\ul{\V{x}}
\\
&\le\bigg(\int_{\RR^\nu}\vr(\ul{\V{x}})
\EE\Big[e^{qu_{\infty,t}^N(\ul{\V{x}})}\vr(\ul{\V{x}}+\ul{\V{b}}_t)\Big]\Id\ul{\V{x}}\bigg)^{\nf{1}{q}},
\quad t\ge1.
\end{align*}
Therefore, we may conclude from Cor.~\ref{corinfspec2} that
\begin{align}\nonumber
\frac{1}{q}E^0_{N,\infty}(q^\eh\ee)
&=-\lim_{t\to\infty}\frac{1}{qt}\ln\bigg(\int_{\RR^\nu}\vr(\ul{\V{x}})
\EE\Big[e^{qu_{\infty,t}^N(\ul{\V{x}})}\vr(\ul{\V{x}}+\ul{\V{b}}_t)\Big]\Id\ul{\V{x}}\bigg)
\\\label{carlo2}
&\le c''\ee^2N^2+NE_{\Lambda,\infty}+E_{N,\Lambda}^{0,<}(\ee).
\end{align}
Here, $E_{\Lambda,\infty}=\int_{\{|\V{m}|<\Lambda\}}f_\infty\beta_\infty\Id\lambda^3$, as 
defined in the proof of Thm.~\ref{lemuexp}, and the number
$E_{N,\Lambda}^{0,<}(\ee):=\inf\mathfrak{q}_{N,\Lambda}^{0,<}$ is the infimum of the spectrum 
of the regularized Nelson Hamiltonian with sharp ultra-violet cut-off at $\Lambda$, i.e., the form 
$\mathfrak{q}_{N,\Lambda}^{0,<}$ is defined as in Def.~\ref{defNelsonHamreg} with 
$\eta:=1_{\{|\V{m}|<\Lambda\}}$, $\chi=1_{\{|\V{m}|<1\}}$, and for some arbitrary $\kappa>\Lambda$. 

\smallskip

\noindent{\em Step~2.} In this step we only assume that $\ee^2N>1$.

Let $g\in\dom(\cE_{\mathrm{P}})$ and put $g^{\otimes_N}(\ul{\V{x}}):=g(\V{x}_1)\dots g(\V{x}_N)$,
a.e. $\ul{\V{x}}\in\RR^\nu$. Let also $h\in\dom(\omega)$ with $\supp(h)\subset\{|\V{m}|\le\Lambda\}$ 
and define $\mr{\zeta}(h):=\|\zeta(h)\|^{-1}\zeta(h)=e^{-\|h\|_{\HP}^2/2}\zeta(h)$, the coherent state
colinear to the exponential vector $\zeta(h)$ given by \eqref{defexpv}. 
Put $f_\Lambda:=\ee\omega^\mh1_{\{|\V{m}|<\Lambda\}}$. Then
we infer from \eqref{vpaad}, \eqref{aexpv}, and \eqref{dGammaexpv} that
\begin{align*}
{\mathfrak{q}}_{N,\Lambda}^{0,<}[g^{\otimes_N}\mr{\zeta}(h)]
&=\frac{N}{2}\|\nabla g\|^2+\|\omega^\eh h\|_{\HP}^2
+N\int_{\RR^3}2|g(\V{x})|^2\Re\SPn{h}{e^{-i\V{m}\cdot\V{x}}f_\Lambda}_{\HP}\Id{\V{x}}.
\end{align*}
If $\rho:=|g|^2$,  then the third term on the right hand 
side equals $2(2\pi)^{\nf{3}{2}}N\Re\SPn{h}{\hat{\rho}f_\Lambda}_{\HP}$. Substituting 
$h_r:=-(2\pi)^{\nf{3}{2}}\ee N\omega^{-\nf{3}{2}}1_{\{r\le|\V{m}|<\Lambda\}}\hat{\rho}\in\dom(\omega)$
with $r>0$ for $h$, we thus obtain
\begin{align}\label{thy1}
E_{N,\Lambda}^{0,<}(\ee)&\le\frac{N}{2}\|\nabla g\|^2-(2\pi)^3\ee^2N^2\sup_{r>0}
\int_{\{r\le|\V{m}|<\Lambda\}}\frac{|\hat{\rho}|^2}{\omega^2}\Id\lambda^3.
\end{align}
For $\vs>0$, put $g_\vs(\V{x}):=\vs^{\nf{3}{2}}g(\vs\V{x})$, $\V{x}\in\RR^3$, and 
$\rho_\vs:=|g_\vs|^2$. Substituting $g_\vs$ for $g$ in \eqref{thy1} and taking the relations
$(\rho_\vs)^\wedge(\V{k})=\hat{\rho}(\V{k}/\vs)$, $\V{k}\in\RR^3$,
and $|\nabla g_\vs(\V{x})|^2=\vs^5|(\nabla g)(\vs\V{x})|$, $\V{x}\in\RR^3$, into account
afterwards, we deduce that
\begin{align*}
E_{N,\Lambda}^{0,<}(\ee)&\le\frac{\vs^2N}{2}\|\nabla g\|^2-(2\pi)^3\vs\ee^2N^2
\int_{\{|\V{m}|\le\Lambda/\vs\}}\frac{|\hat{\rho}|^2}{\V{m}^2+\mu^2/\vs^2}\Id\lambda^3.
\end{align*}
Since $\|\hat{\rho}\|_\infty\le1$ as observed above, we further have
\begin{align*}
\bigg|\int_{\{|\V{m}|\le\Lambda/\vs\}}\Big(\frac{|\hat{\rho}|^2}{\V{m}^2+\mu^2/\vs^2}
-\frac{|\hat{\rho}|^2}{\V{m}^2}\Big)\Id\lambda^3\bigg|
&\le\int_{\RR^3}\frac{\mu^2/\vs^2}{\V{m}^2(\V{m}^2+\mu^2/\vs^2)}\Id\lambda^3
=\frac{2\pi^2\mu}{\vs}.
\end{align*}
Choosing $\vs:=2^{\nf{3}{2}}\pi^2\ee^2N$ 
and $q':=\ee^2N>1$ (so that $\Lambda=70\pi\ee^4N^2$) and plugging in a 
minimizer of the Pekar functional, call its density $\rho_0$, we deduce that
\begin{align*}
E_{N,\Lambda}^{0,<}(\ee)&\le{2^4\pi^5\mu}{\ee^2N^2}
+8\pi^4\ee^4N^3E_{\mathrm{P}}+8\pi^4\ee^4N^3
\int_{\{|\V{m}|>35\ee^2N/\sqrt{2}\pi\}}\frac{|\hat{\rho}_0|^2}{\V{m}^2}\Id\lambda^3.
\end{align*}
Combining this with \eqref{carlo1} and \eqref{carlo2} we arrive at
\begin{align*}
\frac{E_{N,\infty}^0((\ee^2N/(\ee^2N-1))^\eh\ee)}{\ee^2N/(\ee^2N-1)}
&\le8\pi^4\ee^4N^3E_{\mathrm{P}}+NE_{70\pi\ee^4N^2,\infty}+c(1+\mu)\ee^2N^2.
\end{align*}
\noindent{\em Step~3.} Assume now that $\ee^2N>4$.
Then the equation $(\tilde{\ee}^2N/(\tilde{\ee}^2N-1))\tilde{\ee}^2=\ee^2$ for $\tilde{\ee}^2$
has the solution $\tilde{\ee}^2=(\ee^2+\sqrt{\ee^4-4\ee^2/N})/2$ satisfying $\tilde{\ee}^2<\ee^2$ and
$N\tilde{\ee}^2>N\ee^2/2>2$. We conclude by applying the last bound of Step~2 to $\tilde{\ee}$ 
instead of $\ee$ and observing that 
$E_{70\pi\tilde{\ee}^4N^2,\infty}\le c\tilde{\ee}^2(1+\ln(\ee^2N))$.
\end{proof}


\subsection{Continuity}\label{sseccont}

\noindent
Combining our results with some recent results of \cite{Matte2016},
we obtain some information on the continuity properties of elements in the range
of our Feynman-Kac semi-groups and, in particular, of eigenvectors of the ultra-violet
renormalized Nelson Hamiltonian (if any). We shall consider continuity with respect to
$\ul{\V{x}}$, $\ee$, and the external potential. More precisely, we shall consider $\ee$ and $V$
as in Hyp.~\ref{hypNelson}, a sequence of real numbers $\{\ee_n\}_{n\in\NN}$ such that
$\ee_n\to\ee$, as $n\to\infty$, and a sequence of 
Kato decomposable potentials $V_n:\RR^\nu\to\RR$, i.e., $V_n=V_{n+}-V_{n-}$
with $0\le V_{n-}\in K_\nu$ and $0\le V_{n+}\in K_\nu^{\loc}$, satisfying 
\begin{align}\label{Vnconv1}
\lim_{r\downarrow0}\sup_{n\in\NN}\sup_{\ul{\V{x}}\in\RR^\nu}\int_{\RR^\nu}
1_{\{|\ul{\V{x}}-\ul{\V{y}}|<r\}}\frac{V_{n-}(\ul{\V{y}})}{|\ul{\V{x}}-\ul{\V{y}}|^{\nu-2}}\Id\ul{\V{y}}&=0,
\end{align}
and
\begin{align}\label{Vnconv2}
\text{$\forall$ compact $\cC\subset\RR^\nu$:}\quad
\lim_{n\to\infty}\sup_{\ul{\V{x}}\in\RR^\nu}\int_{\cC}1_{\{|\ul{\V{x}}-\ul{\V{y}}|<1\}}
\frac{|V_n(\ul{\V{y}})-V(\ul{\V{y}})|}{|\ul{\V{x}}-\ul{\V{y}}|^{\nu-2}}\Id\ul{\V{y}}&=0.
\end{align}
We recommend \cite{BHL2000} for a discussion of the notion of convergence expressed
by these conditions.

For all $\kappa\in\NN\cup\{\infty\}$ and $n\in\NN$, we let $(T^{(n)}_{\kappa,t})_{t\ge0}$ 
denote the Feynman-Kac semi-group defined by means of $\ee_n$ and $V_n$. We also put
$T_{\kappa,t}^{(\infty)}:=T^V_{\kappa,t}$ (the latter defined by means of $\ee$) to have a
unified notation.

One should keep in mind that $T_{\infty,t}^{(n)}$ is $\chi$-independent, for all $n\in\NN\cup\{\infty\}$, 
while reading the next theorem, whose statement is new in the case $\kappa=\infty$ only.

\begin{thm}\label{thmcont}
Consider the situation described in the preceding paragraphs and
assume that the cut-off function $\chi$ satisfies $\omega^\eh\chi\in\HP$ in addition 
to Hyp.~\ref{hypNelson}(2).
Let $\kappa\in\NN\cup\{\infty\}$, $\tau_2\ge\tau_1>0$, and $p\in[1,\infty]$. 
Then the set of $\sF$-valued functions
\begin{align*}
\big\{ T_{\kappa,s}^{(\infty)}\Phi\,\big|\,
\Phi\in L^p(\RR^\nu,\sF),\,\|\Phi\|_p\le1,\,s\in[\tau_1,\tau_2]\big\}
\end{align*}
is equicontinuous at every point of $\RR^\nu$. 
If $V\in K_\nu$, then it is uniformly equicontinuous on $\RR^\nu$. Furthermore,
\begin{align}
\lim_{n\to\infty}(T_{\kappa,t}^{(n)}\Psi_n)(\ul{\V{x}}_n)=(T_{\kappa,t}^{(\infty)}\Psi)(\ul{\V{x}}),
\quad t>0,
\end{align}
for every converging sequence $\{\Psi_n\}_{n\in\NN}$ in $L^p(\RR^\nu,\sF)$ with limit $\Psi$ and 
every converging sequence $\{\ul{\V{x}}_n\}_{n\in\NN}$ in $\RR^\nu$ with limit $\ul{\V{x}}$.
\end{thm}

\begin{proof}
If $\kappa\in\NN$, then all assertions are special cases of \cite[Thm.~8.1 and Cor.~8.2]{Matte2016}. 
Hence, for $\kappa=\infty$, all assertions follow from the convergence
$$
\sup_{n\in\NN\cup\{\infty\}}\sup_{s\in[\tau_1,\tau_2]}\|T_{\kappa,s}^{(n)}-T_{\infty,s}^{(n)}
\|_{p,\infty}\xrightarrow{\;\;\kappa\to\infty\;\;}0;
$$ 
which is a consequence of Prop.~\ref{propT}(3) and the fact that \eqref{Vnconv2} implies
$$
\sup_{n\in\NN}\sup_{\ul{\V{x}}\in\RR^\nu}
\EE\big[e^{p\int_0^{\tau_2}V_{n-}(\ul{\V{x}}+\ul{\V{b}}_s)\Id s}\big]<\infty,\quad p>0;
$$
confer, e.g., the proof of Lem.~C.1 in \cite{BHL2000}.
\end{proof}


\appendix

\section{The Kolmogorov test}\label{appKol}

\noindent
For the reader's convenience we state the version of Kolmogoroff's lemma
employed in the main text. It can be proved by simple modifications of the arguments 
presented in \cite[\textsection1.4]{Kunita1990}, \cite[pp.~268/9]{Me1982}, 
and \cite[Lem.~13]{Priouret1974}.

\begin{lem}\label{lemKolNev}
For every $\tau\ge0$, let
$(X_t(\tau))_{t\in I}$ be a stochastic process with values in the separable Hilbert space $\sK$
whose paths $[0,\infty)\ni t\mapsto X_t(\tau,\gamma)\in\sK$ are continuous, for all $\gamma$
in the complement of a (possibly $\tau$-dependent) $\PP$-zero set. 
Assume that there exist a non-decreasing function $b:[0,\infty)\to[0,\infty)$, $p\ge1$, and 
$\ve>0$ such that
\begin{align}\label{bd-Kol}
\EE\big[\sup_{s\le t}\|X_s(\sigma)-X_s(\tau)\|^p\big]
&\le b(t)|\tau-\sigma|^{1+\ve},\quad t,\sigma,\tau\ge0.
\end{align}
Then there exists a $\fB([0,\infty)^2)\otimes\fF$-measurable map
$[0,\infty)^2\times\Omega\ni(t,\tau,\gamma)\mapsto X^*_t(\tau,\gamma)\in\sK$
such that the following holds:
\begin{enumerate}[leftmargin=*]
\item[{\rm(1)}]
For all $\tau\ge0$, there is a $\PP$-zero set $\sN_{\tau}\in\fF$ such that
$$
X_t(\tau,\gamma)=X_t^*(\tau,\gamma),\quad t\ge0,\,\gamma\in\Omega\setminus\sN_{\tau}.
$$
\item[{\rm(2)}] Let $\gamma\in\Omega$. Then the map 
$[0,\infty)^2\ni(t,\tau)\mapsto X^*_{t}(\tau,\gamma)\in\sK$ is continuous. 
\item[{\rm(3)}] For all $\delta\in(0,\ve/p)$, there exists $c_{p,\delta,\ve}>0$ with the following
property: If $T>0$, then we can construct an adapted non-negative process $K_{\delta,T}$ such that
\begin{align}\label{Kol1}
\EE[K_{\delta,T,t}^p]&\le c_{p,\delta,\ve}(1+T)^pb(t),\quad t\ge0,
\end{align}
and
\begin{align*}
\sup_{s\le t}\|X_s^*(\sigma,\gamma)-X_s^*(\tau,\gamma)\|&\le K_{\delta,T,t}(\gamma)
(1\wedge|\tau-\sigma|)^\delta,\quad\gamma\in\Omega,
\end{align*}
for all $ t\ge0$ and $\sigma,\tau\in[0,T]$.
\end{enumerate}
\end{lem}

\begin{rem}\label{remKol} 
In the situation of Lem.~\ref{lemKolNev} assume in addition that
$$
\EE\big[\sup_{s\le\tau}\|X_s(\tau)\|^p\big]\le d_p(\tau),\quad\tau\ge0, 
$$
for some increasing function $d_p:[0,\infty)\to[0,\infty)$. On account of Statement~(1) of
Lem.~\ref{lemKolNev} we can replace $X$ by $X^*$ in the previous bound. Then the estimate
\begin{align*}
\sup_{s,\tau\le t}\|X_s^*(\tau)\|&\le\sup_{s\le t}\|X_s^*(t)\|+\sup_{s,\tau\le t}\|X_s^*(\tau)-X_s^*(t)\|
\\
&\le \sup_{s\le t}\|X_s^*(t)\|+\sup_{\tau\le t}K_{\delta,t,t}(1\wedge|t-\tau|)^\delta,
\end{align*}
together with \eqref{Kol1} implies
\begin{align}\label{Kol2}
\EE\big[\sup_{s,\tau\le t}\|X_s^*(\tau)\|^p\big]&\le 2^pd_p(t)+2^pc_{p,\delta,\ve}(1+t)^p
(1\wedge t)^{p\delta}b(t),\quad t\ge0. 
\end{align}
\end{rem}


\section{On exponential moments of sums of pair potentials}\label{apppair}

\noindent
For the sake of completeness we provide the proof of an elementary and doubtlessly well-known
estimate that we used in the proof of Lem.~\ref{lemfrkv} to bound the exponential moment of a
sum of pair potentials. We refer to \cite{Bley2016} for a different discussion of slightly more general
estimates that addresses some optimality aspects as well.

\begin{lem}\label{lempairs}
Suppose that $N\ge2$ and that $X_{(i,j)}$, $i,j\in\{1,\ldots,N\}$, $i<j$, 
are non-negative random variables
all having the same distribution such that $X_{(i,j)}$ and $X_{(k,\ell)}$ are independent whenever
$\{i,j\}\cap\{k,\ell\}=\emptyset$. Let $\cP_N:=\{(i,j):i,j\in\{1,\ldots,N\},i<j\}$ and put
$Y_N(s):=\sup_{p\in\cP_N}\EE[X_p^s]$, $s>0$. Then
\begin{align*}
\EE\Big[\prod_{p\in\cP_N}X_p\Big]&\le\left\{
\begin{array}{ll}
Y_N(N-1)^{N/2},&\text{if $N$ is even,}
\\
Y_N(N)^{(N-1)/2},&\text{if $N$ is odd.}
\end{array}
\right.
\end{align*}
\end{lem}

Note that, $Y_N(N-1)^{N/2}\le Y_N(N)^{(N-1)/2}$, $N\ge2$, by H\"{o}lder's inequality.

\begin{proof}
We prove the claim by induction separately for even and odd $N$.
The claim is trivial for $N=2$ and obvious for $N=3$ by H\"{o}lder's inequality. 
Suppose that $N>2$ is even and that the claim is true for $N-2$. By virtue of
\begin{align*}
\prod_{p\in\cP_N}X_p&=\prod_{j=2}^N\Big\{X_{(1,j)}
\mathop{\prod_{(k,\ell)\in \cP_N}}_{k,\ell\notin\{1,j\}}X_{(k,\ell)}^{1/(N-3)}\Big\},
\end{align*} 
as well as H\"{o}lder's inequality and independence, we obtain
\begin{align*}
\EE\Big[\prod_{p\in\cP_N}X_p\Big]\le\prod_{j=2}^N\EE[X_{(1,j)}^{N-1}]^{1/(N-1)}
\EE\Big[\mathop{\prod_{(k,\ell)\in \cP_N}}_{k,\ell\notin\{1,j\}}X_{(k,\ell)}^{\frac{N-1}{N-3}}
\Big]^{1/(N-1)}.
\end{align*}
Since the induction hypothesis can be applied to the latter expectations, this implies
\begin{align*}
\EE\Big[\prod_{p\in\cP_N}X_p\Big]\le Y_N(N-1)
Y_N\Big({\frac{N-1}{N-3}\cdot(N-3)}\Big)^{(N-2)/2}=Y_N({N-1})^{N/2}.
\end{align*}

If $N>3$ is odd and the claim is true for $N-2$, then we write
\begin{align*}
\prod_{p\in\cP_N}X_p&=\bigg[\prod_{j=2}^N\Big\{X_{(1,j)}
\mathop{\prod_{(k,\ell)\in \cP_N}}_{k,\ell\notin\{1,j\}}X_{(k,\ell)}^{1/(N-2)}\Big\}\bigg]\cdot
\mathop{\prod_{(k,\ell)\in \cP_N}}_{k\not=1}X_{(k,\ell)}^{1/(N-2)}.
\end{align*} 
Here the induction hypothesis applies to the $N$-th power of the multiple products inside the curly 
brackets $\{\cdots\}$ and the already proven case of an even number of random variables applies to
the $N$-th power of the last product. Employing H\"{o}lder's inequality for an $N$-fold product and 
independencies first and taking these remarks into account, we arrive at
\begin{align*}
\EE\Big[\prod_{p\in\cP_N}X_p\Big]\le Y_N(N)^{\frac{N-1}{N}}Y_N\Big(\frac{N(N-2)}{N-2}
\Big)^{\frac{N-1}{N}\cdot\frac{N-3}{2}}Y_N\Big(\frac{N(N-2)}{N-2}
\Big)^{\frac{1}{N}\cdot\frac{N-1}{2}}.
\end{align*}
Here the right hand side equals $Y_N({N})^{(N-1)/2}$.
\end{proof}


\section{The formula for the minimum of the spectrum}\label{appinfspec}

\noindent
In the following theorem and its proof we summarize some conclusions and arguments
we learned from  \cite{AbdesselamHasler2012,LHB2011,LorincziMinlosSpohn2002} in a natural, 
minor generalization.

\begin{thm}\label{thminfspecapp}
Let $(X,\fA,\zeta)$ be a $\sigma$-finite measure space and let $K$ be a self-adjoint operator in
$\sK:=L^2(X,\zeta)$ which is semi-bounded from below such that $e^{-tK}$ is positivity improving, for
all $t>0$. Then
\begin{align}\label{infspecapp}
\inf\sigma(K)=-\lim_{t\to\infty}\frac{1}{t}\ln\SPn{\psi}{e^{-tK}\psi}, \quad 0\not=\psi\in\sK,\,
\psi\ge0.
\end{align}
If we only assume that $e^{-tK}$ is positivity preserving, for all $t>0$, then the identity in 
\eqref{infspecapp} is still valid, for all $\psi\in\sK$ such that
$e^{-\tau K}\psi$ is strictly positive for some $\tau\ge0$.
\end{thm}

\begin{proof}
{\em Step~1.} First, we verify the following claim: Let $p\in[1,\infty)$ and $\psi:X\to\RR$ be 
measurable with $\psi>0$ on $X$. 
Pick any sequence of measurable sets $A_n\uparrow X$ with $\zeta(A_n)<\infty$, $n\in\NN$.
Then the linear span of all bounded measurable $\phi:X\to\RR$ with $0\le\phi\le\psi$ 
and $\phi=1_{A_n}\phi$, for some $n\in\NN$, is dense in $L^p(\cM,\zeta)$.

In fact, the measure space $(X,\fA,\zeta_\psi)$ with $\zeta_\psi:=\psi^p\zeta$ is again $\sigma$-finite
and $\fH:=\{B\cap A_n\cap\{\psi<m\}:B\in\fA,\,m,n\in\NN\}$ is a semi-ring generating $\fA$ and 
consisting only of sets with finite $\zeta_\psi$-measure. Therefore, $\{1_B:B\in\fH\}$ is total in 
$L^p(X,\zeta_\psi)$; see, e.g., \cite[Satz~VI.2.28b)]{Elstrodt1996}. 
Since multiplication by $\psi$ is an isometric isomorphism from
$L^p(X,\zeta_\psi)$ to $L^p(X,\zeta)$, the set $\{\psi1_B:B\in\fH\}$ is total in $L^p(X,\zeta)$.


\smallskip

\noindent
{\em Step~2.} Let $\sX\subset\sK$ such that $\mathrm{span}_{\CC}\sX$ is dense in $\sK$.
Let $E_{\phi}$ be the spectral measure corresponding to $K$ and $\phi\in\sK$. 
We claim that $\inf\sigma(K)=m$, where $m:=\inf\{\inf\supp(E_\phi):\phi\in\sX\}$.

We only need to show the corresponding inequality $\ge$, as the converse inequality is immediate
from the spectral calculus. Define $K':=C\wedge K\in\LO(\sK)$ by the spectral calculus, 
where $C>m$ is some constant. Then $\inf\sigma(K)=\inf\sigma(K')=:\sigma$. 
Now suppose for contradiction that $m>\sigma$ and let $\ve\in(0,m-\sigma)$. 
Then we find some normalized $\psi\in\sK$ such that $\SPn{\psi}{K'\psi}<\sigma+\ve/2$, and
by assumption and boundedness of $K'$ we may pick some normalized 
$\phi\in\mathrm{span}_{\CC}\sX$ with $\SPn{\phi}{K'\phi}<\sigma+\ve$. We have
$\phi=\sum_{i=1}^m\alpha_i\phi_i$, for some $\alpha_i\in\CC$, $\phi_i\in\sX$, and $n\in\NN$.
Pick some $\tau\in(\sigma+\ve,m)$. Then
$E_\tau\phi_i=\phi_i$, $i\in\{1,\ldots,n\}$, where $\{E_s\}_{s\in\RR}$ is the spectral family associated
with $K$. Hence, $\phi=E_\tau\phi$, and we obtain the contradiction
$\sigma+\ve>\SPn{\phi}{K'\phi}=\SPn{\phi}{E_\tau K'E_\tau\phi}\ge\tau$.

\smallskip

\noindent
{\em Step~3.} Let $\psi\in\sK$ be non-negative and $\tau\ge0$.
In what follows we assume that $\psi_\tau:=e^{-\tau K}\psi$ is strictly positive, thus
$\psi\not=0$, and we shall only use that $e^{-tK}$ is positivity preserving, for all $t>0$.

Let $\sX$ be the set of all $\phi\in\sK$ such that $0\le\phi\le\psi_\tau$.
Since the linear span of $\cX$ is dense by Step~1, 
$\inf\sigma(K)=\inf\{\inf\supp(E_\phi):\phi\in\sX\}$ by Step~2. 
Using that the semi-group of $K$ preserves positivity, we further observe that
\begin{align*}
-\frac{1}{t}\ln\SPn{\psi}{e^{-(t+2\tau)K}\psi}
=-\frac{1}{t}\ln\SPn{\psi_\tau}{e^{-tK}\psi_\tau}
&\le-\frac{1}{t}\ln\SPn{\phi}{e^{-tK}\phi},\quad\phi\in\sX,\,t>0.
\end{align*}
A well-known relation valid for any Borel measure on $\RR$ implies the first identity in
\begin{align*}
\inf\supp(E_\phi)=-\lim_{t\to\infty}\frac{1}{t}\ln\bigg(\int_{\RR}e^{-ts}\Id E_\phi(s)\bigg)=
-\lim_{t\to\infty}\frac{1}{t}\ln\SPn{\phi}{e^{-tK}\phi},\quad\phi\in\sK.
\end{align*}
We conclude that $\inf\sigma(K)\le\inf\supp(E_\psi)\le\inf\supp(E_\phi)$, $\phi\in\sX$.
\end{proof}


\bigskip
\noindent{{\bf Acknowledgement.}} We thank Gonzalo Bley, Marcel Griesemer, and 
Fumio Hiroshima for interesting and helpful discussions. Furthermore, we thank Marcel Griesemer 
and Stuttgart University, Masao Hirokawa and Hiroshima University, and 
Fumio Hiroshima and Kyushu University for their great hospitality. 
Finally, we thank the VILLUM foundation for support via the project grant 
``Spectral Analysis of Large Particle Systems'', and the Danish Agency for Science, Technology and 
Innovation for their support via the DFF Research Project grant ``The Mathematics of Dressed 
Particles'' and the International Network Programme grant ``Exciting Polarons''.


\end{document}